\documentclass[A4,10pt]{article}
\usepackage[margin=1in,dvips]{geometry}

\pdfoutput=1 
\usepackage[dvipdfmx]{graphicx}

\usepackage{color, soul}
\usepackage{amsmath, amsthm, slashed}
\usepackage{amssymb}

\usepackage{marvosym}
\usepackage{rotating}
\usepackage{mathrsfs}
\usepackage{upgreek}
\usepackage{verbatim}
\usepackage[affil-it]{authblk}
\usepackage{rotating}
\usepackage{accents}
\usepackage{harmony, slashed}
\usepackage{tikz}
\usepackage{dsfont}
\usepackage{makeidx}
\usepackage{bbm}

\usepackage{tocloft}

\usepackage{thmtools}
\usepackage{thm-restate}

\usepackage{comment}

\usepackage{enumerate}

\usepackage[colorlinks=true,linkcolor=black,citecolor=blue]{hyperref}

\usepackage{cleveref}

\newtheorem{theorem}{Theorem}[section]

\newtheorem*{conjecture*}{Conjecture}
\newtheorem*{theorem*}{Theorem}

\newtheorem{proposition}{Proposition}[subsection]

\newtheorem{corollary}[proposition]{Corollary}
\newtheorem*{corollary*}{Corollary}
\newtheorem{lemma}[proposition]{Lemma}
\newtheorem{remark}[proposition]{Remark}

\newtheorem{assumption}[proposition]{Assumption}

\newtheorem{bigremark}{Remark}[section]

\numberwithin{equation}{section}

\newcommand{\Dk}{\mathfrak{D}^{\mathbf{k}}}
\newcommand{\Dkone}{\mathfrak{D}^{\mathbf{k}_1}}
\newcommand{\Dktwo}{\mathfrak{D}^{\mathbf{k}_2}}

\newcommand{\Fstar}{\accentset{\scalebox{.6}{\mbox{\tiny $(0)$}}}{\underaccent{\scalebox{.8}{\mbox{\tiny $k$}}}{\mathcal{F}}}{}^*}
\newcommand{\Fstarll}{\accentset{\scalebox{.6}{\mbox{\tiny $(0)$}}}{\underaccent{\scalebox{.8}{\mbox{\tiny $\ll\mkern-6mu k$}}}{\mathcal{F}}}{}^*}

\newcommand{\Epstar}{\accentset{\scalebox{.6}{\mbox{\tiny $(p)$}}}{\underaccent{\scalebox{.8}{\mbox{\tiny $k$}}}{\mathcal{E}}}{}^*}

\newcommand{\Ezerostar}{\accentset{\scalebox{.6}{\mbox{\tiny $(0)$}}}{\underaccent{\scalebox{.8}{\mbox{\tiny $k$}}}{\mathcal{E}}}{}^*}
\newcommand{\Ezerostarll}{\accentset{\scalebox{.6}{\mbox{\tiny $(0)$}}}{\underaccent{\scalebox{.8}{\mbox{\tiny $\ll \mkern-6mu k$}}}{\mathcal{E}}}{}^*}

\newcommand{\intEdelta}{\int^*\accentset{\scalebox{.6}{\mbox{\tiny $(-1-\delta)$}}}{\underaccent{\scalebox{.8}{\mbox{\tiny $k$}}}{\mathcal{E}}}}
\newcommand{\intEdeltall}{\int^*\accentset{\scalebox{.6}{\mbox{\tiny $(-1-\delta)$}}}{\underaccent{\scalebox{.8}{\mbox{\tiny $\ll \mkern-6mu k$}}}{\mathcal{E}}}}
\newcommand{\intEpone}{\int^*\accentset{\scalebox{.6}{\mbox{\tiny $(p-1)$}}}{\underaccent{\scalebox{.8}{\mbox{\tiny $k$}}}{\mathcal{E}}}}

\newcommand{\Jpk}{\accentset{\scalebox{.6}{\mbox{\tiny $(p)$}}}{\underaccent{\scalebox{.8}{\mbox{\tiny $k$}}}{J}}}

\newcommand{\Jp}{\accentset{\scalebox{.6}{\mbox{\tiny $(p)$}}}{J}}

\newcommand{\Jzero}{\accentset{\scalebox{.6}{\mbox{\tiny $(0)$}}}{J}}

\newcommand{\Kzero}{\accentset{\scalebox{.6}{\mbox{\tiny $(0)$}}}{K}}

\newcommand{\Kpk}{\accentset{\scalebox{.6}{\mbox{\tiny $(p)$}}}{\underaccent{\scalebox{.8}{\mbox{\tiny $k$}}}{K}}}

\newcommand{\Kp}{\accentset{\scalebox{.6}{\mbox{\tiny $(p)$}}}{K}}

\newcommand{\Hap}{\accentset{\scalebox{.6}{\mbox{\tiny $(p)$}}}{H}}

\newcommand{\Hazero}{\accentset{\scalebox{.6}{\mbox{\tiny $(0)$}}}{H}}

\newcommand{\Hapk}{\accentset{\scalebox{.6}{\mbox{\tiny $(p)$}}}{\underaccent{\scalebox{.8}{\mbox{\tiny $k$}}}{H}}}

\newcommand{\Fzero}{\accentset{\scalebox{.6}{\mbox{\tiny $(0)$}}}{\mathcal{F}}}

\newcommand{\Fp}{\accentset{\scalebox{.6}{\mbox{\tiny $(p)$}}}{\mathcal{F}}}

\newcommand{\Fpprime}{\accentset{\scalebox{.6}{\mbox{\tiny $(p')$}}}{\mathcal{F}}}

\newcommand{\Fzerok}{\accentset{\scalebox{.6}{\mbox{\tiny $(0)$}}}{\underaccent{\scalebox{.8}{\mbox{\tiny $k$}}}{\mathcal{F}}}}

\newcommand{\Fpk}{\accentset{\scalebox{.6}{\mbox{\tiny $(p)$}}}{\underaccent{\scalebox{.8}{\mbox{\tiny $k$}}}{\mathcal{F}}}}

\newcommand{\Fpprimek}{\accentset{\scalebox{.6}{\mbox{\tiny $(p')$}}}{\underaccent{\scalebox{.8}{\mbox{\tiny $k$}}}{\mathcal{F}}}}

\newcommand{\Fpkminusone}{\accentset{\scalebox{.6}{\mbox{\tiny $(p)$}}}{\underaccent{\scalebox{.8}{\mbox{\tiny $k\mkern-6mu -\mkern-6mu 1$}}}{\mathcal{F}}}}

\newcommand{\Fonekminustwo}{\accentset{\scalebox{.6}{\mbox{\tiny $(1)$}}}{\underaccent{\scalebox{.8}{\mbox{\tiny $k\mkern-6mu -\mkern-6mu 2$}}}{\mathcal{F}}}}

\newcommand{\Fzerokminusthree}{\accentset{\scalebox{.6}{\mbox{\tiny $(0)$}}}{\underaccent{\scalebox{.8}{\mbox{\tiny $k\mkern-6mu -\mkern-6mu 3$}}}{\mathcal{F}}}}

\newcommand{\Fzerokminusfour}{\accentset{\scalebox{.6}{\mbox{\tiny $(0)$}}}{\underaccent{\scalebox{.8}{\mbox{\tiny $k\mkern-6mu -\mkern-6mu 4$}}}{\mathcal{F}}}}

\newcommand{\Fzerokminussix}{\accentset{\scalebox{.6}{\mbox{\tiny $(0)$}}}{\underaccent{\scalebox{.8}{\mbox{\tiny $k\mkern-6mu -\mkern-6mu 6$}}}{\mathcal{F}}}}

\newcommand{\Ffancyzero}{\accentset{\scalebox{.6}{\mbox{\tiny $(0)$}}}{\mathfrak{F}}}

\newcommand{\Ffancyp}{\accentset{\scalebox{.6}{\mbox{\tiny $(p)$}}}{\mathfrak{F}}}

\newcommand{\Ffancyzerok}{\accentset{\scalebox{.6}{\mbox{\tiny $(0)$}}}{\underaccent{\scalebox{.8}{\mbox{\tiny $k$}}}{\mathfrak{F}}}}

\newcommand{\Ffancyonek}{\accentset{\scalebox{.6}{\mbox{\tiny $(1)$}}}{\underaccent{\scalebox{.8}{\mbox{\tiny $k$}}}{\mathfrak{F}}}}

\newcommand{\Ffancypk}{\accentset{\scalebox{.6}{\mbox{\tiny $(p)$}}}{\underaccent{\scalebox{.8}{\mbox{\tiny $k$}}}{\mathfrak{F}}}}

\newcommand{\Ep}{\accentset{\scalebox{.6}{\mbox{\tiny $(p)$}}}{\mathcal{E}}}

\newcommand{\Epprime}{\accentset{\scalebox{.6}{\mbox{\tiny $(p')$}}}{\mathcal{E}}}

\newcommand{\Epk}{\accentset{\scalebox{.6}{\mbox{\tiny $(p)$}}}{\underaccent{\scalebox{.8}{\mbox{\tiny $k$}}}{\mathcal{E}}}}

\newcommand{\Xpk}{\accentset{\scalebox{.6}{\mbox{\tiny $(p)$}}}{\underaccent{\scalebox{.8}{\mbox{\tiny $k$}}}{\mathcal{X}}}}

\newcommand{\Epprimekprime}{\accentset{\scalebox{.6}{\mbox{\tiny $(p')$}}}{\underaccent{\scalebox{.8}{\mbox{\tiny {$k'$}}}}{\mathcal{E}}}}

\newcommand{\Epprimek}{\accentset{\scalebox{.6}{\mbox{\tiny $(p')$}}}{\underaccent{\scalebox{.8}{\mbox{\tiny {$k$}}}}{\mathcal{E}}}}

\newcommand{\Xpprimekprime}{\accentset{\scalebox{.6}{\mbox{\tiny $(p')$}}}{\underaccent{\scalebox{.8}{\mbox{\tiny {$k'$}}}}{\mathcal{X}}}}

\newcommand{\Efancypk}{\accentset{\scalebox{.6}{\mbox{\tiny $(p)$}}}{\underaccent{\scalebox{.8}{\mbox{\tiny $k$}}}{\mathfrak{E}}}}

\newcommand{\Etwominusdelk}{\accentset{\scalebox{.6}{\mbox{\tiny $(2-\delta)$}}}{\underaccent{\scalebox{.8}{\mbox{\tiny $k$}}}{\mathcal{E}}}}

\newcommand{\Xtwominusdelk}{\accentset{\scalebox{.6}{\mbox{\tiny $(2-\delta)$}}}{\underaccent{\scalebox{.8}{\mbox{\tiny $k$}}}{\mathcal{X}}}}

\newcommand{\Xtwominusdelkminusone}{\accentset{\scalebox{.6}{\mbox{\tiny $(2-\delta)$}}}{\underaccent{\scalebox{.8}{\mbox{\tiny $k\mkern-6mu -\mkern-6mu 1$}}}{\mathcal{X}}}}

\newcommand{\Xtwominusdelkminustwo}{\accentset{\scalebox{.6}{\mbox{\tiny $(2-\delta)$}}}{\underaccent{\scalebox{.8}{\mbox{\tiny $k\mkern-6mu -\mkern-6mu 2$}}}{\mathcal{X}}}}

\newcommand{\Xtwominusdelkminusthree}{\accentset{\scalebox{.6}{\mbox{\tiny $(2-\delta)$}}}{\underaccent{\scalebox{.8}{\mbox{\tiny $k\mkern-6mu -\mkern-6mu 3$}}}{\mathcal{X}}}}

\newcommand{\Edeltaminusonek}{\accentset{\scalebox{.6}{\mbox{\tiny $(\delta-1)$}}}{\underaccent{\scalebox{.8}{\mbox{\tiny $k$}}}{\mathcal{E}}}}

\newcommand{\Efancytwominusdelk}{\accentset{\scalebox{.6}{\mbox{\tiny $(2-\delta)$}}}{\underaccent{\scalebox{.8}{\mbox{\tiny $k$}}}{\mathfrak{E}}}}

\newcommand{\Efancytwominusdelkminusone}{\accentset{\scalebox{.6}{\mbox{\tiny $(2-\delta)$}}}{\underaccent{\scalebox{.8}{\mbox{\tiny $k\mkern-6mu -\mkern-6mu 1$}}}{\mathfrak{E}}}}

\newcommand{\Efancytwominusdelkminustwo}{\accentset{\scalebox{.6}{\mbox{\tiny $(2-\delta)$}}}{\underaccent{\scalebox{.8}{\mbox{\tiny $k\mkern-6mu -\mkern-6mu 2$}}}{\mathfrak{E}}}}

\newcommand{\Ffancytwominusdelk}{\accentset{\scalebox{.6}{\mbox{\tiny $(2-\delta)$}}}{\underaccent{\scalebox{.8}{\mbox{\tiny $k$}}}{\mathfrak{F}}}}

\newcommand{\Efancyonekminustwo}{\accentset{\scalebox{.6}{\mbox{\tiny $(1)$}}}{\underaccent{\scalebox{.8}{\mbox{\tiny $k\mkern-6mu -\mkern-6mu 2$}}}{\mathfrak{E}}}}

\newcommand{\Efancyzerokminusthree}{\accentset{\scalebox{.6}{\mbox{\tiny $(0)$}}}{\underaccent{\scalebox{.8}{\mbox{\tiny $k\mkern-6mu -\mkern-6mu 3$}}}{\mathfrak{E}}}}

\newcommand{\Efancyonekminusfour}{\accentset{\scalebox{.6}{\mbox{\tiny $(1)$}}}{\underaccent{\scalebox{.8}{\mbox{\tiny $k\mkern-6mu -\mkern-6mu 4$}}}{\mathfrak{E}}}}

\newcommand{\Efancyzerokminussix}{\accentset{\scalebox{.6}{\mbox{\tiny $(0)$}}}{\underaccent{\scalebox{.8}{\mbox{\tiny $k\mkern-6mu -\mkern-6mu 6$}}}{\mathfrak{E}}}}

\newcommand{\Eerrorminusone}{\accentset{\scalebox{.8}{\mbox{\tiny $\xi$}}}{\underaccent{\scalebox{.8}{\mbox{\tiny $ -\mkern-3mu 1$}}}{\mathcal{E}}}}

\newcommand{\Eerrorkminusone}{\accentset{\scalebox{.8}{\mbox{\tiny $\xi$}}}{\underaccent{\scalebox{.8}{\mbox{\tiny $k\mkern-6mu -\mkern-6mu 1$}}}{\mathcal{E}}}}

\newcommand{\Eminusone}{\underaccent{\scalebox{.8}{\mbox{\tiny $ -\mkern-3mu 1$}}}{\mathcal{E}}}

\newcommand{\Epminusone}{\accentset{\scalebox{.6}{\mbox{\tiny $(p\mkern-6mu-\mkern-6mu1)$}}}{\mathcal{E}}}
\newcommand{\Epminusonek}{\accentset{\scalebox{.6}{\mbox{\tiny $(p\mkern-6mu-\mkern-6mu1)$}}}{\underaccent{\scalebox{.8}{\mbox{\tiny $k$}}}{\mathcal{E}}}}

\newcommand{\Xtwominusdeltalesslessk}{\accentset{\scalebox{.6}{\mbox{\tiny $(2\mkern-6mu-\mkern-6mu\delta)$}}}{\underaccent{\scalebox{.8}{\mbox{\tiny $\ll \mkern-6mu k$}}}{\mathcal{X}}}}

\newcommand{\Eoneminusdelk}{\accentset{\scalebox{.6}{\mbox{\tiny $(1-\delta)$}}}{\underaccent{\scalebox{.8}{\mbox{\tiny $k$}}}{\mathcal{E}}}}

\newcommand{\Eoneminusdelkminustwo}{\accentset{\scalebox{.6}{\mbox{\tiny $(1-\delta)$}}}{\underaccent{\scalebox{.8}{\mbox{\tiny $k\mkern-6mu -\mkern-6mu 2$}}}{\mathcal{E}}}}

\newcommand{\Eoneminusdelkminusfour}{\accentset{\scalebox{.6}{\mbox{\tiny $(1-\delta)$}}}{\underaccent{\scalebox{.8}{\mbox{\tiny $k\mkern-6mu -\mkern-6mu 4$}}}{\mathcal{E}}}}

\newcommand{\Efancyonek}{\accentset{\scalebox{.6}{\mbox{\tiny $(1)$}}}{\underaccent{\scalebox{.8}{\mbox{\tiny $k$}}}{\mathfrak{E}}}}

\newcommand{\Efancyonekminusone}{\accentset{\scalebox{.6}{\mbox{\tiny $(1)$}}}{\underaccent{\scalebox{.8}{\mbox{\tiny $k\mkern-6mu -\mkern-6mu 1$}}}{\mathfrak{E}}}}

\newcommand{\Efancyzerokminustwo}{\accentset{\scalebox{.6}{\mbox{\tiny $(0)$}}}{\underaccent{\scalebox{.8}{\mbox{\tiny $k\mkern-6mu -\mkern-6mu 2$}}}{\mathfrak{E}}}}

\newcommand{\Eonek}{\accentset{\scalebox{.6}{\mbox{\tiny $(1)$}}}{\underaccent{\scalebox{.8}{\mbox{\tiny $k$}}}{\mathcal{E}}}}

\newcommand{\Xonek}{\accentset{\scalebox{.6}{\mbox{\tiny $(1)$}}}{\underaccent{\scalebox{.8}{\mbox{\tiny $k$}}}{\mathcal{X}}}}

\newcommand{\Epminusonekminusone}{\accentset{\scalebox{.6}{\mbox{\tiny $(p\mkern-6mu-\mkern-6mu1)$}}}{\underaccent{\scalebox{.8}{\mbox{\tiny $k\mkern-6mu -\mkern-6mu 1$}}}{\mathcal{E}}}}

\newcommand{\Xpminusonekminustwo}{\accentset{\scalebox{.6}{\mbox{\tiny $(p\mkern-6mu-\mkern-6mu1)$}}}{\underaccent{\scalebox{.8}{\mbox{\tiny $k\mkern-6mu -\mkern-6mu 2$}}}{\mathcal{X}}}}

\newcommand{\Eoneminusdelkminusone}{\accentset{\scalebox{.6}{\mbox{\tiny $(1-\delta)$}}}{\underaccent{\scalebox{.8}{\mbox{\tiny $k\mkern-6mu -\mkern-6mu 1$}}}{\mathcal{E}}}}

\newcommand{\Eonekminusone}{\accentset{\scalebox{.6}{\mbox{\tiny $(1)$}}}{\underaccent{\scalebox{.8}{\mbox{\tiny $k\mkern-6mu -\mkern-6mu 1$}}}{\mathcal{E}}}}

\newcommand{\Edeltaminusonekminusone}{\accentset{\scalebox{.6}{\mbox{\tiny $(\delta-1)$}}}{\underaccent{\scalebox{.8}{\mbox{\tiny $k\mkern-6mu -\mkern-6mu 1$}}}{\mathcal{E}}}}

\newcommand{\Eonekminusfour}{\accentset{\scalebox{.6}{\mbox{\tiny $(1)$}}}{\underaccent{\scalebox{.8}{\mbox{\tiny $k\mkern-6mu -\mkern-6mu 4$}}}{\mathcal{E}}}}

\newcommand{\Xonekminusone}{\accentset{\scalebox{.6}{\mbox{\tiny $(1)$}}}{\underaccent{\scalebox{.8}{\mbox{\tiny $k\mkern-6mu -\mkern-6mu 1$}}}{\mathcal{X}}}}

\newcommand{\Xonekminusfour}{\accentset{\scalebox{.6}{\mbox{\tiny $(1)$}}}{\underaccent{\scalebox{.8}{\mbox{\tiny $k\mkern-6mu -\mkern-6mu 4$}}}{\mathcal{X}}}}

\newcommand{\Xonekminusfive}{\accentset{\scalebox{.6}{\mbox{\tiny $(1)$}}}{\underaccent{\scalebox{.8}{\mbox{\tiny $k\mkern-6mu -\mkern-6mu 5$}}}{\mathcal{X}}}}

\newcommand{\Epminusonekminustwo}{\accentset{\scalebox{.6}{\mbox{\tiny $(p\mkern-6mu-\mkern-6mu1)$}}}{\underaccent{\scalebox{.8}{\mbox{\tiny $k\mkern-6mu -\mkern-6mu 2$}}}{\mathcal{E}}}}

\newcommand{\Eonekminustwo}{\accentset{\scalebox{.6}{\mbox{\tiny $(1)$}}}{\underaccent{\scalebox{.8}{\mbox{\tiny $k\mkern-6mu -\mkern-6mu 2$}}}{\mathcal{E}}}}

\newcommand{\Xonekminustwo}{\accentset{\scalebox{.6}{\mbox{\tiny $(1)$}}}{\underaccent{\scalebox{.8}{\mbox{\tiny $k\mkern-6mu -\mkern-6mu 2$}}}{\mathcal{X}}}}

\newcommand{\Epminusonekminusthree}{\accentset{\scalebox{.6}{\mbox{\tiny $(p\mkern-6mu-\mkern-6mu1)$}}}{\underaccent{\scalebox{.8}{\mbox{\tiny $k\mkern-6mu -\mkern-6mu 3$}}}{\mathcal{E}}}}

\newcommand{\Epminusoneminusone}{\accentset{\scalebox{.6}{\mbox{\tiny $(p\mkern-6mu-\mkern-6mu1)$}}}{\underaccent{\scalebox{.8}{\mbox{\tiny $-1$}}}{\mathcal{E}}}}

\newcommand{\Epkminusone}{\accentset{\scalebox{.6}{\mbox{\tiny $(p)$}}}{\underaccent{\scalebox{.8}{\mbox{\tiny $k\mkern-6mu -\mkern-6mu 1$}}}{\mathcal{E}}}}

\newcommand{\Xpkminusone}{\accentset{\scalebox{.6}{\mbox{\tiny $(p)$}}}{\underaccent{\scalebox{.8}{\mbox{\tiny $k\mkern-6mu -\mkern-6mu 1$}}}{\mathcal{X}}}}

\newcommand{\Etwominusdelkminusone}{\accentset{\scalebox{.6}{\mbox{\tiny $(2-\delta)$}}}{\underaccent{\scalebox{.8}{\mbox{\tiny $k\mkern-6mu -\mkern-6mu 1$}}}{\mathcal{E}}}}
\newcommand{\Etwominusdelkminustwo}{\accentset{\scalebox{.6}{\mbox{\tiny $(2-\delta)$}}}{\underaccent{\scalebox{.8}{\mbox{\tiny $k\mkern-6mu -\mkern-6mu 2$}}}{\mathcal{E}}}}

\newcommand{\Etwominusdelkminusfour}{\accentset{\scalebox{.6}{\mbox{\tiny $(2-\delta)$}}}{\underaccent{\scalebox{.8}{\mbox{\tiny $k\mkern-6mu -\mkern-6mu 4$}}}{\mathcal{E}}}}

\newcommand{\Ezero}{\accentset{\scalebox{.6}{\mbox{\tiny $(0)$}}}{\mathcal{E}}}

\newcommand{\Ezerok}{\accentset{\scalebox{.6}{\mbox{\tiny $(0)$}}}{\underaccent{\scalebox{.8}{\mbox{\tiny $k$}}}{\mathcal{E}}}}

\newcommand{\Xzerok}{\accentset{\scalebox{.6}{\mbox{\tiny $(0)$}}}{\underaccent{\scalebox{.8}{\mbox{\tiny $k$}}}{\mathcal{X}}}}

\newcommand{\Xzeroplusk}{\accentset{\scalebox{.6}{\mbox{\tiny $(0+)$}}}{\underaccent{\scalebox{.8}{\mbox{\tiny $k$}}}{\mathcal{X}}}}

\newcommand{\Ezerokminusone}{\accentset{\scalebox{.6}{\mbox{\tiny $(0)$}}}
{\underaccent{\scalebox{.8}{\mbox{\tiny $k\mkern-6mu -\mkern-6mu 1$}}}{\mathcal{E}}}}
\newcommand{\Ezerokminustwo}{\accentset{\scalebox{.6}{\mbox{\tiny $(0)$}}}
{\underaccent{\scalebox{.8}{\mbox{\tiny $k\mkern-6mu -\mkern-6mu 2$}}}{\mathcal{E}}}}
\newcommand{\Ezerokminusthree}{\accentset{\scalebox{.6}{\mbox{\tiny $(0)$}}}
{\underaccent{\scalebox{.8}{\mbox{\tiny $k\mkern-6mu -\mkern-6mu 3$}}}{\mathcal{E}}}}
\newcommand{\Ezerokminusfour}{\accentset{\scalebox{.6}{\mbox{\tiny $(0)$}}}
{\underaccent{\scalebox{.8}{\mbox{\tiny $k\mkern-6mu -\mkern-6mu 4$}}}{\mathcal{E}}}}
\newcommand{\Ezerokminusfive}{\accentset{\scalebox{.6}{\mbox{\tiny $(0)$}}}
{\underaccent{\scalebox{.8}{\mbox{\tiny $k\mkern-6mu -\mkern-6mu 5$}}}{\mathcal{E}}}}
\newcommand{\Ezerokminussix}{\accentset{\scalebox{.6}{\mbox{\tiny $(0)$}}}
{\underaccent{\scalebox{.8}{\mbox{\tiny $k\mkern-6mu -\mkern-6mu 6$}}}{\mathcal{E}}}}

\newcommand{\Efancyzerokminusone}{\accentset{\scalebox{.6}{\mbox{\tiny $(0)$}}}
{\underaccent{\scalebox{.8}{\mbox{\tiny $k\mkern-6mu -\mkern-6mu 1$}}}{\mathfrak{E}}}}

\newcommand{\Xzerokminusone}{\accentset{\scalebox{.6}{\mbox{\tiny $(0)$}}}
{\underaccent{\scalebox{.8}{\mbox{\tiny $k\mkern-6mu -\mkern-6mu 1$}}}{\mathcal{X}}}}

\newcommand{\Xzeropluskminusone}{\accentset{\scalebox{.6}{\mbox{\tiny $(0+)$}}}
{\underaccent{\scalebox{.8}{\mbox{\tiny $k\mkern-6mu -\mkern-6mu 1$}}}{\mathcal{X}}}}

\newcommand{\Xzerokminustwo}{\accentset{\scalebox{.6}{\mbox{\tiny $(0)$}}}
{\underaccent{\scalebox{.8}{\mbox{\tiny $k\mkern-6mu -\mkern-6mu 2$}}}{\mathcal{X}}}}

\newcommand{\Xzerokminusthree}{\accentset{\scalebox{.6}{\mbox{\tiny $(0)$}}}
{\underaccent{\scalebox{.8}{\mbox{\tiny $k\mkern-6mu -\mkern-6mu 3$}}}{\mathcal{X}}}}

\newcommand{\Xzerokminussix}{\accentset{\scalebox{.6}{\mbox{\tiny $(0)$}}}
{\underaccent{\scalebox{.8}{\mbox{\tiny $k\mkern-6mu -\mkern-6mu 6$}}}{\mathcal{X}}}}

\newcommand{\Xzerokminusseven}{\accentset{\scalebox{.6}{\mbox{\tiny $(0)$}}}
{\underaccent{\scalebox{.8}{\mbox{\tiny $k\mkern-6mu -\mkern-6mu 7$}}}{\mathcal{X}}}}

\newcommand{\Xonekminusthree}{\accentset{\scalebox{.6}{\mbox{\tiny $(1)$}}}
{\underaccent{\scalebox{.8}{\mbox{\tiny $k\mkern-6mu -\mkern-6mu 3$}}}{\mathcal{X}}}}

\newcommand{\Efancyzero}{\accentset{\scalebox{.6}{\mbox{\tiny $(0)$}}}
{\mathfrak{E}}}

\newcommand{\Efancyp}{\accentset{\scalebox{.6}{\mbox{\tiny $(p)$}}}
{\mathfrak{E}}}

\newcommand{\Efancyzerok}{\accentset{\scalebox{.6}{\mbox{\tiny $(0)$}}}{\underaccent{\scalebox{.8}{\mbox{\tiny $k$}}}
{\mathfrak{E}}}}

\newcommand{\Ezerolesslessk}{\accentset{\scalebox{.6}{\mbox{\tiny $(0)$}}}{\underaccent{\scalebox{.8}{\mbox{\tiny 
$\ll\mkern-6mu k$}}}{\mathcal{E}}}}

\newcommand{\Xzerolesslessk}{\accentset{\scalebox{.6}{\mbox{\tiny $(0)$}}}{\underaccent{\scalebox{.8}{\mbox{\tiny 
$\ll\mkern-6mu k$}}}{\mathcal{X}}}}

\newcommand{\Xzeropluslesslessk}{\accentset{\scalebox{.6}{\mbox{\tiny $(0+)$}}}{\underaccent{\scalebox{.8}{\mbox{\tiny 
$\ll\mkern-6mu k$}}}{\mathcal{X}}}}

\newcommand{\Xonelesslessk}{\accentset{\scalebox{.6}{\mbox{\tiny $(1)$}}}{\underaccent{\scalebox{.8}{\mbox{\tiny 
$\ll\mkern-6mu k$}}}{\mathcal{X}}}}

\newcommand{\Xplesslessk}{\accentset{\scalebox{.6}{\mbox{\tiny $(p)$}}}{\underaccent{\scalebox{.8}{\mbox{\tiny 
$\ll\mkern-6mu k$}}}{\mathcal{X}}}}

\newcommand{\Ezerominusoneminusdelta}{\accentset{\scalebox{.6}{\mbox{\tiny $(-1\mkern-6mu-\mkern-6mu\delta)$}}}{\mathcal{E}}}

\newcommand{\Ezerominusoneminusdeltak}{\accentset{\scalebox{.6}{\mbox{\tiny $(-1\mkern-6mu-\mkern-6mu\delta)$}}}{\underaccent{\scalebox{.8}{\mbox{\tiny $k$}}}{\mathcal{E}}}}

\newcommand{\Ezerominusthreeminusdeltakminusone}{\accentset{\scalebox{.6}{\mbox{\tiny $(-3\mkern-6mu-\mkern-6mu\delta)$}}}{\underaccent{\scalebox{.8}{\mbox{\tiny $k\mkern-6mu -\mkern-6mu 1$}}}{\mathcal{E}}}}

\newcommand{\Ezerominusoneminusdeltakminusone}{\accentset{\scalebox{.6}{\mbox{\tiny $(-1\mkern-6mu-\mkern-6mu\delta)$}}}{\underaccent{\scalebox{.8}{\mbox{\tiny $k\mkern-6mu -\mkern-6mu 1$}}}{\mathcal{E}}}}

\newcommand{\Ezerominusoneminusdeltakminustwo}{\accentset{\scalebox{.6}{\mbox{\tiny $(-1\mkern-6mu-\mkern-6mu\delta)$}}}{\underaccent{\scalebox{.8}{\mbox{\tiny $k\mkern-6mu -\mkern-6mu 2$}}}{\mathcal{E}}}}

\newcommand{\Ezerominusoneminusdeltakminusthree}{\accentset{\scalebox{.6}{\mbox{\tiny $(-1\mkern-6mu-\mkern-6mu\delta)$}}}{\underaccent{\scalebox{.8}{\mbox{\tiny $k\mkern-6mu -\mkern-6mu 3$}}}{\mathcal{E}}}}

\newcommand{\Ezerominusoneminusdeltakminusfour}{\accentset{\scalebox{.6}{\mbox{\tiny $(-1\mkern-6mu-\mkern-6mu\delta)$}}}{\underaccent{\scalebox{.8}{\mbox{\tiny $k\mkern-6mu -\mkern-6mu 4$}}}{\mathcal{E}}}}

\newcommand{\Ezerominusoneminusdeltakminussix}{\accentset{\scalebox{.6}{\mbox{\tiny $(-1\mkern-6mu-\mkern-6mu\delta)$}}}{\underaccent{\scalebox{.8}{\mbox{\tiny $k\mkern-6mu -\mkern-6mu 6$}}}{\mathcal{E}}}}

\newcommand{\Ezerominusoneminusdeltakminusseven}{\accentset{\scalebox{.6}{\mbox{\tiny $(-1\mkern-6mu-\mkern-6mu\delta)$}}}{\underaccent{\scalebox{.8}{\mbox{\tiny $k\mkern-6mu -\mkern-6mu 7$}}}{\mathcal{E}}}}

\newcommand{\Ezerominusoneminusdeltakminuseight}{\accentset{\scalebox{.6}{\mbox{\tiny $(-1\mkern-6mu-\mkern-6mu\delta)$}}}{\underaccent{\scalebox{.8}{\mbox{\tiny $k\mkern-6mu -\mkern-6mu 8$}}}{\mathcal{E}}}}

\newcommand{\Ezerominusoneminusdeltalesslessk}{\accentset{\scalebox{.6}{\mbox{\tiny $(-1\mkern-6mu-\mkern-6mu\delta)$}}}{\underaccent{\scalebox{.8}{\mbox{\tiny $\ll\mkern-6mu k$}}}{\mathcal{E}}}}

\newcommand{\gslash}{\slashed{g}}

\newcommand{\nablaslash}{\slashed{\nabla}}

\DeclareMathAlphabet\mathbfcal{OMS}{cmsy}{b}{n}

\makeatletter \@addtoreset{equation}{section}  \makeatother

\title{Quasilinear wave equations on  asymptotically flat spacetimes  \\ with applications to Kerr  black holes}

\author[$\ddag$ $\S$]{Mihalis Dafermos}
\author[$\dag$ $*$]{Gustav Holzegel}
\author[$\ddag$]{Igor Rodnianski}
\author[$\dag$]{Martin Taylor}

\affil[$\dag$]{\small Imperial College London,
Department of Mathematics,
South~Kensington~Campus,~London~SW7~2AZ,~United~Kingdom\vskip.2pc \ }

\affil[$\ddag$]{\small Princeton University, Department of Mathematics, Fine~Hall,~Washington~Road,~Princeton,~NJ~08544,~United~States~of~America\vskip.2pc \ }

\affil[$\S$]{\small University of Cambridge, Department of Pure Mathematics and Mathematical
Statistics, Wilberforce~Road,~Cambridge~CB3~0WA,~United~Kingdom\vskip.2pc \ }

\affil[$*$]{\small Westf\"alische Wilhelms-Universit\"at M\"unster,
Mathematisches~Institut, Einsteinstrasse~62~48149~M\"unster,~Bundesrepublik~Deutschland}

\date{26 September 2024} 

\begin{document}

\maketitle

\begin{abstract}
We prove global existence and decay for small-data solutions to
a class of quasilinear wave equations on a wide variety of asymptotically flat
spacetime backgrounds, allowing in particular for the presence of horizons, ergoregions and
 trapped null geodesics, and including as a special case the 
Schwarzschild and very slowly rotating $|a|\ll M$ Kerr family of black holes
 in general relativity.  There are two distinguishing aspects of our approach.
The first aspect is its dyadically localised nature: The nontrivial part of the analysis is  reduced entirely
  to time-translation invariant $r^p$-weighted estimates, in
  the spirit of~\cite{DafRodnew}, to be applied on dyadic time-slabs which for large $r$ are outgoing.
 Global existence and decay then both immediately follow by elementary iteration on consecutive such time-slabs,
  without  further global bootstrap.
The second, and more fundamental, aspect is our direct use of a ``black box'' linear inhomogeneous energy estimate on exactly stationary metrics,
together with a  novel  but elementary physical space top order identity that need not capture the structure of trapping and is robust
to perturbation.
In the specific example of Kerr black holes, the required linear inhomogeneous estimate can then be quoted directly from the literature~\cite{partiii},
while the additional top order physical space identity can be shown easily in many cases
 (we include in the Appendix a proof for the Kerr case $|a|\ll M$, which can in fact be understood in this context simply as a perturbation
 of Schwarzschild).
In particular, the approach  circumvents  the need either for  producing
a purely  physical space identity capturing  trapping  or for a careful analysis of the commutation
properties of frequency projections with the wave operator of time-dependent metrics.
\end{abstract}

  \tableofcontents

\section{Introduction}
We consider here  quasilinear equations of the form
\begin{equation}
\label{theequationzero}
\Box_{g(\psi,x)}\psi = N(\partial\psi, \psi, x)
\end{equation}
on a $4$-dimensional manifold $\mathcal{M}$, where $g(\psi,x)$ and $N(\partial\psi,\psi,x)$ are
appropriate nonlinear terms and 
where $g(0,x)=g_0(x)$ defines a stationary asymptotically flat Lorentzian metric on $\mathcal{M}$
foliated by a suitable family of hypersurfaces $\Sigma(\tau)$, for $\tau\ge0$, which for large $r$ are outgoing null.
In this paper, we prove the  following:
\begin{theorem*}
For equations~\eqref{theequationzero}, under appropriate assumptions on 
the background spacetime $(\mathcal{M}, g_0)$ and on the nonlinearities 
$g(\psi,x)$ and $N(\partial\psi,\psi,x)$ to be described below,
we have 
\begin{itemize}
\item
{\bf Global existence of small data solutions.} Solutions $\psi$ arising from smooth  initial data
on $\Sigma_0$, sufficiently small 
as measured in a suitable weighted Sobolev energy, exist globally on $\mathcal{M}$.
\item
{\bf Orbital stability.} Under the above assumptions,
the above weighted energy flux through $\Sigma(\tau)$
is uniformly bounded by a constant
times its initial value.
\item
{\bf Asymptotic stability.} Under the above assumptions, a suitable  lower order unweighted
energy flux through $\Sigma(\tau)$ decays inverse polynomially in $\tau$
(implying also pointwise inverse polynomial decay for~$\psi$).
\end{itemize}
\end{theorem*}

One possible motivation for studying~\eqref{theequationzero} 
is as an illustrative model problem for issues relating to the
 nonlinear stability  of the spacetimes $(\mathcal{M},g_0)$, when these themselves
are solutions to
the celebrated
 Einstein equations of general relativity. 
 An important example of  such spacetimes $(\mathcal{M},g_0)$ allowed by our theorem
 is provided by the Schwarzschild family of metrics for $M>0$,  and the more general
  Kerr family (parametrised by $a$ and $M$) in the very slowly rotating regime $|a|\ll M$.
  For a discussion of these spacetimes and the stability problem in general relativity,
see~\cite{dhrtplus}. 

 Among other features, the Schwarzschild and Kerr spacetimes
 exhibit \emph{trapped null geodesics}, along which energy
 can concentrate over large time-scales. Moreover, these are \emph{asymptotically flat}
 spacetimes, meaning that linear waves $\psi$ decay like $\sim r^{-1}$ along outgoing null cones.
This then constrains the decay of $\psi$ in the near region to  also be at  best  only inverse polynomial. (In 
the above black hole examples, this slow decay in the near region is
 already a linear effect due to scattering off of far away curvature
 associated to the nontrivial  spacetime mass at infinity; in Minkowski space, this slow decay in the near region
  arises due to purely nonlinear scattering effects, even for compactly supported initial data.)
 It is the combination of these two specific analytical difficulties---trapped null geodesics
 and the relatively slow decay in the near region necessitated by asymptotic flatness---of the black hole stability
 problem (and related problems) which we
 wish to capture with our assumptions in the present paper.
It is for this reason that we include both the nonlinear term
 implicit in the expression $\Box_{g(\psi,x)}\psi$ on the left hand side of~\eqref{theequationzero}  (the ``quasilinearity''),
 as this is the most dangerous term in the vicinity of
 trapping, as well as the nonlinear term $N(\partial\psi,\psi,x)$ on the right hand side (the ``semilinearity''),
 as this term models the true null structure at infinity of the Einstein equations, when the latter
 are written in appropriate geometric gauges. 
 To avoid inessential complications, we will in fact assume 
 that $g(\psi,x)=g_0$ for
 large $r$ and that $N(\partial\psi,\psi,x)$ satisfies a generalised version of the null condition~\cite{KlNull}.

In the case where $(\mathcal{M},g_0)$ is Minkowski space, a version of the above theorem follows from
classical work of~\cite{KlNull,christonull}, and there have been many further amplifications over the years,
particularly in the context of the  obstacle problem, see e.g.~\cite{metsog}.
Concerning specifically the Schwarzschild and  very slowly rotating $|a|\ll M$ Kerr setting, 
versions of the above theorem have been shown previously  in various special cases; see 
already~\cite{MR3082240} in the semilinear and~\cite{lindblad2016global, lindtoh} in the quasilinear
case, as well as the recent~\cite{lindblad2022weak} (the latter three works concerning slightly different classes of equations, but satisfying the
weak null condition~\cite{LindRodweak}). For axisymmetric solutions to certain semilinear equations 
on Kerr in the full subextremal range $|a|<M$, see~\cite{stogin2017nonlinear}.
See also~\cite{Pasqualotto:2017rkh} for a related physically
motivated quasilinear problem and~\cite{hintzvasyglobal, mavrogiannis} for the study
of~\eqref{theequationzero} on 
the non-asymptotically flat Schwarzschild--de Sitter and Kerr--de Sitter black hole
backgrounds, where the cosmological constant $\Lambda$ is positive and decay of $\psi$ is in fact exponential.
For  work specifically connected to the related problem of stability of these spacetimes
themselves under the Einstein equations, see already
Section~\ref{discussionofeinstein}.

In comparison to previous related work, there are two distinguishing features of the present approach.

The first distinguishing feature is our purely
dyadic framework: 
Rather than relying on a global bootstrap based on time-weighted norms,
the  argument is  entirely reduced to $r^p$-weighted but \emph{time-translation invariant}  
estimates to be applied in spacetime slabs
of time length $L$ which are outgoing null  for large $r$. 
Global existence (and decay) is then inferred by proceeding iteratively in time in
consecutive
slabs of length $L= 2^i$, in the spirit of~\cite{DafRodnew}. 
No bootstrap is necessary for the iteration itself, and to estimate the $i$'th slab, no information
is necessary to remember about the past other than information on the data of the slab itself.
In this sense, the argument is truly dyadically localised.
This yields a more streamlined proof even restricted to the semilinear case. (In fact, for both the semilinear and quasilinear cases,
if $(\mathcal{M},g_0)$ is
a suitably small perturbation of Minkowski space, 
the time-translation invariant estimate can be applied directly without iteration, 
and the approach proves global existence without explicit reference to time decay, cf.~the recent~\cite{faccimetcalfe}. 
In our formalism we shall always assume $g_0$ to be stationary, though we note that results can also be obtained on
non-stationary perturbations of Minkowski space~\cite{Yang2013}.)

The second, and more fundamental, feature  is our direct use of a ``black box'' linear inhomogeneous energy estimate  (which in
applications captures  both complicated trapping phenomena as well as low frequency obstructions) on the exactly stationary
background,
together with a physical-space top-order identity which may be applied directly to the quasilinear
equation~\eqref{theequationzero} and is in fact completely insensitive to trapping (and, when it holds,  robust
to perturbation of the metric $g_0$).
In the example of Kerr, in fact for the full subextremal case $|a|<M$, our black box assumption follows
directly from~\cite{partiii}, while 
we show explicitly how to retrieve the additional top-order physical space estimate in the case of $|a|\ll M$
(which in this context, can in fact be thought of simply as a perturbation of Schwarzschild)
in  Appendix~\ref{howto}.
The correct formulation of the companion estimate to be used in connection with~\cite{partiii} 
 for the full subextremal case  $|a|<M$ will appear elsewhere. 
(We emphasise that it is the phenomenon of \emph{superradiance} which is the primary difficulty in producing
this identity, not
the complicated structure of trapping per se; for instance,
 the necessary top order physical space identity can be shown for general 
 stationary spacetimes without horizons or ergospheres, irrespective of the structure of trapping.)
The approach thus does not depend on whether or not there exists
a purely physical-space based identity capturing  trapping (something very fragile!)~nor 
does it require a careful analysis of the commutation
properties of frequency projections with the wave operator $\Box_{g(\psi,x)}$ of
the time-dependent metric $g(\psi,x)$ corresponding to the actual solution $\psi$
(something quite technical, which was successfully done, however, for instance in~\cite{lindtoh}).

We emphasise that the role of the companion physical-space  estimate is purely in order to deal with quasilinear terms.
When our method is restricted to the semilinear case, i.e.~when $g(\psi,x)=g_0$ identically in~\eqref{theequationzero}, then
the black box linear inhomogeneous energy estimate can be applied
on its own. In particular, in view of~\cite{partiii},
our main theorem immediately
applies to semilinear equations on Kerr in the full subextremal range $|a|<M$.

We note that the direct appeal to a linear ``black box'' statement is similar  in spirit
to~\cite{metsog} for instance, where results on quasilinear equations outside obstacles in Minkowski space
were studied under an exponential decay assumption concerning the linear problem.

The remainder of this introduction is structured as follows:
In 
Section~\ref{introassumpgzero} below,
we will discuss in more detail the assumptions we make on the background $(\mathcal{M}, g_0)$, 
introducing both the black box linear inhomogeneous estimate and the additional physical space identity,
followed by the assumptions on the nonlinearity in Section~\ref{intrononlinearassump}. 
We will then sketch our proof of the above theorem
in Section~\ref{introproof}, introducing our purely dyadic approach.
Finally, we shall end  in Section~\ref{introcomparison} with a general discussion, giving various extensions
of the method, describing in more detail
the  relation with the nonlinear stability problem for black holes, and 
comparing in particular with the case of backgrounds modelled on Schwarzschild--de Sitter and Kerr--de Sitter
spacetimes and with the recent work~\cite{mavrogiannis}.

\subsection{Assumptions on $(\mathcal{M}, g_0)$: ``black box'' estimates and physical space identities}
\label{introassumpgzero}

We will make assumptions directly on the geometry of $(\mathcal{M},g_0)$ and on properties
of the linear inhomogeneous wave equation
\begin{equation}
\label{linearhomogeq}
\Box_{g_0}\psi =F
\end{equation}
on the fixed background $g_0$.

\subsubsection{Geometric assumptions on $(\mathcal{M},g_0)$}
\label{purelygeometric}
The purely geometric assumptions on $(\mathcal{M},g_0)$ will include assumptions concerning a suitable notion of asymptotic flatness,
that of stationarity (existence of a Killing field $T$ which is timlike near infinity), a smooth $T$-invariant strictly
positive function $r\ge r_0$, which behaves like the Euclidean area-radius as $r\to \infty$, and
the existence of a $T$-translation invariant foliation of $\mathcal{M}$ by $\Sigma(\tau)$ for $\tau\ge0$
hypersurfaces which are spacelike for $r\le R$ and ``outgoing'' null  for $r\ge R$. (In the case of 
Minkowski space, which will be the most basic example, we note that the function $r$ will only coincide with the usual radial coordinate for large values.)

To encompass the black hole cases of interest, 
we will allow 
for a possibly empty \emph{spacelike} future boundary $\mathcal{S}=\{r=r_0\}$
of $\mathcal{M}$, in which case we will also assume the presence of a Killing horizon $\mathcal{H}^+$
at $r=r_{\rm Killing}$ for an $r_0<r_{\rm Killing}$,  whose  surface
gravity will be required to be positive and whose Killing generator $Z$ need not coincide with $T$ (in 
which case however $Z$ will be required to lie in the span of $T$ and an additional assumed Killing vector field $\Omega_1$).
We will furthermore require that for $r>r_{\rm Killing}$, $T$ (or more generally the span of $T$ and $\Omega_1$, if the latter is assumed Killing) is timelike.
This will incorporate Schwarzschild and Kerr black holes in the full subextremal range $|a|<M$.

With minor modifications, we could also allow $\mathcal{S}$ to be a timelike boundary along which  suitable 
boundary conditions are imposed, and this would allow one to consider waves outside of obstacles, 
cf.~\cite{metsog}.

We will denote by $\mathcal{R}(\tau_0,\tau_1)$ spacetime ``slabs'' 
$\mathcal{R}(\tau_0,\tau_1)=\cup_{\tau_0\le \tau\le \tau_1} \Sigma(\tau) = J^-(\Sigma(\tau_1))\cap J^+(\Sigma(\tau_0))$.
The region $r\ge R$ will also be foliated by $T$-translation invariant  ``ingoing'' null 
hypersurfaces~$\underline{C}_v$.

\begin{figure}
\centering{
\def\svgwidth{15pc}
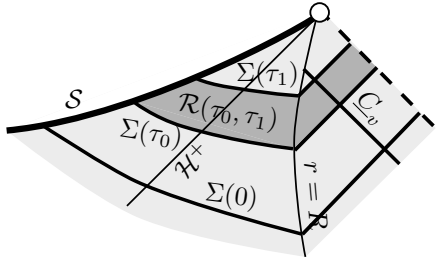}
\caption{The underlying spacetime $(\mathcal{M},g_0)$ in the case $\mathcal{S}\ne\emptyset$ 
with hypersurfaces and regions depicted}\label{introdiagfigure}
\end{figure}

Refer to~Figure~\ref{introdiagfigure} and already to Section~\ref{backgroundgeom} for a detailed discussion.

\subsubsection{Assumptions on~$\Box_{g_0}\psi =F$}
\label{rearranged}

In this section, we discuss the assumptions which we make at the level of the equation~\eqref{linearhomogeq}.

\paragraph{Basic degenerate integrated local energy estimate.}
Our fundamental ``black box'' assumption at the level of solutions of the inhomogeneous linear equation~\eqref{linearhomogeq} will be the validity of an integrated local energy estimate
\begin{equation}
\label{simplifiedestimate}
\mathcal{F}(v), \quad
 \mathcal{E}(\tau) +c \int_{\tau_0}^{\tau_1}{}^\chi\Ezerominusoneminusdelta'(\tau') d\tau' 
 +c \int_{\tau_0}^{\tau_1}\Eminusone'(\tau') d\tau' 
\leq \lambda \mathcal{E}(\tau_0)+ C\int_{\mathcal{R}(\tau_0,\tau_1)}
\left|\left(V^\mu_0\partial_\mu \psi+w_0 \psi\right) F\right|+C\int_{\mathcal{R}(\tau_0,\tau_1)} F^2,
\end{equation}
where $\mathcal{E}(\tau)$ denotes an energy flux through the hypersurface $\Sigma(\tau)$,
$\mathcal{F}(v)$ denotes an energy flux through $\underline{C}_v \cap \mathcal{R}(\tau_0,\tau_1)$,
 $\lambda\ge1$, $C>0$, $c>0$ are constants, and $V_0$, $w$ are a fixed vector field and function, respectively,
 and $\tau_0\leq\tau\leq\tau_1$ are arbitrary.
 In Minkowski space, such estimates go back to Morawetz~\cite{morawetz1968time}.
The energy flux  $\mathcal{E}(\tau)$ will control all first order derivatives of $\psi$ on the spacelike part of $\Sigma(\tau)$ while
it will only control all tangential derivatives
on the null parts of $\Sigma(\tau)$; similarly, $\mathcal{F}(v)$ will control all tangential derivatives 
on~$\underline{C}_v$. The fluxes will also control  a zeroth order term with weight $r^{-2}$.
The two additional terms on the left hand side are as follows: The quantity
${}^\chi\Ezerominusoneminusdelta'(\tau')$  denotes an integral over $\Sigma(\tau')$ controlling all first order derivatives
of $\psi$, whose density
is multiplied by $\chi r^{-1-\delta}$, where
 $0\le \chi\le 1$ denotes a function which is allowed to degenerate but must satisfy $\chi=1$ in a
  region to be discussed below, and
 where $\delta>0$ can be taken to be  an arbitrarily small constant which will be fixed throughout.
The quantity $\Eminusone'(\tau')$ denotes simply the zeroth order
quantity 
\[
\Eminusone'(\tau):=\int_{\Sigma(\tau)} r^{-3-\delta}\psi^2.
\] 
The subscript $-1$ indicates that it is of order $1$ less
in differentiability than the other energies considered.
We recall that in our setup, what we define to be $r$ satisfies $r\ge r_0>0$, and thus
$r$-weights are only relevant as $r\to\infty$.

In the case where $\mathcal{M}$ has no spacelike boundary, i.e.~$\mathcal{S}=\emptyset$,
and $T$ is globally timelike, our only assumption on $\chi$ will be
that $\chi=1$ for large $r$. 
In the case where $\mathcal{S}\neq\emptyset$,  then we will require that $\chi=1$ for $r\le r_2$
for some $r_2>r_{\rm Killing}$. 
See already Section~\ref{basicblackboxestimatesec} for the detailed assumptions.

Let us note that estimate~\eqref{simplifiedestimate}, with $\chi$ satisfying the above properties 
has indeed been shown in~\cite{partiii} 
in the Kerr case for the full subextremal range of parameters $|a|<M$. 
(See  already Theorem~\ref{trueforKerr}.)
One essential feature in deriving~\eqref{simplifiedestimate} in Kerr is the fact
that the  trapped null geodesics of~$g_0$, corresponding to photons in bound orbits around the black hole, are unstable in a suitable sense.
We note that one can show in general that
estimate~\eqref{simplifiedestimate} in the above form \emph{cannot} hold in the presence of
stable trapping. 
(Thus,~\eqref{simplifiedestimate} in particular constrains properties of geodesic flow which can be expressed
purely geometrically in terms of $g_0$.) Similarly, the possibility of an estimate~\eqref{simplifiedestimate}
with $\chi=1$ in a neighbourhood of $\mathcal{H}^+$ is related to the non-degenerate property
of the horizon of subextremal Kerr and its celebrated local \emph{red-shift effect}.
See~\cite{Mihalisnotes}.

\paragraph{Physical space energy identities.}
One way to try to prove~\eqref{simplifiedestimate} is through a
physical space energy \emph{identity} arising from integrating
the divergence identity
\begin{equation}
\label{introphysspace}
\nabla^{\mu}J^{V,w,q,\varpi}_{\mu}[\psi] = K^{V,w,q}[\psi] + (V^\mu\partial_\mu \psi) F+w\psi F
\end{equation}
for well chosen  currents $J^{V,w,q,\varpi}[\psi,g_0]$ (associated to a vector field $V$,
a scalar function $w$, a one-form $q$ and a two-form $\varpi$)
with (degenerate) coercivity properties for the arising bulk and boundary terms.
In this case, if one defines 
\[
\mathfrak{E}(\tau)=\int_{\Sigma(\tau)} J_\mu^{V,w,q,\varpi}[\psi]n^\mu_{\Sigma(\tau)}, \qquad
\mathfrak{F}(v)=\int_{\underline{C}_v\cap\mathcal{R}(\tau_0,\tau_1)} J_\mu^{V,w,q,\varpi}[\psi]n^\mu_{\underline{C}_v},
\]
then one
has $\mathfrak{E}(\tau)\sim \mathcal{E}(\tau)$, $\mathfrak{F}(v)\sim \mathcal{F}(v)$, and one can 
reexpress~\eqref{simplifiedestimate} with  $\mathfrak{E}(\tau)$, $\mathfrak{E}(\tau_0)$ replacing 
$\mathcal{E}(\tau)$, $\mathcal{E}(\tau_0)$, respectively, but
with
$\lambda$ in~\eqref{simplifiedestimate} now equal to  $1$.
This is precisely the situation in the Schwarzschild case, where appropriate coercive purely physical space identities of the form~\eqref{introphysspace}
can be deduced from~\cite{dafrod2007note, MR2563798}.

If~\eqref{simplifiedestimate} indeed holds as a consequence of the coercivity properties of a divergence
 identity~\eqref{introphysspace},
then this situation appears  very advantageous for nonlinear applications.
The advantage of such coercive
identities~\eqref{introphysspace}  is that they are robust and easily amenable to \emph{direct}
application to the quasilinear~\eqref{theequationzero}. This is because
the latter can be viewed as an inhomogeneous wave equation
\emph{with respect to the wave operator $\Box_{g(\psi,x)}$ corresponding to the solution itself},
and both the existence of energy identities like~\eqref{introphysspace} and their coercivity properties  are stable to passing from $g_0$ to $g(\psi,x)$,
i.e.~still hold for $J^{V,w,q,\varpi}[g(\psi,x),\psi]$. 

If one only has estimate~\eqref{simplifiedestimate} 
as a ``black box'', i.e.~not proven via a physical-space identity~\eqref{introphysspace},
one may still of course apply the estimate to~\eqref{theequationzero} by rewriting~\eqref{theequationzero} 
as an inhomogeneous equation with respect
to $\Box_{g_0}$:
\begin{equation}
\label{inhomogformintro}
\Box_{g_0}\psi =  ( \Box_{g_0}-   \Box_{g(\psi,x) })\psi +  N(\partial\psi, \psi, x).
\end{equation}
Applied in this manner, however,
the resulting estimate \emph{loses derivatives}, in view of the quasilinear second order terms on the right hand side
of~\eqref{inhomogformintro}.  (Note in contrast that in the purely semilinear case, where $g(\psi,x)=g_0$, there is no such
loss in estimating~\eqref{inhomogformintro}, and we may base our argument
 entirely on~\eqref{simplifiedestimate}. 
See already Section~\ref{semilinearcase}.)
At first glance, this loss would appear
to present a fundamental difficulty.

\paragraph{An auxiliary physical space estimate.}
Here we come to the   key observation of the present approach:
To avoid the loss of derivatives described above, one does \underline{not} need that~\eqref{simplifiedestimate} be proven
via a  physical space identity~\eqref{introphysspace}; one only needs
a much weaker (from the coercivity point of view) estimate arising from~\eqref{introphysspace}
which need not imply~\eqref{simplifiedestimate} but can be used \emph{in conjunction} with~\eqref{simplifiedestimate}.

Specifically, our main assumption is that, in addition to~\eqref{simplifiedestimate},
 one has a physical space identity~\eqref{introphysspace} with
(a) coercive boundary terms and with (b) bulk terms which are only assumed to be nonnegative \emph{at highest order}, 
i.e.~up to lower order ``error'' terms. Importantly, however, these allowed lower order error terms must moreover be supported
entirely in the region where the degeneration function $\chi$ of~\eqref{simplifiedestimate} satisfies $\chi=1$, i.e.~where the estimate
of~\eqref{simplifiedestimate} is in fact non-degenerate. 
The  identity~\eqref{introphysspace} applied to~\eqref{linearhomogeq}, on its own, 
gives rise thus to an estimate of the form
\begin{equation}
\label{simplifiedestimatewithA}
\mathfrak{F}(v), \quad
 \mathfrak{E}(\tau) +c \int_{\tau_0}^{\tau_1}{}^{\rho}\Ezerominusoneminusdelta'(\tau') d\tau'
+ c \int_{\tau_0}^{\tau_1}{}^{\rho}\Eminusone'(\tau') d\tau'
\leq \mathfrak{E}(\tau_0)+  A \int_{\tau_0}^{\tau_1}\Eerrorminusone'(\tau') d\tau'
+      \int_{\mathcal{R}(\tau_0,\tau_1)}
\left|\left(V^\mu \partial_\mu \psi + w\psi\right) F\right|,
\end{equation}
for all $\tau\in [\tau_0,\tau_1]$ and sufficiently large $v$,
where ${}^\rho\Ezerominusoneminusdelta'$ is defined similarly to ${}^\chi\Ezerominusoneminusdelta'$ but with
degeneration function $\rho$ which in general vanishes on a bigger set than $\chi$. 
The zeroth order term ${}^{\rho}\Eminusone'$ on the left hand side also degenerates, unlike the 
analogous term in~\eqref{simplifiedestimate}.
The presence of the term
multiplying the new constant $A$ on the right hand side is necessary in view of the fact that nonnegativity is only
assumed for the highest order terms.
Here, 
$\Eerrorminusone'$ is a zeroth order energy whose density is supported only in the support of
a function~$\xi(r)$
\[
\Eerrorminusone'(\tau):=\int_{\Sigma(\tau)} \xi (r) \psi^2.
\]
Importantly, in view of our comments above, we will require that $\xi$ vanishes 
identically  where $\chi$ of~\eqref{simplifiedestimate} degenerates, i.e.~$\xi=0$ where $\chi\ne 1$.
We will also require $\xi$ to vanish for large $r$, and, in the case where $\mathcal{S}\ne\emptyset$, in the region $r\le r_2$.

In addition, we will make on $\rho$ the same asymptotic non-vanishing assumption that we 
we made previously on $\chi$, namely that $\rho=1$ for large $r$, and, in the presence of a nonempty spacelike boundary $\mathcal{S}\ne \emptyset$,
that $\rho=1$ in $r\le r_2$.
See already Section~\ref{caseiiiidentity} for the detailed assumptions and  Figure~\ref{supportfigs} 
for a depiction of the supports of $\chi$, $\rho$ and $\xi$.

We note that in the case where
$\mathcal{S}=\emptyset$ and $T$ is globally timelike, one can in fact easily construct a current
satisfying (a) and (b) and leading to~\eqref{simplifiedestimatewithA}, \emph{irrespective of the structure of trapping and the validity
of an estimate of the form~\eqref{simplifiedestimate}}.
We will
show in Appendix~\ref{howto} that for the very slowly rotating Kerr case $|a|\ll M$, although there is no globally
timelike $T$, there is a similar elementary construction, based essentially only
on the fact that superradiance is effectively governed by a small parameter and the 
ergoregion is confined to a part of spacetime  which can be understood using only the red-shift 
effect (cf.~the original treatment of the boundedness problem on Kerr~\cite{DafRodKerr}). 
In general, however, the main difficulty in proving~\eqref{simplifiedestimatewithA} is ensuring the boundary
coercivity~(a).

Though on its own, estimate~\eqref{simplifiedestimatewithA} 
is clearly weaker than~\eqref{simplifiedestimate},
its robustness to direct application to~\eqref{theequationzero} allows
one to avoid the above loss of derivatives, when~\eqref{simplifiedestimatewithA} is used
in conjunction with~\eqref{simplifiedestimate} applied to~\eqref{inhomogformintro}. 
Before turning to describe how this is done in the context of the proof,
we will in fact have to slightly strengthen our assumed estimates~\eqref{simplifiedestimate} and~\eqref{simplifiedestimatewithA}
as follows. 

\paragraph{Extending to $r^p$ weighted estimates.}
First of all, we will need to extend~\eqref{simplifiedestimate} and~\eqref{simplifiedestimatewithA}
to $r^p$-weighted estimates, for $0\le p<2$, of the type first introduced in~\cite{DafRodnew}:
\begin{equation}
\label{simplifiedestimatepweight}
 \Ep(\tau) + \Fp(v)+ \int_{\tau_0}^{\tau_1}{}^\chi\Epminusone'(\tau') d\tau' 
 + \int_{\tau_0}^{\tau_1}\Epminusoneminusone'(\tau') d\tau' 
\lesssim  \Ep(\tau_0)+\text{\rm inhomogeneous terms},
\end{equation}
\begin{equation}
\label{simplifiedestimatewithApweight}
\Ffancyp(v), \qquad
\Efancyp(\tau) +c \int_{\tau_0}^{\tau_1}{}^{\rho}\Epminusone'(\tau') d\tau' +c \int_{\tau_0}^{\tau_1}{}^{\rho}\Eminusone'(\tau') d\tau' +
\leq \Efancyp(\tau_0)+  A \int_{\tau_0}^{\tau_1}\Eerrorminusone'(\tau') d\tau'
 +\text{\rm inhomogeneous terms}.
\end{equation}
In the above, $\Ep$, ${}^\chi \Epminusone'$, etc., are  $r^p$,  respectively
$r^{p-1}$, etc., weighted versions of $\mathcal{E}$, ${}^\chi \Ezerominusoneminusdelta'$, etc, 
while $\Efancyp$, $\Ffancyp$ satisfy
$\Efancyp\sim\Ep$, $\Ffancyp\sim \Fp$ but can moreover be represented as the flux term of a current 
\[
\Efancyp(\tau)=\int_{\Sigma(\tau)} \Jp_\mu^{V,w,q,\varpi}[\psi]n^\mu,
\qquad
\Ffancyp(v)=\int_{\underline{C}_v} \Jp_\mu^{V,w,q,\varpi}[\psi]n^\mu,
\]
and we in fact assume that~\eqref{simplifiedestimatewithApweight}, just as~\eqref{simplifiedestimatewithA}, is
the result of a pointwise coercive energy identity~\eqref{introphysspace} associated to the current $\Jp$.

To obtain~\eqref{simplifiedestimatepweight} and~\eqref{simplifiedestimatewithApweight}, given~\eqref{simplifiedestimate} and~\eqref{simplifiedestimatewithA}, it suffices to
assume the existence of currents defined only in the region $r\ge \tilde{R}$
with suitable far away coercivity properties. For maximal generality, we will here simply directly
postulate the existence of such currents as an additional assumption. 
See already Section~\ref{rpsection} for the precise formulation.
This assumption 
can be shown to follow from suitable pointwise asymptotically flat assumptions on the metric $g_0$ and
holds of course in all our examples of interest, where the currents can easily be explicitly constructed.
See Appendix~\ref{section:appedixcurrents}, and also~\cite{Moschnewmeth} for an even more general setting.

\paragraph{Extending to higher order estimates.}
Secondly, we will need, through suitable commutations,  
to extend~\eqref{simplifiedestimate}, \eqref{simplifiedestimatewithA}, \eqref{simplifiedestimatepweight} and~\eqref{simplifiedestimatewithApweight}
to higher order statements.
We will distinguish two energies
\[
\Epk, \qquad  \Efancypk
\]
at $k$'th order. 

The energy $\Epk$ is the sum of the usual energy, with $r^p$ weights $0\le p<2$,
and with commutation vector fields $\widetilde{\mathfrak{D}}^k$ up to order $k$, \emph{spanning the entire tangent space}.  
The set  will include the vector fields $\Omega_i$, $i=1,\ldots, 3$, which for large $r$ correspond to the usual 
rotational vector fields.
(Note that the latter will be the only commutation vector fields which are  $r$-weighted.)
This is a fundamental energy with respect to which
initial data can be measured and with respect to which suitable Sobolev inequalities hold.

The energy $\Efancypk$, on the other hand, 
denotes the precise energy flux of the current leading to~\eqref{introphysspace} applied
moreover  with  a new set  $\mathfrak{D}^k$ of commutation vector fields, which will include $T$ (and $\Omega_1$, if this is assumed to be Killing), but for which \emph{all non-Killing vectors are cut off to vanish
in a suitable region of finite $r$}.
(See already 
Sections~\ref{commutatordefsec}--\ref{mastercurrentsandfluxdef} for the definition of these commutation vector fields and energies.)
Nonetheless, by our geometric assumptions and elliptic estimates, we will have 
that the energies
\begin{equation}
\label{similarlabel}
\Efancypk\sim \Epk
\end{equation}
 are equivalent for solutions of $\Box_{g_0}\psi=0$. 
 
 To extend our estimates to higher order, we will need some properties of the commutation
 errors arising from $[\Box_{g_0}, \mathfrak{D}^k]$. 
 Let us note that
in the case where $\mathcal{S}=\emptyset$ and $T$ is globally timelike,
these errors are only supported in the asymptotic region of large $r$ and can be controlled if the metric
is suitably asymptotically flat. For maximum generality, we will formulate the relevant
asymptotic assumption directly in terms of pointwise decay bounds for these commutators.
This will include all examples of interest. See already Section~\ref{farawaycommutationerror}.

In the  case where $\mathcal{S}\ne\emptyset$, there will be 
additional commutation errors $[\Box_{g_0}, \mathfrak{D}^k]$ arising in  a neighbourhood of the boundary $\mathcal{S}$ including 
the Killing horizon $\mathcal{H}^+$. Key to their control is the assumption of positive surface gravity of $\mathcal{H}^+$,
discussed in Section~\ref{purelygeometric}, and the inclusion among the  collection $\mathfrak{D}$  of a well-chosen 
vector field $Y$ which is null and transversal to $\mathcal{H}^+$. As first shown in~\cite{Mihalisnotes},
commutation by $Y$ generates
a term with a good sign at $\mathcal{H}^+$ (see already Proposition~\ref{fromDR} of Section~\ref{commutationerrorsnearhor}), 
and this fact, together with elliptic estimates in $r>r_{\rm Killing}$, an enhanced red-shift estimate near $\mathcal{H}^+$ and enhanced positivity in the black hole interior
up to $\mathcal{S}$ 
(see already Propositions~\ref{enhancedprop}  and~\ref{enhancedinteriorprop}
of Section~\ref{enhanced}), can be used to absorb the errors arising from commutation
and to extend the estimates to higher order. We note that the above-mentioned good sign  is yet another manifestation of the local red-shift
at the horizon $\mathcal{H}^+$ associated to the positivity of the surface gravity.
 To understand how all commutation errors can indeed be absorbed,
 see already Section~\ref{globalcontrolsection}.
 
The final  extended estimates on~\eqref{linearhomogeq} 
 resulting from~\eqref{simplifiedestimate} and \eqref{simplifiedestimatewithA} take
 the form
 \begin{equation}
\label{simplifiedestimatepweighthigherorder}
 \Epk(\tau) +\Fpk+ \int_{\tau_0}^{\tau_1}{}^\chi\Epminusonek'(\tau') d\tau' 
 + \int_{\tau_0}^{\tau_1}\Epminusonekminusone'(\tau') d\tau' \lesssim \Ep(\tau_0)+\text{\rm inhomogeneous terms},
\end{equation}
\begin{equation}
\label{simplifiedestimatewithApweighthigherorder}
\Ffancypk(v), \quad \Efancypk(\tau) +c \int_{\tau_0}^{\tau_1}{}^{\rho}\Epminusonek'(\tau') d\tau'
+c \int_{\tau_0}^{\tau_1}{}^{\rho}\Ezerominusthreeminusdeltakminusone'(\tau') d\tau'
\leq \Efancypk(\tau_0)+  A \int_{\tau_0}^{\tau_1}\Eerrorkminusone'(\tau') d\tau'
   +\text{\rm inhomogeneous terms},
\end{equation}
where again~\eqref{simplifiedestimatewithApweighthigherorder} derives
from the pointwise coercivity properties of the energy identity~\eqref{introphysspace} associated to a current $\Jpk$,
whose  flux terms are  precisely given by $\Efancypk$, $\Ffancypk$.
See already Section~\ref{finalhighestimlinsec} for the precise form of the estimates and inhomogeneous terms.

\subsection{Assumptions on the nonlinearities $g(\psi, x)$ and $N(\partial\psi,\psi,x)$}
\label{intrononlinearassump}

We now turn to the nonlinear equation~\eqref{theequationzero}. 
In formulating assumptions on the allowed nonlinearities, we are here largely motivated 
by geometric formulations of the Einstein equations, as in~\cite{dhrtplus}, which produce
``almost decoupled'' equations similar to~\eqref{theequationzero}, which  satisfy an appropriate
form of the null condition. We emphasise, however, that the resulting equations in such a formulation
do not constitute a pure system
of wave equations in the form~\eqref{theequationzero}, and the quasilinear structure is mitigated
through coupling with transport equations. Thus, in this context, one should really only think
of~\eqref{theequationzero} as an indicative model equation.
For more specific discussion of
the Einstein equations, see already Section~\ref{discussionofeinstein}.

In addressing the combined difficulties of the interaction of quasilinear terms, trapped null geodesics, and merely polynomial decay at the level of  model scalar equations of the form~\eqref{theequationzero},  the most natural  simplest setting is   to  require
that the \emph{quasilinear} term $g(\psi,x)-g_0(x)$ be supported entirely in the region $r\le R$,
and   to impose on the \emph{semilinear} term $N(\partial\psi,\psi,x)$ 
a generalised version of the null condition~\cite{KlNull}.
This 
 also allows us to use the exact null hypersurfaces of the background $(\mathcal{M},g_0)$
without additional
complications, making the null hierarchical structure of  the estimates clearer. 
(In the case of the Einstein equations, in our framework, these would be replaced by  null hypersurfaces
of the actual dynamic spacetime $(\mathcal{M},g)$, for instance associated to a double null gauge.)  We note, however, that
use of exact null hypersurfaces is in no way fundamental;  one can always replace these
with suitable hyperboloidal hypersurfaces, for instance, using the setup of~\cite{Moschnewmeth}.

For maximum generality, rather than
attempt a general algebraic definition of the null condition
for the semilinear terms $N(\partial\psi,\psi,x)$ of~\eqref{theequationzero},
we will formalise a ``null condition assumption'' in the form of
a time-translation translation invariant $r^p$ weighted estimate for the inhomogeneous terms
of~\eqref{simplifiedestimatepweighthigherorder}--\eqref{simplifiedestimatewithApweighthigherorder}
when applied to~\eqref{theequationzero},
restricted entirely to the asymptotic region
$r\ge R$.

To describe the condition, it will be convenient to introduce ``master'' energies
\begin{equation}
\label{defforintroone}
{}^\rho \Xpk(\tau_0,\tau_1):= 
\sup_v\Fzerok(v)+\sup_{\tau_0\le\tau\le \tau_1}\Epk(\tau)+\int_{\tau_0}^{\tau_1}\left( 
 {}^\rho \Epminusonek'(\tau)+ {}^{\rho} \Ezerominusthreeminusdeltakminusone'(\tau) \right) d\tau
  \,\, 
\end{equation}
for all
$\delta\le p\le 2-\delta$.
In the case $p=0$, which is anomalous, we will define
\begin{equation}
\label{defforintrotwo}
\begin{aligned}
{}^\rho \Xzerok(\tau_0,\tau_1)&:= 
\sup_v\Fzerok(v)+\sup_{\tau_0\le\tau\le \tau_1}\Ezerok(\tau)+\int_{\tau_0}^{\tau_1}\left( {}^\rho \Ezerominusoneminusdeltak'(\tau) +
{}^{\rho} \Ezerominusthreeminusdeltakminusone'(\tau) \right)
d\tau \,\, , \\
{}^\rho \Xzeroplusk(\tau_0,\tau_1)&:= 
\sup_v\Fzerok(v)+\sup_{\tau_0\le\tau\le \tau_1}\Ezerok(\tau)+\int_{\tau_0}^{\tau_1} \left(  {}^{\rho} \Edeltaminusonek' (\tau')+ 
{}^{\rho}\Ezerominusthreeminusdeltakminusone'(\tau')
			 						\right)
d\tau \,\, .
\end{aligned}
\end{equation}
We may define similar master energies ${}^\chi\Xpk$, etc., with $\chi$ replacing $\rho$, 
and where  extra non-degenerate lower order terms $\Epminusonek'$, $\Ezerominusoneminusdeltakminusone'$
are added to the integrands on the right hand side of~\eqref{defforintroone} and~\eqref{defforintrotwo}, respectively,
and finally $\Xpk$ where $\rho$ is removed.
Let us note that
\[
{}^\rho \Xpk\lesssim {}^\chi \Xpk\lesssim \Xpk, \qquad \Xpkminusone\lesssim {}^\chi \Xpk.
\]

Specifically, our null condition assumption
is that in any spacetime slab $\mathcal{R}(\tau_0,\tau_1)$,
the far away  (i.e.~$r\ge R$) contribution of the inhomogeneous terms arising from $N(\partial\psi,\psi,x)$ in the 
estimates~\eqref{simplifiedestimatepweighthigherorder},~\eqref{simplifiedestimatewithApweighthigherorder} 
can be bounded by the expressions
\begin{equation}
\label{nullcondassumpintro}
{}^\rho \Xpk(\tau_0,\tau_1) \sqrt{\Xzerolesslessk(\tau_0,\tau_1)} + \sqrt{{}^\rho \Xpk(\tau_0,\tau_1)}\sqrt{{}^\rho \Xzerok(\tau_0,\tau_1)}\sqrt{\Xplesslessk(\tau_0,\tau_1)},
\qquad   {}^\rho\Xzeroplusk(\tau_0,\tau_1) \sqrt{ \Xzeropluslesslessk(\tau_0,\tau_1)}
\end{equation}
for $\delta\le p\le 2-\delta$ and $p=0$, respectively, and where
all energies  may in fact be restricted to the region $r\ge 8R/9$, where $\rho=1$.  (Note how  the $p=0$ weighted estimate is 
anomalous, in that the nonlinear terms require the boundedness of a bulk term associated to the energy identity of a higher ($p>0$) 
weight so as to be estimated.) See already Section~\ref{formofnullcondsec} for the precise formulation of the assumption.

The
assumption that  one can bound the relevant inhomogeneous terms arising from $N(\partial\psi,\psi,x)$ by~\eqref{nullcondassumpintro} isolates precisely that aspect of the usual null condition
which is relevant for global existence in our method.
We will then show in Appendix~\ref{nonlineartermsatinf} that our assumption indeed holds in particular for
the usual nonlinearities  allowed by the classical null condition~\cite{KlNull} and includes also
 the class of semilinear
 terms $N(\partial\psi,\psi,x)$ on  Kerr  studied previously in~\cite{MR3082240}.

\subsection{The dyadically localised approach to global existence: sketch of the  proof of the main theorem}
\label{introproof}

We now turn to the proof of the main theorem and the other novel aspect
of the present work, namely the replacement of a global bootstrap built on global decay-in-time estimates
with dyadic iteration in consecutive
spacetime slabs based entirely on dyadically localised, $r^p$ weighted---but time-translation invariant---estimates.

We sketch the argument here. For details, see already Section~\ref{estimatehier}.
In our discussion below, we will  focus directly  on our main case of interest, referred to later in the paper
as case (iii), where the black box inhomogeneous estimate~\eqref{simplifiedestimate} has nontrivial degeneration
and indeed does not arise
from a physical space identity, and thus must be used in conjunction with the auxiliary
 identity~\eqref{simplifiedestimatewithA}.  We emphasise, however, that the dyadic approach described here
is already novel in the case, referred to later in the paper as case (ii),
where the black box estimate~\eqref{simplifiedestimate} again has degeneration but is actually a consequence of a  physical space identity (as in Schwarzschild, see the discussion following~\eqref{introphysspace}).
Finally, in the case where the black box estimate~\eqref{simplifiedestimate} derives from a physical space identity and is
moreover non-degenerate (i.e.~$\chi=1$ identically), referred to in the paper as case (i), then dyadic iteration is in fact unnecessary, and the approach reduces to a completely elementary time-translation invariant
estimate.
(This latter case applies for instance when $g_0$ is a small stationary
perturbation of Minkowski space.)
\emph{In the paper, we will provide self-contained treatments of all three cases so the reader can
compare with the more traditional approach.}

\subsubsection{The $r^p$ weighted estimate hierarchy on a spacetime slab of time-length $L$}
\label{allabouthierarchy}

In the dyadic approach, the key element is an appropriate time-translation invariant $r^p$ weighted hierarchy on a slab
of time length $L$. The hierarchy  will be formed by combining the $r^p$ and higher $k$ order
estimates associated to~\eqref{simplifiedestimate} and~\eqref{simplifiedestimatewithA}
with various choices of $p$ weights and differentiability order $k$.
The estimates will depend only on a time-translation invariant bootstrap assumption,
to be retrieved within the slab, concerning only lower order  energy quantities.

Specifically, from the assumptions of Section~\ref{introassumpgzero} and~\ref{intrononlinearassump}, 
on any spacetime slab $\mathcal{R}(\tau_0,\tau_0+L)$ 
of time-length $L\ge 1$, then, we obtain
for $p=2-\delta$ and $p=1$, and for $0<\delta<\frac1{10}$,
the following hierarchy of estimates:
\begin{align}
\nonumber
\Ffancypk(v),\quad
\Efancypk(\tau),\quad
{}^\rho \Xpk(\tau_0,\tau)
&\leq \Efancypk(\tau_0) 
+\boxed{A\int_{\tau_0}^\tau \Eerrorkminusone'}+   C\,
{}^\rho \Xpk(\tau_0,\tau) \sqrt{\Xzerolesslessk(\tau_0,\tau)} +C\, \sqrt{{}^\rho \Xpk(\tau_0,\tau)}\sqrt{{}^\rho \Xzerok(\tau_0,\tau)}\sqrt{\Xplesslessk(\tau_0,\tau)}
\\
\label{highes}
&\qquad
+\boxed{C\sup_{\tau_0\le \tau'\le\tau}\Ezerok(\tau')
\sqrt{\Xzerolesslessk(\tau_0,\tau)} \sqrt{L}},
\end{align}
\begin{equation}
\begin{aligned}
\label{secondhigh}
{}^\chi \Xpkminusone (\tau_0,\tau) &\lesssim \Epkminusone(\tau_0) +  
{}^\rho \Xpkminusone(\tau_0,\tau) \sqrt{\Xzerolesslessk(\tau_0,\tau)} + \sqrt{{}^\rho \Xpkminusone(\tau_0,\tau)}\sqrt{{}^\rho \Xzerokminusone(\tau_0,\tau)}\sqrt{\Xplesslessk(\tau_0,\tau)}
 \\
&
\qquad +\boxed{\sup_{\tau_0\le \tau'\le\tau}\Ezerok(\tau')}
\sqrt{\Xzerolesslessk(\tau_0,\tau)} \sqrt{L}.
\end{aligned}
\end{equation}
The estimates are in fact
 contingent on an appropriate  time-translation invariant bootstrap assumption 
on the lower order energy 
\begin{equation}
\label{contingent}
\Xplesslessk(\tau_0,\tau) \lesssim \varepsilon
\end{equation}
for $p=0$.
The hierarchy~\eqref{highes}--\eqref{secondhigh} descends also to  
$p=0$, but where 
${}^\rho\Xzeroplusk(\tau_0,\tau) \sqrt{ \Xzeropluslesslessk(\tau_0,\tau)}$ and ${}^\rho\Xzeropluskminusone(\tau_0,\tau) \sqrt{ \Xzeropluslesslessk(\tau_0,\tau)}$ replace 
the sum of the third to last and second to last terms
of~\eqref{highes} and~\eqref{secondhigh}, respectively.
(This anomaly, referred to already immediately after~\eqref{nullcondassumpintro},  is a fundamental aspect of the estimates.)
See already Section~\ref{hierarchiiisec}.

Estimate~\eqref{highes} arises from applying the physical space identity
associated to the current $\Jpk$, which in the linear case led to~\eqref{simplifiedestimatewithApweighthigherorder},
\emph{directly} to
the $k$ times commuted equation~\eqref{theequationzero}, 
i.e.~it arises from the divergence identity~\eqref{introphysspace} associated to the  current $\Jpk[g(\psi,x), \psi]$,
where $g(\psi,x)$ replaces $g_0$ covariantly in the definition of~$\Jpk$. (The energies $\Efancypk$ and $\Ffancypk$ in~\eqref{highes} now in fact
denote the exact fluxes arising from these currents; in view of the bootstrap assumption~\eqref{contingent}, one can again show the
equivalence~\eqref{similarlabel} for solutions of~\eqref{theequationzero}.)
As discussed already in Section~\ref{rearranged}, it follows that~\eqref{highes}
 does not ``lose derivatives'', as is
clear from examining the order of differentiability of the terms on the right hand side. 
We remark that the first boxed term in~\eqref{highes} reflects the analogous term in~\eqref{simplifiedestimatewithApweighthigherorder}. 
On the right hand side of~\eqref{highes}, we recognise the bound for the far away contribution of the nonlinear term~\eqref{nullcondassumpintro}  from our assumption
capturing the null
condition, while
the second boxed term in~\eqref{highes}
is necessary to control the nonlinear terms on the set where $\rho$ degenerates,
for there, this boxed term cannot be absorbed in the  spacetime bulk term on the left hand side.
(The bad explicit
dependence on the length $L$ arises because one must take the supremum over $\tau$
for the energy of the highest order terms there.)

Estimate~\eqref{secondhigh}, on the other hand, 
arises from applying the ``black box'' estimate~\eqref{simplifiedestimatepweighthigherorder}
to the nonlinear equation written
in the form~\eqref{inhomogformintro}, i.e.~now thought of simply as an inhomogeneous equation with respect to the
background $g_0$.
The boxed term in~\eqref{secondhigh}, which is $k$'th order, reflects the loss of derivatives due
to the quasilinearity discussed already
in Section~\ref{rearranged}.

The estimates can in principle close because, in view of our restrictions on the support of $\xi$, we have the fundamental
relation
\begin{equation}
\label{fundamentalrelationhere}
\Eerrorkminusone' \lesssim  {}^\chi \Ezerominusoneminusdeltakminusone' +\Ezerominusoneminusdeltakminustwo'
\end{equation}
and thus we may bound the term
\[
A\int_{\tau_0}^\tau \Eerrorkminusone'(\tau')d\tau'  \lesssim {}^\chi \Xzerokminusone(\tau_0,\tau),
\]
which  is in turn bounded by the term on the left hand side of~\eqref{secondhigh} for any $p$.

\subsubsection{Global existence in a slab}

The first task is to show global existence in a given slab $\mathcal{R}(\tau_0,\tau_0+L)$ for arbitrary $L\ge 1$,
given suitable ($L$-dependent) smallness assumptions at $\Sigma(\tau_0)$.

We first note that the estimate~\eqref{highes} alone can be used to easily show \emph{local} existence,
by which we mean existence
in an entire smaller slab of the form~$\mathcal{R}(\tau_0,\tau_0+\epsilon)$, provided that
some $\Epk \lesssim \varepsilon_0$ for sufficiently high $k$.
This is in fact true also for $p=0$.
(Note, that  this ``local''  result already uses non-trivially the null condition
assumption of Section~\ref{intrononlinearassump}.
This is because our foliation $\Sigma(\tau)$ is outgoing null for $r\ge R$. Equations~\eqref{theequationzero} 
not satisfying some version of the null condition
can fail to yield solutions in $\mathcal{R}(\tau_0,\tau_0+\epsilon)$ for any $\epsilon>0$, no matter
how small the data are.)
Moreover, we show that smallness of a suitably high $\Epk$ norm defines a continuation
criterion. 

In view of the above, for  global existence, it suffices to show that, under suitable assumptions
at $\tau=\tau_0$, the quantity
$\Epk$, for some $p\in \{0\} \cup [\delta, 2-\delta]$ and suitably high $k$ remains globally small in the slab.
Examining the hierarchy, 
for global existence in a slab $\mathcal{R}(\tau_0,\tau_0+L)$ of length $L$,
it seems necessary to have at least
\begin{equation}
\label{bootmovoedw}
\Xzerolesslessk(\tau_0,\tau) \lesssim\varepsilon L^{-1}.
\end{equation}
Indeed, in view of the nonlinear terms,
the $L^{-1}$ factor in~\eqref{bootmovoedw} represents  a natural threshold,
as $s=1$ represents the minimum value of $s$ for which $\sqrt{L^{-s}}\sqrt{L} \lesssim 1$, allowing
the quantities in a fixed slab to be easily bounded by their initial values.

Assuming
\begin{equation}
\label{forglobexistence}
\Eonek(\tau_0) \lesssim \varepsilon_0, 
\qquad  \Ezerokminustwo(\tau_0) \lesssim L^{-1}\varepsilon_0 , 
\end{equation}
which of course imply, for $\varepsilon_0\ll \varepsilon$, both~\eqref{contingent} (for $p=1$) and~\eqref{bootmovoedw} at $\tau=\tau_0$,
 an easy continuity
argument, with~\eqref{contingent} and~\eqref{bootmovoedw} taken as bootstrap assumptions,
shows that if estimates~\eqref{forglobexistence}
hold for $\tau=\tau_0$, then the solution indeed exists globally in the entire slab $\mathcal{R}(\tau_0,\tau_0+L)$.

See already Section~\ref{globpropthree} for the details of the proof.
Note that our above appeal to~\eqref{contingent} and~\eqref{bootmovoedw} is in fact the only instance that a bootstrap assumption must be used in the proof
of our main theorem. 

\subsubsection{Improved estimates and the pigeonhole principle}

Once global existence within a slab has been established, one can improve a posteriori the estimates given 
stronger assumptions on initial data. Anticipating dyadic iteration, one seeks a set of initial estimates at $\tau=\tau_0$ with the property that they are
\emph{exactly recoverable} at the top of the slab $\tau_0+L$, with suitable redefinition of $\varepsilon_0$ and with
$L$ replaced by $2L$.

Re-examining the hierarchy~\eqref{highes}--\eqref{secondhigh}, we show that  
for an appropriate large constant $\alpha$, and any $L\ge 1$, $\tau_0\ge 1$,
then for $\hat\epsilon_0$ sufficiently small,
given for example the stronger estimates
\begin{equation}
\label{iteratintro}
\Efancyonek(\tau_0)\leq\hat\varepsilon_0 , \qquad
\Efancytwominusdelkminustwo(\tau_0)\leq\hat\varepsilon_0, \qquad \Efancyzerokminustwo(\tau_0)\leq \hat\varepsilon_0\alpha L^{-1}, \qquad \Efancyonekminusfour(\tau_0)\leq\hat\varepsilon_0 \alpha L^{-1+\delta},
\qquad \Efancyzerokminussix(\tau_0)\leq \hat\varepsilon_0\alpha^{2} L^{-2+\delta},
\end{equation}
then~\eqref{iteratintro} holds also at the final time $\tau_0+L$  where
 $\tau_0$, $L$ and $\hat\varepsilon_0$ are now however replaced by 
\begin{equation}
\label{toinduct}
\tau_0+L, \qquad 2L, \qquad  \hat\varepsilon_0(1+\alpha  L^{-\frac14}),
\end{equation}
respectively.
Note that the statement that the last three inequalities of~\eqref{iteratintro}  hold at $\tau_0+L$ with $2L$ in place of $L$
is a \emph{stronger} smallness assumption 
than that assumed for these quantities at $\tau_0$, signifying that these quantities have in fact decayed, and relies on  a pigeonhole argument as in~\cite{DafRodnew},
based on the fundamental relation that, for $p-1\ge 0$, the \emph{bulk} integrand 
 $\Epminusonekminustwo'$ of the $p$-weighted identity~\eqref{secondhigh} at order $k-1$ (when integrated, included in the ${}^\chi \Xpminusonekminustwo$
 master energy)
 controls the \emph{boundary} term 
of the $p-1$ weighted identity at arbitrary order $k$, denoted without the prime:
\begin{equation}
\label{comparison}
\Epminusonek' \gtrsim \Epminusonek.
\end{equation}

Briefly, the pigeonhole principle is applied as follows: The integral of the left hand side of~\eqref{comparison} for $p=1$ and $k$
replaced by $k-2$ can be shown to be bounded 
$\lesssim\hat\varepsilon_0$, whence we obtain the existence of a $\tau=\tau''$ slice
for which the right hand side $\Ezerokminustwo(\tau'')$  is bounded  $\lesssim \hat\varepsilon_0L^{-1}$,
where the $L^{-1}$ factor arises from the length of the interval.
We can then propagate this to $\tau_0+L$, and use that $\alpha \gg 1$ to arrange that we have the precise new
version of the third inequality of~\eqref{iteratintro} at $\tau_0+L$, i.e.~for $\Efancyzerokminustwo(\tau_0+L)$.
Similarly, the integral of the left hand side  of~\eqref{comparison} for $p=2-\delta$ and $k$ replaced by $k-4$
can be shown to be bounded $\lesssim\hat\varepsilon_0$, whence 
we obtain the existence of a $\tau=\tau''$ slice
for which the right hand side $\Eoneminusdelkminusfour(\tau'')$  is bounded $\lesssim\hat\varepsilon_0L^{-1}$,
 the $L^{-1}$ factor again arising from the length of the interval.
Interpolating with the boundedness statement $\Etwominusdelkminusfour(\tau'')\lesssim \hat\varepsilon_0$ we obtain
that on this slice  $\Eonekminusfour(\tau'')\lesssim\hat\varepsilon_0 L^{-1+\delta}$.
Again, we can then propagate this to $\tau_0+L$, and use that $\alpha\gg1$  to arrange that we have the precise
new version of the fourth inequality of~\eqref{iteratintro} at $\tau_0+L$, i.e.~for $\Efancyonekminusfour(\tau_0+L)$. 
Finally, a similar argument, applying~\eqref{comparison} for
$p=1$ and $k-6$, yields the new version of the
 last inequality of~\eqref{iteratintro}. 
 
 Note that 
to obtain the new versions of the first two inequalities of~\eqref{iteratintro} at $\tau_0+L$, which refer to quantities that  do not
decay with respect to their initial values, it 
is important that the constant
appearing in the first term on the right hand side of~\eqref{highes} 
is indeed $1$.

\subsubsection{Iteration over consecutive  spacetime slabs of dyadic time-length $L_i=2^i$}

Let us assume that our initial data on $\tau_0:=1$ satisfy
\[
\Eonek(\tau_0)+  \Etwominusdelkminustwo(\tau_0) \leq\varepsilon_0 
\]
for sufficiently small $\varepsilon_0$. It follows that~\eqref{iteratintro} is satisfied for $\tau_0=1$, $L:=L_0=1$ and appropriate $\hat\varepsilon_0$.

Examining closely~\eqref{iteratintro} and~\eqref{toinduct}, we see that we have now obtained a series of estimates
which, in addition to being sufficient to obtain existence in $\mathcal{R}(\tau_0, \tau_0+L)$, are moreover
preserved under dyadic iteration. 
We
simply iterate the above statement on consecutive spacetime slabs of dyadic time-length $L_i=2^i$,
 defining $\tau_{i+1}=\tau_i+L_i=2^i$, and $\hat\varepsilon_{i+1}=
\hat\varepsilon_i(1+\alpha  L_i^{-\frac14})$.
Note that $\hat\varepsilon_{i+1}\lesssim \hat\varepsilon_0\lesssim \varepsilon_0$ for all $i$. 

This immediately yields global existence in $\mathcal{R}(\tau_0,\infty) = \cup_{i\ge 0}\mathcal{R}(\tau_i,\tau_{i+1})$,
the top order boundedness statement
\begin{equation}
\label{finalestimates}
\Eonek(\tau)+ \Etwominusdelkminustwo(\tau)  \lesssim\varepsilon_0
\end{equation}
and the decay estimate 
\[
\Ezerokminussix(\tau_i)\lesssim \varepsilon_0\tau^{-2+\delta}.
\]
Thus, our result is a true orbital stability statement at highest order, in addition to yielding asymptotic stability.

\subsection{Discussion}
\label{introcomparison}

We end with some additional discussion.

\subsubsection{The semilinear case}
\label{semilinearcase}
We have already remarked that in the semilinear case, i.e.~the case where $g(\psi,x)=g_0$ identically and the equation
takes the form
\begin{equation}
\label{semil}
\Box_{g_0} \psi = N(\partial\psi,\psi, x),
\end{equation}
there
is no loss of derivatives when estimating~\eqref{semil} as a linear inhomogeneous equation. Thus, 
one can base the entire argument directly on the black box estimate~\eqref{simplifiedestimate}, without the need
for the auxiliary physical space based estimate~\eqref{simplifiedestimatewithA}, provided
of course that the appropriate  asymptotic flatness
and commutation assumptions are made so that~\eqref{simplifiedestimate} can be extended to the higher order
weighted~\eqref{simplifiedestimatepweighthigherorder}.
See already Remark~\ref{bigremarkonsemi} for the modified statement  appropriate
in this case and a guide to its proof. 
In particular, our theorem applies directly to the purely semilinear version of~\eqref{theequationzero} 
on Kerr in the full subextremal range $|a|<M$. 

\subsubsection{Additional examples and extensions}

We collect here some additional examples to which our results apply to, as well as various
natural future extensions.

\paragraph{Other spacetimes $(\mathcal{M},g_0)$ satisfying our assumptions.}
We have already specifically mentioned  the examples where $(\mathcal{M},g_0)$ is
Minkowski space, Schwarzschild and Kerr, but let us remark that our 
black box assumption~\eqref{simplifiedestimate} holds in fact
for the full  subextremal  Kerr--Newman family~\cite{civinthesis}. 
We also note that combining the  use of an energy current construction similar to  Appendix~\ref{howto}
with~\cite{wunschzworski}, 
one can show that~\eqref{simplifiedestimate}  holds
for \emph{arbitrary}  stationary perturbations $(\mathcal{M},g_0)$ of Schwarzschild, provided that they are close to
Schwarzschild in a suitably regular norm. Since
we have already remarked that the physical space based estimate~\eqref{simplifiedestimatewithA} holds also for such spacetimes, in fact under much weaker regularity assumptions, this means that
our main theorem indeed applies in this case.
Thus, for very slowly rotating $|a|\ll M$ Kerr spacetimes, one sees that the special additional structure of Kerr (like 
the various manifestations of Carter separability, the  Killing tensor, etc.)~has no significance
for the problem, and the spacetimes can just be understood as a
very small stationary  perturbation of Schwarzschild with a Killing horizon. \emph{It is only in the   case  $|a|\sim M$
where true Kerr properties are of any relevance.} A treatment of 
this case, reducing directly to the black box estimate proven in~\cite{partiii}, will appear elsewhere.

\paragraph{The axisymmetric case.} 
In the case where we assume that the metric $g_0$ admits an additional Killing field~$\Omega_1$, one may also 
restrict to nonlinear equations of the form~\eqref{theequationzero} 
which moreover preserve this symmetry, i.e.~with the property that if the data are preserved by $\Omega_1$,
then so is the solution. Let us call such equations and solutions axisymmetric. Under the geometric assumptions of this paper,
for axisymmetric equations~\eqref{theequationzero} and restricted to axisymmetric solutions, one may in fact easily derive
the existence of the necessary auxiliary physical space identity yielding~\eqref{simplifiedestimatewithA}, just as in Schwarzschild.
This in particular applies to Kerr--Newman in the full subextremal case.
Thus, in view of~\cite{partiii, civinthesis} which obtain~\eqref{simplifiedestimate}, the main theorem of the paper applies also in these settings.
We leave the details as an exercise for the reader.
Note that in this case, if one prefers,
one may quote the original~\cite{DafRodsmalla} for~\eqref{simplifiedestimate} for axisymmetric solutions,
in place of the more general~\eqref{simplifiedestimate} obtained for all solutions in~\cite{partiii}.

\paragraph{Potentials.}
Our method applies also if a suitably decaying
$T$-independent potential $\mathcal{V}(x)$ is added to~\eqref{theequationzero},
i.e.~for
\[
\Box_{g(\psi,x)}\psi -\mathcal{V}(x)\psi  =  N(\partial\psi, \psi, x),
\]
provided that both the black box assumption~\eqref{simplifiedestimate} and the physical space 
based~\eqref{simplifiedestimatewithA} hold for solutions of the inhomogeneous version of
the new linearised equation
equation $\Box_{g_0}\psi -\mathcal{V}(x)\psi=F$ in place of~\eqref{linearhomogeq}.
We note that in the case where $T$ is globally timelike and $\mathcal{V}$ decays suitably and is nonnegative, 
then, one can always obtain also the analogue of~\eqref{simplifiedestimatewithA}.  The same statement applies for any such potential in the slowly
rotating Kerr case $|a|\ll M$.

\paragraph{Systems.} For notational convenience, we have restricted here to scalar equations. The considerations generalise readily to systems,
for which a generalised null condition can again be formulated in terms of bounds of the nonlinear terms by~\eqref{nullcondassumpintro}; 
this now includes many examples. See also the discussion of the weak null condition in Section~\ref{wnc} below.

\paragraph{The obstacle problem.} We have already remarked that, with a mild adaptation of the setup,
 one can also consider the case where the boundary
 $\mathcal{S}$ is in fact timelike, imposing
now also suitable boundary conditions on $\mathcal{S}$.
This is for instance the setting for the classical obstacle problem. The analogue of~\eqref{simplifiedestimate} 
has indeed been proven in the nontrapping case (with $\chi=1$ identically, i.e.~without degeneration), but also, with suitable degeneration functions $\chi$,
for many examples with nontrivial hyperbolic trapping (see for example~\cite{lafontaine}).
We note  again that the additional physical space estimate~\eqref{simplifiedestimatewithA} 
can always be retrieved in this case. Thus, this situation can also be incorporated in our framework.

\paragraph{Extremal black holes.}
Extremal black holes do not satisfy our assumptions on $(\mathcal{M},g_0)$, already
because the stationary vector field $T$ will not  be tangential to the boundary $\mathcal{S}$, but
more seriously, because $\chi$ must degenerate at the black hole horizon~\cite{SbierskiGauss}, where $T$ will not
be timelike. Nonetheless, a version of~\eqref{simplifiedestimate} 
has been proven in the extremal Reissner--Nordstr\"om
case~\cite{aretakisstability} (and in the extremal Kerr case, under the assumption of axisymmetry~\cite{Aretakis}), with an additional
hierarchy of degeneration associated to the horizon, and
nonlinear stability for semilinear equations has been proven~\cite{nonlinearextreme}, where, however,
an additional null structural condition is
required at the horizon, in full analogy with the situation at null infinity. It would be interesting to reformulate this latter work in the dyadic
setup used here. This is related to the nonlinear stability problem of extremal black holes.
See the discussion in~\cite{dhrtplus} and Conjecture IV.2 therein.

\paragraph{The asymptotically AdS case.}
Finally, although our assumptions are modelled on the asymptotically flat setting, one could also
try to reformulate things in the asymptotically anti-de Sitter case, appropriate for solutions of the Einstein equations with negative cosmological constant $\Lambda<0$ (for a discussion of the $\Lambda>0$
case see already Section~\ref{compwithlambdapos} below). Here, one must also impose boundary conditions
\emph{at infinity}, since infinity itself can be thought of as an asymptotic \emph{timelike} boundary. We note that under
reflecting boundary conditions, there exist periodic (and thus non-decaying)
solutions of the wave equation on pure AdS, while general solutions  on
the Schwarzschild--anti de Sitter and  Kerr--anti de Sitter case~\cite{holzegel2013decay, HolzSmulevici}, again
under such boundary conditions, decay
only inverse logarithmically. 
(This is due to stable photon orbits repeatedly reflecting off
of infinity only to return later having been repelled centrifugally by the black hole.)
Based on this lack of decay, pure AdS has in fact been proven to be nonlinearly \underline{unstable} under
reflecting boundary conditions as a solution to various Einstein--matter systems, where the problem
can be studied already under spherical symmetry~\cite{moschidis2018newproof}.
Thus, if one hopes to show nonlinear stability for an equation of the type~\eqref{theequationzero} on
asymptotically AdS
backgrounds, it is
more natural to consider \emph{dissipative} boundary conditions like those considered 
in~\cite{holzegeletalasmyptot}.
Note that on Kerr--AdS, 
one expects to indeed satisfy the analogue of~\eqref{simplifiedestimatewithA} with the help 
 of the Hawking--Reall Killing vector field, provided that the parameters satisfy the Hawking--Reall bound. 
 Thus, given a version of~\eqref{simplifiedestimate}, nonlinear
 stability for a suitable class of equations~\eqref{theequationzero} 
should be   tractable for all such Kerr--AdS parameters. 

\subsubsection{The dyadic approach vs.~the traditional approach}

The dyadic approach followed here is of course closely related to the more traditional approach.
Our $L$-weighted smallness translates easily into $t$-decay assumptions, and  if one wishes,
one can rewrite
the argument using a bootstrap with $t$-decaying norms.
We believe, however, that the dyadic localisation of the argument both makes the proof
more modular and  may serve to better
identify possible future refinements with respect to regularity and decay. 
One sees clearly, for instance,
that the  $L^{-1}$ smallness assumptions necessary for global existence
\emph{within a slab} are manifestly weaker than the $L^{-2+\delta}$  smallness necessary
to improve and iterate (in fact one may weaken this to $L^{-1-\epsilon}$).  Since bootstrap is only used to show existence \emph{within a slab}, this already 
simplifies the argument considerably.
We also note that the $L^{-1-\epsilon}$  threshold corresponds to pointwise decay $t^{-\frac12-\frac\epsilon2}$,
considerably weaker than the $t^{-1-\epsilon}$
``improved decay'' which is often invoked for
global existence and stability for~\eqref{theequationzero}. 
Proving sharper decay is of course an extremely important problem in itself (see~\cite{Angelopoulos:2016moe, hintz2022sharp, angelopoulos2021latetime, kehrberger2022case})
but is not necessary (or even particularly helpful) for nonlinear stability.

Let us also point out that for problems with gauge invariance,
like the Einstein equations,  dyadic localisation also provides a convenient opportunity to refresh the gauge.
This suggests an alternative approach to the global teleological normalisations of 
gauge  done in~\cite{holzstabofschw, dhrtplus}.

\subsubsection{Applications to the  Teukolsky equation and derived equations
 and to the nonlinear stability of black holes}
 \label{discussionofeinstein}

As discussed in Section~\ref{intrononlinearassump}, one motivation for the precise assumptions on
the nonlinearities of~\eqref{theequationzero} made here
is viewing these as a model
for understanding the Einstein equations,
when the latter are considered under suitable geometrically defined gauges.  Such gauges were
first  exploited analytically in the monumental proof of the nonlinear stability of Minkowski space~\cite{CK}.

The geometric-gauge approach to black hole stability was taken up in~\cite{holzstabofschw}, where
the problem was studied  in double null gauge, and the full linear stability of Schwarzschild was first proven. 
(See~\cite{Chr} for an introduction to
double null gauge, including a discussion of the characteristic initial value problem and a derivation of all equations.)
The setup is of course intimately
connected to the Newman--Penrose formalism~\cite{newmanpenrose}.  In the   system resulting from
linearising the reduced Einstein equations in double null gauge around Schwarzschild, 
the gauge invariant quantities are determined
by the extremal curvature components $\alpha$ and $\underline\alpha$, which each satisfy
the Teukolsky equation,  first derived  in this setting by~\cite{bardeen1973} (and generalised
to Kerr in~\cite{teukolsky1973}).
The transport equations satisfied by the residual gauge dependent quantities, on the other hand, are
always already \emph{linearly} coupled to the gauge invariant ones.  
To analyse first the gauge invariant quantities,~\cite{holzstabofschw} introduced
a pair of novel physical space  quantities $P$ and $\underline{P}$, related to $\alpha$ and $\underline\alpha$ 
by second order weighted null differential operators, but satisfying the more tractable Regge--Wheeler
equation (an equation first derived in a different context in~\cite{Regge}), 
which, unlike the Teukolsky equation, could be understood by the exact same methods as
the linear wave equation. In particular, an analogue
of~\eqref{simplifiedestimate} was proven, via a  physical space identity, for $P$ and $\underline{P}$, and this led to 
boundedness and decay through the hierarchical structure of the system, first for the quantities $P$ and $\underline{P}$ 
themselves, then for 
the original gauge-independent quantities $\alpha$ and $\underline\alpha$, and then, after teleological
normalisation of the double null gauge, of all residual gauge dependent quantities as well. 
(For a complete scattering theory for this system, see~\cite{masaood2020scattering, masaood2022scattering}. For other approaches to the gauge in Schwarzschild under linear theory, see~\cite{benomio2022new} and the references discussed in Section~\ref{wnc} below.) 
The origin of this relation
between Teukolsky and Regge--Wheeler
goes back to the fixed frequency transformation theory of~\cite{Chandraschw}.
On the basis of this linear theory, the full nonlinear stability of the Schwarzschild family, without symmetry,
was  proven
in our~\cite{dhrtplus}. (For previous nonlinear results for Schwarzschild under symmetry assumptions, see~\cite{maththeory, DRPrice} in the case of the Einstein-scalar field system under spherical symmetry, 
and~\cite{holzbiax,klainerman2020global} for vacuum, the former in the higher dimensional case under biaxial Bianchi symmetry, and the latter under polarised axisymmetry, reducing to $2+1$ dimensions, the first such result beyond $1+1$ dimensional reductions.)
Note that since Schwarzschild is a co-dimension $3$ subfamily of Kerr, 
nonlinear asymptotic stability of Schwarzschild  refers to constructing the full (teleologically defined) codimension $3$ set of initial data which asymptote to Schwarzschild in the future.

A generalisation of the quantities $P$ and $\underline{P}$ of~\cite{holzstabofschw} to the slowly
rotating Kerr case $|a|\ll M$
was  given independently by~\cite{DHRteuk, Ma:2017yui2}.
The equations satisfied by these quantities are now however (weakly) coupled to the quantities
satisfying the Teukolsky equation.  The works~\cite{DHRteuk, Ma:2017yui2} both show an analogue
of~\eqref{simplifiedestimate} for this system, not via a physical space identity however but
 based on the framework for frequency localised estimates on Kerr introduced in~\cite{DafRodsmalla}. 
These results were followed by full linear stability statements for the gauge dependent
quantities in various gauges~\cite{Andersson:2019dwi, Hafner:2019kov} (the work~\cite{Hafner:2019kov} considers
in fact harmonic gauge, cf.~the discussion in Section~\ref{wnc}).
The estimates of~\cite{DHRteuk, Ma:2017yui2} have recently been reformulated by~\cite{giorgietal}
in the language of the higher order physical space commutation by 
the Carter tensor first introduced in~\cite{AndBlue}. The work~\cite{giorgietal} then  uses this, among other ingredients,   in
the context of the formalism~\cite{giorgi2020general} to obtain
nonlinear stability results for the very slowly
rotating Kerr case $|a|\ll M$, culminating an impressive series of preprints.
For  generalisations to the  Einstein--Maxwell system
see~\cite{giorgi, giorgi2021carter, apetroaie2022instability}.

 The full subextremal Kerr case $|a|< M$  is  more subtle, as it cannot be understood as a small
 perturbation of Schwarzschild. Remarkably,  however,
an  analogue of~\eqref{simplifiedestimate} for this system has been obtained in the full subextremal case~\cite{RitaShlap, RitaShlap2},
using the already highly non-trivial mode stability results~\cite{Whiting, SRT, daCosta:2019muf}
(see also~\cite{doi:10.1063/1.4991656}), but, also, non-trivial new high frequency 
structure with no apparent precise
analogue in the context of the wave equation. To exploit frequency analysis, the proof of~\cite{RitaShlap}
 uses a version of the frequency localisation framework of~\cite{partiii}.
(In connection with Teukolsky equation, we also note~\cite{ma2023sharp, millet2023optimal} 
for precise asymptotics
in the special case of smooth compactly supported data. For a discussion of what are the
appropriate initial assumptions on data for the Teukolsky equation in various physical settings, see~\cite{Gajic_2022}.)

In view of the method introduced in the present paper, it should be clear 
that there is absolutely nothing  to fear in these type of frequency 
localisations for nonlinear applications, provided that they are  indeed used to prove a spacetime localised statement which 
can be expressed in the form~\eqref{simplifiedestimate}. The results of~\cite{DHRteuk, Ma:2017yui2}  and~\cite{RitaShlap} 
can thus in principle be used \emph{directly} for the nonlinear problem, in fact, as ``black box'' results. 
Thus, in our view, given the breakthrough of~\cite{RitaShlap},
the technical complications  in extending the nonlinear Schwarzschild analysis of~\cite{dhrtplus}  
to the full subextremal case  $|a|<M$ of Kerr
may not be as severe as one might naively have thought.

\subsubsection{Harmonic gauge and the weak null condition}
\label{wnc}

In our discussion in Section~\ref{discussionofeinstein} we have focussed above primarily on ``geometric'' gauges,
but as is well known, 
another way to relate~\eqref{theequationzero} to the Einstein equations is via the harmonic gauge
condition (see for instance~\cite{LindRodAnn}, where the nonlinear stability of Minkowski space is proven
in this gauge). 
The resulting reduced equations (for $\psi^{\mu\nu}=g^{\mu\nu}-g^{\mu\nu}_0$)
produce  additional complicated linear terms, and moreover fail to satisfy the null condition, although
they do satisfy the so-called \emph{weak null condition} introduced in~\cite{LindRodweak}.
Note that a linear stability result has been given for this system
in the Schwarzschild case~\cite{Johnson:2018yci} (see also~\cite{2018arXiv180303881H}), and as mentioned earlier, also
on very slowly rotating Kerr  $|a|\ll M$ in~\cite{Hafner:2019kov}.

In order to model the Einstein equations in harmonic gauge,
 classes of equations~\eqref{theequationzero} satisfying the weak null condition 
have  been studied in the recent~\cite{lindblad2016global, lindtoh, lindblad2022weak}.
Though we do not here implement our method in this latter setting,
nonetheless, we emphasise that our analysis
can in principle be extended to equations or systems
 satisfying the weak null condition, following~\cite{Keir:2018qzh}, under an appropriate
  black box assumption for the linearisation.

\subsubsection{Comparison with the $\Lambda>0$ case}
\label{compwithlambdapos}

Finally, it is interesting to compare with the $\Lambda>0$ case.
Here, the analogue of the Schwarzschild and Kerr black holes are the 
 Schwarzschild--de Sitter and Kerr--de Sitter family.  See~\cite{Mihalisnotes}
for a discussion of their geometry.
The study of decay for semilinear and quasilinear
equations of the type~\eqref{theequationzero} on Kerr--de Sitter was pioneered 
by Hintz and Vasy~\cite{hintzvasyglobal}, eventually leading to their groundbreaking
proof of the nonlinear stability of very slowly rotating ($|a|\ll M,\Lambda$)
 Kerr--de Sitter in harmonic gauge. The proof appeals to extensive machinery from
microlocal analysis, with an elaborate compactification of spacetime, and with a Nash--Moser iteration.
For a more recent approach
replacing  Nash--Moser iteration with a global bootstrap, see~\cite{fang}.

Returning to the model equation~\eqref{theequationzero}, an elementary  new
 method was recently introduced by Mavrogiannis~\cite{mavrogiannis, mavrogiannis_2022} 
to treat nonlinear stability on Schwarzschild--de Sitter and Kerr--de Sitter backgrounds.
In~\cite{mavrogiannis}, the entire analysis is reduced
to a ``black box'' estimate for the linear problem on timeslabs of some \emph{fixed} length $L$.
The required black box linear estimate, however, is not just the analogue of the (degenerate)
integrated local energy estimate~\eqref{simplifiedestimate}  (first proven
 in this context in~\cite{Dafermos:2007jd}), 
but a higher order 
refinement of~\eqref{simplifiedestimate} which is \emph{relatively
non-degenerate}, i.e.~where the bulk and boundary term have identical degeneration function and
 thus again are comparable. 
This is possible through a commutation with a
globally defined well chosen vector field orthogonal to the  photon sphere and vanishing at the horizons.
This is an analogue of an energy originally constructed in 
the Schwarzschild case in~\cite{holzegel2020note}. (In the Kerr--de Sitter case, both the degenerate
integrated local energy decay statement and the accompanying relatively non-degenerate statement
are now proven using frequency localisation in a framework similar to~\cite{partiii}.)
Exponential decay is then a derived statement which follows from iterating a suitable estimate
on consecutive slabs, each now of fixed length $L$.

In comparison to the asymptotically flat case studied in the present paper, 
it is noteworthy that in the Kerr--de Sitter case, one does 
\emph{not}  require an additional non-trivial physical space-based ingredient, analogous to~\eqref{simplifiedestimatewithA},
 other than this new, relatively non-degenerate version of the black box linear estimate~\eqref{simplifiedestimate}. 
From the time-slab localised point of view of the two works, the reason for the difference is clear:
Since estimates in~\cite{mavrogiannis} have been reduced to a fixed time scale $L$, 
the role of~\eqref{simplifiedestimatewithA} is essentially provided by \emph{Cauchy stability},
a completely soft statement.

 From this point of view, the difference between
the Schwarzschild and Kerr case on the one hand and their de Sitter analogues on the other is more than simply one between
polynomial and exponential decay, but is one between \emph{dyadically localised} and \emph{truly local}.
In this approach, it is this fundamentally local nature
of the analysis in the de Sitter case that renders 
the \emph{nonlinear} aspects of   stability problems on such de Sitter backgrounds to be essentially soft.

\subsection{Outline of the paper}
We end with an outline of the remainder of the paper.

In Section~\ref{backgroundgeom} we shall describe the geometric assumptions on the background
spacetime, followed in Section~\ref{waveequassump} by the assumptions on properties of the linear inhomogeneous
equation~\eqref{linearhomogeq}.
We shall then introduce in Section~\ref{prelimsec} the class of
nonlinear equations~\eqref{theequationzero}  that we shall consider, 
 stating in particular
a local well posedness theorem and continuation criterion. We shall give the precise statement of the main theorem
in Section~\ref{maintheoremsection}.  The proof will then be carried
out in Section~\ref{estimatehier}.

The paper also contains three appendices.  In Appendix~\ref{howto}, we show how to 
obtain~\eqref{simplifiedestimatewithA} from a physical space identity in the very slowly rotating Kerr case.
In Appendix~\ref{section:appedixcurrents}, we will show how to define the far away currents encoding the $r^p$ method
necessary to extend the estimates~\eqref{simplifiedestimate},~\eqref{simplifiedestimatewithA}
to~\eqref{simplifiedestimatepweight},~\eqref{simplifiedestimatewithApweight}.  
Finally, in Appendix~\ref{nonlineartermsatinf}, we show that our null condition assumption encoded
by the bounds~\eqref{nullcondassumpintro} indeed includes the classical null condition~\cite{KlNull}
and also the more general class of semilinear terms considered on Kerr
in~\cite{MR3082240}. 
In Appendix~\ref{theinhomo}, 
we show explicitly how to obtain the inhomogeneous estimate~\eqref{simplifiedestimate}  
in the Kerr case from the homogeneous estimates of~\cite{partiii}.
Together, these appendices show that our main theorem holds in particular
for a wide class of equations of the form~\eqref{theequationzero} on very slowly rotating $|a|\ll M$
Kerr backgrounds (and, for the semilinear case $g=g_0$, in the full subextremal range $|a|<M$).

\subsection{Acknowledgements}
MD acknowledges support through NSF grant DMS-2005464. GH acknowledges support by the Alexander von Humboldt Foundation in the framework of the Alexander von Humboldt Professorship endowed by the Federal Ministry of Education and Research and ERC Consolidator Grant 772249. 
IR acknowledges support through NSF grants DMS-2005464 and a Simons Investigator Award.
MT acknowledges support through Royal Society Tata University Research Fellowship URF\textbackslash R1\textbackslash 191409.

\section{Geometric assumptions on the background spacetime}
\label{backgroundgeom}

We consider a manifold $\mathcal{M}$ with a background metric $g_0$
satisfying certain assumptions. 
In this section, we collect the assumptions which concern directly the
geometry of $(\mathcal{M},g_0)$. 
The assumptions here will not be the most general possible but are sufficient to include
the main examples of interest. They can be easily further generalised in various directions if desired.
Some of the assumptions are in principle redundant but we have not attempted to derive them from
the most minimal considerations.
We note already
that the assumptions of this section
are modelled on (and are satisfied by---see Section~\ref{seeexamples}!)~the basic cases
of Minkowski space and the (extended) exterior regions of Schwarzschild and subextremal $|a|<M$ Kerr spacetimes.

\subsection{Underlying differential structure and the positive function $r$}
\label{thatwhichunderlies}
For definiteness, we consider the underlying differential structure of $\mathcal{M}$
to be given by $\mathbb R^{4}_{(x^0,x^1,x^2,x^3)}$ or alternatively by
the manifold with boundary $\mathbb R^4\setminus \{ |x^1|^2+|x^2|^2+|x^3|^2< r_0^2\}$.
(The black hole examples in fact correspond to the latter case.) In this latter case, let us define
\begin{equation}
\label{rdefhere}
r: =\sqrt{(x^1)^2+(x^2)^2+(x^3)^2}.
\end{equation}
In the former case where the underlying manifold is $\mathbb R^4$, let us pick an arbitrary $r_0>0$
and define 
\begin{equation}
\label{rdefherealt}
r: = f\left(\sqrt{(x^1)^2+(x^2)^2+(x^3)^2}\right).
\end{equation}
where $f:[0,\infty)\to [r_0,\infty)$ is a smooth function with $f'(v)>0$ such that
$f(z)=\sqrt{r_0^2+z^2}$ near $z=0$ and $f(z)=z$  for $z\ge 2r_0$. 
Thus, in all cases $r$ is a smooth positive function on $\mathcal{M}$
satisfying $r\ge r_0>0$. 
Let us also fix a large $R\ge 20r_0$.

We may also define ambient spherical coordinates in the usual way by the relation
\begin{equation}
\label{standardspherical}
(x^1,x^2,x^3) =(r\cos\varphi\sin\vartheta, r\sin\varphi\sin\vartheta, r\cos\vartheta).
\end{equation}

In the case of non-empty boundary, we will denote by  $\Omega_1=\partial_\varphi$, $\Omega_2$, and $\Omega_3$
the standard rotation vector fields associated to 
the ambient spherical coordinates $(x^0,r,\vartheta,\varphi)$. These are globally regular vector fields.

In the case where the underlying structure is $\mathbb R^{4}$, the above vector fields would only be regular in $r>r_0$.
In this case then, we will use the notation $\Omega_i$, $i=1,\ldots, 3$ for the above standard vector fields multiplied by $\omega(r)$, for a smooth cutoff function 
satisfying  $\omega(r)=1$ for $r\ge R/2$,
and $\omega(r)=0$ in a neighbourhood of $r_0$.
These are again globally regular vector fields.

\subsection{The Lorentzian metric $g_0$, time orientation and the causality of the boundary}
We assume that $g_0$ is a time oriented Lorentzian metric on $\mathcal{M}_0$
such that the coordinate vector field $\partial_{x^0}$ is future directed timelike for $r\ge R/6$.

In the case where $\mathcal{M}$ is a manifold with boundary $\mathcal{S}=\{r=r_0\}$, we assume that 
the boundary $\mathcal{S}$ is spacelike and the future directed normal points out of the spacetime.

\subsection{The stationary Killing field $T$}
Denoting now by $T$ the coordinate vector field 
\[
T=\partial_{x^0}
\]
with respect to the ambient coordinates $(x^0,x^1,x^2,x^3)$,
we assume that $T$ is Killing with respect to the metric $g_0$.
It follows by the previous assumptions
that the ambient function $r$ above is $T$-invariant
and that $T$  is future directed
timelike in the region $r\ge R/6$. Let us further assume that  $g(T,T)\to -1$ as $r\to \infty$.

We will denote by $\phi_\tau$ the $1$-parameter group of diffeomorphisms generated by $T$.

In the case where $\mathcal{S}=\emptyset$, it is natural to already assume that $T$ is globally
timelike, in view of the Friedman instability~\cite{mosxfried} and its evanescent analogue~\cite{keirevan} which applies
in the marginal case, both of which would be incompatible with assumptions we will make later on
concerning the behaviour of waves.

In the case where $\mathcal{S}\neq\emptyset$, we notice that $T$ cannot be globally timelike as $T$
is tangential to and thus  spacelike on $\mathcal{S}=\{r=r_0\}$. 
In this case we will need to assume the existence of a Killing horizon.

\subsection{The Killing horizon $\mathcal{H}^+$}
\label{Killinghorizonsection}

If $\mathcal{S}\ne\emptyset$ we will assume the following further properties:

We assume that there exists an $r_{\rm Killing}$ such that
the hypersurface $\mathcal{H}^+:=\{r=r_{\rm Killing}\}$ is a Killing horizon with future directed null generator $Z=T$, or more generally, with future directed null generator $Z$ in the span 
of $T$ and~$\Omega_1$, in which case we will assume that $\Omega_1$ is globally Killing.  
We will denote by $Z$ the globally defined Killing field given by this combination of $T$ and $\Omega_1$. 

We will assume furthermore that $\mathcal{H}^+$ has strictly positive surface gravity, i.e.~$\nabla_ZZ=\kappa(\vartheta,\varphi) Z$ for some $\kappa>0$
and we will assume that  $T$, or more generally the span of $T$ and $\Omega_1$, is timelike for $r>r_{\rm Killing}$.

Finally, in the case $\mathcal{S}\ne \emptyset$ we will assume that the vector field $\nabla r$ is future directed timelike in the region $r<r_{\rm Killing}$.
Note that $\nabla r$ on $\mathcal{H}^+=\{r=r_{\rm Killing}\}$ will be null and in the direction of $Z$.

\subsection{Foliations, subregions and volume forms}
\subsubsection{The foliation $\Sigma(\tau)$}
\label{thefoliationsection}

We will assume $\mathcal{M}$  admits a  hypersurface $\Sigma_0$ with $2$-dimensional 
corner at $\Sigma_0\cap \{r=R\}$, such that
$\Sigma_0\cap\{r \le R\}$ is strictly spacelike and $\Sigma_0\cap\{r\ge R\}$ is null
and $\Sigma_0$ is transversal to $T$.

We assume that, for $r'\ge R/2$, $\Sigma_0$ is transversal to the hypersurface $\{r=r'\}$, and that
$\Omega_1$, $\Omega_2$, $\Omega_3$ 
 are tangent to $\Sigma_0\cap \{r \ge R/2\}$.
In particular the $2$-dimensional space
$T_p(\Sigma_0\cap \{r \le R\}) \cap T_p(\Sigma_0\cap \{ r\ge R\})$
should coincide with the span of these vectors.

If $\mathcal{S}\ne\emptyset$, we assume that
$\Sigma_0$ is transversal to the hypersurface $\mathcal{S}$, and
$\Sigma_0 \cap\mathcal{S}$ is  diffeomorphic to the $2$-sphere.
We assume finally that the vector field $Z$ of Section~\ref{Killinghorizonsection} is orthogonal to $\Sigma_0\cap \mathcal{H}^+$,
while $\Omega_1$, $\Omega_2$, $\Omega_3$ are tangent to  $\Sigma_0\cap \mathcal{H}^+$.

We assume that
$\Sigma_0$ separates $\mathcal{M}$ into two connected components and that
$\Sigma_0$ is a past Cauchy hypersurface for $J^+(\Sigma_0)$. 

Denoting by $\Sigma(\tau):=(\phi_\tau)_*(\Sigma_0)$ we assume
\[
J^+(\Sigma_0) =\cup_{\tau\ge 0} \Sigma(\tau),
\]
where $J^+$ denotes causal future in $\mathcal{M}$ with respect to the metric $g_0$.

Clearly $\{\Sigma(\tau)\}_{\tau\in \mathbb R}$ defines a foliation of $\mathcal{M}$
and thus defines globally on $\mathcal{M}$ a Lipschitz function $\tau$, which is smooth separately on $r\le R$ and
$r\ge R$. 

Note that by our assumptions above, for $r\ge R/2$, we have that $\tau=\tau(x^0, r)$ and the triple $(r,\vartheta,\varphi)$
represent a smooth coordinate system on $\Sigma(\tau)\cap \{r \ge R\}$  (modulo the spherical degeneration).

We will denote by $L$ a smooth choice of  future directed null generator of $\Sigma_0\cap \{r \ge R\}$ normalised to satisfy
the constraint
$g(L,T)\sim -1$. By translation invariance, this extends to a smooth vector field on all of $\{r \ge R\}$ in the direction
of the null generator of $\Sigma(\tau)$.
(Note that we also use $L$ to denote a general length parameter; in practice these two notations will not interfere with one another.)

\subsubsection{The regions $\mathcal{R}(\tau_0,\tau_1)$}
Let us denote
\[
\mathcal{R}(\tau_0,\tau_1) :=\cup_{\tau_0\tau\le \tau_1} \Sigma(\tau).
\]
We shall refer to such regions as spacetime slabs. We will also use the notation
$\mathcal{R}(\tau_0,\infty) :=\cup_{\tau \ge \tau_0} \Sigma(\tau)$. Note that 
$\mathcal{R}(\tau_0,\infty) = J^+(\Sigma_0)$. 

We shall denote
\[
\mathcal{S}(\tau_0,\tau_1):= \mathcal{S}\cap \mathcal{R}(\tau_0,\tau_1).
\]

\subsubsection{The ingoing cones $\underline{C}_v$ and truncated regions}
We also assume that the region $r\ge R$ is foliated by translation invariant ``ingoing'' null cones
$\underline{C}_v$ parameterised by a smooth function $v$ defined on $r\ge R$, increasing towards the future, which may moreover
be expressed as $v(x^0, r)$. In particular, the vector fields $\Omega_i$ are tangent to $\underline{C}_v$.
We again define  a smooth future directed null generator $\underline{L}$ of $\underline{C}_v$ normalised
by
the constraint $g(\underline{L},T)\sim-1$, and translation invariant;  this defines a smooth vector field on $r\ge R$.

Let us assume moreover that $g(\underline{L},L)\sim -1$ globally in $r\ge R$.

Note that, under the above assumptions, $r$ is constant on $\underline{C}_v \cap \Sigma(\tau)$, and
the future boundary of $\underline{C}_v$ is the set $\{r=R\}\cap \Sigma(\tau(v))$, where this relation defined $\tau(v)$.

If $\tau_0\le \tau_1 \le \tau(v)$,  
we shall denote
\[
\mathcal{R}(\tau_0,\tau_1,v) :=\mathcal{R}(\tau_0,\tau_1)\setminus\bigcup_{\tilde{v}> v} \underline{C}_{\tilde{v}}
\]
and
\[
\Sigma(\tau,v):=\Sigma(\tau)\setminus\bigcup_{\tilde{v}> v} \underline{C}_{\tilde{v}} .
\]

The spacetime region $\mathcal{R}(\tau_0,\tau_1,v)$ is a compact subset of spacetime
with past boundary  $\Sigma(\tau_0,v)$ 
and future boundary $\mathcal{S}(\tau_0,\tau_1)\cup \Sigma(\tau_1,v) \cup 
 \underline{C}_{v}$.

\subsubsection{Comparison of volume forms}
\label{volumecomparisonsection}

We will assume that in the region $r\ge R$,  writing the volume form of $(\mathcal{M},g)$ as
\[
dV_{\mathcal{M}} =\upsilon(r, \vartheta,\varphi) d\tau\,  r^2 dr  \sin\vartheta d\vartheta d\varphi
\]
and, for $\Sigma(\tau)\cap \{r \ge R \}$, with the choice $L$ for the normal, as
\[
dV_{\Sigma(\tau)\cap \{r\ge R\} } := \tilde\upsilon(r,\vartheta,\varphi) r^2 dr \sin\vartheta d\vartheta d\varphi 
\] 
and for  $\underline{C}_v$, with the choice $\underline{L}$ for the normal, as
\[
dV_{ \underline{C}_v } := \tilde{\tilde\upsilon}(r,\vartheta,\varphi) r^2 dr \sin\vartheta d\vartheta d\varphi  ,
\]
then  
\[
\upsilon \sim \tilde\upsilon \sim \tilde{\tilde\upsilon} \sim 1.
\]

With this assumption,  it follows by the coarea formula and 
compactness that 
globally, the volume
form of $(\mathcal{M},g_0)$ is related to the volume form of $\Sigma(\tau)$
\[
dV_{\mathcal{M}} \sim d\tau \, dV_{\Sigma(\tau)},
\]
where $\sim$ is interpreted for $4$-forms in the obvious sense.

Note that when volume forms are omitted from integrals, the above induced volume forms from the metric $g_0$ will be 
understood, unless otherwise noted.

\subsection{Other vector fields}
It will be useful to extend the vector fields defined above to a global set of vector fields which span the tangent space
$T_x\mathcal{M}$ for all $x\in \mathcal{M}$.

\subsubsection{The global extensions of the  vector fields $L$, $\underline{L}$ and $\Omega_i$}
For notational convenience, in the case where $\mathcal{S}= \emptyset$,
let us define $\Omega_4=(1-\omega(r))\partial_{x^1}$, $\Omega_5=(1-\omega(r))\partial_{x^2}$, $\Omega_6=(1-\omega(r))\partial_{x^3}$,
where $\omega$ is as in Section~\ref{thatwhichunderlies},
and let us define $L$ and $\underline{L}$ to be translation invariant extensions of the vector fields defined already with the property
that 
$L$, $\underline{L}$ and $\Omega_1,\ldots, \Omega_6$ span the tangent space globally.
Note that this is easy to satisfy in general, and moreover we are not requiring that $L$ and $\underline{L}$ be null,
and the $\Omega_i$, $i=1,\ldots, 6$ are not required to be tangential to the ambient spheres for $r\le R/2$.
(See already Section~\ref{Minkowexam}.)

In the case where $\mathcal{S}\ne\emptyset$ we will define $L$ and $\underline{L}$ to be smooth $\phi_\tau$-invariant
 extensions of $L$ and $\underline{L}$ to $r_0\le r\le R$, with the property that 
 $L$, $\underline{L}$  everywhere span the  same plane spanned by the coordinate vector fields $\partial_{x^0}$ and $\partial_r$ of
 ambient spherical coordinates $(x^0,r,\vartheta,\varphi)$.
 This again can easily seen to be possible. Note again that we are not requiring these vector fields to be globally null.

\subsubsection{The notation $|\slashed\nabla\psi|^2$}
We define the notation 
\[
|\slashed\nabla \psi|^2 := \sum_{i} r^{-2} |\Omega_i\psi|^2.
\]
We note that under our assumptions, the above expression is comparable to the induced gradient squared
on the spacelike spheres $\Sigma(\tau)\cap \{r=r'\}$  in the region $r\ge R/2$.

In view, however, of our spanning assumptions, in both the case $\mathcal{S}=\emptyset$ and $\mathcal{S}\ne\emptyset$, the expression
\[
|L\psi|^2+|\underline L\psi|^2+ |\slashed\nabla \psi|^2 
\]
will always be a translation invariant
coercive expression on first derivatives of $\psi$, similar (by compactness) to any other such coercive
expression in $r\le R$, while  in the region $r\ge R$ it will have the right weights in $r$ so as to be a suitable
reference point for natural energy expressions.

\subsubsection{The one-forms $\varrho^L$, $\varrho^{\underline{L}}$, $\varrho^1$, $\varrho^2$ and $\varrho^3$}
\label{threeoneformssec}

We will need to introduce a set of weighted spanning $1$-forms on the sphere so as to properly measure weighted boundedness of forms.
Define one forms $\varrho^L$, $\varrho^{\underline{L}}$, and $\varrho^1$, $\varrho^2$, $\varrho^3$ on the far region 
$\mathcal{M} \cap \{ r \ge R\}$ as follows.

The region $\mathcal{M} \cap \{ r \ge R\}$ can be written as the product manifold $\mathcal{M} \cap \{ r \ge R\} = \mathbb{R} \times [R,\infty) \times S^2$. 
 Let $\pi \colon \mathcal{M} \cap \{ r \ge R\} \to \mathbb{R} \times [R,\infty)$ denote the canonical projection.  Let $\varrho^L$, $\varrho^{\underline{L}}$ be defined by $\varrho^L = \pi^* \tilde{\varrho}^L$ and $\varrho^{\underline{L}} = \pi^* \tilde{\varrho}^{\underline{L}}$, where $\tilde{\varrho}^L, \tilde{\varrho}^{\underline{L}}$ is the dual coframe of $\pi_* L, \pi_* \underline{L}$ on $\mathbb{R} \times [R,\infty)$.  It follows that
\[
	\varrho^L(L) = \varrho^{\underline{L}}(\underline{L}) = 1,
	\qquad
	\varrho^{L} (\underline{L}) = \varrho^{\underline{L}}(L) = \varrho^{L}(\Omega_i) = \varrho^{\underline{L}}(\Omega_i) = 0,
	\qquad
	i=1,2,3.
\]

For all smooth functions $\psi$, one can then write
\[
	d \psi = L\psi \, \varrho^L + \underline{L} \psi \, \varrho^{\underline{L}} + \slashed{d} \psi,
\]
where $\slashed{d} \psi (L) = \slashed{d} \psi(\underline{L}) = 0$.

For each $i=1,2,3$, there is a polar coordinate system $(\vartheta^i,\varphi^i)$ on $S^2$, associated to the corresponding rotation vector field $\Omega_i$, with the property that, in the $(x^0,r,\vartheta^i,\varphi^i)$ coordinate system for $\mathcal{M} \cap \{ r \ge R\}$, one has $\partial_{\varphi^i} = \Omega_i$.  Such a $(\vartheta^i,\varphi^i)$ coordinate system on $S^2$ is unique up to a choice of meridian.  Define then, for $i=1,2,3$,
\[
	\varrho^1 = r \sin \vartheta^1 d\varphi^1,
	\qquad
	\varrho^2 = r \sin \vartheta^2 d\varphi^2,
	\qquad
	\varrho^3 = r \sin \vartheta^3 d\varphi^3.
\]
There exist (non-unique, but universal) bounded functions ${a^i}_j(\vartheta,\varphi)$ for $i,j=1,2,3$, such that, for any smooth function 
$\psi \colon \mathcal{M} \cap \{ r \ge R\} \to \mathbb{R}$,
\begin{equation} \label{eq:dpsidualframe}
	d \psi = L\psi \, \varrho^L + \underline{L} \psi \, \varrho^{\underline{L}} + r^{-1} {a^i}_j(\vartheta,\varphi) \, \Omega_i \psi \, \varrho^j.
\end{equation}

\subsubsection{The vector field $Y$}
\label{Yvectorfieldsec}
In the case $\mathcal{S}\ne\emptyset$, 
we  may define $Y$ to be a $\phi_\tau$ invariant vector field such that $Y$ is future directed null on $\mathcal{H}^+$ and transversal to $\mathcal{H}^+$
and orthogonal to $\mathcal{H}^+\cap\Sigma_0$, $Y$ is supported in $r\le r_1+(r_2-r_1)/2$, 
for some $r_2>r_1>r_{\rm Killing}$,
and such that $Z$, $Y$, $\Omega_i$ span the tangent space
in $r\le r_1+(r_2-r_1)/4$.

 The existence of such a vector field in a neighbourhood of a Killing horizon
follows from~\cite{Mihalisnotes}.

In the case where $\mathcal{S}=\emptyset$, since we are assuming that $T$ is globally timelike,
we may simply set $Y=0$.

\subsection{Examples: Minkowski, Schwarzschild and Kerr}
\label{seeexamples}

We note that Schwarzschild and Kerr in the full subextremal black hole range of parameters $|a|< M$
satisfy the assumptions of this section with an appropriate definition of the underlying differential structure.

\subsubsection{Minkowski}
\label{Minkowexam}
For the Minkowski case, we consider the underlying manifold to be $\mathbb R^4$, i.e.~without boundary, and
 we define the metric to be the familiar expression
 \[
 g_0 =-(dx^0)^2 + (dx^1)^2+(dx^2)^2+(dx^3)^2.
 \]
We define $r_0:=1$ and this determines the function $r$. 
Let us distinguish this from $\tilde{r}:=\sqrt{(x^1)^2+(x^2)^2+(x^3)^2}$ which we may think of
as a function $\tilde{r}(r)$.
Fixing any $R\ge 20$, we note that $u=t-\tilde{r}$, $v=t+\tilde{r}$ are null coordinates.
We may define $L=\partial_u$, $\underline{L}=\partial_v$, with respect to coordinates $(u,v,\vartheta,\varphi)$,
 in the region $r\ge 2$.

We may take $\Sigma_0 =(\{t=0\}\cap \{r\le R\})\cup (\{  v=\tilde{r}(R)\} \cap\{ r\ge R\})$
and $\underline{C}_v$ defined by the level sets of $v$. 

Note that $r/\tilde{r}\sim 1$ for $r\ge r_0$ and that 
the spacetime volume forms is $\tilde{r}^2 \sin\vartheta d\tilde{r} d\tau d\vartheta d\varphi$,
while the volume form on $\Sigma_0\cap \{r\ge R\} $ and 
$\underline{C}_v$  is $\tilde{r}^2\sin\vartheta d\tilde{r} d\vartheta d\varphi$
with our choice of $L$ and $\underline{L}$.

Note also in this case, we define $\Omega_4=\eta(r) \partial_{x^1}$, $\Omega_5=\eta(r) \partial_{x^2}$
and $\Omega_6=\eta(r) \partial_{x^3}$  where $\eta$ is a cutoff vanishing with $\eta=1$ for $r\le \frac14$ and $\eta=0$ for $r \ge\frac12 $,
and we can extend $L$ and $\underline L$ simply by $L=(1-\eta)\partial_u+ \eta\partial_v$, 
$\underline{L}=(1-\eta)\partial_v$. (We may then define the cutoff $\omega$ of Section~\ref{thatwhichunderlies} so that $\omega=1$ 
for $r\ge \frac14$ and $\omega=0$
for $r\le \frac18$.)
We have of course $Y=0$ in this case.

\subsubsection{Schwarzschild}
For the Schwarzschild case, 
given a real parameter $M>0$, we may define
$r_0:=(2-\delta_1) M$ for sufficiently small  $0<\delta_1\ll M$, denote $x^0$ by $t^*$,
let $r$ be defined by~\eqref{theequationzero} and standard spherical coordinates $(\vartheta,\varphi)$ on $\mathbb R^3$ 
by~\eqref{standardspherical}
and define the metric
$g_M$ to be
\[
-(1-2M/r) (dt^*)^2 +4M/r dr dt^* + (1+2M/r)dr^2+r^2(d\vartheta^2+\sin^2\vartheta d\varphi^2)
\]
We define the vector fields $\Omega_i$ to be the standard Killing vector fields associated to  spherical symmetry
and we define $T=\partial_{t^*}$ to be the coordinate vector field with respect to the above coordinates.

Note that $r=2M$ is a Killing horizon with null generator $T$ and positive surface gravity, and the hypersurfaces
$r=r'$ for $r'\in [r_0,2M)$ are indeed spacelike.

We may define the function $t$ in the region $r>2M$ by $t= t^* -2M\log(r-2M)$.
We note that in the coordinates $(t,r,\theta,\phi)$, the metric takes the familiar form
\[
-(1-2M/r) dt^2+ (1-2M/r)^{-1}dr^2+ r^2(d\vartheta^2+\sin^2\vartheta d\varphi^2).
\]
We may define
$r^*=r+2M\log (r-2M)$ and we may define coordinates $u$ and $v$ in $r>2M$ by $u=t-r^*$, $v=t+r^*$.

We may fix  $R\ge 20 r_0$ and define now $L=\partial_v$ and $\underline{L}=\partial_u$ to be the coordinate vector fields
with respect to  $(u,v,\theta,\phi)$ coordinates in $r\ge R$. We may extend these to  be globally
defined linearly independent translation invariant vector fields in the span of the coordinate vector fields $\partial_{t*}$ and $\partial_r$.

We then may define $\Sigma_0 =(\{t^*=0\}\cap \{r\le R\})\cup (\{  v=R\} \cap\{ r\ge R\})$  and this satisfies all the required transversality properties, etc.   

Finally, we may define for instance $Y$ to be a translation invariant  vector field in the
span of  $\partial_{t*}$ and $\partial_r$ which is null and future directed and satisfies $g(\partial_{t*}, Y)=-2$
at $r=2M$, such that $Y$ is supported entirely in $r\le r_1+(r_2-r_1)/2$, with $r_1:=(2+\delta_1)M$.

\subsubsection{Kerr}
\label{kerrsubsubsection}
To put Kerr in our preferred form, one uses a combination of
Kerr star coordinates (based on Boyer--Lindquist coordinates) and double null coordinates. We leave the details to the reader, but note
the following. 

Given subextremal parameters $|a|<M$, and defining $r_+=M+\sqrt{M^2-a^2}$,
we set $r_0=r_+-\delta_1$ for a small $\delta_1$. 

One may define the ambient differential structure so that the $r$ of~\eqref{rdefhere} will coincide with the Boyer--Lindquist $r$ of Appendix~\ref{howto} in the region $r\le R/3$,
while for $r\ge R/2$ it will coincide with the coordinate $r_*$ of Appendix~\ref{section:Kerrcase}. The coordinate
$v$ coincides with the double null $v$ of Appendix~\ref{section:Kerrcase}, and 
$\Sigma(\tau)\cap \{r\ge R\}$ will be a hypersurface of constant $u$, where again $u$
is as defined in Appendix~\ref{section:Kerrcase}.

Further, one may set things up so that $\Omega_1=\partial_\varphi$ is the axisymmetric Killing field.

Let us note finally that extremal  Kerr (corresponding to $|a|=M$) cannot
in fact be put into our preferred form already because of our
requirement that $T$ be tangential to $\mathcal{S}$ and thus spacelike. Note also that the Killing horizon of extremal Kerr 
has zero surface gravity, which would also contradict our assumptions of Section~\ref{Killinghorizonsection}.

\subsection{Table of $r$-parameters}
\label{tableofrparameters}
We collect finally a list of important $r$-values in increasing order. Some will only be introduced later
in the paper:
\begin{center}
\begin{tabular}{ |c|c| } 
\hline
$r_0$ & $r\ge r_0$; moreover if $\mathcal{S}\ne\emptyset$, then $\mathcal{S}=\{r=r_0\}$ \\
\hline
$r_{\rm Killing}$ & $\mathcal{H}^+=\{r=r_{\rm Killing}\}$ a Killing horizon with positive surface gravity;\\
&  span of $T$ and $\Omega_1$ is timelike for $r> r_{\rm Killing}$ \\
\hline
$r_1$  & parameter related to the vector field $Y$  if $\mathcal{S}\ne\emptyset$ \\
\hline
$r_1+(r_2-r_1)/4$ &  $Z$, $Y$, $\Omega_i$ span the tangent space for $r_0\le r_1+(r_2-r_1)/4$\\
\hline
$r_1+(r_2-r_1)/2$ &  commutation vector fields $\mathfrak{D}$  all Killing for $r_1+(r_2-r_1)/2\le r\le R/2$\\
\hline
$r_2$ &  $\rho=1$, $\chi=1$ for $r\le r_2$\\
\hline
$R/6$ &   $T$ timelike for $r\ge R/6$ \\
\hline
$R/4$ &   $\rho=1$, $\chi =1$ for $r\ge R/4$ \\
\hline
$R/2$ &  commutation vector fields $\mathfrak{D}$  all Killing for $r_1+(r_2-r_1)/2\le r\le R/2$\, ; \\
& $g=g_0$ for $r\ge R/2$ \\
\hline
$8R/9$ &   the generalised null condition assumption concerns $r\ge 8R/9$  \\
\hline
$R$ &   $\Sigma(\tau)\cap \{r\ge R\}$ is null,  $\underline{C}_v\subset \{r\ge R\}$ \\
\hline
$\tilde{R}$ &   parameter related to positivity properties of far-away currents \\
\hline
\end{tabular}
\end{center}

The parameters $r_{\rm Killing}$, $r_1$, $r_2$ only occur if $\mathcal{S}\ne \emptyset$.

\section{Assumed identities and estimates for $\Box_{g_0}\psi=F$}
\label{waveequassump}

Our fundamental assumptions in this paper are connected with the behaviour of  solutions of the linear
inhomogeneous equation~\eqref{linearhomogeq} on the exactly stationary background $g_0$.

\subsection{Constants and parameters}
\label{constantsandparameterssec}
Before stating assumptions, we make a remark concerning constants and parameters.

Given a spacetime $(\mathcal{M},g_0)$ satisfying the assumptions of Section~\ref{backgroundgeom},
we will consider the parameters of Section~\ref{tableofrparameters} as fixed.
Let us also  fix once and for all a 
\begin{equation}
\label{fixeddelta}
0<\delta<\frac1{10}.
\end{equation}

We will use $k$ to denote integers $\ge0$ which will parametrise number of derivatives.

In inequalities, 
we will denote by $C$ and $c$ generic positive constants depending only on (a) $(\mathcal{M},g_0)$
(with the choice of $r$-parameters), (b) the above choice of $\delta$, and (c)
if there is $k$-dependence
in the relevant statement, also on $k$.
(We use $C$ for large constants and $c$ for small constants.)

For nonnegative quantities $A$ and $B$, the notation 
\[
A\lesssim B
\]  
means $A\le CB$, while 
\[
A\sim B
\]
means $cB\le A\le CB$.

The reader should be aware to distinguish between $\leq$ and $\lesssim$,
as both will appear!

For a discussion of additional smallness parameters depending also on the nonlinearity, see already Section~\ref{explanofsmall}.

\subsection{Basic degenerate integrated local energy estimate}
\label{basicblackboxestimatesec}

As discussed in the introduction, our basic assumption will be that of a (degenerate)
spacetime-localised integrated local energy decay statement for the inhomogeneous equation~\eqref{linearhomogeq}.

The statement is that
a certain energy flux on $\Sigma(\tau)$
plus a
 nonnegative bulk are controlled by an initial energy flux on $\Sigma(\tau_0)$ and a spacetime
integral over the slab $\mathcal{R}(\tau_0,\tau)$  relating to the inhomogeneous term $F$. The controlled
 bulk integral is allowed to degenerate where a certain degeneration function $\chi= \chi(r)$
vanishes, but is assumed to control a zeroth order term without degeneration (with decaying weight at infinity).

Let us thus define
\begin{equation}
\label{chidef}
\chi:[r_0,\infty)\to [0,1]
\end{equation}
to be a function such that $\chi=1$ for $r\ge R/4$ and $r\le r_2$.

The assumed estimate is:
\begin{equation}
\begin{aligned}
\label{inhomogeneous}
\sup_{v: \tau\le \tau(v)}
\Fzero(v,\tau_0,\tau),\qquad& \Ezero(\tau) + c \Ezero_{\mathcal{S}}(\tau) 
+ c\int_{\tau_0}^{\tau}{}^\chi\Ezerominusoneminusdelta'(\tau')
d\tau' +c \int_{\tau_0}^{\tau} \Eminusone' (\tau') d\tau' \\
&\leq \lambda \Ezero(\tau_0)+ C \int_{\mathcal{R}(\tau_0,\tau)}\left|(V_0^\mu\partial_\mu \psi +w_0\psi)   F\right| +C \int_{\mathcal{R}(\tau_0,\tau)}F^2,
\end{aligned}
\end{equation}
for some constants $\lambda\ge 1$,  $C\ge1$, $0<c<1$,
and where
$V_0$ is a fixed vector field and $w_0$ a fixed function satisfying
\begin{equation}
\label{heresomeboundssat}
|g_0(V_0,L)| \lesssim 1, \qquad  |g_0(V_0,\underline{L})|\lesssim 1, \qquad  \sum |g_0(V_0,\Omega_i)|^2  \lesssim 1, \qquad
|w_0| \lesssim r^{-1}.
\end{equation}
Here, the unprimed energies are defined by 
\begin{align}
\label{firstdef}
\Ezero(\tau)&:=  \int_{\Sigma(\tau)} |L\psi|^2+|\slashed\nabla\psi|^2 + \iota_{r\le R} |\underline{L}\psi|^2+r^{-2}\psi^2,\\
\label{extradef}
\Ezero_{\mathcal{S}}(\tau)&:=  \int_{\mathcal{S}(\tau_0,\tau)} |L\psi|^2+|\slashed\nabla\psi|^2 + |\underline{L}\psi|^2+\psi^2,\\
\label{seconddef}
\Fzero(\tau_0,\tau,v)&: = \int_{\underline{C}_v\cap \mathcal{R}(\tau_0,\tau)} |\underline{L}\psi|^2+|\slashed\nabla\psi|^2+r^{-2}\psi^2,
\end{align}
while the primed energies are defined by
\begin{align}
\nonumber
{}^\chi\Ezerominusoneminusdelta'(\tau)&:=
 \int_{\Sigma(\tau)} 
  r^{-1-\delta}\chi(r) (|L\psi|^2+|\underline{L}\psi|^2+|\slashed\nabla\psi|^2),\\
\label{noteherelower}
\Eminusone'(\tau)&:=\int_{\Sigma(\tau)} r^{-3-\delta}\psi^2.
\end{align}
(The $'$ notation will be used in general to denote energies that naturally appear in bulk terms.)
The estimate~\eqref{inhomogeneous} is to hold for all smooth $\psi$ such that the right hand side
of~\eqref{inhomogeneous} is finite.  In the above, we already see the~$\delta$ fixed in~\eqref{fixeddelta}.
For future reference let us also define the quantity
\begin{align}
\label{refertothislaternow}
\Ezerominusoneminusdelta'(\tau)&:=
 \int_{\Sigma(\tau)} 
r^{-1-\delta} (|L\psi|^2+|\underline{L}\psi|^2+|\slashed\nabla\psi|^2) +r^{-3-\delta}\psi^2.
\end{align}
Note that if $\chi=1$ identically, then ${}^\chi\Ezerominusoneminusdelta'+ \Eminusone' =\Ezerominusoneminusdelta'$.
We note that all expressions defined above are $T$-invariant.

We distinguish the constants $C$ and $\lambda$ in~\eqref{inhomogeneous} to highlight the significance
of the case 
$\lambda=1$,  when~\eqref{inhomogeneous} is derived from a
suitable energy identity and the energies are replaced by exact fluxes.
See already Section~\ref{caseidiscussion}.

Note that in cases (i) and (ii), we shall
see that~\eqref{inhomogeneous} holds, and in fact one may drop the $\int F^2$ term
on the right hand side.

In fact, in all of the above cases more precise estimates than~\eqref{inhomogeneous} are available with
respect to what can actually be controlled by the right hand side; we shall not need to make use of
this here.

Again, the assumptions are motivated by our model cases of Minkowski space, Schwarzschild
and sub-extremal Kerr black hole exteriors.
The estimate~\eqref{inhomogeneous} indeed holds in the case of Kerr  in the full subextremal range $|a|<M$ (this is the statement of Theorem~\ref{trueforKerr} of Section~\ref{theinhomo}).
We note, however, the following remark:
\begin{remark}
\label{weakerassumption}
In our arguments in subsequent sections,  we will in fact only use the statement that 
follows from~\eqref{inhomogeneous} by
partitioning the middle term on the right hand side into the regions $r\le R$ and $r\ge R$,
and
 applying Cauchy--Schwarz to the former, replacing thus the full middle term
 with the expression
 \begin{equation}
\label{replacedterm}
\sqrt{
\int_{\mathcal{R}(\tau_0,\tau_1)\cap \{ r\le R \}} |L\psi|^2+|\underline L\psi|^2 +|\slashed\nabla\psi|^2 
+r^{-2}|\psi|^2}
\sqrt{\int_{\mathcal{R}(\tau_0,\tau_1)\cap \{ r\le R \}} F^2 
}
+
 \int_{\mathcal{R}(\tau_0,\tau)\cap \{r\ge R\}}\left|(V_0^\mu\partial_\mu \psi +w_0\psi)   F\right| .
\end{equation}
We could thus alternatively consider this weaker statement, i.e.~\eqref{inhomogeneous} but with~\eqref{replacedterm} replacing the middle term on the right hand side,
as our main assumption.  We prefer, however, to
 keep our general assumption in the form~\eqref{inhomogeneous} because it is more
compact and we do not know of a physical example where this weaker
assumption holds but our original assumption does not.  On the other hand,
it is often easier to prove the weaker assumption
directly than to prove~\eqref{inhomogeneous}. This is indeed the 
case  for Kerr in the full subextremal range $|a|<M$, where in 
  Section~\ref{theinhomo} we will in fact 
only give a proof of this weaker version.
\end{remark}

\subsection{Physical space identities on a general Lorentzian metric}
\label{letsgetphysical0}
As explained in the introduction, for nonlinear applications to~\eqref{theequationzero}, it is most convenient 
when~\eqref{inhomogeneous} is the result of a \underline{physical space energy} identity~\eqref{introphysspace}
for solutions $\psi$ of~\eqref{linearhomogeq}. Thus, we will recall some general properties of such identities here.

\subsubsection{Definition of energy currents}
Physical space energy  identities can be associated to a quadruple $(V,w,q,\varpi)$ where $V^\mu$ is a vector field on $\mathcal{M}$, $w$ is a scalar function, $q_\mu$ is a one-form and $\varpi_{\mu\nu}$ is a two-form.
Given such a quadruple, a general Lorentzian metric $g$ and a suitably regular function $\psi$, we define
\begin{align}
\label{generalJdef}
J^{V,w,q,\varpi}_\mu [g, \psi] &:= T_{\mu\nu} [g,\psi] V^\nu + w  \psi \partial_\mu \psi + q_\mu \psi^2  
+ * d \big( \psi^2 \varpi \big)_{\mu}
\\
\label{generalKdef}
K^{V,w,q}[g, \psi]  &:=  \pi^{V}_{\mu\nu} [g]  T^{\mu\nu}[g, \psi]  +  \nabla^\mu w \psi \partial_\mu \psi +w\nabla^\mu\psi \partial_\mu \psi +  \nabla^\mu q_\mu \psi^2 + 2\psi q_\mu g^{\mu\nu}
\partial_\nu \psi \\
\label{generalHdef}
H^{V,w}[\psi] &:= V^\mu\partial_\mu \psi + w\psi , 
\end{align}
where
\[
T_{\mu\nu} [g,\psi] =  \partial_\mu\psi\partial_\nu\psi-\frac12 g_{\mu\nu} g^{\alpha\beta}\partial_\alpha\psi\partial_\beta\psi,
\qquad \pi^X_{\mu\nu} [g] = \frac12 (\mathcal{L}_X g)_{\mu\nu} = \frac12(\nabla_\mu X_\nu+\nabla_\nu X_\mu),
\]
and where $* \colon \Lambda^3 \mathcal{M} \to \Lambda^1 \mathcal{M}$ denotes the Hodge star operator.

Let us note already that it is often more natural to parametrise choices of currents in a slightly different way,
as twisted currents~\cite{HOLZEGEL20142436}, for instance, given a scalar function, one may define a twisted energy momentum
tensor $\tilde{T}_{\mu\nu}[g,\psi]$ according to~\eqref{twistedT} and then consider for instance currents of the form
$\tilde{J}^V_\mu [g,\psi]:=\tilde{T}_{\mu\nu}[g,\psi] V^\nu$, etc. 
Such a current can always be rewritten as~\eqref{generalJdef} for some $w$, $q$, $\varpi$. 

We note that the additional $\varpi$ component in~\eqref{generalJdef} arises
naturally for twisted currents and  is useful in generating positive zeroth order flux terms on the boundary.

\subsubsection{The divergence identity}
With the above definitions, one can compute the following identity for a general function $\psi$:
\begin{equation}
\label{dividentitynotintegrated}
\nabla_{g}^\mu J^{V,w,q, \varpi}_{\mu} [g,\psi] = K^{V,w,q}[g,\psi ] + H^{V,w} [\psi] \Box_{g} \psi  .
\end{equation}

In particular, for solutions of the covariant wave equation $\Box_g\psi=0$, one obtains a divergence relation
between the currents $J$ and $K$ which both depend only on the $1$-jet of $\psi$. 
See~\cite{christodoulou2016action} for a discussion of the classification of currents with this property.
Note that $\varpi$ does not contribute to the bulk current $K$, and neither $\varpi$ nor $q$ contribute
to $H$, which moreover is independent of the metric $g$.

The significance of identity~\eqref{dividentitynotintegrated} is that it can be integrated in a spacetime
region bounded by homologous  hypersurfaces to obtain a relation between boundary fluxes
of $J$ and a bulk integral of $K$. We give the  form of this relation below in the special case
of the region $\mathcal{R}(\tau_0,\tau_1,v)\subset \mathcal{M}$.

\subsubsection{The integrated identity on $\mathcal{R}(\tau_0,\tau_1,v)$}

For a solution of~\eqref{linearhomogeq}, identity~\eqref{dividentitynotintegrated}
upon integration in $\mathcal{R}(\tau_0,\tau_1,v)$ yields
\begin{align}
\nonumber
\int_{\Sigma(\tau_1,v)} J[\psi]\cdot n +\int_{\mathcal{S}(\tau_0,\tau_1)}J[\psi]\cdot n +
\int_{\underline{C}_v(\tau_0,\tau_1)}J[\psi]\cdot n +
\int_{\mathcal{R}(\tau_0,\tau_1,v)} K[\psi] \\
\label{energyidentity}
=
\int_{\Sigma(\tau_0,v)} J[\psi]\cdot n-\int_{\mathcal{R}(\tau_0,\tau_1,v)} H[\psi]  F,
\end{align}
where the normals and volume forms are with respect to the metric $g_0$ according
to the convention of Section~\ref{volumecomparisonsection}.

Identity~\eqref{energyidentity} will be useful when it satisfies suitable \underline{positivity properties for its bulk and boundary terms}. 

\subsection{Assumed unweighted first order physical space identities: cases (i)--(iii)}
\label{letsgetphysical}

As discussed already in the introduction, we shall \underline{not} in general require that~\eqref{inhomogeneous} is the result of
an integrated divergence identity~\eqref{energyidentity} associated to currents with a pointwise coercivity property.
We shall, however, require, \emph{in addition} to assuming~\eqref{inhomogeneous},
 the existence of currents generating an identity~\eqref{energyidentity} with much 
weaker non-negativity properties, properties which in particular are insensitive to the presence and structure of possible trapping.

It is indeed useful to see first, however,
how the existence of a physical space    proof of~\eqref{inhomogeneous}
via an identity of type~\eqref{energyidentity}  simplifies the considerations.
Thus, we shall distinguish two simpler cases, to be called case (i) and (ii),
to be discussed in Sections~\ref{caseidiscussion} and~\ref{caseiidiscussion} below, where 
indeed~\eqref{inhomogeneous} is proven via an identity~\eqref{energyidentity},
(with case (i) corresponding to the even simpler setting where there is no degeneration
at all in the coercivity).

The most general case, case (iii), which represents the main goal of this paper,
will be discussed in Section~\ref{caseiiiidentity}.
We will introduce some helpful common notation in Section~\ref{sumaryonetwothree}. 

Finally, in Section~\ref{enhanced}, for the case $\mathcal{S}\ne\emptyset$,
we will provide a family of currents with enhanced red-shift control at $\mathcal{H}^+$ and in the black hole interior, parametrised
by a parameter $\varsigma$,
which will be useful for obtaining higher order estimates in Section~\ref{higherorderestimates}.

\subsubsection{Case (i)}
\label{caseidiscussion}

The simplest case to consider is when~\eqref{inhomogeneous} indeed
follows from a physical space energy identity~\eqref{energyidentity},
and when moreover, there is in fact no degeneration in the estimate~\eqref{inhomogeneous}, i.e.~the
function $\chi$ of~\eqref{chidef} satisfies $\chi=1$ identically.

That is to say,  we assume that there exists a $T$-invariant quadruple $(V,w,q,\varpi)$,
where $V$ is a vector field, $w$ is a scalar function, $q$ is a one form and $\varpi$ is a two form, satisfying
\begin{equation}
\begin{aligned}
\label{seedboundedness}
|g_0(V,L)| \lesssim 1, \qquad  |g_0(V,\underline{L})|\lesssim 1, \qquad  \sum |g_0(V,\Omega_i)|^2  \lesssim 1, \\
|w| \lesssim r^{-1}, \qquad
|q_\mu L^\mu| \lesssim r^{-2}, \qquad |q_\mu \underline{L}^\mu|\lesssim r^{-2}, \\
	\vert (*\, d\varpi)_{\mu} L^{\mu} \vert \lesssim r^{-2},
	\qquad
	\vert (*\, d\varpi)_{\mu} \underline{L}^{\mu} \vert \lesssim r^{-2},\\
	\vert *(\varrho^{L} \wedge \varpi)_{\mu} L^{\mu} \vert \lesssim r^{-1},
	\quad
	\vert *(\varrho^{L} \wedge \varpi)_{\mu} \underline{L}^{\mu} \vert \lesssim r^{-1},\\
	\vert *(\varrho^{\underline{L}} \wedge \varpi)_{\mu} L^{\mu} \vert \lesssim r^{-1},
	\quad
	\vert *(\varrho^{\underline{L}} \wedge \varpi)_{\mu} \underline{L}^{\mu} \vert \lesssim r^{-1},
\\
	\vert *(\varrho^{i} \wedge \varpi)_{\mu} L^{\mu} \vert \lesssim r^{-1},
	\quad
	\vert *(\varrho^{i} \wedge \varpi)_{\mu} \underline{L}^{\mu} \vert \lesssim r^{-1},
	\qquad
	i = 1,2,3,
\end{aligned}
\end{equation}
 such that,
defining $J^{V,w,q,\varpi}[\psi]$, $K^{V,w,q}$  by~\eqref{generalJdef}--\eqref{generalKdef}, 
the energy identity~\eqref{energyidentity} corresponding to these
currents has the following properties: 
(a) the boundary terms of~\eqref{energyidentity} on $\Sigma(\tau)$ are coercive,
(b) the remaining boundary terms nonnegative,
and (c)  the  bulk term $K$ is nonnegative and coercive, with no degeneration, except ``at infinity'',
where standard derivatives are in general only
controlled with weight $r^{-1-\delta}$. 

More precisely, we assume the following pointwise bulk coercivity relation
\begin{equation}
\begin{aligned}
\label{insymbolsi}
K^{V,w,q}[\psi] &\gtrsim r^{-1-\delta}\left((L\psi)^2+(\underline L\psi)^2 +|\slashed\nabla\psi|^2\right)+r^{-3-\delta}\psi^2,
\end{aligned}
\end{equation}
and pointwise boundary coercivity relations:
\begin{equation}
\begin{aligned}
\label{insymbolsiboundary}
\qquad J^{V,w,q,\varpi}_\mu[\psi] n_{\Sigma(\tau)}^\mu &\gtrsim (L\psi)^2 +|\slashed\nabla \psi|^2+\iota_{r\le R}( \underline{L}\psi)^2+r^{-2}\psi^2,\\
J^{V,w,q,\varpi}_\mu[\psi] n_{\underline{C}_v}^\mu &\gtrsim (\underline{L}\psi)^2 +|\slashed\nabla \psi|^2 +r^{-2}\psi^2
 ,\\
J^{V,w,q,\varpi}_\mu[\psi] n_{\mathcal{S}}^\mu &\gtrsim \left( (L\psi)^2+(\underline L\psi)^2 +|\slashed\nabla\psi|^2+\psi^2 \right) .
\end{aligned}
\end{equation}
Here  the normals are of course taken with respect to the metric $g_0$.

Let us note that the boundary coercivity statement on $\Sigma(\tau)$ can only possibly hold if $V$ is timelike on $\Sigma(\tau)$.

Defining
\begin{equation}
\begin{aligned}
\label{frakturenergies}
\Efancyzero(\tau)&:= \int_{\Sigma(\tau)}  J^{V,w,q,\varpi}_\mu[\psi] n_{\Sigma(\tau)}^\mu , \qquad
\Efancyzero_{\mathcal{S}}(\tau ) := \int_{\mathcal{S}} J^{V,w,q,\varpi}_\mu[\psi] n_{\mathcal{S}}^\mu \, ,\\
\Ffancyzero(v,\tau_0,\tau_1) &:=\int_{\underline{C}_v\cap \mathcal{R}(\tau_0,\tau_1)} J^{V,w,q,\varpi}_\mu[\psi] n_{\underline{C}_v}^\mu,
\end{aligned}
\end{equation}
it follows from~\eqref{insymbolsiboundary} that
\begin{equation}
\label{theyrelessthansimilar}
\Ezero(\tau)\lesssim \Efancyzero(\tau), 
\qquad 
\Ezero_{\mathcal{S}}(\tau) \lesssim \Efancyzero_{\mathcal{S}}(\tau),
\qquad  \Fzero(v,\tau_0,\tau_1)\lesssim \Ffancyzero(v, \tau_0,\tau_1).
\end{equation}
From~\eqref{seedboundedness}, 
it follows that in addition to the coercivity~\eqref{insymbolsiboundary}, we have the corresponding boundedness
\begin{equation}
\begin{aligned}
\label{infacttheyresim}
J^{V,w,q,\varpi}_\mu[\psi] n_{\Sigma(\tau)}^\mu &\lesssim (L\psi)^2 +|\slashed\nabla \psi|^2+\iota_{r\le R}( \underline{L}\psi)^2+r^{-2}\psi^2,\\
J^{V,w,q,\varpi}_\mu[\psi] n_{\underline{C}_v}^\mu &\lesssim (\underline{L}\psi)^2 +|\slashed\nabla \psi|^2 +r^{-2}\psi^2 ,\\
J^{V,w,q,\varpi}_\mu[\psi] n_{\mathcal{S}}^\mu &\lesssim  (L\psi)^2+(\underline L\psi)^2 +|\slashed\nabla\psi|^2+\psi^2  ,
\end{aligned}
\end{equation}
and thus
\begin{equation}
\label{theyresimilar}
\Ezero(\tau)\sim \Efancyzero_r(\tau), \qquad 
\Ezero_{\mathcal{S}}(\tau) \sim \Efancyzero_{\mathcal{S}}(\tau),
\qquad  \Fzero(v,\tau_0,\tau_1)\sim \Ffancyzero(v,\tau_0,\tau_1).
\end{equation}
Note that to estimate the boundary terms arising from the two-form
$\varpi$, we have used~\eqref{eq:dpsidualframe} and the fact that
\[
	*\, d(\psi^2 \varpi)
	=
	2 \psi *(d \psi \wedge \varpi) + \psi^2 * d \varpi.
\]

Under the above assumptions,
in the notation of Section~\ref{basicblackboxestimatesec} (recall now~\eqref{refertothislaternow}), the 
identity~\eqref{energyidentity}
 gives rise
to the estimate:
\begin{equation}
\begin{aligned}
\label{firstinstance}
\sup_{v: \tau\le \tau(v)}
\Ffancyzero(v,\tau_0,\tau)   ,\qquad& \Efancyzero(\tau) + \Efancyzero_{\mathcal{S}}(\tau)+c \int_{\tau_0}^{\tau_1}\Ezerominusoneminusdelta'(\tau')d\tau' 
\leq \Efancyzero(\tau_0)+ \int_{\mathcal{R}(\tau_0,\tau_1)}
\left|H[\psi] F\right|.
\end{aligned}
\end{equation}
We note that, in view of~\eqref{theyrelessthansimilar},~\eqref{theyresimilar}
and the fact that we may express
\[
H[\psi] = V^\mu\partial_\mu\psi +w\psi,
\]
 this gives~\eqref{inhomogeneous}, for some $\lambda\ge 1$ and
 without degeneration, i.e.~with $\chi=1$, and with $V_0=V$ and $w_0=w$,  and where 
  including the final
 $\int F^2$ term in~\eqref{inhomogeneous} is here unnecessary.
The point of expressing the estimate in terms of the fraktur energies~\eqref{frakturenergies} is that~\eqref{firstinstance} is
a sharper  statement than~\eqref{inhomogeneous}, corresponding to $\lambda=1$,
which will be useful for us.

Let us note immediately that Minkowski space itself, but also suitably small stationary perturbations of the Minkowski metric, satisfy
the assumptions of this section (see Appendix~\ref{section:appedixcurrents}). More generally, given a metric $g_0$
as in Section~\ref{backgroundgeom} and a $T$-invariant quadruple $(V,w,q,\varpi)$ satisfying~\eqref{seedboundedness},
whose energy currents $J^{V,w,q, \varpi}[g_0,\psi]$, $K^{V,w,q}[g_0,\psi]$ 
satisfy the above coercivity properties~\eqref{insymbolsi},~\eqref{insymbolsiboundary} and~\eqref{infacttheyresim},
it is clear that $J^{V,w,q, \varpi}[g_\epsilon,\psi]$, $K^{V,w,q}[g_\epsilon,\psi]$ retain
the coercivity properties~\eqref{insymbolsi},~\eqref{insymbolsiboundary} and~\eqref{infacttheyresim} for
any stationary perturbation $g_\epsilon$ of $g_0$ satisfying the assumptions of Section~\ref{backgroundgeom}, 
sufficiently close to $g$, such that $g=g_\epsilon$ in $r\ge R$. Thus, we see that when,
 as in the present section, estimate~\eqref{inhomogeneous} 
is proven via~\eqref{firstinstance},
and there
is no degeneration, i.e.~$\chi=1$, then the estimate~\eqref{inhomogeneous} 
can immediately  be inferred to be stable to suitably small perturbations
of the metric~$g_0$.

Unfortunately, however, for most of our examples of spacetimes $(\mathcal{M},g_0)$ of interest, it turns out that
the assumptions of the present section cannot in fact hold.
Specifically, if $(\mathcal{M},g_0)$ contains  trapped null geodesics, then one can show
that~\eqref{inhomogeneous} cannot hold with $\chi=1$ identically, and thus,
no quadruple $(V,w,q,\varpi)$ as above can exist satisfying~\eqref{insymbolsi}; 
see~\cite{SbierskiGauss}. 
In particular, the assumptions of this section do not encompass
 the black hole cases of interest like Schwarzschild or Kerr.

\subsubsection{Case (ii)}
\label{caseiidiscussion}

The next simplest case is when estimate~\eqref{inhomogeneous} again
follows from the coercivity properties of a suitable physical space energy identity~\eqref{energyidentity},
but where the degeneration function~$\chi$ of~\eqref{chidef} is now non-trivial, potentially
vanishing on some set.

More precisely, we assume that
there exists a $T$-invariant quadruple $(V,w,q,\varpi)$, 
again bounded in the sense of~\eqref{seedboundedness}, but now satisfying the following relaxed coercivity properties.
Defining currents $J^{V,w,q,\varpi}$, $K^{V,w,q}$  by~\eqref{generalJdef}--\eqref{generalKdef}, we again
assume the boundary coercivity properties~\eqref{insymbolsiboundary} on $J^{V,w,q,\varpi}$, but we weaken
the bulk coercivity assumption on $K^{V,w,q}$  to
\begin{equation}
\label{insymbolsii}
K^{V,w,q}[\psi] \gtrsim \chi(r) r^{-1-\delta}\left((L\psi)^2+(\underline L\psi)^2 +|\slashed\nabla\psi|^2\right)+r^{-3-\delta}\psi^2,
\end{equation}
where $\chi$ is the function~\eqref{chidef} in Section~\ref{basicblackboxestimatesec}.

We define the fraktur energies again by~\eqref{frakturenergies}, and we note that again we 
have~\eqref{theyrelessthansimilar}, and, in view of~\eqref{seedboundedness}, also~\eqref{infacttheyresim}, and thus  \eqref{theyresimilar}.
Under the above assumptions, the identity~\eqref{energyidentity} gives rise 
to the estimate
\begin{equation}
\begin{aligned}
\label{firstinstancebutnot}
\sup_{v: \tau\le \tau(v)}\Ffancyzero(v,\tau_0,\tau),    \qquad &\Efancyzero(\tau)  + \Efancyzero_{\mathcal{S}}(\tau)+c \int_{\tau_0}^{\tau}{}^\chi\Ezerominusoneminusdelta'(\tau')
d\tau' +c \int_{\tau_0}^{\tau} \Eminusone' (\tau') d\tau' \\
&\leq \Efancyzero(\tau_0)+ \int_{\mathcal{R}(\tau_0,\tau)}
\left| H[\psi] F\right|.
\end{aligned}
\end{equation}
We note again that, in view of~\eqref{theyrelessthansimilar} and~\eqref{theyresimilar},
 estimate~\eqref{firstinstancebutnot} indeed implies~\eqref{inhomogeneous}, for some $\lambda\ge 1$,
 with $V_0=V$ and $w_0=w$
  and where 
  including the final
 $\int F^2$ term in~\eqref{inhomogeneous} is here unnecessary.
As with for~\eqref{firstinstance} of case (i), the
point of expressing the estimate in terms of the fraktur energies~\eqref{frakturenergies} is that~\eqref{firstinstance} is
a sharper  statement than~\eqref{inhomogeneous}, corresponding to $\lambda=1$,
which will be useful for us.

We note that a current-defining quadruple $(V,w,q,\varpi)$ satisfying the properties of this section indeed exists for
the Schwarzschild metric and can be constructed from the 
considerations in~\cite{dafrod2007note, MR2563798}.   
Note that a pre-requisite for even the degenerate bulk coercivity property~\eqref{insymbolsii}
is that any trapped null geodesics be ``unstable'' in a suitable sense (cf.~the Schwarzschild--AdS case with reflective boundary conditions at infinity,
where there exist stably trapped null geodesics~\cite{HolzSmulevici}). 

In the Kerr case for all $|a|\ne0$, even
though all trapped null geodesics are again unstable, 
and even though estimate~\eqref{inhomogeneous} is true (by~\cite{partiii}),
one can show that no quadruple $(V,w,q,\varpi)$ can give rise to currents satisfying
the coercivity properties of this section (see~\cite{Alinhacpert}).  
(For a higher order current defined using second order operators which gives an analogue of the coercivity properties here
in the $|a|\ll M$ case,
see however Andersson--Blue~\cite{AndBlue}.)

One of the main motivations of this paper is to show that, from a suitable point of view, not only is a purely physical
space proof of~\eqref{inhomogeneous} unnecessary for nonlinear applications, but such a proof would only result in a minor and inessential simplification of the argument. We now turn to our main case of interest,
case (iii).

\subsubsection{Case (iii)}
\label{caseiiiidentity}
The most general case we wish to consider in this paper
is where the estimate~\eqref{inhomogeneous} is assumed as a ``black box'', i.e.~it is \underline{not}
necessarily
the consequence of the coercivity properties of some more fundamental physical  space identity~\eqref{energyidentity}.
We note for instance that~\eqref{inhomogeneous} indeed holds when $(\mathcal{M},g_0)$ is Kerr, in fact for the full subextremal range
$|a|<M$ of parameters~\cite{partiii}, but as discussed above,   it
 does not arise (see~\cite{Alinhacpert}) from a current as in case (ii).

In this most general case, however, in addition to assuming~\eqref{inhomogeneous},
we will still make a further assumption on the existence of an auxiliary pair of currents 
$J^{V,w,q,\varpi}$, $K^{V,w,q}$,
 whose
energy identity will \underline{not} in general imply~\eqref{inhomogeneous} but rather will be used \emph{in combination} with~\eqref{inhomogeneous}. This auxiliary current will have the following properties:
The bulk $K$ current will be nonnegative in a neighbourhood of the set $\{\chi\ne 1\}$, where $\chi$
is the function~\eqref{chidef} of Section~\ref{basicblackboxestimatesec} appearing in~\eqref{inhomogeneous}. In nontrivial applications, $K$ will
often vanish identically in this region, making it completely insensitive to the possible presence and nature of trapping. 
Where $\chi=1$, on the other hand, the bulk $K$ current will be  assumed nonnegative only modulo lower order terms,
provided these lower order terms are supported in the region $r_2\le r\le R/2$.  Finally, for $r\ge R$, the bulk current $K$ will control the terms
familiar from cases (i) and (ii).

We now lay out the assumptions in detail. 

We will define functions
\begin{equation}
\label{rhoxidef}
\rho : [r_0, \infty): \to [0,1], \qquad \xi: [r_0,\infty)\to [0,1]
\end{equation}
 such that 
$\rho=1$ for $r\ge R/4$ and $r\le r_2$,
and such that 
\begin{equation}
\label{xivanishing}
\xi=0\,\,\,  {\rm in}\,\,\, \{\chi\ne 1\}\cup \{r\le r_2\} \cup \{r\ge R/4\}, 
\end{equation}
where $\chi$ is the function~\eqref{chidef} in 
Section~\ref{basicblackboxestimatesec}, appearing
in the assumed estimate~\eqref{inhomogeneous}.
(In nontrivial applications, $\rho$ will vanish identically in a neighbourhood
containing the set where $\chi\ne 1$.) The supports of the various functions is indicated in Figure~\ref{supportfigs}.
We note finally that we lose no generality in the present paper in always taking $\rho$, $\chi$, $\xi$ to be indicator functions of appropriate sets,
 but the sharpest estimates one could prove would often involve degeneration on sets of positive codimension.

\begin{figure}
\centering{
\def\svgwidth{15pc}
\begingroup%
  \makeatletter%
  \providecommand\color[2][]{%
    \errmessage{(Inkscape) Color is used for the text in Inkscape, but the package 'color.sty' is not loaded}%
    \renewcommand\color[2][]{}%
  }%
  \providecommand\transparent[1]{%
    \errmessage{(Inkscape) Transparency is used (non-zero) for the text in Inkscape, but the package 'transparent.sty' is not loaded}%
    \renewcommand\transparent[1]{}%
  }%
  \providecommand\rotatebox[2]{#2}%
  \newcommand*\fsize{\dimexpr\f@size pt\relax}%
  \newcommand*\lineheight[1]{\fontsize{\fsize}{#1\fsize}\selectfont}%
  \ifx\svgwidth\undefined%
    \setlength{\unitlength}{112.20245361bp}%
    \ifx\svgscale\undefined%
      \relax%
    \else%
      \setlength{\unitlength}{\unitlength * \real{\svgscale}}%
    \fi%
  \else%
    \setlength{\unitlength}{\svgwidth}%
  \fi%
  \global\let\svgwidth\undefined%
  \global\let\svgscale\undefined%
  \makeatother%
  \begin{picture}(1,0.53020137)%
    \lineheight{1}%
    \setlength\tabcolsep{0pt}%
    \put(0,0){\includegraphics[width=\unitlength,page=1]{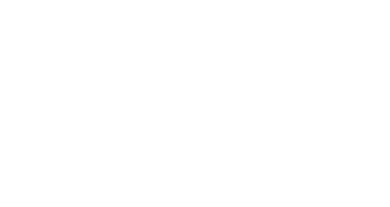}}%
    \put(0.036468,0.40333805){\color[rgb]{0,0,0}\makebox(0,0)[lt]{\lineheight{0}\smash{\begin{tabular}[t]{l}$\chi$\end{tabular}}}}%
    \put(0,0){\includegraphics[width=\unitlength,page=2]{therhos.pdf}}%
    \put(0.03552442,0.32292233){\color[rgb]{0,0,0}\makebox(0,0)[lt]{\lineheight{0}\smash{\begin{tabular}[t]{l}$\rho$\end{tabular}}}}%
    \put(0.03419278,0.24435112){\color[rgb]{0,0,0}\makebox(0,0)[lt]{\lineheight{0}\smash{\begin{tabular}[t]{l}$\xi$\end{tabular}}}}%
    \put(0.1098509,0.12340938){\color[rgb]{0,0,0}\makebox(0,0)[lt]{\lineheight{0}\smash{\begin{tabular}[t]{l}$r_0$\end{tabular}}}}%
    \put(0.23507469,0.1217792){\color[rgb]{0,0,0}\makebox(0,0)[lt]{\lineheight{0}\smash{\begin{tabular}[t]{l}$r_2$\end{tabular}}}}%
    \put(0.69052936,0.11917541){\color[rgb]{0,0,0}\makebox(0,0)[lt]{\lineheight{0}\smash{\begin{tabular}[t]{l}$R/4$\end{tabular}}}}%
    \put(0,0){\includegraphics[width=\unitlength,page=3]{therhos.pdf}}%
  \end{picture}%
\endgroup%
}
\caption{The supports of $\chi$, $\rho$ and $\xi$}\label{supportfigs}
\end{figure}

We  assume the existence of 
a $T$-invariant quadruple $(V,w,q,\varpi)$ satisfying the boundedness
properties~\eqref{seedboundedness},
 such that the associated current
$J^{V,w,q,\varpi}[\psi]$ still satisfies~\eqref{insymbolsiboundary},  but
where the bulk coercivity assumption~\eqref{insymbolsii} on $K^{V,w,q}[\psi]$
is now further relaxed to
\begin{equation}
\label{insymbolsiii}
K^{V,w,q}[\psi] + \tilde{A}\xi(r)\psi^2 \gtrsim  \rho (r) r^{-1-\delta} \left((L\psi)^2+(\underline L\psi)^2 +|\slashed\nabla\psi|^2\right)+\rho(r)r^{-3-\delta}\psi^2,
\end{equation}
where $\tilde{A}\ge 0$ is a possibly large constant. 

With this current, we define the fraktur energies again by~\eqref{frakturenergies}, and we again  
have~\eqref{theyrelessthansimilar}, and, in view of~\eqref{seedboundedness}, also~\eqref{infacttheyresim}, and thus  \eqref{theyresimilar}.

In view of the relaxed coercivity properties, the identity~\eqref{energyidentity}
gives rise
 to
\begin{equation}
\begin{aligned}
\label{firstinstanceforkerr}
\sup_{v: \tau\le \tau(v)}
\Ffancyzero(v, \tau_0,\tau), \qquad &\Efancyzero(\tau) + \Efancyzero_{\mathcal{S}}(\tau) +c\int_{\tau_0}^{\tau} {}^{\rho} \Ezerominusoneminusdelta'(\tau')
 +{}^\rho \Eminusone'(\tau')d\tau'
\\
 &\leq \Efancyzero(\tau_0) 
+  A \int_{\tau_0}^{\tau'}\Eerrorminusone'(\tau') d\tau'+  \int_{\mathcal{R}(\tau_0,\tau)}
\left| H[\psi] F\right|,
\end{aligned}
\end{equation}
for an $A\ge 0$,
where $ {}^{\rho} \Ezerominusoneminusdelta'(\tau)$ is defined in analogy with 
${}^{\chi}  \Ezerominusoneminusdelta'(\tau)$, i.e.
\[
{}^\rho \Ezerominusoneminusdelta'(\tau) := \int_{\Sigma(\tau)} 
  r^{-1-\delta}\rho(r) (|L\psi|^2+|\underline{L}\psi|^2+|\slashed\nabla\psi|^2) , \qquad
  {}^\rho \Eminusone'(\tau): = \int_{\Sigma(\tau)}  r^{-3-\delta}\rho(r)\psi^2
  \]
and
\begin{equation}
\label{analogyerror}
\Eerrorminusone'(\tau) := \int_{\Sigma(\tau)} \xi(r)\psi^2.
\end{equation}
Again, we note that it is expressing things with respect to the fraktur energies which allows the constant to be exactly $1$ in the 
above
estimate~\eqref{firstinstanceforkerr}, and this fact will be useful for us.

The existence of currents $J^{V,w,q,\varpi}$, $K^{V,w,q}$ satisfying the above assumptions
can indeed be shown for Kerr in the range  $|a| \ll M$ (and in fact for general stationary suitably small perturbations of Schwarzschild 
satisfying the assumptions of Section~\ref{backgroundgeom} and appropriate assumptions at infinity).
See Appendix~\ref{howto}.
Note that the estimate~\eqref{firstinstanceforkerr} is manifestly weaker than~\eqref{inhomogeneous}. The point, as discussed
in the introduction, is that, being derived from the (relaxed) coercivity properties of~\eqref{insymbolsiii} applied
to~\eqref{dividentitynotintegrated}, estimate~\eqref{firstinstanceforkerr}, 
or more properly, the identity~\eqref{dividentitynotintegrated} itself, can be applied directly to~\eqref{theequationzero}.
See already Section~\ref{stabilityofenergyidentitiessec}.

\subsubsection{Summary of the unweighted assumptions for cases (i), (ii) and (iii)}
\label{sumaryonetwothree}

So as to not have to always refer to separate formulas in the distinct cases (i), (ii) and (iii), we summarise
the assumptions in a way which can be subsequently interpreted for all cases simultaneously.

In case (i), we set $\tilde\rho=\rho=\chi=1$ and $A=\tilde{A}=0$.

In case (ii), we set $\chi$ to be the  function~\eqref{chidef} appearing
in both~\eqref{inhomogeneous} and~\eqref{insymbolsii}, and we set
$\rho=\chi$, $\tilde\rho=1$ and $A=\tilde{A}=0$.

Finally, in case (iii), we set $\chi$ to be the function~\eqref{chidef} of Section~\ref{basicblackboxestimatesec}
appearing in~\eqref{inhomogeneous}, 
we let $\rho$ and $\xi$ to be the functions~\eqref{rhoxidef} 
and the constants $A$ and $\tilde{A}$ to be as in Section~\ref{caseiiiidentity},
and we set $\tilde\rho=\rho$.

Our assumptions, applicable for all cases, are (a)  that~\eqref{inhomogeneous} holds and (b) that there exist
$T$-invariant $(V,w,q,\varpi)$,
satisfying the bounds~\eqref{seedboundedness}
such that,
defining the currents $J^{V,w,q,\varpi}[\psi]$, $K^{V,w,q}[\psi]$  and $H^{V,w}[\psi]$ by~\eqref{generalJdef}--\eqref{generalHdef}, we have
the bulk coercivity property
\begin{equation}
\label{bulkunweightedcoercivity}
K^{V,w,q}[\psi] + \tilde{A}\xi(r)\psi^2 \gtrsim  \rho (r) r^{-1-\delta} \left((L\psi)^2+(\underline L\psi)^2 +|\slashed\nabla\psi|^2+r^{-3}\psi^2\right) +\tilde\rho(r)r ^{-3-\delta}\psi^2
\end{equation}
and the boundary coercivity properties
\begin{equation}
\begin{aligned}
\label{boundunweightedcoercivity}
\qquad J^{V,w,q,\varpi}_\mu[\psi] n_{\Sigma(\tau)}^\mu &\gtrsim (L\psi)^2 +|\slashed\nabla \psi|^2+\iota_{r\le R}( \underline{L}\psi)^2+r^{-2}\psi^2,\\
J^{V,w,q,\varpi}_\mu[\psi] n_{\underline{C}_v}^\mu &\gtrsim (\underline{L}\psi)^2 +|\slashed\nabla \psi|^2 +r^{-2}\psi^2,
\\
J^{V,w,q,\varpi}_\mu[\psi] n_{\mathcal{S}}^\mu &\gtrsim (L\psi)^2+(\underline L\psi)^2 +|\slashed\nabla\psi|^2+\psi^2.
\end{aligned}
\end{equation}

With these currents, we define the fraktur energies again by~\eqref{frakturenergies}, and we again  
have~\eqref{theyrelessthansimilar}, and, in view of~\eqref{seedboundedness}, also~\eqref{infacttheyresim}, and thus  \eqref{theyresimilar}.

The energy identity~\eqref{energyidentity}
gives rise then to
\begin{equation}
\begin{aligned}
\label{unifiedhere}
\sup_{v: \tau\le \tau(v)}
\Ffancyzero(v,\tau_0,\tau), \qquad& \Efancyzero(\tau) + \Efancyzero_{\mathcal{S}}(\tau)+c\int_{\tau_0}^{\tau} {}^{\rho} \Ezerominusoneminusdelta'(\tau')d\tau'
+c\int_{\tau_0}^{\tau} {}^{\tilde\rho} \Eminusone'(\tau')d\tau'\\
&
 \leq \Efancyzero(\tau_0) 
+  A \int_{\tau_0}^{\tau}\Eerrorminusone'(\tau') d\tau'+  \int_{\mathcal{R}(\tau_0,\tau)}
\left|H[\psi] F\right|,
\end{aligned}
\end{equation}
where we define
\[
{}^{\tilde\rho} \Eminusone'(\tau) := \int_{\Sigma(\tau)} 
 \tilde\rho (r)r^{-3-\delta} \psi^2 .
\]

We emphasise again that in cases (i) and (ii) the statement that~\eqref{inhomogeneous} holds  need not be taken
as an independent assumption, as~\eqref{inhomogeneous} in fact 
follows from~\eqref{unifiedhere} with the above definitions in these two cases.

\subsubsection{An enhanced redshift  current and enhanced positivity in the black hole interior}
\label{enhanced}

In the case $\mathcal{S}\ne\emptyset$, for the purpose of higher order estimates to be considered in Section~\ref{higherorderestimates},
we will need to enhance the positivity near $\mathcal{H}^+$.

We first state the following proposition which allows us to introduce an arbitrary largeness factor $\varsigma$ in front of some of the terms of our coercivity estimate:
\begin{proposition}
\label{enhancedprop}
Under the assumptions of Section~\ref{sumaryonetwothree},  there exists a constant $c>0$ such that the following holds.

Given arbitrary $\varsigma\ge 1$, there exist parameters $r_0\le r_0'(\varsigma)< r_{\rm Killing} < r_1(\varsigma)\le r_1$,
and a translation invariant vector field $V'_\varsigma$ and one form $q'_\varsigma$, 
such that, defining $J^{V'_\varsigma,w,q'_\varsigma,\varpi}[\psi]$, $K^{V'_\varsigma,w,q'_\varsigma}[\psi]$, $H^{V'_\varsigma,w}[\psi]$
as in Section~\ref{sumaryonetwothree}, then
\begin{equation}
\label{enhancedpositivity}
K^{V'_\varsigma, w, q'_\varsigma}[\psi] \ge  c (Y\psi)^2 + c\varsigma \left( \sum_{i=1}^3 (\Omega_i\psi)^2  +  (T\psi)^2 +\psi^2 \right)
\end{equation}
in $r_0'(\varsigma)\le r\le r_1(\varsigma)$,  and  all properties of Section~\ref{sumaryonetwothree} hold for these currents 
with implicit constants in~\eqref{seedboundedness},~\eqref{bulkunweightedcoercivity},~\eqref{boundunweightedcoercivity},~\eqref{infacttheyresim} which may be taken independent of $\varsigma$.
Finally, $V'_\varsigma=V$, $q'_\varsigma=q$ for $r\ge r_1$.
\end{proposition}

\begin{proof} 
This is clear by examining the proof of Theorem~7.1 of~\cite{Mihalisnotes}. 
\end{proof}

Using the above and the timelike character of $\nabla r$ in the ``black hole interior'' region $r<r_{\rm Killing}$, 
we may further modify our current to obtain the following enhanced positivity in $r_0\le r\le r_1(\varsigma)$, in particular, all the way up to $\mathcal{S}$,
at the expense of a lower order term. The resulting modified current will be used for higher order estimates:

\begin{proposition}
\label{enhancedinteriorprop}
Under the assumptions of Section~\ref{sumaryonetwothree}, there exist constants $C>0$, $c>0$  such that
the following holds. 

Given arbitrary $\varsigma \ge 1$, let $V'_\varsigma$, $q'_\varsigma$ be as given in Proposition~\ref{enhancedprop}.
Then  there exists a translation invariant  quadruple $(V_{\varsigma}, w_{\varsigma}, q_{\varsigma}, \varpi_{\varsigma})$ 
with $V_{\varsigma}=V'_{\varsigma}$, $w_{\varsigma}=w$,  $q_{\varsigma}=q'$,  $\varpi_{\varsigma}=\varpi$ for $r\ge  r_{\rm Killing}$, 
and a positive function $\lambda_\varsigma(r)$, 
such that, defining $J^{V_{\varsigma},w_{\varsigma} ,q_{\varsigma},\varpi_{\varsigma}}[\psi]$, $K^{V_{\varsigma},w_\varsigma,q_{\varsigma}}[\psi]$, 
$H^{V_{\varsigma},w_\varsigma}[\psi]$
as in Section~\ref{sumaryonetwothree},   the following enhanced positivity (modulo lower order terms):
\begin{eqnarray}
\nonumber
K^{V_{\varsigma}, w_{\varsigma}, q_{\varsigma}}[\psi] &\ge& c  \varsigma \lambda_{\varsigma}(r)  \left( \sum_{i=1}^3  ( \Omega_i\psi)^2 +(T\psi)^2   + |r_{\rm Killing}-r|(Y\psi)^2\right)\\
&&\qquad
\label{enhancedherefinal}
 + c\lambda_\varsigma(r)    (Y\psi)^2 -C\varsigma \lambda_\varsigma(r)   \psi^2
 \end{eqnarray}
holds in $r_0\le r\le r_{\rm Killing}$. 

Moreover,
all boundedness and coercivity properties of the currents of Section~\ref{sumaryonetwothree} hold for these currents
as well with constants independent of $\varsigma$ in the region $r\ge r_{\rm Killing}$, 
while in the region $r< r_{\rm Killing}$, the properties
 of Section~\ref{sumaryonetwothree} again hold for these currents but with~\eqref{bulkunweightedcoercivity} replaced by~\eqref{enhancedherefinal},  
and with coercivity bounds that now however depend in general on $\varsigma$. 
In particular, we have
\begin{equation}
\label{Hbound}
|H^{V_{\varsigma},w_{\varsigma} }[\psi]| \leq C \lambda_\varsigma(r)  \left( |Y\psi| + \sum_{i=1}^3 |\Omega_i\psi| +|T\psi| \right) +
C\lambda_\varsigma(r) \psi^2.
\end{equation}
\end{proposition}

\begin{proof}
Given $\varsigma$, let $r_0'(\varsigma)$ be as given by Proposition~\ref{enhancedprop}.
Let us define $\lambda_\varsigma(r)$ to be a smooth positive function such that $\lambda_\varsigma(r)=1$ for $r\ge r_{\rm Killing}$,
\begin{equation}
\label{mainpartofdef}
 \lambda_\varsigma(r) = e^{-\varsigma (r_{\rm Killing}-r)}
 \end{equation}
for $  \{r_0 \le  r\le r_{\rm Killing} \} \cap \left( \{ \varsigma \ge (r_{\rm Killing}-r)^{-1}\}\cup \{ r \le r_0'(\varsigma)\}\right)$, and 
\[
0\le \frac{d\lambda_\varsigma (r) }{dr}\le 2\varsigma \lambda_\varsigma(r).
\]

We define now
\[
V_\varsigma :=\lambda_\varsigma(r) V'_{\varsigma}, \qquad w_\varsigma:=\lambda_\varsigma(r) w , \qquad q_\varsigma :=\lambda_\varsigma(r) q -*(d\lambda_\varsigma \wedge \varpi), \qquad
\varpi_\varsigma =\lambda_\varsigma(r)\varpi .
\]
We note that under these definitions  $J^{V_{\varsigma},w_{\varsigma} ,q_{\varsigma},\varpi_{\varsigma}}[\psi] =
\lambda_\varsigma(r) J^{V'_{\varsigma}, w, q'_{\varsigma}, \varpi}[\psi]$, and thus the positivity properties of the boundary currents are preserved,
but with constants that now depend on $\varsigma$.

We have that
\[
\nabla^\mu V_{\varsigma}^\nu  = \nabla^\mu (\lambda_\varsigma(r) ) V'_{\varsigma}{}^\nu
+ \lambda_\varsigma(r)  \nabla^\mu V'_{\varsigma}{}^\nu\\ =\frac{d \lambda_\varsigma(r)}{dr}   \nabla^\mu r V'_{\varsigma}{}^\nu 
+ \lambda_\varsigma(r)  \nabla^\mu V'_{\varsigma}{}^\nu
\]
and thus
\begin{align*}
K^{V_{\varsigma}^\nu}[\psi]+w_\varsigma \nabla^\mu\psi \partial_\mu \psi &=
T_{\mu\nu} \frac12 (\nabla^\mu V_\varsigma^\nu +\nabla^\nu V_\varsigma^\mu)+w_\varsigma \nabla^\mu\psi \partial_\mu \psi  \\
& =
\frac{d \lambda_\varsigma(r)}{dr}  T_{\mu\nu} \nabla^\mu r V'_{\varsigma}{}^\nu 
+\lambda_\varsigma(r) (T_{\mu\nu}\frac12( \nabla^\mu V'_{\varsigma}{}^\nu +\nabla^\nu V'_{\varsigma}{}^\mu) + w  \nabla^\mu\psi \partial_\mu \psi )\\
&\ge c 
\frac{d\lambda_\varsigma(r)}{dr}   \left( \sum_{i=1}^3  ( \Omega_i\psi)^2 +(T\psi)^2   + |r_{\rm Killing}-r|(Y\psi)^2\right)\\
&\qquad
 +c \lambda_\varsigma(r)  \left(     (Y\psi)^2 +   \sum_{i=1}^3 (\Omega_i\psi)^2 + (T\psi)^2 \right)\\
 &\ge c \varsigma  \lambda_\varsigma(r)  \left( \sum_{i=1}^3  ( \Omega_i\psi)^2 +(T\psi)^2   + |r_{\rm Killing}-r|(Y\psi)^2\right)
 +c \lambda_\varsigma(r)    (Y\psi)^2  .
\end{align*}
For the first inequality, we are using the fact that $\nabla r$ is a smooth vector field which is
future directed null on $\mathcal{H}^+$ and future directed timelike in $r<r_{\rm Killing}$, by our assumptions from Section~\ref{Killinghorizonsection},
as well as the fact that  
\begin{equation}
\label{alsopos}
T_{\mu\nu}[\psi] \frac12 (\nabla^\mu V'_{\varsigma}{}^\nu +\nabla^\nu V'_{\varsigma}{}^\mu) + w  \nabla^\mu\psi \partial_\mu \psi  \ge c   \left(     (Y\psi)^2 +   \sum_{i=1}^3 (\Omega_i\psi)^2 + (T\psi)^2 \right),
\end{equation}
which follows because coercivity of the  current $K^{V'_\varsigma,w,q'_{\varsigma}}$ asserted in~\eqref{enhancedpositivity}
requires in particular coercivity of its highest order terms.
Note that the second inequality above  clearly follows from the first wherever~\eqref{mainpartofdef} holds. On the other hand, 
wherever~\eqref{mainpartofdef} does not hold, 
we again obtain the second inequality in view of the definition of $\lambda_{\varsigma}(r)$, since 
in this region, the enhanced positivity~\eqref{enhancedpositivity} applies, allowing us to 
put extra $\varsigma$ factors on the second two terms of~\eqref{alsopos}, as well as 
 the bound $\varsigma (r_{\rm Killing}-r ) \le 1$.
 
On the other hand,
\begin{align*}
|K^{w_\varsigma,q_\varsigma}[\psi] -w_\varsigma \nabla^\mu\psi \partial_\mu \psi |&=  |\nabla^\mu w_\varsigma \psi \partial_\mu \psi  +  \nabla^\mu (q_\varsigma)_\mu \psi^2 + 2\psi (q_\varsigma)_\mu g^{\mu\nu}
\partial_\nu \psi| \\
&\le C\left( \frac{d\lambda_\varsigma(r)}{dr}  +\lambda_\varsigma(r)\right )  |\psi|  \left( |Y\psi| + \sum_{i=1}^3 |\Omega_i\psi| +|T\psi| \right)\\
&\qquad +C\left( \frac{d\lambda_\varsigma(r)}{dr}  +\lambda_\varsigma(r)\right)  \psi^2+ C\left( \frac{d\lambda_\varsigma(r)}{dr}  +C\lambda_\varsigma(r)\right)| \psi|   \left( |Y\psi| + \sum_{i=1}^3 |\Omega_i\psi| +|T\psi| \right),
\end{align*}
whence we deduce
\[
|K^{w_\varsigma,q_\varsigma}[\psi-w_\varsigma \nabla^\mu\psi \partial_\mu \psi]  | \le \frac12 (K^{V_{\varsigma}^\nu}[\psi]+w_\varsigma \nabla^\mu\psi \partial_\mu \psi ) +   C \varsigma \lambda_\varsigma(r) \psi^2 .
\]

On the other hand,
\begin{align*}
|H^{V_{\varsigma}, w_{\varsigma}}[\psi]|	&=	 |V_{\varsigma}^\nu \partial_\nu \psi |=  |\lambda_{\varsigma}(r)  V'_{\varsigma}{}^\nu \psi
+w_\varsigma(r)\psi|\\
						&\le   C\lambda_{\varsigma}(r)  \left( |Y\psi| + \sum_{i=1}^3 |\Omega_i\psi| +|T\psi| \right)+C\lambda_\varsigma(r)|\psi|,
\end{align*}
giving~\eqref{Hbound}.
The statement~\eqref{enhancedherefinal} now follows since $K^{V_\varsigma, w_\varsigma,q_\varsigma} = K^{V_\varsigma}+w_\varsigma \nabla^\mu\psi \partial_\mu \psi+K^{w_\varsigma,q_\varsigma}-w_\varsigma \nabla^\mu\psi \partial_\mu \psi$ and the above bounds.
\end{proof}

\subsection{The $r^p$ hierarchy}
\label{rpsection}

For nonlinear applications, we will need to extend our estimates to suitable weighted estimates
satisfying the $r^p$ hierarchy~\cite{DafRodnew}.

In a suitable framework,
very general assumptions on stationary metrics $g_0$ allowing one to apply the $r^p$
 hierarchy are contained in~\cite{Moschnewmeth}.
Here, let us note that this class was shown in particular to  contain Minkowski space, Schwarzschild and Kerr in the full range of
parameters.

So as not to translate to the setup of~\cite{Moschnewmeth}, however, 
rather than formulate the most general asymptotically flat assumptions which we will allow in terms of $g_0$, 
it will be convenient to make the assumptions 
\emph{directly} in terms of coercivity properties of appropriate
currents defined in a region $r\ge \tilde{R} \ge R$ for some large $\tilde{R}$.

To give our precise assumption, let us first define
$\accentset{\scalebox{.6}{\mbox{\tiny $(p)$}}}{{V}}^{\circ}_{\rm far}:=r^pL$ and let
$\accentset{\scalebox{.6}{\mbox{\tiny $(p)$}}}{{w}}^{\circ}_{\rm far}$, $\accentset{\scalebox{.6}{\mbox{\tiny $(p)$}}}{{q}}^{\circ}_{\rm far}$,
$\accentset{\scalebox{.6}{\mbox{\tiny $(p)$}}}{{\varpi}}^{\circ}_{\rm far}$ be as defined in Appendix~\ref{subsec:pcurrentMink}.
We assume then that for any $\delta \le p \le 2-\delta$,
there exists a $T$-invariant quadruple $(\accentset{\scalebox{.6}{\mbox{\tiny $(p)$}}}V_{\rm far}, \accentset{\scalebox{.6}{\mbox{\tiny $(p)$}}}w_{\rm far},\accentset{\scalebox{.6}{\mbox{\tiny $(p)$}}}q_{\rm far}, 
\accentset{\scalebox{.6}{\mbox{\tiny $(p)$}}}\varpi_{\rm far})$, 
with $\accentset{\scalebox{.6}{\mbox{\tiny $(p)$}}}V_{\rm far}=\accentset{\scalebox{.6}{\mbox{\tiny $(p)$}}}{{V}}^{\circ}_{\rm far} +\tilde{V}_{\rm far}$ a vector field,  $\accentset{\scalebox{.6}{\mbox{\tiny $(p)$}}}w_{\rm far}=\accentset{\scalebox{.6}{\mbox{\tiny $(p)$}}}{{w}}^{\circ}_{\rm far}+\tilde{w}_{\rm far}$ a scalar function, 
$\accentset{\scalebox{.6}{\mbox{\tiny $(p)$}}}q_{\rm far}=\accentset{\scalebox{.6}{\mbox{\tiny $(p)$}}}{{q}}^{\circ}_{\rm far} +\tilde{q}_{\rm far}$ a one form,
and  $\accentset{\scalebox{.6}{\mbox{\tiny $(p)$}}}\varpi_{\rm far}=\accentset{\scalebox{.6}{\mbox{\tiny $(p)$}}}{{\varpi}}^{\circ}_{\rm far} +\tilde{\varpi}_{\rm far}$ 
a two-form,
defined on $r\ge \tilde{R}$,
satisfying 
\begin{equation}
\begin{aligned}
\label{seedfarboundedness}
 |g(\tilde{V}_{\rm far},L)| \lesssim 1, \qquad  |g(\tilde{V}_{\rm far},\underline{L})|\lesssim 1, \qquad  \sum |g(\tilde{V}_{\rm far},\Omega_i)|^2  \lesssim 1,\\
|  \tilde{w}_{\rm far} | \lesssim r^{-1}, 
\qquad |L^\mu \tilde{q}_{\rm far}{}_{\mu}| \lesssim r^{-2}, \qquad |\underline{L}^\mu \tilde{q}_{\rm far}{}_{\mu}| \lesssim r^{-2}, \\
	\vert (*\, d\tilde\varpi_{\rm far})_{\mu} L^{\mu} \vert \lesssim r^{-2},
	\qquad
	\vert (*\, d\tilde\varpi_{\rm far})_{\mu} \underline{L}^{\mu} \vert \lesssim r^{-2},\\
	\vert *(\varrho^{L} \wedge \tilde\varpi_{\rm far})_{\mu} L^{\mu} \vert \lesssim r^{-1},
	\quad
	\vert *(\varrho^{L} \wedge \tilde\varpi_{\rm far})_{\mu} \underline{L}^{\mu} \vert \lesssim r^{-1},\\
	\vert *(\varrho^{\underline{L}} \wedge \tilde\varpi_{\rm far})_{\mu} L^{\mu} \vert \lesssim r^{-1},
	\quad
	\vert *(\varrho^{\underline{L}} \wedge \tilde\varpi_{\rm far})_{\mu} \underline{L}^{\mu} \vert \lesssim r^{-1},
\\
	\vert *(\varrho^{i} \wedge \tilde\varpi_{\rm far})_{\mu} L^{\mu} \vert \lesssim r^{-1},
	\quad
	\vert *(\varrho^{i} \wedge \tilde\varpi_{\rm far})_{\mu} \underline{L}^{\mu} \vert \lesssim r^{-1},
	\qquad
	i = 1,2,3\, ,
\end{aligned}
\end{equation}
such that, defining the associated currents $\Jp_{\rm far}$, $\Kp_{\rm far}$, by \eqref{generalJdef},~\eqref{generalKdef} 
these satisfy the weighted
bulk coercivity properties:
  \begin{equation}
  \begin{aligned}
\label{insymbolsiiwithpweightKfar}
\Kp_{\rm far}[\psi] &\gtrsim     r^{p-1}\left( (r^{-1}L(r\psi))^2 +(L\psi)^2 + |\slashed\nabla \psi|^2 \right) 
 + r^{-1-\delta}(\underline L\psi)^2 +r^{p-3}\psi^2
 \end{aligned}
\end{equation}
and the weighted boundary coercivity property
\begin{equation}
\begin{aligned}
\label{insymbolsiiwithpweightJfar}
 \Jp_{\rm far}{}_\mu[\psi] n_{\Sigma(\tau)}^\mu &\gtrsim r^p( r^{-1}L(r\psi))^2 +r^{\frac{p}2}(L\psi)^2 +|\slashed\nabla \psi|^2
+r^{\frac{p}2-2}\psi^2 ,\\
\Jp_{\rm far}{}_\mu[\psi] n_{\underline{C}_v}^\mu &\gtrsim (\underline{L}\psi)^2 +r^p|\slashed\nabla \psi|^2 +r^{p-2}\psi^2.
\end{aligned} 
\end{equation}

We note that the $r^{p-1} (r^{-1}L(r\psi))^2$ term is redundant  in~\eqref{insymbolsiiwithpweightKfar} 
as it can be estimated pointwise from
$r^{p-1}(L\psi)^2$ and $r^{p-3}\psi^2$. We retain it to compare with the boundary term~\eqref{insymbolsiiwithpweightJfar} where it is necessary
to retain explicitly 
the $(r^{-1}L(r\psi))^2$ term, as it is not controlled by the terms $r^{\frac{p}2}(L\psi)^2$ and $r^{\frac{p}2-2}\psi^2$.

See Appendix~\ref{section:appedixcurrents} for the construction of such a current on Minkowski space and a broad
class of spacetimes with suitable asymptotic flatness assumptions at infinity, including Schwarzschild
and Kerr in the full subextremal range $|a|<M$.

Let us note immediately that given such currents satisfying the far away coercivity assumptions~\eqref{insymbolsiiwithpweightKfar} and \eqref{insymbolsiiwithpweightJfar}, 
and given  $(V,w,q,\varpi)$ as in Section~\ref{sumaryonetwothree},
then by defining a suitable cutoff function
 $\zeta(r)$ with $\zeta=0$ for $r\le \tilde{R}$  and $\zeta(r)=1$ for $r\ge \tilde{R}+1$, and introducing a small fixed parameter
 $e>0$, and defining the $T$-invariant vector field, functions and forms
 \begin{equation}
 \label{thetriple}
 \accentset{\scalebox{.6}{\mbox{\tiny $(p)$}}}V= V+e\zeta(r) \accentset{\scalebox{.6}{\mbox{\tiny $(p)$}}}V_{\rm far} , \qquad
  \accentset{\scalebox{.6}{\mbox{\tiny $(p)$}}}w = w+e\zeta  \accentset{\scalebox{.6}{\mbox{\tiny $(p)$}}}w_{\rm far}, \qquad 
   \accentset{\scalebox{.6}{\mbox{\tiny $(p)$}}}q = q+e\zeta  \accentset{\scalebox{.6}{\mbox{\tiny $(p)$}}}q_{\rm far} , \qquad
      \accentset{\scalebox{.6}{\mbox{\tiny $(p)$}}}\varpi = \varpi+e\zeta  \accentset{\scalebox{.6}{\mbox{\tiny $(p)$}}}\varpi_{\rm far} 
 \end{equation}
 one sees immediately that the associated currents  $\Jp$, $\Kp$ satisfy the global (relaxed) weighted
 bulk coercivity assumptions
  \begin{equation}
  \begin{aligned}
\label{insymbolsiiwithpweightK}
\Kp[\psi] + \tilde{A}\xi(r)\psi^2 &\gtrsim    r^{p-1} \rho(r) \left(  (r^{-1}L(r\psi))^2 + (L\psi)^2+ |\slashed\nabla \psi|^2\right) 
 + r^{-1-\delta}\rho (r) (\underline L\psi)^2 +r^{p-3} \tilde\rho(r) \psi^2,
 \end{aligned}
 \end{equation}
 and the global weighted boundary coercivity assumptions
 \begin{equation}
 \begin{aligned}
 \label{insymbolsiiwithpweightJ}
 \Jp{}_\mu[\psi] n_{\Sigma(\tau)}^\mu &\gtrsim r^p (r^{-1}L(r\psi))^2  +
 r^{\frac{p}2}(L\psi)^2+ |\slashed\nabla \psi|^2+\iota_{r\le R}( \underline{L}\psi)^2
+r^{\frac{p}2-2}\psi^2 , \\
\Jp{}_\mu[\psi] n_{\underline{C}_v}^\mu &\gtrsim (\underline{L}\psi)^2 +r^p|\slashed\nabla \psi|^2 +r^{p-2}\psi^2,\\
\Jp_\mu[\psi] n_{\mathcal{S}}^\mu &\gtrsim  (L\psi)^2+(\underline L\psi)^2 +|\slashed\nabla\psi|^2+\psi^2,
\end{aligned} 
\end{equation}
as one can absorb the error terms in the region  $\tilde{R}\le r\le \tilde{R}+1$ arising  from the cutoff by terms
controlled by the coercivity relations~\eqref{bulkunweightedcoercivity} and~\eqref{boundunweightedcoercivity},
provided $e$ is suitably chosen. We define also 
the associated $ \Hap= \accentset{\scalebox{.6}{\mbox{\tiny $(p)$}}}V^\mu\partial_\mu \psi +    \accentset{\scalebox{.6}{\mbox{\tiny $(p)$}}}w$ 
so that the divergence identity~\eqref{dividentitynotintegrated} holds with $\Jp$, $\Kp$ and
$\Hap$.

From the bounds~\eqref{seedfarboundedness} and~\eqref{insymbolsiiwithpweightJ}, we see that we have in fact
 \begin{equation}
 \begin{aligned}
 \label{insymbolsiiwithpweightJequiv}
 \Jp{}_\mu[\psi] n_{\Sigma(\tau)}^\mu &\sim r^p (r^{-1}L(r\psi))^2  +
 r^{\frac{p}2}(L\psi)^2+ |\slashed\nabla \psi|^2+\iota_{r\le R}( \underline{L}\psi)^2
+r^{\frac{p}2-2}\psi^2 , \\
\Jp{}_\mu[\psi] n_{\underline{C}_v}^\mu &\sim (\underline{L}\psi)^2 +r^p|\slashed\nabla \psi|^2 +r^{p-2}\psi^2,\\
\Jp_\mu[\psi] n_{\mathcal{S}}^\mu &\sim  (L\psi)^2+(\underline L\psi)^2 +|\slashed\nabla\psi|^2+\psi^2.
\end{aligned} 
\end{equation}

To state the resulting weighted versions of estimate~\eqref{inhomogeneous} let us introduce some notation.
For $\delta\le p \le 2-\delta$ we define
  \begin{align}
  \label{coeffshereand}
  \Ep(\tau) &: = \Ezero(\tau)+  \int_{\Sigma(\tau)\cap \{r\ge R\} } r^p (r^{-1}L(r\psi))^2+ r^{\frac{p}2}(L\psi)^2
  +r^{\frac{p}2-2}\psi^2, \\
   \label{weightedpflux}
 \Fp(v, \tau_0,\tau) &:= \Fzero(v, \tau_0,\tau) + \int_{\underline{C}_v\cap \mathcal{R}(\tau_0,\tau)} r^p|\slashed\nabla\psi|^2 + r^{p-2}\psi^2,\\
   \label{coeffsheretootoosmaller}
   \Epminusone'(\tau) &: = \Ezerominusoneminusdelta'(\tau)+ \int_{\Sigma(\tau)\cap \{r\ge R\} } r^{p-1}\left( (r^{-1}L(r\psi))^2 
    +  (L\psi)^2 +  |\slashed\nabla \psi|^2\right) +r^{p-3}\psi^2 .
  \end{align}

We also define versions of~\eqref{coeffsheretootoosmaller} where 
the first term is replaced by ${}^\chi \Ezerominusoneminusdelta'(\tau)$ or ${}^\rho \Ezerominusoneminusdelta'(\tau)$.
We will call these ${}^\chi\Epminusone'(\tau)$ and ${}^\rho \Epminusone'(\tau)$.
Note that these latter two expressions do not control a zero'th order term in the region $r<R$.

Note the following properties. For $2-\delta \ge p \ge \delta $,
we have
\begin{equation}
\label{fluxbulkrelation} 
\Ep \gtrsim \Epprime, \, \, \Fp \gtrsim \Fpprime {\rm\ \ for\ } p\ge p'\ge \delta {\rm\ or\ }p'=0, \qquad  
  \Epminusone' \gtrsim \Epminusone  {\rm\ for\ } p\ge 1+\delta, \qquad   \Epminusone' \gtrsim \Ezero  {\rm\ for\ }p\ge 1 .
\end{equation}

Let us also define the fluxes
\begin{equation}
\label{frakturepnergies}
\Efancyp(\tau):= \int_{\Sigma(\tau)}  \Jp_\mu[\psi] n_{\Sigma(\tau)}^\mu,  \qquad
\Efancyp_{\mathcal{S}}(\tau ) := \int_{\mathcal{S}(\tau_0,\tau)}  \Jp_\mu[\psi] n_{\mathcal{S}}^\mu , \qquad
\Ffancyp(v, \tau_0,\tau) :=\int_{\underline{C}_v\cap \mathcal{R}(\tau_0,\tau)} \Jp_\mu[\psi] n_{\underline{C}_v}^\mu.
\end{equation}

We note that by~\eqref{insymbolsiiwithpweightJequiv}, it follows that
\begin{equation}
\label{btermsequiv}
\Efancyp(\tau)\sim \Ep(\tau) , \qquad \Efancyp_{\mathcal{S}}(\tau_0,\tau ) = \Efancyzero_{\mathcal{S}}(\tau_0,\tau)\sim 
\Ezero_{\mathcal{S}}(\tau_0,\tau)
\end{equation}
\begin{equation}
\label{btermsequivag}
\Ffancyp(v, \tau_0,\tau)
\sim  \Fp(v, \tau_0,\tau).
\end{equation}

We may now state the main result of this section:
\begin{proposition}
\label{finaluncommuted}
Let $(\mathcal{M},g_0)$ satisfy the assumptions of Section~2 and Section~\ref{sumaryonetwothree}, and 
let $(V,w,q,\varpi)$ be as in Section~\ref{sumaryonetwothree}. 
(In particular, the estimates~\eqref{inhomogeneous} and~\eqref{unifiedhere} are assumed to hold.)
Assume, for all $0<\delta\le p\le 2-\delta$, the existence of
a quadruple $(\accentset{\scalebox{.6}{\mbox{\tiny $(p)$}}}V_{\rm far}, \accentset{\scalebox{.6}{\mbox{\tiny $(p)$}}}w_{\rm far},\accentset{\scalebox{.6}{\mbox{\tiny $(p)$}}}q_{\rm far},\accentset{\scalebox{.6}{\mbox{\tiny $(p)$}}}\varpi_{\rm far})$ as above 
satisfying~\eqref{seedfarboundedness} and the far away coercivity properties~\eqref{insymbolsiiwithpweightKfar} and~\eqref{insymbolsiiwithpweightJfar}. 

Then for all $0<\delta\le p \le 2-\delta$, the following estimate holds for all
$\tau_0\le \tau$:
\begin{equation}
\begin{aligned}
\label{rpidentherefluxes}
\sup_{v: \tau\le \tau(v)}\Ffancyp(v, \tau_0,\tau),  \qquad   &\Efancyp(\tau) + \Efancyp_{\mathcal{S}}(\tau_0,\tau ) +c \int_{\tau_0}^{\tau}{}^\rho\Epminusone'(\tau')
d\tau' 
+c \int_{\tau_0}^{\tau} {}^{\tilde\rho} \Eminusone' (\tau') d\tau' \\
&\leq  \Efancyp(\tau_0)+  A \int_{\tau_0}^{\tau}\Eerrorminusone'(\tau') +\int_{\mathcal{R}(\tau_0,\tau)}
\left| \Hap[\psi] F\right|
+C \int_{\mathcal{R}(\tau_0,\tau)}F^2
\end{aligned}
\end{equation}
as well as the estimate
\begin{equation}
\begin{aligned}
\label{rpidenthere}
\sup_{v: \tau\le \tau(v)} \Fp(v,\tau_0,\tau)&+  \Ep(\tau) + \Ezero_{\mathcal{S}}(\tau_0,\tau) +\int_{\tau_0}^{\tau}{}^\chi\Epminusone'(\tau')
d\tau' + \int_{\tau_0}^{\tau} \Eminusone' (\tau') d\tau' \\
&\lesssim \Ep(\tau_0)+  \int_{\mathcal{R}(\tau_0,\tau)}
\left(|r^pr^{-1} L (r\psi)|+|\tilde{V}_p^\mu\partial_\mu \psi |+|\tilde{w}_p\psi| \right) |F|
+\int_{\mathcal{R}(\tau_0,\tau)}F^2,
\end{aligned}
\end{equation} 
where $\tilde{V}_p$ is a fixed  vector field and  $\tilde{w}_p$ a fixed function satisfying
\begin{equation}
\label{inhomogboundednessassump}
 |g(\tilde{V}_p,L)| \lesssim 1, \qquad  |g(\tilde{V}_p,\underline{L})|\lesssim 1, \qquad  \sum |g(\tilde{V}_p,\Omega_i)|^2  \lesssim 1, \qquad
| \tilde{w}_p | \lesssim r^{-1} .
\end{equation}
\end{proposition}

\begin{remark}
In the case where we replace the middle term of~\eqref{inhomogeneous} 
with~\eqref{replacedterm}, we 
should add the first term of~\eqref{replacedterm} to the right hand side of~\eqref{rpidenthere}.
\end{remark}

\begin{proof}
The estimate~\eqref{rpidentherefluxes} follows immediately from the energy identity~\eqref{energyidentity} 
corresponding
to the current $\Jp$, in view of the properties assumed. 

Note that in cases (i) or (ii), estimate~\eqref{rpidentherefluxes} already implies~\eqref{rpidenthere},
in view also of~\eqref{seedfarboundedness}. In general, to obtain~\eqref{rpidenthere}, we add a large multiple
of  estimate~\eqref{inhomogeneous} to~\eqref{rpidentherefluxes}. This allows us to absorb the term
multiplying $A$ on the right hand side of~\eqref{rpidentherefluxes}
in view of the trivial relation
\[
\Eerrorminusone' \lesssim \Eminusone'.
\]
\end{proof}

For consistency, we will in what follows often denote the quadruple $(V,w,q,\varpi)$ of Section~\ref{sumaryonetwothree}
as $\accentset{\scalebox{.6}{\mbox{\tiny $(0)$}}}V$,  $\accentset{\scalebox{.6}{\mbox{\tiny $(0)$}}}w$,  
   $\accentset{\scalebox{.6}{\mbox{\tiny $(0)$}}}q$,
      $\accentset{\scalebox{.6}{\mbox{\tiny $(0)$}}}\varpi$, and the currents as $\Jzero$, $\Kzero$, $\Hazero$.
      
Finally, in view of Proposition~\ref{enhancedinteriorprop}, given $\varsigma\ge 1$, we may define versions of the above currents where $V_{\varsigma}$, $w_{\varsigma}$,
$q_{\varsigma}$, $\varpi_{\varsigma}$
replace $V$, $w$, $q$, $\varpi$, respectively, in definition~\eqref{thetriple}. 
We will denote these currents as 
\[
\Jp_{\varsigma},\qquad \Kp_{\varsigma}, \qquad \Hap_{\varsigma}.
\]
All boundedness and coercivity inequalities will continue to hold with constants independent of $\varsigma$ in the region~$r\ge r_{\rm Killing}$,
while in the region $r<r_{\rm Killing}$, the resulting $\Kp_\varsigma$ will satisfy the enhanced positivity property~\eqref{enhancedherefinal}, at the expense
of lower order terms. 
(In the region $r<r_{\rm Killing}$, however, we emphasise the dependence of the coercivity constants on $\varsigma$.)

\subsection{Higher order estimates}
\label{higherorderestimates}

As is well known, in applications to nonlinear problems, one must be able to prove
higher order estimates in order for the estimates to close.

\subsubsection{The commutation vector fields $\mathfrak{D}$ and the auxiliary  $\widetilde{\mathfrak{D}}$}
\label{commutatordefsec}

Let $\zeta(r)$ denote a  smooth cutoff function 
such that  $\zeta=1$ for $r\ge 3R/4$ and $\zeta=0$ for $r\le R/2$. If $\mathcal{S}=\emptyset$, let $\hat\zeta=0$.
Otherwise, let $\hat\zeta(r)$ denote a smooth cutoff function 
such that $\hat\zeta=1$ for $r\le r_1+(r_2-r_1)/4$ and $\hat\zeta=0$ for $r\ge r_1+(r_2-r_1)/2$. 
If $\Omega_1$ is Killing, we define $\upsilon=1$, otherwise we set $\upsilon=0$. 
(More generally, we may take $\upsilon=0$ if $T$ coincides with the Killing generator and is timelike in $r>r_{\rm Killing}$.)

We define
\begin{equation}
\label{commutatordefine}
\mathfrak{D}_{1}=T, \qquad \mathfrak{D}_2=\upsilon\Omega_1, \qquad
 \mathfrak{D}_{3}=   Y, \qquad
  \mathfrak{D}_{4} = \zeta L, \qquad
\mathfrak{D}_{5}=\zeta\underline L, \qquad
  \mathfrak{D}_{5+i}= (\zeta+\hat\zeta) \Omega_i, \, i=1,\ldots, 3.
\end{equation}

We will denote ${\bf k}=(k_1,k_2,k_3,k_4, k_5, k_6, k_7, k_8)$, $|{\bf k}|=\sum_{i=1}^{8} k_i$
and 
by $\mathfrak{D}^{{\bf k}}\psi$   the  expression
\begin{equation}
\label{commutators}
\mathfrak{D}^{{\bf k}}\psi = 
 \mathfrak{D}_1^{k_1}\mathfrak{D}_2^{k_2}\cdots \mathfrak{D}_{8}^{k_{8}}\psi.
\end{equation}

Note that the vector fields~\eqref{commutatordefine}
span the tangent space for $r\ge 3R/4$ and, if $\mathcal{S}\ne\emptyset$, in $r\le r_1+(r_2-r_1)/4$,
but not so in general for $r_1+(r_2-r_1)/4\le r\le 3R/4$. It is also useful to have a collection of operators
that do. We thus define $\widetilde{\mathfrak{D}}^{\bf k}$
to denote commutation  strings of operators from the collection
\begin{equation}
\label{commutatortildedefine}
\widetilde{\mathfrak{D}}_{1}=L, \qquad \widetilde{\mathfrak{D}}_{2}=  \underline{L} , \qquad \widetilde{\mathfrak{D}}_{2+i}=  \Omega_i, \qquad i=1,\ldots, 6, 
\end{equation}
i.e.~we define
\begin{equation}
\label{commutatorsalt}
\widetilde{\mathfrak{D}}^{{\bf k}}\psi = 
\widetilde{\mathfrak{D}}_1^{k_1}\widetilde{\mathfrak{D}}_2^{k_2}\cdots \widetilde{\mathfrak{D}}_{8}^{k_8}\psi.
\end{equation}
(Recall that the additional $\Omega_i$, $i=4,5,6$ were introduced in the case $\mathcal{S}=\emptyset$ to ensure the existence of a convenient
globally defined
spanning set of vectors. In the case $\mathcal{S}\ne\emptyset$, we may understand the above formula  with $\Omega_4=\Omega_5=\Omega_6=0$.)

\subsubsection{Assumption on commutation errors in $r\ge  R_k$}
\label{farawaycommutationerror}

As is to be expected, in order to obtain higher order estimates on~\eqref{inhomogeneous},
we will need to strengthen our asymptotic flatness assumption
to a higher order statement on $g_0$.
Again, instead of formulating sufficient conditions in terms of the decay properties of $g_0$ for large $r$,
it will be convenient here to directly assume exactly the statement we shall need, 
in terms of decay properties of  the coefficients appearing in 
$ [\mathfrak{D}^{\bf k},\Box_{g_0}]\psi$.

Our assumption is thus simply the following: For all $k\ge 1$, there exists an $R_k$ such that 
 the following pointwise  bound for $p=0$ and for $\delta \le p \le 2-\delta$ holds in $r\ge R_k$
\begin{equation}
\label{commutatorerrorformula}
\left| \sum_{|{\bf k}|= k} \Hap[\mathfrak{D}^{\bf k}\psi] [\mathfrak{D}^{\bf k},\Box_{g_0}]\psi \right|
\leq\frac1{12}  \sum_{|{\bf k}|= k}
\Kp[{\mathfrak{D}}^{\bf k}\psi]+ C \sum_{|{\bf k}|\le k-1}
\Kp[{\mathfrak{D}}^{\bf k}\psi]
 \end{equation}
 for all smooth functions $\psi$. 
 
 Again, we emphasise that by our conventions of Section~\ref{constantsandparameterssec}, 
 the constant $C$ on the right hand side of~\eqref{commutatorerrorformula} in general
 depends on $k$.  (In actuality, we need only assume~\eqref{commutatorerrorformula} for all $k\le k_{\rm asympt}$ for some sufficiently large
 $k_{\rm asympt}$,  but   the statement
 of our theorem will then be restricted to such $k$.) 

Assumption~\eqref{commutatorerrorformula} is easily seen to be satisfied in all examples of Section~\ref{seeexamples}.
We note that the $k$ dependence of $R_k$ is in general necessary, even
for Minkowski space, as we must take $R_k \to \infty$ as $k \to \infty$.

\subsubsection{The red-shift commutation}
\label{commutationerrorsnearhor}

We recall the following
\begin{proposition}[\cite{Mihalisnotes}]
\label{fromDR}
Let $Y$ be the vector field of Section~\ref{Yvectorfieldsec}, 
and let $\mathcal{H}^+$ be a Killing horizon with positive surface gravity as assumed in 
Section~\ref{Killinghorizonsection},
such that the generator $Z$ lies in the span of $T$ and $\Omega_1$. Then, along $\mathcal{H}^+$, we have 
\[
 [Y^{k},\Box_{g_0}]\psi  = \kappa_k(\vartheta,\varphi) (Y^{k+1}\psi)+ Y^{k}\Box_{g_0}\psi +
  \sum_{|{\bf k}|\le k+1, k_3< k+1}   \alpha_{\bf k} (\vartheta,\phi)({\mathfrak{D}}^{\bf k} \psi) ,
\]
for a $\kappa_k\ge c>0$.
\end{proposition}

We have
\begin{corollary}
\label{corollaryofnotes}
Let   $\Hap$, $\Hap_\varsigma$ be as defined in Section~\ref{rpsection}.  Then there exist constants $c>0$, $C>0$ (independent of $\varsigma$)
such that along $\mathcal{H}^+$, we have
\begin{equation}
\label{thisistheassumptionimplied}
\Hap [Y^{k}\psi] [Y^{k},\Box_{g_0}]\psi  \geq c (Y^{k+1} \psi)^2 - C \sum_{|{\bf k}|\le k+1, k_3\ne k+1} ({\mathfrak{D}}^{\bf k} \psi)^2.
\end{equation}
\end{corollary}
Again, we emphasise that by our conventions of Section~\ref{constantsandparameterssec}, 
 the constants $c$ and $C$ on the right hand side of~\eqref{thisistheassumptionimplied} in general
 depend on $k$.

\subsubsection{Divergence identity for the higher order master currents}
\label{mastercurrentsandfluxdef}

We define the following positive signature function for $1\le |{\bf k}|\le k$
\begin{equation}
\begin{aligned}
\label{signaturedefinition}
\sigma({\bf k},k)&= \sigma_{12}(k) {\rm\ if\ } k_1+k_2= |{\bf k}| \\
		&=1 {\rm\ otherwise}\\
\end{aligned}
\end{equation}
where $\sigma_{12}$ will be chosen later,  such that
moreover $\sigma_{12}\ge1$.

We define an additional signature function $\varsigma(k)$, for $k\ge 1$ also to be determined later, which will be used
to select the parameter of Proposition~\ref{enhancedinteriorprop} to be used in the currents for higher order estimates.

Let us finally fix a   positive  function 
\begin{equation}
\label{varkappadef}
\varkappa({\bf k}, k)=\varkappa(|{\bf k}|, k) = (\varkappa_0(k))^{1-|{\bf k}|}
\end{equation}
 for a
$\varkappa_0(k) \ge 1$  to be determined later.

Given the above commutation vector fields~\eqref{commutatordefine}, and the notation~\eqref{commutators},
we may now define currents
\begin{equation}
\begin{aligned}
\label{JpkdefKpkdef}
 \Jpk[\psi]&:=\varkappa(0,k) \Jp[\psi]+  \sum_{1\le |{\bf k}|\leq k} \varkappa ({\bf k},k) \sigma({\bf k},k ) \Jp_{\varsigma(k)} [\mathfrak{D}^{\bf k}\psi], \\
 \Kpk[\psi]&:=\varkappa(0,k) \Kp[\psi]+  \sum_{1\le |{\bf k}|\leq k} \varkappa ({\bf k},k) \sigma({\bf k},k ) \Kp_{\varsigma(k)} [ \mathfrak{D}^{\bf k}\psi],
 \end{aligned}
\end{equation}
and given a collection $G_{k}=\{G_{\bf k}\}_{|{\bf k}|\le k} $ of scalar functions, the current
\begin{equation}
\label{Hapkdef}
 \Hapk[\psi]\cdot G_k:=
 \varkappa_0(k)\Hap[\psi]G_0+
 \sum_{1\le |{\bf k}|\leq k}\varkappa ({\bf k},k) \sigma({\bf k}, k) \Hap_{\varsigma( k)}[\mathfrak{D}^{\bf k}\psi]G_{\bf k} ,
\end{equation}
where $\Jp$, $\Kp$ and $\Hap$, as well as  $\Jp_{\varsigma}$, $ \Kp_{\varsigma}$ and  $\Hap_{\varsigma}$, 
 are as defined in Section~\ref{rpsection}.

If $\psi$ satisfies the inhomogeneous equation~\eqref{inhomogeneous},
then $\mathfrak{D}^{\bf k}\psi$ satisfies the equation
\[
\Box_{g_0} (\mathfrak{D}^{\bf k}\psi )= [\mathfrak{D}^{\bf k},\Box_{g_0}]\psi + \mathfrak{D}^{\bf k}F
\]
and the currents satisfy  
\begin{equation}
\label{commutedenergyidentity}
\nabla^\mu \Jpk{}_\mu[\psi] =  \Kpk[\psi] +  \Hapk[\psi]\cdot \{ [\mathfrak{D}^{\bf k},\Box_{g_0}]\psi\} + 
\Hapk[\psi]\cdot \{ \mathfrak{D}^{\bf k}F\},
\end{equation}
which can again be integrated in region $\mathcal{R}(\tau_0,\tau_1,v)$ to yield
\begin{align}
\nonumber
\int_{\Sigma(\tau_1,v)} \Jpk[\psi]\cdot n +\int_{\mathcal{S}(\tau_0,\tau_1)}\Jpk[\psi]\cdot n +
\int_{\underline{C}_v(\tau_0,\tau_1)}\Jpk[\psi]\cdot n +
\int_{\mathcal{R}(\tau_0,\tau_1,v)} \Kpk[\psi] \\
\label{energyidentityhigh}
=
\int_{\Sigma(\tau_0,v)} \Jpk[\psi]\cdot n
-\int_{\mathcal{R}(\tau_0,\tau_1,v)} 
 \Hapk[\psi]\cdot \{ [\mathfrak{D}^{\bf k},\Box_{g_0}]\psi\}
-\int_{\mathcal{R}(\tau_0,\tau_1,v)} \Hapk[\psi]\cdot \{ \mathfrak{D}^{\bf k}F\}.
\end{align}

Note that
\begin{equation}
\label{wherevanishes}
 \Hapk[\psi]\cdot \{ [\mathfrak{D}^{\bf k},\Box_{g_0}]\psi\} =0
 \end{equation}
 in $r_1+(r_2-r_1)/2\le r\le R/2$.

Finally, if 
\begin{equation}
\label{nowcentredvarkappa}
\varkappa_0(k)\gg \varsigma(k)
\end{equation}
is sufficiently large, then we notice that $\Kpk[\psi]\ge 0$ even for $r\le r_{\rm Killing}$, 
and in fact in $r_0\le r\le r_1(\varsigma)$, we have
\begin{align}
\label{noticelower}
&\Kpk[\psi]  \ge c\varkappa_0(k)  \Kp[\psi]\\
\nonumber
&\qquad+c \sum_{1\le |{\bf k}| \le k} \varkappa({\bf k} ,k)   \varsigma \lambda_{\varsigma}(r)  \left( \sum_{i=1}^3  ( \Omega_i\mathfrak{D}^{\bf k} \psi)^2 +(T\mathfrak{D}^{\bf k}\psi)^2   + |r_{\rm Killing}-r|(Y\mathfrak{D}^{\bf k}\psi)^2\right)
 + \varkappa({\bf k} ,k)   \lambda_\varsigma(r)    (Y\mathfrak{D}^{\bf k}\psi)^2,
\end{align}
where we have absorbed the lower order term in~\eqref{enhancedherefinal} with the wrong sign by the largeness of $\varkappa_0(k) = \varkappa({|\bf k}|,k)^{-1}\varkappa(|{\bf k}|-1,k)$ and the positivity of $\Kp[\psi]$, and we have dropped the factor $\sigma({\bf k}, k)$ using that $\sigma\ge 1$.
We will always assume~\eqref{nowcentredvarkappa} in what follows so that~\eqref{noticelower} holds.

Defining
\begin{equation}
\begin{aligned}
\label{frakturepknowenergies}
\Efancypk(\tau)&:= \int_{\Sigma(\tau)}  \Jpk_\mu[\psi] n_{\Sigma(\tau)}^\mu,  \qquad
\Efancypk_{\mathcal{S}}(\tau_0,\tau ) := \int_{\mathcal{S}(\tau_0, \tau)}  \Jpk_\mu[\psi] n_{\mathcal{S}}^\mu , \\
\Ffancypk (v,\tau_0,\tau) &:=\int_{\underline{C}_v\cap \mathcal{R}(\tau_0,\tau)} \Jpk_\mu[\psi] n_{\underline{C}_v}^\mu,
\end{aligned}
\end{equation}
note that by  construction
\[
\Efancypk(\tau) \ge 0, \qquad \Efancypk_{\mathcal{S}}(\tau_0,\tau )\ge 0, \qquad \Ffancypk(v,\tau_0,\tau) \ge 0,
\]
and thus, from~\eqref{energyidentityhigh} 
we have in particular 
\begin{align}
\nonumber
\sup_{v:\tau_1\le \tau(v)}
\Ffancypk(v,\tau_0,\tau_1) +
\Efancypk(\tau)  + \Efancypk_{\mathcal{S}}(\tau_0,\tau )  +
\int_{\mathcal{R}(\tau_0,\tau_1)} \Kpk[\psi]& \leq\Efancypk(\tau_0) 
-\int_{\mathcal{R}(\tau_0,\tau_1)\cap \{r\le r_2\} } 
 \Hapk[\psi]\cdot \{ [\mathfrak{D}^{\bf k},\Box_{g_0}]\psi\}\\
 \nonumber
 &\qquad
+\int_{\mathcal{R}(\tau_0,\tau_1)\cap \{r\ge R/2\} } 
 | \Hapk[\psi]\cdot \{ [\mathfrak{D}^{\bf k},\Box_{g_0}]\psi\}|\\
 &\qquad
 \label{globalidentone}
+\int_{\mathcal{R}(\tau_0,\tau_1)} |\Hapk[\psi]\cdot \{ \mathfrak{D}^{\bf k}F\}|,
\end{align}
provided the terms on the right hand side of the above inequalities  are suitably integrable.

\subsubsection{Elliptic estimates for $\Box_{g_0}\psi=F$}
\label{ellipticlinear}

Since in the region $r\le 3R/4$ our commutation operators $\mathfrak{D}^{\bf k}$ do not necessarily span the tangent space, we will need
to also invoke elliptic estimates. 
These  will allow one to estimate, for solutions of $\Box_{g_0}\psi=F$, 
all highest order derivatives  $\widetilde{\mathfrak{D}}^{\bf k}\psi$ from only  derivatives $\mathfrak{D}^{\bf k}\psi$ with respect to 
our commutation operators together with appropriate terms involving $F$. 
These estimates require integration over hypersurfaces  $\Sigma(\tau)$ or spacetime regions $\mathcal{R}(\tau_0,\tau_1)$, 
and rely on the fact
that the commutation operators  $\mathfrak{D}_i$ span an appropriate timelike direction.

For estimates at order $k$, we will in fact more generally need spacetime elliptic estimates in  $r\le 9R_k/8$, despite
the fact that the  $\mathfrak{D}^{\bf k}$ span the tangent space as long as $r\ge 3R/4$, 
in order to absorb commutation error
terms  where the formula~\eqref{commutatorerrorformula} does not yet apply.   

We have the following proposition:
\begin{proposition}
\label{ellipticprop}
Let $(\mathcal{M},g_0)$ satisfy the assumptions of Section~\ref{backgroundgeom}.
Let $\psi$ be a solution of the inhomogeneous equation~\eqref{inhomogeq} in $\mathcal{R}(\tau_0,\tau_1)$
and let $\tau_0\le \tau'\le \tau_1$.  
Then for all $k\ge 1$ and for all $r_{\rm Killing}< r'_-<r'<r''<r''_+\le R$,
\begin{equation}
\label{ellipticestimnoY}
\int_{\Sigma(\tau')\cap \{ r' \le r\le r'' \}  } \sum_{|{\bf  k}|\le k+1}
(\widetilde{\mathfrak{D}}^{\bf {{k}}}\psi)^2\\
 \lesssim 
\int_{\Sigma(\tau')\cap \{r'_- \le r\le  r''_+\} } \sum_{1\le |{\bf k}|\le k+1, k_1+k_2\ge |{\bf k}|-1} (\mathfrak{D}^{\bf k}\psi)^2 
+\sum_{|{\bf  k}|\le 1}
(\widetilde{\mathfrak{D}}^{\bf{k}}\psi)^2
+\sum_{|{\bf  k}|\le k-1} (\widetilde{\mathfrak{D}}^{\bf  k}F)^2,
\end{equation}
where here $\lesssim$ depends on the choice of $r'_-<r'<r''<r_+''$. 
We also have the estimate, for all $r' \le r_1$,
\begin{equation}
\label{ellipticestimnorestriction}
\int_{\Sigma(\tau')\cap \{r\ge r' \}  } \sum_{|{\bf  k}|\le k+1}
(\widetilde{\mathfrak{D}}^{\bf {{k}}}\psi)^2\\
 \lesssim 
\int_{\Sigma(\tau')\cap \{r\ge r' \} } \sum_{|{\bf k}|\le  k+1} (\mathfrak{D}^{\bf k}\psi)^2 
+\sum_{|{\bf  k}|\le k-1} (\widetilde{\mathfrak{D}}^{\bf  k}F)^2.
\end{equation}

Note that the analogous statements to~\eqref{ellipticestimnoY}, \eqref{ellipticestimnorestriction} 
with integration on $\mathcal{R}(\tau_0,\tau_1) \cap \{r'\le r \le r''\}$, etc.,
follows immediately in view of the coarea formula.

In fact, even without assuming $r''_+\le R$, we still have the spacetime elliptic estimate:
\begin{equation}
\label{ellipticestimnoYspacetime}
\int_{\mathcal{R}(\tau_0,\tau_1)\cap \{ r' \le r\le r'' \}  } \sum_{|{\bf  k}|\le k+1}
(\widetilde{\mathfrak{D}}^{\bf {{k}}}\psi)^2\\
 \lesssim 
\int_{\mathcal{R}(\tau_0,\tau_1)\cap \{r'_- \le r\le  r''_+\} } \sum_{1\le |{\bf k}|\le k+1, k_1+k_2\ge |{\bf k}|-1} (\mathfrak{D}^{\bf k}\psi)^2 
+\sum_{|{\bf  k}|\le 1}
(\widetilde{\mathfrak{D}}^{\bf{k}}\psi)^2
+\sum_{|{\bf  k}|\le k-1} (\widetilde{\mathfrak{D}}^{\bf  k}F)^2,
\end{equation}
where again $\lesssim$ depends on the choice of $r'_-<r'<r''<r_+''$.

\end{proposition}

\begin{proof}
Estimate~\eqref{ellipticestimnoY} is a standard elliptic estimate, using the fact that the span of
$T, \Omega_1$ always  contains a timelike direction
in the region $r\ge r'_{-} > r_{\rm Killing}$. Estimate~\eqref{ellipticestimnorestriction} then follows from the previous in view
of the fact that  the $\mathfrak{D}^k$
span the tangent space for $r\le r_1$ and $r\ge R/2$.

As we have remarked, in 
the case $r''_+\le R$, estimate~\eqref{ellipticestimnoYspacetime} follows from~\eqref{ellipticestimnoY} by the coarea formula. 
It can be obtained more generally, even if $R\le r''$ or $R\le r''_+$ by a suitable integration by parts argument. We sketch this here for 
the case $k=1$.  Let us consider the most difficult case where $R< r'' <r''_+$. One first does the usual elliptic estimates on $\Sigma(\tau) \cap \{r\le R\}$, then
integrating over $\tau$,
to estimate the left hand side of~\eqref{ellipticestimnoYspacetime} integrated over $\mathcal{R}(\tau_0,\tau_1)\cap \{r \le R\}$ from the right hand
side, but with an additional boundary term on $r=R$. Squaring the inhomogeneous wave equation in $r\ge R$ 
and integrating by parts along the null cone $\Sigma(\tau)\cap \{r \ge R\}$, one may again bound the left hand side of~\eqref{ellipticestimnoYspacetime}, integrated over
$\mathcal{R}(\tau_0, \tau_1)\cap \{R\le r\le r''\}$, from the right hand side and an additional boundary term on $r=R$. These boundary
terms can be arranged to cancel modulo a term of the form $\int_{\mathcal{R}(\tau_0,\tau_1)\cap\{r=R\}}  T\psi \slashed\Delta \psi$,
where $\slashed\Delta$ denotes the Laplacian on $\Sigma(\tau)\cap \{r=R\}$, and this term can in turn be related
to a bulk integral which can again be controlled by the right hand side of~\eqref{ellipticestimnoYspacetime}.
\end{proof}

\subsubsection{Global control of the commutation errors}
\label{globalcontrolsection}
We may now give the final statement allowing for the global
control of the commutation error terms in the identity~\eqref{commutedenergyidentity}.

\begin{proposition}
\label{finalcoercivityprop}
Under the assumptions of Section~\ref{farawaycommutationerror},
 we may choose our weight functions $\sigma_{12}(k)$ in the definition~\eqref{signaturedefinition}, $\varkappa_0(k)$
 in the definition~\eqref{varkappadef}  and  $\varsigma(k)$,
so that
the following holds:

For all $k\ge 1$, there exists $r_{\rm Killing}<r_1(k)\le r_1$, such that, 
for all
 $\psi$ satisfying the inhomogeneous equation~\eqref{inhomogeq} in $\mathcal{R}(\tau_0,\tau_1)$, the following estimates hold:
\begin{align}
  \label{commuterrorone}
-  \Hapk[\psi]\cdot \{ [\mathfrak{D}^{\bf k},\Box_{g_0}]\psi\} 
&\leq
 \frac13\Kpk[\psi]
 , \qquad r_0\le r \le r_1(k)
 \\
  \label{commuterroronenhalf}
\int_{\mathcal{R}(\tau_0,\tau_1)\cap \{ r_1(k) \le r_2 \}  }  
| \Hapk[\psi]\cdot \{ [\mathfrak{D}^{\bf k},\Box_{g_0}]\psi\}| 
 &\leq
 \int_{\mathcal{R} (\tau_0,\tau_1)\cap \{r_{\rm Killing} \le r \le r_2\} }\frac13\Kpk[\psi]+ C\sum_{|{\bf  k}|\le k-1} |\widetilde{\mathfrak{D}}^{\bf  k}F|^2,\\
 \label{commuterrortwo}
  \Hapk[\psi]\cdot \{ [\mathfrak{D}^{\bf k},\Box_{g_0}]\psi\} 
 &=0, \qquad r_1+\frac12(r_2-r_1) \le r \le R/2\\
 \label{commuterrorthree}
 \int_{\mathcal{R}(\tau_0,\tau_1)\cap \{ R/2 \le r \le R_k  \}  }  
| \Hapk[\psi]\cdot \{ [\mathfrak{D}^{\bf k},\Box_{g_0}]\psi\} |
 &\leq
 \int_{\mathcal{R} (\tau_0,\tau_1)\cap \{ R/4 \le r\le 9R_k/8 \} }\frac13\Kpk[\psi]+C\sum_{|{\bf  k}|\le k-1} |\widetilde{\mathfrak{D}}^{\bf  k}F|^2,
 \\
 \label{commuterrorfour}
| \Hapk[\psi]\cdot \{ [\mathfrak{D}^{\bf k},\Box_{g_0}]\psi\} |
 &\leq
\frac1{10}\Kpk[\psi], \qquad r\ge  R_k.
\end{align}
\end{proposition}
\begin{proof}

We first prove~\eqref{commuterrorfour}. Here we will use~\eqref{commutatorerrorformula}.  
We have that the left hand side of~\eqref{commuterrorfour} is bounded in $r\ge R_k$ by
\begin{align*}
\left| \Hapk[\psi]\cdot \{ [\mathfrak{D}^{\bf k},\Box_{g_0}]\psi\}_{|{\bf k}|\le k} \right|
&= \left| \sum_{|{\bf k}|\le k} \varkappa({\bf k}, k) \Hap[\mathfrak{D}^{\bf k} \psi] [\mathfrak{D}^{\bf k},\Box_{g_0}]\right|\\
&\leq   \frac1{12}  \sum_{|{\bf k}|\le  k}\sigma({\bf k},k)\varkappa({\bf k},k) 
\Kp[{\mathfrak{D}}^{\bf k}\psi]+  C \sum_{|{\bf k}|\le k-1}\varkappa(|{\bf k}|+1, k) \sigma({\bf k},k)
\Kp[{\mathfrak{D}}^{\bf k}\psi]
 \\
 &\leq   \frac1{12}  \sum_{|{\bf k}|\le  k}\sigma({\bf k},k)\varkappa({\bf k},k) 
\Kp[{\mathfrak{D}}^{\bf k}\psi]\\
&\qquad +  C \sum_{|{\bf k}|\le k-1}(\varkappa(|{\bf k}|+1, k) \varkappa^{-1} (|{\bf k}|, k)) \varkappa (|{\bf k}|, k)\sigma({\bf k},k)
\Kp[{\mathfrak{D}}^{\bf k}\psi]
 \\
&\leq \frac1{12} \Kpk[\psi] +\frac1{100} \Kpk[\psi]\leq \frac1{10} \Kpk[\psi]
\end{align*}
provided that  $\varkappa({\bf k},k)$ is defined so that
\begin{equation}
\label{varkappaconstraint}
\varkappa_0^{-1} (k) \ll 1.
\end{equation}

To prove~\eqref{commuterrorone}, we note that in the region $r_0\le r\le r_1$, we may estimate:
\begin{align}
\label{labelonthefirst}
- & \Hapk[\psi]\cdot \{ [\mathfrak{D}^{\bf k},\Box_{g_0}]\psi\}  
= 
-  \sum_{1\le |{\bf k}|\leq k, k_3<k}\varkappa (|{\bf k}|,k) \Hap_{\varsigma( k)}[\mathfrak{D}^{\bf k}\psi] \,
[\mathfrak{D}^{\bf k},\Box_{g_0}]\psi
 - \varkappa( k,k)  \Hap_{\varsigma(k)} [Y^k\psi] [Y^k,\Box_{g_0}]\psi  \\
 \nonumber
&\le C\varsigma(k)^{\frac12} \lambda_{\varsigma(k)} (r) 
 \sum_{1\le |{\bf k}|\le k} \varkappa(|{\bf k}|,k) \left( (T{\mathfrak{D}}^{\bf k} \psi)^2 +\sum_{i=1}^3 (\Omega_i{\mathfrak{D}}^{\bf k} \psi)^2 +
  ({\mathfrak{D}}^{\bf k} \psi)^2\right)\\
  & \qquad +
 C\varsigma(k)^{-\frac12}  \lambda_{\varsigma(k)} (r)  \sum_{1\le |{\bf k}|\le k+1}  \varkappa(|{\bf k}|-1,k)({\mathfrak{D}}^{\bf k} \psi)^2   \\
 \nonumber
 &\qquad
 - c \lambda_{\varsigma(k)} (r)  \varkappa(k,k) 
 (Y^{k+1}\psi )^2  +C \lambda_{\varsigma(k)} (r)  |r-r_{\rm Killing}| \varkappa(k,k) \sum_{1\le |{\bf k}|\le k+1}({\mathfrak{D}}^{\bf k} \psi)^2 .
\end{align}
Here,  we have applied~\eqref{thisistheassumptionimplied} from 
Corollary~\ref{corollaryofnotes},  and
the bounds on $\Hap_{\varsigma(k)}$  following from~\eqref{Hbound} of Proposition~\ref{enhancedinteriorprop}. 
Note the dependence on $\varsigma(k)$ through
$\lambda_{\varsigma(k)}(r)$,
and note that the terms in 
 the commutation identity not present on $\mathcal{H}^+$ appear with an extra $(r-r_{\rm Killing})$ factor. 
(We emphasise again the conventions from Section~\ref{constantsandparameterssec} that constants $C$, $c$, also depends on $k$!)

We now have     that  given $\varsigma=\varsigma(k)$, from~\eqref{noticelower}, it follows that
 in the region $r_0\le r\le r_1(\varsigma)$,
 provided that~\eqref{nowcentredvarkappa} holds,
we have 
\begin{align*}
\varsigma(k)^{\frac12} \lambda_{\varsigma(k)} (r) 
   \sum_{1\le |{\bf k}|\le k}  \varkappa(|{\bf k}|, k)  \left( (T{\mathfrak{D}}^{\bf k} \psi)^2 +\sum_{i=1}^3 (\Omega_i{\mathfrak{D}}^{\bf k} \psi)^2 +
  ({\mathfrak{D}}^{\bf k} \psi)^2\right)  
\le  C \varsigma^{-\frac12 }(k) \Kpk[\psi] .
  \end{align*}
\begin{align*}
\varsigma^{-\frac12}(k)  \lambda_{\varsigma(k)} (r)  \sum_{1\le |{\bf k}|\le k+1}  \varkappa(|{\bf k}|-1)({\mathfrak{D}}^{\bf k} \psi)^2
\le  C \varsigma^{-\frac12 }(k) \Kpk[\psi] .
  \end{align*}
  
  Finally, again provided~\eqref{nowcentredvarkappa} holds, and provided we further restrict $r_1(\varsigma)$ such that $0<r_1(\varsigma)-r_{\rm Killing}\leq \varsigma^{-1}$,
   we have
  \begin{align*}
 \lambda_{\varsigma(k)} (r)  |r-r_{\rm Killing}| \varkappa(k,k) \sum_{|{\bf k}|\le k+1}({\mathfrak{D}}^{\bf k} \psi)^2 \le  C\varsigma^{-1}(k)  \Kpk[\psi]  .
 \end{align*}

It follows that, for $r_0\le r\le r_1(\varsigma)$, we have 
 \begin{equation}
 \label{thisyields}
- \Hapk[\psi]\cdot \{ [\mathfrak{D}^{\bf k},\Box_{g_0}]\psi\}  
 \le   C \varsigma^{-\frac12}(k) \Kpk[\psi].
 \end{equation}
  Now we fix $\varsigma=\varsigma(k)$ to be
  sufficiently large such that $C\varsigma^{-\frac12}(k) \le \frac13$. 
  This yields~\eqref{commuterrorone}.

Since  $\varsigma(k)$ is now fixed, as well as $r_1(k):=r_1(\varsigma(k))$,
according to our conventions, the dependence on these parameters may now be absorbed into the constants $C$, $c$, etc.

To prove~\eqref{commuterroronenhalf},
we estimate
\begin{align*}
\sum _{|{\bf k}|\le k} \int_{\mathcal{R}(\tau_0,\tau_1)\cap \{r_1(k) \le r\le  r_2\} } &|  \Hapk[\psi]\cdot \{ [\mathfrak{D}^{\bf k},\Box_{g_0}]\psi\} |\\
&\le
C \int_{\mathcal{R} (\tau_0,\tau_1)\cap \{ r_1(k)  \le r\le r_2\} }
 \sum_{|{\bf k}|\le k+1}  \varkappa(|{\bf k}|-1, k) (\widetilde{\mathfrak{D}}^{\bf k}\psi)^2  \\
 &\le C \int_{\mathcal{R} (\tau_0,\tau_1)\cap \{ r_{\rm Killing}  \le r\le r_2\} }
 \sum_{1\le |{\bf k}|\le k+1, k_1+k_2\ge |{\bf k}|-1}  \varkappa(|{\bf k}|-1, k) ({\mathfrak{D}}^{\bf k}\psi)^2\\
 &\qquad\qquad +   \sum_{|{\bf k}|\le 1}(\widetilde{\mathfrak{D}}^{\bf k}\psi)^2  
 +C \sum_{|{\bf  k}|\le k-1} (\widetilde{\mathfrak{D}}^{\bf  k}F)^2 \\
  &\le C \int_{\mathcal{R} (\tau_0,\tau_1)\cap \{ r_{\rm Killing}  \le r\le r_2\} }
\sigma_{12}^{-1}(k) \Kpk[\psi] +   C\varkappa^{-1}(0,k)\Kpk[\psi]
 +C \sum_{|{\bf  k}|\le k-1} (\widetilde{\mathfrak{D}}^{\bf  k}F)^2. 
\end{align*}
(We emphasise again the conventions from Section~\ref{constantsandparameterssec} that constants $C$ also depends on $k$.)
We have used the spacetime version of estimate~\eqref{ellipticestimnoY} of Proposition~\ref{ellipticprop} in the above (i.e.~\eqref{ellipticestimnoYspacetime}), with 
$r'_-:= r_{\rm Killing} +\frac12  (r_1(k)-r_{\rm Killing})$, $r''_+:= r_1+3(r_2-r_1)/4$, $r'= r_1(k)$, $r''= r_1+(r_2-r_1)/2$, where
we also use that the integrand on the left hand side is nonzero only in $r'\le r\le r''$.

Requiring now $\sigma_{12}(k)\gg 1$ to be sufficiently large and  $\varkappa(0,k)=\varkappa_0(k)\gg 1$ to be sufficiently large, we obtain~\eqref{commuterroronenhalf}.

The proof of~\eqref{commuterrorthree} is analogous to~\eqref{commuterroronenhalf},
and also constrains $\sigma_{12}(k)$, $\varkappa_0(k)$ to be sufficiently large. 
This finally fixes $\sigma_{12}(k)$ and $\varkappa_0(k)$.
Note here we use~\eqref{ellipticestimnoYspacetime}
with $r''= R_k$ and $r''_+=9R_k/8$ and in general we may have $9R_k/8\ge R$. 

Identity~\eqref{commuterrortwo} follows immediately from the definition of the commutation
vector fields $\mathfrak{D}$.
\end{proof}

In what follows, we shall consider $\varsigma$, $\sigma$ and $\varkappa$ fixed so as to satisfy the above Proposition.
Thus, from now on, dependence of constants on $\varsigma$, $\sigma$ and $\varkappa$ will be absorbed in the $C$ and
$\lesssim$ notation.
We emphasise again that, according to our conventions from Section~\ref{constantsandparameterssec}, 
all our constants $c$, $C$ will in general depend on $k$. 

We have the following immediate corollary of Proposition~\ref{finalcoercivityprop}:
\begin{corollary}
\label{globalonesided}
Let $k\ge 1$ and let
 $\psi$ be a solution of the inhomogeneous equation~\eqref{inhomogeq} in $\mathcal{R}(\tau_0,\tau_1)$.
Then
we have the global (one-sided)  bound:
\begin{equation}
\begin{aligned}
  \label{finalcommuter}
\int_{\mathcal{R} (\tau_0,\tau_1)  \cap \{r\le r_2\} }  &
-  \Hapk[\psi]\cdot \{ [\mathfrak{D}^{\bf k},\Box_{g_0}]\psi\} 
+\int_{\mathcal{R} (\tau_0,\tau_1)  \cap \{r\ge r_2\} }  
| \Hapk[\psi]\cdot \{ [\mathfrak{D}^{\bf k},\Box_{g_0}]\psi\} |\\
 &\leq
 \int_{\mathcal{R} (\tau_0,\tau_1)\cap\{r\le r_2\} \cap  \{r\ge R/4\} }\frac12\Kpk[\psi]+ 
  C\sum_{|{\bf  k}|\le k-1} |\widetilde{\mathfrak{D}}^{\bf  k}F|^2.
 \end{aligned}
\end{equation}
\end{corollary}
\begin{proof}
Note that the presence of $\frac12$ instead of $\frac13$ is due to the fact that both estimates~\eqref{commuterrorthree}
and~\eqref{commuterrorfour} borrow from the bulk of the region $r\ge R_k$. Note that restricted to the region $r\ge R$, if desired, we may of course
replace the $\widetilde{\mathfrak{D}}$ commutation on the last term on the right hand side by $\mathfrak{D}$.
\end{proof}

\subsubsection{The higher order energy notation}
\label{higherorderenergynot}

We are now ready to derive the higher order weighted estimates which will allow
us in particular to address nonlinear problems.
We will turn to the estimates themselves in Section~\ref{finalhighestimlinsec}.
In the meantime, let us introduce some notation below.

We proceed to summarise the definitions:
Given $k\ge 1$,  $1\le\tau_0\le \tau$, $v$,
and a spacetime function $\psi$, we define the following energies. 

In general, our fundamental  energies without degeneration functions $\chi$ or $\rho$
will be defined with $\widetilde{\mathfrak{D}}^{\bf k}$ as commutation vector fields, i.e.~they will contain all derivatives at the relevant
order.
Specifically, we define first the unweighted energies:
\begin{align}
\label{firstdefk}
\Ezerok(\tau)&:=  \sum_{|{\bf k}|\leq k} 
   \int_{\Sigma(\tau)} (L\widetilde{\mathfrak{D}}^{\bf k}\psi)^2+|\slashed\nabla \widetilde{\mathfrak{D}}^{\bf k}\psi|^2 + \iota_{r\le R} (\underline{L}\widetilde{\mathfrak{D}}^{\bf k}\psi)^2+r^{-2}(\widetilde{\mathfrak{D}}^{\bf k}\psi)^2,\\
\label{extradefk}
\Ezerok_{\mathcal{S}}(\tau_0,\tau)&:= \sum_{|{\bf k}|\leq k} 
  \int_{\mathcal{S}(\tau_0,\tau)} (L\widetilde{\mathfrak{D}}^{\bf k}\psi)^2+|\slashed\nabla\psi|^2 +
   (\underline{L}\widetilde{\mathfrak{D}}^{\bf k}\psi)^2+(\widetilde{\mathfrak{D}}^{\bf k}\psi)^2,\\
\label{seconddefk}
\Fzerok(v, \tau_0,\tau)&: = \sum_{|{\bf k}|\leq k}
   \int_{\underline{C}_v\cap \mathcal{R}(\tau_0,\tau)} (\underline{L}\widetilde{\mathfrak{D}}^{\bf k}\psi)^2
   +(\slashed\nabla\widetilde{\mathfrak{D}}^{\bf k}\psi)^2+r^{-2}(\widetilde{\mathfrak{D}}^{\bf k}\psi)^2,\\
\Ezerominusoneminusdeltak'(\tau)&:= \sum_{|{\bf k}|\leq k}
\int_{\Sigma(\tau)} 
r^{-1-\delta} \left((L\widetilde{\mathfrak{D}}^{\bf k}\psi)^2+(\underline{L}\widetilde{\mathfrak{D}}^{\bf k}\psi)^2
+|\slashed\nabla\widetilde{\mathfrak{D}}^{\bf k}\psi|^2\right)+r^{-3-\delta}(\widetilde{\mathfrak{D}}^{\bf k}\psi)^2.
\end{align}
We now define (in analogy with~\eqref{coeffshereand}--\eqref{coeffsheretootoosmaller}) the higher order $p$-weighted energies for $\delta\le p\le 2-\delta$:
  \begin{align}
  \label{coeffshereandkorder}
  \Epk(\tau) &: = \Ezerok(\tau)+  \sum_{|{\bf k}|\leq k} \int_{\Sigma(\tau)\cap \{r\ge R\} } r^p (r^{-1}L(r\widetilde{\mathfrak{D}}^{\bf k}\psi))^2 + r^{\frac{p}2}(L\widetilde{\mathfrak{D}}^{\bf k}\psi)^2
  +r^{\frac{p}2-2}(\widetilde{\mathfrak{D}}^{\bf k}\psi)^2, \\
     \label{weightedpfluxkorder}
 \Fpk(v,\tau_0,\tau) &:= \Fzerok(v) +     \sum_{|{\bf k}|\leq k}  \int_{\underline{C}_v\cap \mathcal{R}(\tau_0,\tau)} r^p|\slashed\nabla\widetilde{\mathfrak{D}}^{\bf k}\psi|^2 + r^{p-2}(\widetilde{\mathfrak{D}}^{\bf k}\psi)^2,\\
      \label{coeffsheretootoosmallerkorder}
   \Epminusonek'(\tau) &: = \Ezerominusoneminusdeltak'(\tau)+
    \sum_{|{\bf k}|\leq k}  \int_{\Sigma(\tau) } r^{p-1} \left( (r^{-1}L(r\widetilde{\mathfrak{D}}^{\bf k}\psi))^2
     + (L\widetilde{\mathfrak{D}}^{\bf k}\psi)^2 
     +      |\slashed\nabla\widetilde{\mathfrak{D}}^{\bf k} \psi|^2\right)
  +r^{p-3}(\widetilde{\mathfrak{D}}^{\bf k}\psi)^2 .
  \end{align}
Note the important relation:
\begin{equation}
\label{higherorderfluxbulkrelation}
\Epk \gtrsim \Epprimek, \, \, \Fpk \gtrsim \Fpprimek {\rm\ \ for\ } p\ge p'\ge \delta {\rm\ or\ }p'=0, \qquad  
  \Epminusonek' \gtrsim \Epminusonek  {\rm\ for\ } p\ge 1+\delta, \qquad   \Epminusonek' \gtrsim \Ezerok  {\rm\ for\ }p\ge 1 ,
\end{equation}
representing the higher order analogue of~\eqref{fluxbulkrelation}

In contrast, we define the higher order energies carrying the degenerations
functions $\chi$, $\rho$ and $\tilde\rho$ in terms of the $\mathfrak{D}$ commutators:
\begin{align}
{}^\chi\Ezerominusoneminusdeltak'(\tau)&:= \sum_{|{\bf k}|\leq k} \int_{\Sigma(\tau)}
  r^{-1-\delta}\chi(r) \left( (L\mathfrak{D}^k\psi)^2+(\underline{L}\mathfrak{D}^k\psi)^2+|\slashed\nabla\mathfrak{D}^k\psi|^2\right),\\
  {}^\rho\Ezerominusoneminusdeltak'(\tau)&:= \sum_{|{\bf k}|\leq k} \int_{\Sigma(\tau)}
    r^{-1-\delta}\rho(r) \left(  (L\mathfrak{D}^k\psi)^2+(\underline{L}\mathfrak{D}^k\psi)^2+|\slashed\nabla\mathfrak{D}^k\psi|^2\right),\\
{}^{\tilde\rho} \Ezerominusthreeminusdeltakminusone'(\tau) &:=  \sum_{|{\bf k}|\leq k} \int_{\Sigma(\tau)} 
 \tilde\rho (r)r^{-3-\delta} (\mathfrak{D}^k\psi)^2 
 \end{align}
and then
\begin{align}
   \label{coeffsheretootookorderchi}
 {}^\chi  \Epminusonek'(\tau) &: = {}^\chi\Ezerominusoneminusdeltak'(\tau)+  \sum_{|{\bf k}|\leq k} \int_{\Sigma(\tau)\cap \{r\ge R\} } r^{p-1} \left( (r^{-1}L(r\widetilde{\mathfrak{D}}^{\bf k}\psi))^2
     + (L\widetilde{\mathfrak{D}}^{\bf k}\psi)^2 
     +      |\slashed\nabla\widetilde{\mathfrak{D}}^{\bf k} \psi|^2\right)
  +r^{p-3}(\widetilde{\mathfrak{D}}^{\bf k}\psi)^2 ,\\
     \label{coeffsheretootookorderrho}
 {}^\rho  \Epminusonek'(\tau) &={}^\rho\Ezerominusoneminusdeltak'(\tau)+ \sum_{|{\bf k}|\leq k} \int_{\Sigma(\tau)\cap \{r\ge R\} } r^{p-1} \left( (r^{-1}L(r\widetilde{\mathfrak{D}}^{\bf k}\psi))^2
     + (L\widetilde{\mathfrak{D}}^{\bf k}\psi)^2 
     +      |\slashed\nabla\widetilde{\mathfrak{D}}^{\bf k} \psi|^2\right)
  +r^{p-3}(\widetilde{\mathfrak{D}}^{\bf k}\psi)^2 .
 \end{align}
 (Note that in the integrals over $r\ge R$,  it would not matter whether we use $\widetilde{\mathfrak{D}}$ or $\mathfrak{D}$ as these would here define
 directly comparable energies.) 
 
  In analogy with~\eqref{analogyerror}
 we define
 \begin{align}
  \Eerrorkminusone'(\tau)& := \sum_{|{\bf k}|\leq k} \int_{\Sigma(\tau)} \xi(r)(\mathfrak{D}^k\psi)^2.
\end{align}
 We note the fundamental relation
 \begin{equation}
 \label{glue}
   \Eerrorkminusone' \lesssim  {}^\chi  \Ezerominusoneminusdeltakminusone' + \Ezerominusoneminusdeltakminustwo' 
 \end{equation}
 which follows from our constraints on the support of $\xi$ and the degeneration of $\chi$.

\begin{proposition}
\label{somerelationshere}
For a general smooth function $\psi$, we have the following relations between energies
\begin{equation}
\label{firstrelatsone}
\Efancypk(\tau)\lesssim \Epk(\tau),
\end{equation}
\begin{equation}
\label{firstrelatstwo}
\Efancypk_{\mathcal{S}}(\tau)=\Efancyzerok_{\mathcal{S}}(\tau) \sim \Epk_{\mathcal{S}}(\tau), \qquad
 \Ffancypk(v,\tau_0,\tau)\sim \Fpk(v,\tau_0,\tau)
\end{equation}
\begin{equation}
\begin{aligned}
\label{firstrelatsthree}
\int_{\mathcal{R}(\tau_0,\tau_1)} \Kpk[\psi] +\sum_{|{\bf k}|\leq k} \tilde{A}\xi(r)(\mathfrak{D}^k\psi)^2 &\gtrsim \int_{\tau_0}^{\tau_1}   
{}^\rho \Epminusonek'(\tau')d\tau' + \int_{\tau_0}^{\tau_1} {}^{\tilde\rho} \, \Ezerominusthreeminusdeltakminusone' (\tau') d\tau' , \qquad & 2-\delta\ge p \ge \delta\\
\int_{\mathcal{R}(\tau_0,\tau_1)} \Kpk [\psi]  + \sum_{|{\bf k}|\leq k} \tilde{A}\xi(r)(\mathfrak{D}^k\psi)^2 &\gtrsim \int_{\tau_0}^{\tau_1}    
{}^\rho \Ezerominusoneminusdeltak'(\tau')d\tau' + \int_{\tau_0}^{\tau_1} {}^{\tilde\rho} \, \Ezerominusthreeminusdeltakminusone' (\tau') d\tau', \qquad  &p=0.
\end{aligned}
\end{equation}
For $\psi$ a solution of the inhomogeneous equation $\Box_{g_0}\psi=F$, we have
\begin{equation}
\label{witheliptone}
 \Epk(\tau) \lesssim \Efancypk(\tau) +\int_{\Sigma(\tau)\cap \{ r\le R\}  }\sum_{|{\bf  k}|\le k-1} (\widetilde{\mathfrak{D}}^{\bf  k}F)^2.
\end{equation}
If $\chi=1$ and $\tilde\rho=1$ identically (as in case (i)), 
then 
\begin{align}
\label{withelipttwo}
 \Epminusonek'(\tau)   &\lesssim {}^{\chi} \Epminusonek'(\tau)+  {}^{\tilde\rho} \, \Ezerominusthreeminusdeltakminusone'(\tau)
+ \int_{\Sigma(\tau) \cap \{ r\le R\}  }\sum_{|{\bf  k}|\le k-1} (\widetilde{\mathfrak{D}}^{\bf  k}F)^2,
\quad
2-\delta\ge p\ge \delta \\
\label{witheliptthree}
\Ezerominusoneminusdeltak'(\tau) &\lesssim  {}^{\chi} \Ezerominusoneminusdeltak'(\tau) 
+  {}^{\tilde\rho}\, \Ezerominusthreeminusdeltakminusone'(\tau)
+\int_{\Sigma(\tau) \cap \{ r\le R\}  }\sum_{|{\bf  k}|\le k-1} (\widetilde{\mathfrak{D}}^{\bf  k}F)^2, 
\end{align}
and if $\tilde\rho=1$ identically (i.e.~as in cases (i) and (ii)), then
\begin{align}
\label{witheliptfour}
\Epminusonekminusone'(\tau) &\lesssim  
{}^\rho\Epminusonek'(\tau)+
{}^{\tilde\rho} \Ezerominusthreeminusdeltakminusone'(\tau) +\int_{\Sigma(\tau)\cap \{ r \le R\}  }\sum_{|{\bf  k}|\le k-2} (\widetilde{\mathfrak{D}}^{\bf  k}F)^2, \quad
2-\delta\ge p\ge \delta \\
\label{witheliptfive}
\Ezerominusoneminusdeltakminusone'(\tau)    &\lesssim 
{}^\rho\Ezerominusoneminusdeltak'(\tau)+
 {}^{\tilde\rho} \Ezerominusthreeminusdeltakminusone'(\tau) +\int_{\Sigma(\tau)\cap \{ r \le R\}  }\sum_{|{\bf  k}|\le k-2} (\widetilde{\mathfrak{D}}^{\bf  k}F)^2.
\end{align}
\end{proposition}

\begin{proof}
The relations~\eqref{firstrelatsone}--\eqref{firstrelatsthree}
follow immediately from our coercivity and boundedness assumptions on the currents
while the inequalities~\eqref{witheliptone}--\eqref{witheliptfive} follow easily from the elliptic 
estimate~\eqref{ellipticestimnorestriction}
of Proposition~\ref{ellipticprop}.
\end{proof}

We note the following immediate corollary of the above proposition:
\begin{corollary}
\label{corollaryforhomog}
Let $\psi$ be a solution of $\Box_{g_0}\psi=0$ in $\mathcal{R}(\tau_0,\tau_1)$.
Then calligraphic and fraktur energies on $\Sigma(\tau)$ 
are equivalent
\[
\Epk(\tau) \sim \Efancypk(\tau).
\]
If $\chi=1$ and $\tilde\rho=1$ identically (as in case (i)), then
\[
 \Epminusonek'(\tau)    \sim  {}^{\chi} \Epminusonek'(\tau)+
 {}^{\tilde\rho} \, \Ezerominusthreeminusdeltakminusone'(\tau), \qquad \Ezerominusoneminusdeltak'(\tau) \sim  
 {}^{\chi} \Ezerominusoneminusdeltak'(\tau)+{}^{\tilde\rho} \, \Ezerominusthreeminusdeltakminusone'(\tau) .
\]
If $\tilde\rho=1$ identically (i.e.~as in cases (i) and (ii)), then
\[
\Epminusonekminusone'(\tau)\lesssim {}^\rho\Epminusonek'(\tau)+  {}^{\tilde\rho} \Ezerominusoneminusdeltakminusone'(\tau) , \qquad
\Ezerominusoneminusdeltakminusone'(\tau)  \lesssim {}^\rho\Ezerominusoneminusdeltak'(\tau)+     {}^{\tilde\rho} \Ezerominusthreeminusdeltakminusone'(\tau) .
\] 
\end{corollary}

\subsubsection{The final higher order estimates}
\label{finalhighestimlinsec}

We conclude with our higher order version of Proposition~\ref{finaluncommuted}.

\begin{proposition}
Under the assumptions of Proposition~\ref{finaluncommuted}, 
let us assume the additional assumptions of  Section~\ref{farawaycommutationerror}.

Fix $k\ge 1$.
Then for all $0<\delta\le p \le 2-\delta$ and for all
$\tau_0\le \tau\le \tau_1$ we have the following statement.

Let
$\psi$ be a solution of the inhomogeneous equation~\eqref{inhomogeq} in $\mathcal{R}(\tau_0,\tau_1)$.
Then 
\begin{align}
\nonumber
\sup_{v:\tau\le \tau(v)}
\Ffancyp(v,\tau_0,\tau), & \qquad   \Efancypk(\tau) + \Efancyzerok_{\mathcal{S}}(\tau_0,\tau) + c \int_{\tau_0}^{\tau}{}^\rho\Epminusonek'(\tau')
d\tau' +c \int_{\tau_0}^{\tau} {}^{\tilde\rho} \, \Ezerominusthreeminusdeltakminusone' (\tau') d\tau' \\
\nonumber
&\leq  \Efancypk(\tau_0)+  A \int_{\tau_0}^{\tau}\Eerrorkminusone'(\tau') +\int_{\mathcal{R}(\tau_0,\tau)}\left|\Hapk[\psi]\cdot  \{ \mathfrak{D}^{\bf k}F\}\right|
+C\int_{\mathcal{R}(\tau_0,\tau)}\sum_{|{\bf k}| \le k}({\mathfrak{D}}^{\bf k}F)^2\\
\label{rpidentherefluxescommuted}
&\qquad+C\int_{\mathcal{R}(\tau_0,\tau)\cap\{r\le R\}}\sum_{|{\bf k}|\le k-1}({\widetilde{\mathfrak{D}}}^{\bf k}F)^2
\end{align}
as well as the estimate
\begin{align}
\nonumber
\sup_{v:\tau\le \tau(v)}\Fpk(v,\tau_0,\tau)&+ \Epk(\tau) +\Ezerok_{\mathcal{S}}(\tau_0,\tau)+ \int_{\tau_0}^{\tau}{}^\chi\Epminusonek'(\tau')
d\tau' + \int_{\tau_0}^{\tau}  \Epminusonekminusone'(\tau') d\tau' \\
\nonumber
&\lesssim  \Epk(\tau_0) +\int_{\mathcal{R}(\tau_0,\tau)}
\sum_{|{\bf k}|\le k}\left(|V^\mu_p({\mathfrak{D}}^{\bf k}\psi)_\mu|+|w_p{\mathfrak{D}}^{\bf k}\psi| \right)|{\mathfrak{D}}^{\bf k}F|+\int_{\mathcal{R}(\tau_0,\tau)}\sum_{|{\bf k}| \le k}({\mathfrak{D}}^{\bf k} {F})^2\\
\label{rpidenbboxcommuted}
&\qquad+C \int_{\mathcal{R}(\tau_0,\tau)\cap \{r\le R\}}\sum_{|{\bf k}| \le k-1}({\widetilde{\mathfrak{D}}}^{\bf k} F )^2
+\int_{\Sigma(\tau)\cap\{r\le R\}}\sum_{|{\bf k}|\le k-1}(\widetilde{\mathfrak{D}}^{\bf k}F)^2.
\end{align}
For $p=0$, identical statements hold with $p-1$ replaced by $-1-\delta$.
\end{proposition}

\begin{remark}
\label{theextratermremarkfirst}
In the case where we replace the middle term of~\eqref{inhomogeneous} 
with~\eqref{replacedterm}, we 
should add 
\[
\sum_{|{\bf k}|\le k} \sqrt{
\int_{\mathcal{R}(\tau_0,\tau)\cap \{ r\le R \}} |L{\mathfrak{D}}^{\bf k}\psi|^2+|\underline L{\mathfrak{D}}^{\bf k}\psi|^2 +|\slashed\nabla{\mathfrak{D}}^{\bf k}\psi|^2 
+r^{-2}|{\mathfrak{D}}^{\bf k}\psi|^2}
\sqrt{\int_{\mathcal{R}(\tau_0,\tau_1)\cap \{ r\le R \}} ({\mathfrak{D}}^{\bf k}F)^2 }
\]
 to the right hand side of~\eqref{rpidenbboxcommuted}.
\end{remark}
 
\begin{proof}
We first prove~\eqref{rpidentherefluxescommuted}.
This  follows from~\eqref{globalidentone}
and applying
Proposition~\ref{finalcoercivityprop} (in the form of Corollary~\ref{globalonesided})
and the properties of Section~\ref{higherorderenergynot}.

To prove~\eqref{rpidenbboxcommuted}, let us first
commute the equation by $T^{\tilde k}$ (and $\Omega_1^{\tilde{k}}$, if assumed Killing) 
and apply the black box estimate~\eqref{rpidenthere} to $T^{\tilde k}\psi$ (and $\Omega_1^{\tilde k}\psi$)
for $\tilde{k}\le k$.
This gives
\begin{equation}
\begin{aligned}
\label{rpidenthereTk}
\sup_v \Fp[T^{\tilde k}\psi] (v,\tau_0,\tau)&+  \Ep[T^{\tilde k}\psi](\tau) + \Ep_{\mathcal{S}}[T^{\tilde k}\psi](\tau) 
+\int_{\tau_0}^{\tau}{}^\chi\Epminusone'[T^{\tilde k}\psi](\tau')
d\tau + \int_{\tau_0}^{\tau} \Eminusone'[T^{\tilde k}\psi] (\tau') d\tau' \\
&\lesssim \Ep[T^{\tilde k}\psi](\tau_0)+  \int_{\mathcal{R}(\tau_0,\tau)}
\left( |r^pr^{-1} L (rT^{\tilde k}\psi)|+|\tilde{V}_p^\mu\partial_\mu T^{\tilde k}\psi |+|\tilde{w}_p\psi| \right) |T^{\tilde k}F|\\
&\qquad
+\int_{\mathcal{R}(\tau_0,\tau)}(T^{\tilde k}F)^2,
\end{aligned}
\end{equation} 
and similarly with $\Omega_1^{\tilde{k}}\psi$.

Now note that in the support of $\xi$, all $\mathfrak{D}$ vanish except for $T$ (and $\Omega_1$), whence we
manifestly have 
\[
  A \int_{\tau_0}^{\tau}\Eerrorkminusone'(\tau')d\tau'
   \lesssim \sum_{\tilde k\le k} \int_{\tau_0}^{\tau}\Eminusone'[T^{\tilde k}\psi](\tau')+\Eminusone'[\upsilon\Omega_1^{\tilde k}\psi](\tau'),
  \]
where we recall that $\upsilon=1$ only if $\Omega_1$ is Killing.
It follows that 
the left hand side of~\eqref{rpidentherefluxescommuted} is bounded by the right hand side 
of~\eqref{rpidenbboxcommuted}.

Now apply  the black box estimate~\eqref{rpidenthere} to $\mathfrak{D}^{\bf k}\psi$, and note that
we can bound the terms arising from $[\Box_{g_0} , \mathfrak{D}^{\bf k}]\psi$ by 
 $c \int_{\tau_0}^{\tau}{}^\rho\Epminusonek'(\tau') +\Ezerominusthreeminusdeltakminusone' (\tau') d\tau'$ 
and the spacetime term involving $\widetilde{\mathcal{D}}^{\bf k}F$ 
on the right hand side of~\eqref{rpidenbboxcommuted}.

We thus have that~\eqref{rpidenthereTk} holds with $\mathfrak{D}^{\bf k}\psi$ in place of $T^k\psi$.

We  apply our elliptic estimate of Proposition~\ref{somerelationshere} to bound~$\Epk(\tau)$ from $\Efancypk(\tau)$,
generating the last term on the right hand side of~\eqref{rpidenbboxcommuted}.

Finally, by~\eqref{ellipticestimnorestriction} of Proposition~\ref{ellipticprop}, we may estimate
\[
 \int_{\tau_0}^{\tau}  \Epminusonekminusone'(\tau') d\tau' \lesssim  \int_{\tau_0}^{\tau}\left(  {}^\rho \Epminusonekminusone'(\tau') + \Ezerominusthreeminusdeltakminusone' (\tau')  \right)d\tau'
 + \int_{\mathcal{R}(\tau_0,\tau)\cap \{r\le R\}}\sum_{|{\bf k}| \le k-1} |{\widetilde{\mathfrak{D}}}^{\bf k} F|^2.
\]
\end{proof}

\begin{remark}
\label{remarkforseminlinear}
Note that for convenience we have used~\eqref{rpidentherefluxescommuted} to obtain~\eqref{rpidenbboxcommuted}. 
This was because
our precise commutation assumptions were stated with respect to the global currents.
If one does not assume estimate~\eqref{unifiedhere}, and thus one does not have~\eqref{rpidentherefluxescommuted},
one may still obtain~\eqref{rpidenbboxcommuted} under the assumption of asymptotic flatness of Sections~\ref{rpsection} 
and~\ref{farawaycommutationerror},
where we insert simply the far-away currents $\Jp_{\rm far}$, $\Kp_{\rm far}$ in~\eqref{commutatorerrorformula}. 
(We also use that, from the assumptions of Section~\ref{backgroundgeom}, we can still define currents as in 
Section~\ref{enhanced} giving enhanced positivity
near $\mathcal{H}^+$ and in the black hole interior.)
Thus, in particular, estimate~\eqref{rpidenbboxcommuted} holds in the Kerr case in the full subextremal range $|a|<M$.
\end{remark}

\subsubsection{Sobolev inequalities and interpolation of $p$-weighted energies}
We end this section recording two easy statements about the energies we have defined.

In anticipation of studying nonlinear equations, one will need to estimate lower order pointwise quantities from higher order energies by
Sobolev inequalities. 
 We have the following

\begin{proposition}
\label{sobolevforfunctions}
Let $\psi$ be a smooth function on a neighbourhood of $\Sigma(\tau)$.
Then  we have
\begin{equation}
\label{simplesobolev}
\sum_{|{\bf k}|\le k-3}
\sup_{x\in \Sigma(\tau) \cap \{r \le R\} }(\widetilde{\mathfrak{D}}^k \psi(x))^2 \lesssim
\min\{ 
\Ezerominusoneminusdeltak'(\tau), \Ezerok(\tau) \}.
\end{equation}
\end{proposition}

\begin{proof}
We note that the left hand side of~\eqref{simplesobolev} is in fact bounded by the energy restricted to $\Sigma\cap \{r\le R\}$, which is why
we may choose either of the two quantities on the right hand side~\eqref{simplesobolev}.
\end{proof}

Weighted Sobolev inequalities will also hold globally on $\Sigma(\tau)$, and in practice these are used to estimate nonlinearities in the region near infinity. 
Because this use is  incorporated in our assumption on the null condition (see already Section~\ref{formofnullcondsec}), we 
shall not need to state such inequalities, although, in practice, 
they will appear in the context of verifying the assumptions of 
Section~\ref{formofnullcondsec}. See already Appendix~\ref{nonlineartermsatinf}.

Finally, we have the following easy 
interpolation result:
\begin{proposition}
\label{interpolationprop}
For $\delta$ as fixed in~\eqref{fixeddelta}, one has 
the following interpolation inequalities
\begin{align}
\label{interpolationstatementnewback}
\Eonek (\tau)
&\lesssim 
  \left( \, \, \Eoneminusdelk  (\tau)\right)^{1-\delta} \left(\, \, \Etwominusdelk (\tau)\right)^{\delta},\\
  \label{interpolationstatementbulk}
 \Edeltaminusonek' (\tau) &\lesssim  \left( \, \, \Ezerominusoneminusdeltak'(\tau) \right)^{\frac{1-\delta}{1+\delta}} \left( \, \, \Ezerok'(\tau) \right)^{\frac{2\delta}{1+\delta}},\\
\label{interpolationwithk}
 \int_{\tau_0}^{\tau_0+\frac12} \Ezerokminusone(\tau) &\lesssim \sqrt{ \int_{\tau_0}^{\tau_0+\frac12} \Ezerok(\tau)}\sqrt{ \int_{\tau_0}^{\tau_0+\frac12} \Ezerokminustwo(\tau)}.
\end{align}
\end{proposition}

\begin{proof} 
The proof of these inequalities is standard and is left to the reader.
\end{proof}

\section{Quasilinear equations: preliminaries, the null condition and local existence} 
\label{prelimsec}

We assume throughout that
$(\mathcal{M}, g_0)$ satisfies the assumptions of
Sections~\ref{backgroundgeom} and~\ref{waveequassump} (for cases (i), (ii) or (iii)).
We will introduce in this section the class of quasilinear equations to be considered in this 
paper, and derive some preliminary results which will be used in the proof of our main theorem.

\subsection{The class of equations}
\label{classstruggle}

\label{classofequations}
We will consider solutions $\psi$ to quasilinear equations  of the form
\begin{equation} 
\label{theequation}
\Box_{g(\psi, x)} \psi = N^{\mu\nu}(\psi,x)\partial_\mu\psi\,\partial_\nu \psi,
\end{equation}
where 
\begin{equation}
\label{thenonlinearityfunctions}
g: \mathbb R\times \mathcal{M} \to T^*\mathcal{M}\otimes T^*\mathcal{M}, \qquad
N:\mathbb R\times \mathcal{M}\to T \mathcal{M}\otimes T \mathcal{M}
\end{equation}
are such that $\pi\circ g(\psi,x)=x$, $\pi\circ N(\psi, x)=x$
where $\pi : T \mathcal{M}\otimes T \mathcal{M}\to \mathcal{M}$, $\pi : T^*\mathcal{M}\otimes T^*\mathcal{M}\to \mathcal{M}$
denote the canonical projections, and 
 such that 
 $g(0,x)=g_0(x)$ for all $x$ while $g(\cdot,x)=g_0(x)$ for $r(x)\ge R/2$, and $N$ and $g$ are smooth maps. 

We will assume moreover that, for each $k$,
$\partial_{{\bf x}}^{{\bf k}}\partial_\xi^s g^{\alpha\beta}(\xi, x)$ and 
$\partial_{{\bf x}}^{{\bf k} } \partial_\xi^s N^{\alpha\beta}(\xi, x)$ are
uniformly bounded for all $|{\bf k}|+s\le k$, all $r\le R$ and $|\xi|\le 1$.
Here $\partial_{{\bf x}}^{{\bf k}}$ denote multi-indices with respect to the ambient Cartesian coordinates
of Section~\ref{thatwhichunderlies}.

For $N^{\mu\nu}(\psi,x)$, we will eventually
need in addition to assume some version of the null condition.
We will only introduce this in Section~\ref{formofnullcondsec} (see already Assumption~\ref{nullcondassumphere}).
Let us note that our assumption on the support of $g-g_0$ is merely so as not to deal with
formulations of the null condition for the quasilinear part.

We will sometimes view the equation in~\eqref{theequation} as an inhomogeneous equation
on a fixed $g_0$ background, i.e.~we will write it as 
\begin{equation}
\label{equationwithloss}
\Box_{g_0} \psi = N^{\mu\nu}(\psi,x)\partial_\mu\psi\,\partial_\nu \psi + (\Box_{g_0}-\Box_{g(\psi,x)})\psi,
\end{equation}
and similarly its commuted versions.
We may thus apply estimates to the inhomogeneous equation
\begin{equation}
\label{inhomogeq}
\Box_{g_0} \psi = F.
\end{equation}

\subsection{Smallness parameters}
\label{explanofsmall}

Starting in this section, we shall introduce smallness parameters (i.e.~parameters related to making smallness assumptions on solutions), 
denoted using the symbol $\varepsilon$ and various subscripts, e.g.~$\varepsilon_{\rm prelim}$, $\varepsilon_{\rm local}$. 
Unless otherwise noted, these 
will  depend only on $(\mathcal{M},g_0)$, 
on the nonlinearities of~\eqref{theequation} defined by~\eqref{thenonlinearityfunctions}
and, in general, on $k$, if there is $k$ dependence in the statement.

When these smallness parameters depend on an additional quantity, this will be indicated in parenthesis, e.g.~$\hat\varepsilon_{\rm slab}(\alpha)$.

Note finally that our convention on constants denoted $C$, $c$, etc.~remains the same as set in Section~\ref{constantsandparameterssec}, i.e.~these will \underline{not} 
depend on~\eqref{thenonlinearityfunctions}.

\subsection{Energy currents for $\Box_{g(\psi,x)}$ and the stability of coercivity properties}
\label{thisisprimitive}

Let $\psi$  denote a solution of~\eqref{theequation}
on a domain $\mathcal{R}(\tau_0,\tau_1)$. 
Because we shall use energy identities connected with $\Box_{g(\psi,x)}$, we shall
require certain basic smallness  assumptions on $\psi$ which ensure that the causal nature
of relevant hypersurfaces is retained and that induced volume forms of $g$ and $g_0$
are comparable.

In the sections to follow,  the main smallness parameter we will consider will be  $\varepsilon_{\rm prelim}>0$.
We emphasise that according to our conventions of Section~\ref{explanofsmall},  $\varepsilon_{\rm prelim}>0$ will in general depend
on the nonlinearity~\eqref{thenonlinearityfunctions}  (and
 also on $k$). 
The parameter  $\varepsilon_{\rm prelim}>0$ can be taken  as fixed everywhere in the paper, but note that
its smallness  is constrained in multiple propositions whose preambles  refer to its existence.

As we shall see, the propositions in this section will always refer  to solutions  satisfying
\begin{equation}
\label{mostprimitive}
\sum_{|{\bf k}|\le 1}  (\widetilde{\mathfrak{D}}^{\bf k}  \psi)^2 \leq \sqrt{\varepsilon}
\end{equation}
in $r\le R$ for $0<\varepsilon\le \varepsilon_{\rm prelim}$. 
(We recall here that for $r(x)\ge R$, we have $g(\psi,x)=g_0$.)
Eventually,~\eqref{mostprimitive} will  be the consequence of a stronger estimate (see already~\eqref{basicbootstrap}).

\subsubsection{Currents for $\Box_{g(x,\psi)}\psi=F$ and the stability of coercivity properties}
\label{stabilityofenergyidentitiessec}

Let us for the moment consider more generally the equation
\begin{equation}
\label{BoxgF}
\Box_{g(\psi,x)}\psi =F
\end{equation}
for arbitrary $F$, where $\psi$ satisfies~\eqref{mostprimitive} for sufficiently small $\varepsilon$.

We may define now the currents
\begin{equation}
\label{psidepcur}
\Jpk[g(\psi,x),\psi], \qquad \Kpk[g(\psi,x),\psi],
\end{equation}
again by the expressions~\eqref{JpkdefKpkdef}, where the constituent~\eqref{generalJdef},~\eqref{generalKdef} 
are defined with the same quadruples~\eqref{thetriple} as before, 
but now with $g=g(\psi,x)$ replacing $g_0$. These currents
satisfy~\eqref{energyidentity},
with respect now to the normals and volume forms of the metric $g$, and where 
 $\Hapk[\psi]$ is defined  by~\eqref{generalHdef}. (Notice that the definition of $\Hapk[\psi]$  does not depend on the
metric.) In the region $r\ge R$ of course, all currents coincide with their $g_0$ versions,
since $g=g_0$ in this region.

We have the following
\begin{proposition}
\label{stabofenergyid}
There exists an $\varepsilon_{\rm prelim}>0$ such that
the following statement holds.

Let $\psi$ be a solution of~\eqref{BoxgF} in $\mathcal{R}(\tau_0,\tau_1)$ 
satisfying~\eqref{mostprimitive} for $0<\varepsilon\le \varepsilon_{\rm prelim}$. Then $x\mapsto g(\psi(x),x)$ defines a Lorentzian metric 
on $\mathcal{R}(\tau_0,\tau_1)$ and  the identity~~\eqref{energyidentity} holds in $\mathcal{R}(\tau_0,\tau_1,v)$, for all $v$ such that $\tau_1\le \tau(v)$,
where
the coefficients, normals and volume forms
are $\varepsilon^{1/4}$ close to those of the currents~\eqref{psidepcur} corresponding to $\Box_{g_0}\phi=F$.

In particular, there exist constants $C$, $c$, such that, for $p=0$ or $\delta\le p\le 2-\delta$, the corresponding 
coercivity properties~\eqref{insymbolsiiwithpweightJequiv} 
are retained for the currents~\eqref{psidepcur}, while~\eqref{insymbolsiiwithpweightK} is replaced by 
\begin{eqnarray}
\nonumber
\Kp[g(\psi,x),\psi] + \tilde{A}\xi(r)\psi^2 &\ge   c  r^{p-1} \rho(r) \left( ({L}\psi)^2 + |\slashed\nabla \psi|^2\right) 
 + c r^{-1-\delta}\rho (r) (\underline L\psi)^2 + c r^{-3} \tilde\rho(r) \psi^2\\
 \label{perturbedcoercive}
 &\qquad -C\varepsilon^{1/4}\,\iota_{\{r\le R\}\cap \{\rho \le C\varepsilon^{1/4} <\frac12\} } 
 \left( (L\psi)^2 +  (\underline L\psi)^2 + |\slashed\nabla \psi|^2 +\psi^2 \right) .
\end{eqnarray} 
In particular, we still have~\eqref{insymbolsiiwithpweightK}  
in $r\ge R$, and also in $\{ r \le R \} \cap \{ \rho \ge C\varepsilon^{1/4} \}.$

Given $\varsigma=\varsigma(k)$ as fixed previously, 
then the above applies also for the currents $\Jp_\varsigma[g(\psi,x),\psi]$, $\Kp_\varsigma[g(\psi,x),\psi]$ in the region $r\ge r_{\rm Killing}$, 
while in the region~$r\le r_1(\varsigma)$,
 the coercivity properties of Propositions~\ref{enhancedprop} and~\ref{enhancedinteriorprop} hold for these currents.
\end{proposition}

\begin{proof}
This is clear from the properties of tensor identities. Note how  the smallness constraint on $\varepsilon_{\rm prelim}$  provided by this
proposition
depends of course on the map $g$ of~\eqref{thenonlinearityfunctions} and is needed even simply to ensure for $g$ to be Lorentzian and for the inverse metric to be well defined.
\end{proof}

Note that since  $\rho=\tilde\rho=1$  in case (i), the bound~\eqref{perturbedcoercive} implies
that the full coercivity applies in that case. (Note that if $\rho$ is a step function valued in $\{0,1\}$, then $\{\rho\ge C\varepsilon^{1/4}\}=\{\rho\ge \frac12\} =\{\rho=1\}$ and $\{\rho \le C\varepsilon^{1/4}\} =\{\rho =0\}$.)

\begin{corollary}
\label{refertothisnow}
Under the assumption of Proposition~\ref{stabofenergyid}, we have the coercivity statement 
\[
\Kpk[g(\psi,x),\psi] \ge \frac89 \Kpk[g_0,\psi]
\]
in $\{r\le r_2\}\cup \{r\ge R/4\}$  and an analogue of the coercivity statement of~\eqref{finalcommuter} holds  in the form
\begin{equation}
\begin{aligned}
  \label{finalcommuterperturbed}
\int_{\mathcal{R} (\tau_0,\tau_1)  \cap \{r\le r_2\} }  &
-  \Hapk[\psi]\cdot \{ [\mathfrak{D}^{\bf k},\Box_{g_0}]\psi\} 
+\int_{\mathcal{R} (\tau_0,\tau_1)  \cap \{r\ge r_2\} }  
| \Hapk[\psi]\cdot \{ [\mathfrak{D}^{\bf k},\Box_{g_0}]\psi\} |\\
 &\leq
 \int_{\mathcal{R} (\tau_0,\tau_1)\cap\{r\le r_2\} \cap  \{r\ge R/4\} }\frac23\Kpk[g(\psi,x), \psi]
  +C\sum_{|{\bf  k}|\le k-1} (\widetilde{\mathfrak{D}}^{\bf  k}(  F+ (\Box_{g_0}-\Box_{g(\psi,x)})\psi))^2.
 \end{aligned}
\end{equation}
We have moreover the following version of~\eqref{firstrelatsthree},
\begin{equation}
\begin{aligned}
\label{firstrelatsthreeperturbed}
\int_{\mathcal{R}(\tau_0,\tau_1)} \Kpk[g(\psi,x),\psi] +\sum_{|{\bf k}|\leq k} \tilde{A}\xi(r)(\mathfrak{D}^k\psi)^2 &\gtrsim \int_{\tau_0}^{\tau_1}   
{}^\rho \Epminusonek'(\tau')d\tau' + \int_{\tau_0}^{\tau_1} {}^{\tilde\rho} \, \Ezerominusthreeminusdeltakminusone' (\tau') d\tau' , \qquad & 2-\delta\ge p \ge \delta\\
\int_{\mathcal{R}(\tau_0,\tau_1)} \Kpk [g(\psi,x), \psi]  + \sum_{|{\bf k}|\leq k} \tilde{A}\xi(r)(\mathfrak{D}^k\psi)^2 &\gtrsim\int_{\tau_0}^{\tau_1}    
{}^\rho \Ezerominusoneminusdeltak'(\tau')d\tau' + \int_{\tau_0}^{\tau_1} {}^{\tilde\rho} \, \Ezerominusthreeminusdeltakminusone' (\tau') d\tau', \qquad  &p=0.
\end{aligned}
\end{equation}
\end{corollary}

\begin{proof}
Note that we have retained $g_0$ on the left hand side of~\eqref{finalcommuterperturbed}, 
so this is simply a statement about the stability of coercivity properties of the first term on the right hand side.
\end{proof}

\subsubsection{Relations of energy fluxes} 
In view of the above we will define a new set of energy quantities, $\Ffancyzerok(g)$, $\Efancypk(g)$, 
etc., defined in parallel with those of $\Ffancyzerok$, $\Efancypk$ 
of Section~\ref{higherorderenergynot} (to be denoted below as  $\Ffancyzerok(g_0)$, $\Efancypk(g_0)$), 
but where the flux is defined
with respect to the divergence identity with respect to $g=g(\psi,x)$, i.e.~we define
\begin{equation}
\begin{aligned}
\label{newfrakturenergies}
\Efancypk(g, \tau):= \int_{\Sigma(\tau)}  \Jpk_\mu[g, \psi] n(g)_{\Sigma(\tau)}^\mu , \quad
\Efancypk_{\mathcal{S}}(g, \tau):= \int_{\mathcal{S}}  \Jpk_\mu[g, \psi] n(g)_{\mathcal{S}}^\mu , \\
\Ffancypk(g, v, \tau_0,\tau) :=\int_{\underline{C}_v\cap \mathcal{R}(\tau_0, \tau)} \Jpk_\mu[g, \psi] n(g)_{\underline{C}_v}^\mu, \tau_0\le \tau \le \tau(v)
\end{aligned}
\end{equation}
where the omitted volume forms above are here understood with respect to $g=g(\psi,x)$.

\begin{corollary}
We have that under the assumptions of Proposition~\ref{stabofenergyid},
for all $\tau_0\le \tau\le \tau_1$, $v$ such that $\tau_1\le \tau(v)$,
\[
\qquad \Efancypk(g,\tau)\sim \Efancypk(g_0,\tau) \lesssim \Epk(\tau),
\qquad \Efancyzero_{\mathcal{S}}(g, \tau) = \Efancyp_{\mathcal{S}}(g, \tau)\sim \Efancypk_{\mathcal{S}}(g_0, \tau)
 =\Efancyzerok_{\mathcal{S}}(g_0, \tau) \sim \Ezerok_{\mathcal{S}}(\tau),
\]
\[
\Ffancypk(g,v, \tau_0, \tau)= \Ffancypk(g_0,v,\tau_0,\tau) \sim \Fpk(v,\tau_0,\tau),
\]
in fact
\[
\Efancyzerok(g) \leq (1+C\varepsilon^{\frac14})  \Efancyzerok (g_0) , \qquad
\Efancyzerok(g_0)  \leq (1+C\varepsilon^{\frac14})  \Efancyzerok(g) ,
\]
\begin{equation}
\label{comparethemlater}
\Epk(\tau) \lesssim \Efancypk[g](\tau) + \int_{\Sigma(\tau)\cap \{r_1 \le r  \le R\} } \sum_{|{\bf  k}|\le k-1}
(\widetilde{\mathfrak{D}}^{\bf  k} (F+ (\Box_{g_0}-\Box_{g(\psi,x)})\psi))^2.
\end{equation}
\end{corollary}

As always, the significance of expressing the above with respect to the energies $\Efancypk(g)$, etc.,  is that the constants in our nonlinear
inequalities will be exactly~$1$.

Note that in all  integrals below with volume form omitted, we will continue to use the volume form induced by $g_0$ unless otherwise noted.
(We shall only use the volume form induced by $g(\psi,x)$ when applying the divergence identity involving~\eqref{newfrakturenergies}.
The corresponding volume forms are of course equivalent under the assumptions of Proposition~\ref{stabofenergyid},
but  again one must distinguish where appropriate so as to obtain the exact constant.)

\subsection{Higher order estimates for the quasilinear equation}
\label{hoeforquas}

Using the above, we  finally 
obtain the following:

\begin{proposition}
\label{higherorderestimquasi}
For all $k\ge 1$.
 There exists an $\varepsilon_{\rm prelim}>0$ such that
the following statement holds.

Let $0<\delta\le p \le 2-\delta$,
$\tau_0\le \tau_1$ and
let $\psi$ be a solution of~\eqref{theequation} in $\mathcal{R}(\tau_0,\tau_1)$ 
satisfying~\eqref{mostprimitive} 
for $0<\varepsilon\le \varepsilon_{\rm prelim}$.  Then for all $\tau\in [\tau_0,\tau_1]$, we have
\begin{align}
\nonumber
\sup_{v:\tau\le \tau(v)} &\Ffancypk(v,\tau_0,\tau),  \qquad   \Efancypk(\tau) + \Efancypk_{\mathcal{S}}(\tau_0,\tau) 
+c \int_{\tau_0}^{\tau}{}^\rho\Epminusonek'(\tau')
d\tau' +c \int_{\tau_0}^{\tau} {}^{\tilde\rho} \, \Ezerominusthreeminusdeltakminusone' (\tau') d\tau'\\
\nonumber
&\leq  \Efancypk(\tau_0)+  A \int_{\tau_0}^{\tau}\Eerrorkminusone'(\tau')d\tau'\\
\label{firstofthenonlinear}
&\qquad+\int_{\mathcal{R}(\tau_0,\tau)}
\left|\Hapk[\psi]\cdot  \{ \mathfrak{D}^{\bf k}(N^{\alpha\beta}(\psi,x)\partial_\alpha\psi\partial_\beta \psi)\}\right|\\
\label{toexpand}
&\qquad
+\int_{\mathcal{R}(\tau_0,\tau)\cap \{r\le R\}}\left|\Hapk[\psi]\cdot  \{ [\Box_{g(\psi,x)}-\Box_{g_0},\mathfrak{D}^{\bf k}]\psi\}\right|
\\
\nonumber
&\qquad
+C \int_{\mathcal{R}(\tau_0,\tau)}\sum_{|{\bf k}| \le k} 
({{\mathfrak{D}}}^{\bf k} (N^{\alpha\beta}(\psi,x)\partial_\alpha\psi\partial_\beta \psi))^2+
C\int_{\mathcal{R}(\tau_0,\tau)\cap \{r\le R\}} \sum_{|{\bf k}| \le k} \left( [\Box_{g(\psi,x)}-\Box_{g_0},\mathfrak{D}^{\bf k}] \psi
\right)^2
\\
\label{nonlinearrpidentherefluxescommuted}
&\qquad+
C \int_{\mathcal{R}(\tau_0,\tau)\cap \{r\le R\}}\sum_{|{\bf k}| \le k-1} (
\widetilde{\mathfrak{D}}^{\bf k} 
(N^{\alpha\beta}(\psi,x)\partial_\alpha\psi\,\partial_\beta \psi+(\Box_{g_0}-\Box_{g(\psi,x)})\psi))^2\\
\label{failureofcoercivity}
&\qquad+C\boxed{\varepsilon^{1/4}} \int_{\mathcal{R}(\tau_0,\tau)\cap \{r\le R\} \cap \{\rho\le C\sqrt{\varepsilon}\}}
\sum_{|{\bf {k}}|\le k} \left( (L(\widetilde{\mathfrak{D}}^{\bf k} \psi))^2
+(\underline{L}(\widetilde{\mathfrak{D}}^{\bf k}\psi))^2+|\slashed\nabla(\widetilde{\mathfrak{D}}^{{\bf k}}\psi)|^2\right)+\psi^2
\end{align}
as well as the estimate
\begin{align}
\label{whyanumberidontknow}
\sup_{v:\tau\le \tau(v)} & \Fpkminusone(v, \tau_0,\tau)+ \Epkminusone(\tau) + \Ezerokminusone_{\mathcal{S}}(\tau_0,\tau)
+\int_{\tau_0}^{\tau}{}^\chi\Epminusonekminusone'(\tau')
d\tau' + \int_{\tau_0}^{\tau} \Epminusonekminustwo'(\tau') d\tau' \\
\label{addedalabel}
&\lesssim  \Epkminusone(\tau_0) +\int_{\mathcal{R}(\tau_0,\tau)}
\sum_{|{\bf k}|\le k-1}(|V^\mu_p\partial_\mu ({\mathfrak{D}}^{\bf k}\psi) |+|w_p{\mathfrak{D}}^{\bf k}\psi| )|{\mathfrak{D}}^{\bf k}
( N^{\alpha\beta}(\psi,x)\partial_\alpha\psi\,\partial_\beta \psi + (\Box_{g_0}-\Box_{g(\psi,x)})\psi)|
\\
\label{nonlinearrpidenbboxcommuted}
&\qquad+ \int_{\mathcal{R}(\tau_0,\tau)}\sum_{|{\bf k}| \le k-1} \left(  {{\mathfrak{D}}}^{\bf k} 
(N^{\alpha\beta}(\psi,x)\partial_\alpha\psi\,\partial_\beta \psi + (\Box_{g_0}-\Box_{g(\psi,x)})\psi ) \right)^2 \\
\label{ellipticone}
&\qquad+ \int_{\mathcal{R}(\tau_0,\tau)\cap \{r\le R\}}\sum_{|{\bf k}| \le k-2} \left({\widetilde{\mathfrak{D}}}^{\bf k} ( N^{\alpha\beta}(\psi,x)\partial_\alpha\psi\,\partial_\beta \psi + (\Box_{g_0}-\Box_{g(\psi,x)})\psi   )\right)^2\\
\label{elliptictwo}
&\qquad
+\int_{\Sigma(\tau)\cap\{r\le R\}}\sum_{|{\bf k}|\le k-2} \left(
\widetilde{\mathfrak{D}}^{\bf  k}(N^{\alpha\beta}(\psi,x)\partial_\alpha\psi\,\partial_\beta \psi + (\Box_{g_0}-\Box_{g(\psi,x)})\psi)
\right)^2.
\end{align}
Moreover, for $p=0$, identical statements hold with $-1-\delta$ replacing $p-1$.
\end{proposition}

\begin{remark}
\label{theextratermremarksecond}
In the case where we replace the middle term of~\eqref{inhomogeneous} 
with~\eqref{replacedterm}, we 
should add 
\begin{align*}
\sum_{|{\bf k}|\le k-1} \sqrt{
\int_{\mathcal{R}(\tau_0,\tau)\cap \{ r\le R \}} |L{\mathfrak{D}}^{\bf k}\psi|^2+|\underline L{\mathfrak{D}}^{\bf k}\psi|^2 +|\slashed\nabla{\mathfrak{D}}^{\bf k}\psi|^2 
+r^{-2}|{\mathfrak{D}}^{\bf k}\psi|^2}\,\, \cdot \\
\qquad \cdot \, \sqrt{\int_{\mathcal{R}(\tau_0,\tau_1)\cap \{ r\le R \}} {\mathfrak{D}}^{\bf k}( N^{\alpha\beta}(\psi,x)\partial_\alpha\psi\,\partial_\beta \psi + (\Box_{g_0}-\Box_{g(\psi,x)})\psi)^2 }
\end{align*}
 to the right hand side of~\eqref{addedalabel}. Cf.~Remark~\ref{theextratermremarkfirst}.
\end{remark}

\begin{proof}
Note the term on line~\eqref{toexpand} arose by expanding the term appearing in the actual identity as follows:
\begin{eqnarray*}
\int_{\mathcal{R}(\tau_0,\tau)\cap \{r\le R\}}  \Hapk[\psi]\cdot  \{ [\Box_{g(\psi,x)},\mathfrak{D}^{\bf k}]\psi\}
&=&
\int_{\mathcal{R}(\tau_0,\tau)\cap \{r\le R\}}\Hapk[\psi]\cdot  \{ [\Box_{g(\psi,x)}-\Box_{g_0},\mathfrak{D}^{\bf k}]\psi\}\\
&&\qquad +
\int_{\mathcal{R}(\tau_0,\tau)\cap \{r\le R\}}\Hapk[\psi]\cdot  \{ [\Box_{g_0},\mathfrak{D}^{\bf k}]\psi\}
\end{eqnarray*}
and then bringing the second term to the left hand side to use 
Corollary~\ref{refertothisnow} and absorb it in the bulk.  (Recall that in the region $r\ge R$, we have $g=g_0$.)
The term on line~\eqref{nonlinearrpidentherefluxescommuted} 
arose from our application of elliptic estimates to $\psi$, considering
it as a solution of~\eqref{equationwithloss}. 
The term on line~\eqref{failureofcoercivity} arose from the term on the last line 
of~\eqref{perturbedcoercive}.
We note that restricted to $r\ge R$, we may replace $\widetilde{\mathfrak{D}}$ commutation with 
$\mathfrak{D}$ commutation.

The inequality~\eqref{whyanumberidontknow}--\eqref{elliptictwo}, on the other hand,
simply arises from applying the estimate~\eqref{rpidenbboxcommuted} to the nonlinear equation
written in the form~\eqref{equationwithloss}. Note that we have applied it at one order less, i.e.~with $k-1$
in place of $k$, in view of the fact that the right hand side is of order $k$.
\end{proof}

\subsection{Summed norm notation and the lower order smallness assumption}
In addition to the primitive  assumption~\eqref{mostprimitive}, 
in the context of our proof, 
we will need  to introduce stronger a priori
energy smallness assumptions on $\psi$ in our region $\mathcal{R}(\tau_0,\tau_1)$ under consideration.
In the context of the proof of the main theorem, this will appear as a bootstrap assumption. 
This will ensure that higher order nonlinear terms are indeed
absorbable.

We first explain some additional notation.
We define the following summed quantities for $\delta\le p \le 2-\delta$:
\begin{eqnarray*}
\Xpk(\tau_0,\tau_1)					&:=&		\sup_{\tau'\in[\tau_0,\tau_1]} \Epk (\tau')+\sup_{v:\tau_1\le v(\tau)} \Fpk(v,\tau_0,\tau_1)
			+ \int_{\tau_0}^{\tau_1}
			 \,\, \Epminusonek'(\tau') 
			 d\tau' ,\\
{}^\rho \Xpk(\tau_0,\tau_1)			&:=& \sup_{\tau'\in[\tau_0,\tau_1]} \Epk(\tau')+\sup_{v:\tau_1\le \tau(v)}  \Fpk(v,\tau_0,\tau_1)
			 +\int_{\tau_0}^{\tau_1}\left( \,\,    {}^{\rho} \Epminusonek'(\tau')+ {}^{\tilde\rho}\,\Ezerominusthreeminusdeltakminusone'
			 										(\tau')	\right)d\tau'	,	\\
{}^\chi \Xpk(\tau_0,\tau_1)				&:=&  \sup_{\tau'\in[\tau_0,\tau_1]}\Epk (\tau')+\sup_{v:\tau_1\le \tau(v)} \Fpk(v,\tau_0,\tau_1) 
			+\int_{\tau_0}^{\tau_1}\left( \, \,   {}^{\chi} \Epminusonek' (\tau') + \Epminusonekminusone' (\tau') \right)d\tau' .
\end{eqnarray*}
For $p=0$ we first define the analogous quantities, where $p-1$ is replaced by $-1-\delta$:
\begin{eqnarray*}
\Xzerok(\tau_0,\tau_1) 				&:=&   \sup_{\tau'\in[\tau_0,\tau_1]}\Ezerok (\tau')+\sup_{v:\tau_1\le \tau(v)} \Fzerok(v,\tau_0,\tau_1)
			 +\int_{\tau_0}^{\tau_1}  \, \Ezerominusoneminusdeltak'(\tau') d\tau' ,\\
 {}^\rho \Xzerok(\tau_0,\tau_1) 			&:=&  \sup_{\tau'\in[\tau_0,\tau_1]} \Ezerok (\tau')+\sup_{v:\tau_1\le \tau(v)}  \Fzerok(v,\tau_0,\tau_1)
 			+\int_{\tau_0}^{\tau_1}  \left(\,\, {}^{\rho} \Ezerominusoneminusdeltak' (\tau') + 
								{}^{\tilde\rho}\,\Ezerominusthreeminusdeltakminusone'	(\tau')  \right)d\tau' , \\
{}^\chi \Xzerok(\tau_0,\tau_1)			&:=&	  \sup_{\tau'\in[\tau_0,\tau_1]}\Ezerok (\tau')+\sup_{v:\tau_1\le \tau(v)} \Fzerok(v,\tau_0,\tau_1)
			+\int_{\tau_0}^{\tau_1} \left(\,\,  {}^{\chi} \Ezerominusoneminusdeltak'(\tau') 
								+ \Ezerominusoneminusdeltakminusone'(\tau')\right) d\tau'.
\end{eqnarray*}
Because $p=0$ is anomalous, however, we will need in addition the following stronger energies which will appear \emph{on the right hand side} of $p=0$ estimates:
\begin{eqnarray*}
\Xzeroplusk(\tau_0,\tau_1) 			&:=&  \sup_{\tau'\in[\tau_0,\tau_1]}\Ezerok (\tau')+\sup_{v:\tau_1\le \tau(v)} \Fzerok(v,\tau_0,\tau_1) 
			+\int_{\tau_0}^{\tau_1}    \, \Edeltaminusonek'(\tau') d\tau' , 			\\
 {}^\rho \Xzeroplusk(\tau_0,\tau_1) 		&:=&   \sup_{\tau'\in[\tau_0,\tau_1]}\Ezerok (\tau')+\sup_{v:\tau_1\le \tau(v)} \Fzerok(v,\tau_0,\tau_1)
 			 +\int_{\tau_0}^{\tau_1}  \left(  {}^{\rho} \Edeltaminusonek' (\tau')+ {}^{\tilde\rho}\,\Ezerominusthreeminusdeltakminusone'(\tau')
			 						\right) d\tau' , \\
{}^\chi \Xzeroplusk(\tau_0,\tau_1)		&:=&  \sup_{\tau'\in[\tau_0,\tau_1]}\Ezerok (\tau')+\sup_{v:\tau_1\le \tau(v)} \Fzerok(v,\tau_0,\tau_1) 
			+\int_{\tau_0}^{\tau_1}   \left( {}^{\chi} \Edeltaminusonek'(\tau') + \Edeltaminusonekminusone'(\tau')\right) d\tau'.
\end{eqnarray*}
Note the general properties that $p'\ge p$, $k'\ge k$ 
implies $\Xpprimekprime\gtrsim \Xpk$, ${}^\chi\Xpprimekprime\gtrsim {}^\chi \Xpk$,
${}^\rho \Xpprimekprime\gtrsim {}^\rho \Xpk$, while
\begin{equation}
\label{therelationweknow}
\Xpkminusone \lesssim {}^\chi \Xpk.
\end{equation}

We will use the notation $\ll \mkern-6mu k$ to denote some particular positive integer, depending on $k$,
which may be different on different instances of our use of the notation,
such that $\ll \mkern-6mu k \le k$, and in fact, $\ll \mkern-6mu k$ is ``much less than $k$'', 
 provided $k$ is sufficiently
large. In particular, for all positive integers $n$, 
we assume there exists a $k(n)$ for which $k\ge k(n)$ implies $\ll \mkern-6mu k \le k-n$.

We have the following:
\begin{proposition}
\label{improvetheprimitive}
Let $k\ge 4$ be sufficiently large. There exists an $\varepsilon_{\rm prelim}>0$ such that for all  $0<\varepsilon\le \varepsilon_{\rm prelim}$, the following holds.

Let $\psi$ satisfy
\begin{equation}
\label{basicbootstrap}
\Xplesslessk \leq\varepsilon
\end{equation}
 in  $\mathcal{R}(\tau_0,\tau_1)$, for $p=0$ (or  some $\delta \le p\le 2-\delta$),
where $\ll\mkern-6mu k \ge 4$. 
Then the following improved version of~\eqref{mostprimitive} holds:
\begin{equation}
\label{mostprimitiveimproved}
\sup_{r \le R}\sum_{|{\bf k}|\le 1}  (\widetilde{\mathfrak{D}}^{\bf k}  \psi)^2 \lesssim \varepsilon \leq \sqrt{\varepsilon}
\end{equation}
\end{proposition}
\begin{proof}
This follows of course immediately from the Sobolev inequality~\eqref{simplesobolev}.
\end{proof}

In the context of the proof of the main theorem,
inequality
\eqref{basicbootstrap} will be introduced as a
bootstrap assumption, and Proposition~\ref{improvetheprimitive} will be applied to 
retrieve the assumption~\eqref{mostprimitive},
necessary for the 
results of Sections~\ref{thisisprimitive}--\ref{hoeforquas}. Note that with the assumption~\eqref{basicbootstrap},
we may replace the boxed $\varepsilon^{1/4}$ in~\eqref{failureofcoercivity} with $\sqrt{\varepsilon}$.

We note that in what follows, restrictions on $k$ sufficiently large  will always include the condition $\ll\mkern-6mu k \ge 4$.

\subsubsection{Comparability of  $\mathfrak{E}$ and $\mathcal{E}$ energies}
\label{nonlinearelliptic}

We note first the following:
\begin{proposition}
\label{useithere}
Let $k\ge 4$ be sufficiently large. There exists an $\varepsilon_{\rm prelim}>0$ such that  the following holds.

Let  $\psi$ be a solution of~\eqref{theequation} in $\mathcal{R}(\tau_0,\tau_1)$ 
satisfying~\eqref{basicbootstrap} for $p=0$ and some $0<\varepsilon\le \varepsilon_{\rm prelim}$. 
Then setting $F= N^{\mu\nu}(\psi,x)\partial_\mu\psi\,\partial_\nu \psi + (\Box_{g_0}-\Box_{g(\psi,x)})\psi$,
we have
\begin{equation}
\label{toabsorbit}
\int_{\Sigma (\tau')\cap \{ r\le R\}  } \sum_{|{\bf k}|\le k-1}
(\widetilde{\mathfrak{D}}^{\bf k} F)^2 \lesssim
\min\left\{ \, \Ezerolesslessk (\tau')   \, ,\,\, \Ezerominusoneminusdeltalesslessk'(\tau')  \right\}  \,\Ezerok(\tau') \lesssim \varepsilon\, 
 \Ezerok(\tau').
\end{equation}
\end{proposition}
\begin{proof}
This is standard in view of our assumptions on $N$ and $g(\psi,x)$ from Section~\ref{classstruggle},
and can be proven using the Sobolev inequality~\eqref{simplesobolev} 
on $\Sigma(\tau')\cap\{r\le R\}$.
\end{proof}

We may now obtain the following result, which can be viewed as a corollary of 
Proposition~\ref{somerelationshere}.
\begin{corollary}
\label{nomoreroman}
Let $k\ge 4$ be sufficiently large. Then
 there exists an $\varepsilon_{\rm prelim}>0$ such that
the following statement holds.

Let $\psi$ be a solution of~\eqref{theequation} in $\mathcal{R}(\tau_0,\tau_1)$ 
satisfying~\eqref{basicbootstrap} for $p=0$ and some $0 <\varepsilon\le \varepsilon_{\rm prelim}$.
We have the analogue of Corollary~\ref{corollaryforhomog}:
The calligraphic and fraktur (when applicable) energies are comparable, i.e.~
\[
\Epk(\tau)\sim  \Efancypk(\tau) (g)  ,
\]
for all $\tau_0\le \tau\le \tau_1$.
Moreover, if $\chi=1$ and $\tilde\rho=1$ identically (as in case (i)), then
\[
 \Epminusonek'(\tau)    \sim  {}^{\chi} \Epminusonek'(\tau)+
 {}^{\tilde\rho} \, \Ezerominusthreeminusdeltakminusone'(\tau), \qquad \Ezerominusoneminusdeltak'(\tau) \sim  {}^{\chi} \Ezerominusoneminusdeltak'(\tau)+{}^{\tilde\rho} \, \Ezerominusthreeminusdeltakminusone'(\tau) 
\]
and thus
\[
{}^\chi \Xpk(\tau_0,\tau_1) \sim  {}^\rho\Xpk (\tau_0,\tau_1) \sim   \Xpk (\tau_0,\tau_1).
\]
If $\tilde\rho=1$ identically, $\rho=\chi$, (i.e.~as in cases (i) and (ii)), then
\[
{}^\chi\Epminusonek'(\tau)+ \Epminusonekminusone' (\tau)\sim{}^\rho\Epminusonek'(\tau) +  {}^{\tilde\rho} \Ezerominusthreeminusdeltakminusone' (\tau) , \qquad
{}^\chi\Ezerominusoneminusdeltak'(\tau) +\Ezerominusoneminusdeltakminusone'(\tau)  \sim {}^\rho\Ezerominusoneminusdeltak'(\tau) + 
 {}^{\tilde\rho} \Ezerominusthreeminusdeltakminusone'(\tau) 
\] 
and thus 
\[
{}^\chi \Xpk(\tau_0,\tau_1) \sim {}^\rho\Xpk (\tau_0,\tau_1).
\]
\end{corollary}

\begin{proof}
This follows from~\eqref{witheliptone}--\eqref{witheliptfive}, with $F$ defined as in Proposition~\ref{useithere},
absorbing the inhomogeneous term using the estimate~\eqref{toabsorbit}, for sufficiently small $\varepsilon_{\rm prelim}$.
\end{proof}

\subsection{Estimates on the nonlinear terms in the near region}
\label{additionaltermsec}

\begin{proposition}
\label{nearprop}
Let $k\ge 4$ be sufficiently large. 
There exists an $\varepsilon_{\rm prelim}>0$ such that  the following holds.

Let $\psi$ be a solution of~\eqref{theequation} in $\mathcal{R}(\tau_0,\tau_1)$ 
satisfying~\eqref{basicbootstrap}
for $p=0$ and some $0<\varepsilon \le\varepsilon_{\rm prelim}$.
Then   all  integrals on lines~\eqref{firstofthenonlinear}--\eqref{failureofcoercivity},
\underline{restricted to $\mathcal{R} (\tau_0,\tau_1)\cap \{r \le R\}$}, 
may be estimated by
\begin{align}
\label{inopposition}
\ldots &\lesssim \int_{\tau_0}^\tau  \Ezerok(\tau') \sqrt{\,\,\Ezerominusoneminusdeltalesslessk'(\tau')  \,\, }d\tau' ,
\end{align}
while the integrals on lines~\eqref{addedalabel}--\eqref{nonlinearrpidenbboxcommuted},
again \underline{restricted to $\mathcal{R}(\tau_0,\tau_1)\cap \{r \le R\}$}, 
with or without the extra term of Remark~\ref{theextratermremarksecond},
may be similarly estimated by
\begin{align}
\label{remarktheloss}
\ldots &\lesssim \int_{\tau_0}^{\tau} \Ezerok(\tau') \sqrt{\,\,\Ezerominusoneminusdeltalesslessk'(\tau')  \,\, }d\tau' .
\end{align}
\end{proposition}

\begin{proof}
This is standard given the assumptions in Section~\ref{classstruggle} and can be proven using the Sobolev inequality~\eqref{simplesobolev} 
on $\Sigma(\tau)\cap\{r\le R\}$, cf.~the proof of~\eqref{toabsorbit} which in fact already bounds several of the terms.
 (Note that the boxed $\varepsilon^{1/4}$ in~\eqref{failureofcoercivity}
may clearly now be replaced by the quantity $\sqrt{\,\,\Ezerominusoneminusdeltalesslessk'(\tau')}$ within the integral.)
\end{proof}

We remark again that, as opposed to~\eqref{inopposition},
 the control~\eqref{remarktheloss} 
loses a derivative, as the energy on the left hand side of~\eqref{addedalabel} is 
$k-1$'th order while the right hand side of~\eqref{remarktheloss}
is  of $k$'th order.

\subsection{Assumptions for the nonlinearity in the far region: the ``null'' condition}
\label{formofnullcondsec}

We may now state our
 specific assumption capturing the null condition for the semilinear terms.

Rather than formulate algebraically the null condition in terms of the form of $N$, for maximum generality we will make our basic assumption directly 
at the level of an estimate for  terms on the right hand side of the inequalities of~Proposition~\ref{higherorderestimquasi}.

Our assumptions will only depend on the region $r\ge 8R/9$, so we will need some notation to denote the restriction of energies, etc., to this region.
Given $r_*\ge r_0$ and  $v$,  by
\[
\Epk_{r_*,v}, \qquad  \Xpk_{r_*,v}, \qquad {}^\rho \Xpk_{r_*,v}, \qquad \ldots
 \]
we shall mean  the expressions defined as before, but where all domains of integration are restricted to the region $\mathcal{R}(\tau_0,\tau_1,v)\cap \{r\ge r_*\}$. 
We may thus define these expressions for functions $\psi: \mathcal{R}(\tau_0,\tau_1,v)\cap \{r\ge r_*\}\to \mathbb R$. 
Similarly, the above expressions without the $v$ subscript, e.g.~$\Epk_{r_*}$, will be defined where all domains of intergation
are restricted to $\mathcal{R}(\tau_0,\tau_1)\cap \{r\ge r_*\}$. These can then be defined for functions
$\psi: \mathcal{R}(\tau_0,\tau_1)\cap \{r\ge r_*\}\to \mathbb R$ . 
Let us note that
\begin{equation}
\label{implicationhere}
r_*\ge 8R/9 \implies \Xpk_{r_*,v}={}^\chi \Xpk_{r_*,v}={}^\rho \Xpk_{r_*,v}
\end{equation}
and similarly for the quantities  without the $v$ subscript.

\begin{assumption}[Null condition for semilinear terms]
\label{nullcondassumphere}
There exists a $k_{\rm null}$, and for all $k\ge k_{\rm null}$, 
 an $\varepsilon_{\rm null}>0$, such that the following holds for all $\tau_0\le \tau_1$, $v$.

Let $\psi$ be a smooth function defined on $\mathcal{R}(\tau_0,\tau_1,v)\cap \{ r\ge 8R/9\}$ 
 satisfying the bound
 \begin{equation}
 \label{basicbootstrapinnullcond}
 \Xzerolesslessk_{8R/9, v} \leq \varepsilon
 \end{equation}
  for
some $0<\varepsilon\leq \varepsilon_{null}$. 
Then for all $\delta\le p\le 2-\delta$,
\begin{align}
\nonumber
\int_{\mathcal{R}(\tau_0,\tau_1,v)\cap \{r\ge R\}}&\sum_{|{\bf k}|\le k}(|r^pr^{-1}L(r{\mathfrak{D}}^{\bf k}\psi)|+|L{\mathfrak{D}}^{\bf k}\psi|
+|\underline{L}{\mathfrak{D}}^{\bf k}\psi|+|r^{-1}\slashed\nabla{\mathfrak{D}}^{\bf k}\psi|+ |r^{-1}{\mathfrak{D}}^{\bf k}\psi| )  |{\mathfrak{D}}^{\bf k}
( N^{\alpha\beta}(\psi,x)\partial_\alpha\psi\,\partial_\beta \psi) |\\
\nonumber
&\qquad\qquad\qquad\qquad\qquad
+ (\mathfrak{D}^{\bf k} ( N^{\alpha\beta}(\psi,x)\partial_\alpha\psi\,\partial_\beta \psi))^2 \\
\label{nullcondassump}
&\lesssim
 \Xpk_{\frac{8R}9,v}(\tau_0,\tau_1) \sqrt{\Xzerolesslessk_{\frac{8R}9,v}(\tau_0,\tau_1)} + \sqrt{\Xpk_{\frac{8R}9,v}(\tau_0,\tau_1)}\sqrt{\Xzerok_{\frac{8R}9,v}(\tau_0,\tau_1)}\sqrt{\Xplesslessk_{\frac{8R}9,v}(\tau_0,\tau_1)},
\end{align}
while, corresponding to $p=0$ we have
\begin{align}
\nonumber
\int_{\mathcal{R}(\tau_0,\tau_1,v)\cap \{r\ge R\}}&\sum_{|{\bf k}|\le k}(|r^{-1}L(r{\mathfrak{D}}^{\bf k}\psi)|
+|L{\mathfrak{D}}^{\bf k}\psi|
+|\underline{L}{\mathfrak{D}}^{\bf k}\psi|+|r^{-1}\slashed\nabla{\mathfrak{D}}^{\bf k}\psi|+ |r^{-1}{\mathfrak{D}}^{\bf k}\psi| )  |{\mathfrak{D}}^{\bf k}
( N^{\alpha\beta}(\psi,x)\partial_\alpha\psi\,\partial_\beta \psi) |\\
\nonumber
&\qquad\qquad\qquad\qquad\qquad
+ (\mathfrak{D}^{\bf k} ( N^{\alpha\beta}(\psi,x)\partial_\alpha\psi\,\partial_\beta \psi))^2 \\
\label{nullcondassumptwo}
 &\lesssim 
 \Xzeroplusk_{\frac{8R}9,v} (\tau_0,\tau_1) \sqrt{ \Xzeropluslesslessk_{\frac{8R}9,v}(\tau_0,\tau_1)}.
\end{align}
\end{assumption}

We note the following
\begin{proposition}
\label{itholdsthenull}
Assumption~\ref{nullcondassumphere}  holds for equations~\eqref{theequation} 
where $g_0$ is Minkowski and where the semilinear terms $N$ satisfy
the classical null condition of Klainerman~\cite{KlNull}, and more generally
when $g_0$ is the Kerr metric and 
the semilinear terms $N$ belong to the class considered by Luk~\cite{MR3082240}. 
\end{proposition}
\begin{proof}
See Appendix~\ref{nonlineartermsatinf}.
\end{proof}

\begin{remark}
\label{alreadyreplace}
Let us note that for $\psi$ defined  on the slab  $\mathcal{R}(\tau_1,\tau_2)\cap \{r \ge 8R/9\}$, 
from the trivial bound $\Xpk_{r_*,v}\leq   \Xpk_{r_*}$ we may drop the $v$ subscripts.
For $\psi$ defined globally on 
  the slab  $\mathcal{R}(\tau_1,\tau_2)$, then 
 in view of~\eqref{implicationhere} and the trivial bounds
 ${}^\rho \Xpk_{r_*}\leq {}^\rho \Xpk$,  etc.,
we may replace the $\frac{8M}9,v$ subscripts in all the $\Xpk$, etc., energies on the right hand side  with  
the $\rho$ superscript, i.e.~replace $\Xpk_{\frac{8M}9,v}$, etc.,  with ${}^\rho \Xpk$, etc.
Also, for such $\psi$, we may clearly replace~\eqref{basicbootstrapinnullcond} with assumption~\eqref{basicbootstrap}.
\end{remark}

\subsection{The final estimates}

Putting everything together we have the following:

\begin{proposition} 
\label{finalpropest}
Consider $(\mathcal{M},g_0)$ 
satisfying the assumptions of Sections~\ref{backgroundgeom} and~\ref{waveequassump}
(for cases (i), (ii) or (iii))
and equation~\eqref{theequation} satisfying the assumptions
of Section~\ref{classofequations} and~\ref{formofnullcondsec}. 

Then for all $k\ge 4$ sufficiently large,  there exist
constants $C>0$ (also implicit in the $\lesssim$ below), $c>0$ 
and an $\varepsilon_{\rm prelim}>0$, 
 such that the following is true.

Let $\tau_0\le \tau_1$ and 
let $\psi$ be a solution of~\eqref{theequation} in $\mathcal{R}(\tau_0,\tau_1)$ 
satisfying~\eqref{basicbootstrap} with $p=0$ for some $0<\varepsilon\le \varepsilon_{\rm prelim}$.
Then for $\delta \le p \le 2-\delta$, one has the estimates
\begin{align}
\nonumber
\sup_{v:\tau\le \tau(v)} \Ffancyp(v,\tau_0,\tau),&\quad
c \, {}^\rho \Xpk(\tau_0,\tau_1), \quad
  \Efancypk(\tau_1)\\
 \nonumber
  &\leq \Efancypk(\tau_0) +
A\int_{\tau_0}^{\tau_1}\Eerrorkminusone'(\tau')d\tau'\\
\nonumber
&\qquad+
C\, \Xpk_{\frac{8R}9}(\tau_0,\tau_1) \sqrt{\Xzerolesslessk_{\frac{8R}9}(\tau_0,\tau_1)}  + C\,\sqrt{\Xpk_{\frac{8R}9}(\tau_0,\tau_1)}\sqrt{\Xzerok_{\frac{8R}9}(\tau_0,\tau_1)}\sqrt{\Xplesslessk_{\frac{8R}9}(\tau_0,\tau_1)}
\\
 \label{addtermoneherewithflux}
 &\qquad + C\int_{\tau_0}^{\tau_1}  \Ezerok(\tau') \sqrt{\,\,\Ezerominusoneminusdeltalesslessk'(\tau')  \,\, }d\tau' \\
 \nonumber
{}^\chi \Xpkminusone(\tau_0,\tau_1) &\lesssim \Epkminusone(\tau_0) +
 \Xpkminusone_{\frac{8R}9}(\tau_0,\tau_1) \sqrt{\Xzerolesslessk_{\frac{8R}9}(\tau_0,\tau_1)} 
 + \sqrt{\Xpkminusone_{\frac{8R}9}(\tau_0,\tau_1)}\sqrt{\Xzerokminusone_{\frac{8R}9}(\tau_0,\tau_1)}
 \sqrt{\Xplesslessk_{\frac{8R}9}(\tau_0,\tau_1)}
 \\
\label{addtermoneherechi}
 &\qquad+
\int_{\tau_0}^{\tau_1}  \Ezerok(\tau') \sqrt{\,\,\Ezerominusoneminusdeltalesslessk'(\tau')  \,\, }d\tau' 
\end{align}
while for $p=0$, one has the estimates
\begin{align}
\nonumber
\sup_{v:\tau\le \tau(v)} \Ffancyzero(v,\tau_0,\tau), \quad &
c \, {}^\rho \Xzerok(\tau_0,\tau_1), \quad \Efancyzerok(\tau_1) \leq \Efancyzerok(\tau_0) +
A\int_{\tau_0}^{\tau_1}\Eerrorkminusone'(\tau')d\tau'+
C\, \Xzeroplusk_{\frac{8R}9} (\tau_0,\tau_1) \sqrt{ \Xzeropluslesslessk_{\frac{8R}9}(\tau_0,\tau_1)}
\\
 \label{addtermoneherewithfluxzero}
 &\qquad \qquad\qquad\qquad \qquad\qquad+ C\int_{\tau_0}^{\tau_1}  \Ezerok(\tau') \sqrt{\,\,\Ezerolesslessk(\tau')  \,\, }d\tau', \\
 \label{addtermoneherechizero}
{}^\chi \Xzerokminusone(\tau_0,\tau_1) &\lesssim \Ezerokminusone(\tau_0) +
 \Xzeroplusk_{\frac{8R}9} (\tau_0,\tau_1) \sqrt{ \Xzeropluslesslessk_{\frac{8R}9}(\tau_0,\tau_1)}
+
\int_{\tau_0}^{\tau_1}  \Ezerok(\tau') \sqrt{\,\,\Ezerolesslessk(\tau')  \,\, }d\tau' .
\end{align}

Finally, in both the case $p=0$ and the case $\delta\le p\le 2-\delta$, one has 
the alternative estimate
\begin{equation}
\label{Afd}
 \Xpk(\tau_0,\tau_1)  \lesssim \Epk (\tau_0)  +
  (1+\tau_1-\tau_0)\Xpk(\tau_0,\tau_1) \sqrt{{}\Xzerolesslessk(\tau_0,\tau_1)},
\end{equation}
depending on $\tau_1-\tau_0$.
\end{proposition}

\begin{proof}
This follows from Proposition~\ref{higherorderestimquasi}, Corollary~\ref{nomoreroman}, Proposition~\ref{nearprop} and 
Assumption~\ref{nullcondassumphere}.
\end{proof}

\begin{remark}
Note how the case $p=0$ is anomalous in that one sees $0+$ on the right hand side 
of~\eqref{addtermoneherewithfluxzero}--\eqref{addtermoneherechizero}.
It is this fact which will require us to use $p>0$ weights to close the 
global estimates~\eqref{addtermoneherewithflux}--\eqref{addtermoneherechi},
even in the case (i). On the other hand,~\eqref{Afd} is sufficient to show local existence 
even using only the $p=0$ weight.
\end{remark}

\begin{remark}
In using the above we shall always simply replace terms like
$\Xpk_{\frac{8R}9}(\tau_0,\tau_1)$, etc., by ${}^\rho \Xpk(\tau_0,\tau_1)$, etc., according to Remark~\ref{alreadyreplace}. We have kept the terms in the above form
simply to highlight their origin in the estimate of Assumption~\ref{nullcondassumphere}.
\end{remark}

\begin{remark}
\label{semiremsemi}
In the case where we only assume the black box inequality~\eqref{inhomogeneous} and 
the assumptions of asymptotic flatness of Sections~\ref{rpsection} 
and~\ref{farawaycommutationerror} as interpreted in Remark~\ref{remarkforseminlinear} (necessary
to obtain~\eqref{rpidenbboxcommuted})
and moreover, the equation~\eqref{theequation} is semilinear, i.e.~$g(x,\psi)=g_0$ for all $x\in\mathcal{M}$,
then we obtain the statement of Proposition~\ref{finalpropest}  without inequalities~\eqref{addtermoneherewithflux} and~\eqref{addtermoneherewithfluxzero} 
and where~\eqref{addtermoneherechi} and~\eqref{addtermoneherechizero} are
replaced by
\begin{align}
\nonumber
 {}^\chi \Xpk(\tau_0,\tau_1) &\lesssim \Epk(\tau_0) +
 \Xpk_{\frac{8R}9}(\tau_0,\tau_1) \sqrt{\Xzerolesslessk_{\frac{8R}9}(\tau_0,\tau_1)} 
 + \sqrt{\Xpk_{\frac{8R}9}(\tau_0,\tau_1)}\sqrt{\Xzerok_{\frac{8R}9}(\tau_0,\tau_1)}
 \sqrt{\Xplesslessk_{\frac{8R}9}(\tau_0,\tau_1)}
 \\
\label{addtermoneherechinewhere}
 &\qquad+
\int_{\tau_0}^{\tau_1}  \Ezerok(\tau') \sqrt{\,\,\Ezerominusoneminusdeltalesslessk'(\tau')  \,\, }d\tau' 
 ,\\
 \label{addtermoneherechizeroforsemi}
{}^\chi \Xzerok(\tau_0,\tau_1) &\lesssim \Ezerok(\tau_0) +
 \Xzeroplusk_{\frac{8R}9} (\tau_0,\tau_1) \sqrt{ \Xzeropluslesslessk_{\frac{8R}9}(\tau_0,\tau_1)}
+
\int_{\tau_0}^{\tau_1}  \Ezerok(\tau') \sqrt{\,\,\Ezerolesslessk(\tau')  \,\, }d\tau' ,
\end{align}
respectively.
(In particular, equation~\eqref{Afd} also still holds as stated.)
\end{remark}

\subsection{Local well posedness}

Finally, we give a local well posedness statement and an extension criterion.

Since our initial data hypersurfaces $\Sigma(\tau_0)$ are in part null, it will be convenient to assume that our
initial data are smooth. 
We define a smooth initial data set on $\Sigma(\tau_0)$
to be a pair $(\uppsi, \uppsi')$ where $\uppsi$ is  a function on
$\Sigma(\tau_0)$ smooth on 
$\Sigma(\tau_0)\cap\{r\le R\}$ and smooth on $\Sigma(\tau_0)\cap \{r \ge R\}$, and 
$\uppsi'$ is a smooth function on $\Sigma(\tau_0)\cap \{r\le R\}$,
such that moreover  there exists a smooth a function $\Psi$ on $\mathcal{R}(\tau_0,\tau_1)$
for some $\tau_1$ such that $\Psi|_{\Sigma(\tau_0)}=\uppsi$, and $n\Psi|_{\Sigma(\tau_0)\cap\{r\le R\}}=\uppsi'$,
where $n$ denotes the normal to $\Sigma(\tau_0)\cap\{r\le R\}$.

Note that given such initial data on $\Sigma(\tau_0)$ one may define
\begin{equation}
\label{tobedfined}
 \Epk[\uppsi,\uppsi']
\end{equation}
as follows: 

Using the equation~\eqref{theequation} we may
 compute along $\Sigma(\tau_0)$
 the $k+1$-jet of any smooth solution $\psi$ of~\eqref{theequation} such that 
 \begin{equation}
 \label{attainingthedata}
 \psi|_{\Sigma(\tau_0)} =\uppsi, \qquad n\psi|_{\Sigma(\tau_0)\cap \{r\le R\}} = \uppsi'.
 \end{equation}
We may define  $\Epk[\uppsi,\uppsi']$ to equal the usual $\Epk(\tau_0)$ where the derivatives
are interpreted in terms of this formal computation.
Note of course that the expression~\eqref{tobedfined} may well be infinite.

Alternatively, we may express this  as follows. Using the above computation, we may
in fact define a smooth  $\Psi$ as in the previous paragraph, such that 
$\Psi$ moreover satisfies~\eqref{theequation} along $\Sigma(\tau_0)$,
and, 
 for all $|{\bf k}| \le k$,
  $\widetilde{\mathfrak{D}}^{\bf k} \Psi$ satisfies along $\Sigma(\tau_0)$ the equation one obtains by commuting~\eqref{theequation}  by $\widetilde{\mathfrak{D}}^{\bf k}$. 
We may then simply define~\eqref{tobedfined} substituting this $\Psi$ for $\psi$ in the usual expression
  for $\Epk(\tau_0)$.

With this, we can now state our local well posedness:

\begin{proposition}[Local well posedness]
\label{localexistence}
Consider $(\mathcal{M},g_0)$ 
satisfying the assumptions of Sections~\ref{backgroundgeom} and~\ref{waveequassump}
(for cases (i), (ii) or (iii))
and equation~\eqref{theequation} satisfying the assumptions
of Section~\ref{classofequations} and~\ref{formofnullcondsec}. Fix either $p=0$ or 
$\delta \le p\le 2-\delta$.

There exists a positive integer $k_{\rm loc}\ge 4$ such that the following holds. Let $k\ge k_{\rm loc}$.
Then there exists a positive real constant $C>0$ sufficiently large,
a positive real parameter $\varepsilon_{\rm loc}>0$ sufficiently small, and a decreasing positive function $\tau_{\rm exist}:(0,\varepsilon_{\rm loc})\to \mathbb R$
such that for all smooth initial data $(\uppsi, \uppsi')$ on $\Sigma(\tau_0)$ such that
\[
 \Epk[\uppsi,\uppsi'] \leq \varepsilon_0
 \leq \varepsilon_{\rm loc},
\]
there exists a smooth solution  $\psi$
of~\eqref{theequation}  
in $\mathcal{R}(\tau_0,\tau_1)$ for $\tau_1:=\tau_0+\tau_{\rm exist}$ which satisfies
\[
 \Xpk(\tau_0,\tau_1) \leq C \varepsilon_0.
\]
Moreover, for all $\tau_0\le \tau\le \tau_1$,
 any other smooth $\widetilde\psi$ defined on $\mathcal{R}(\tau_0,\tau)$ satisfying~\eqref{theequation} in
 $\mathcal{R}(\tau_0,\tau)$ and attaining the initial data, i.e.~satisfying~\eqref{attainingthedata} (with $\widetilde{\psi}$
replacing $\psi$), coincides with the restriction of $\psi$, i.e.~$\widetilde\psi=\psi|_{\mathcal{R}(\tau_0,\tau)}$.

We have in addition the following propagation of higher order regularity and/or higher weighted estimates. 
Given $p$, $k\ge k_{\rm loc}$, $\varepsilon_0\le \varepsilon_{\rm loc}$, $\tau_0$, $\tau_1$  as above, 
$2-\delta \ge p' \ge  p$ if $p\ge \delta$  and either $2-\delta\ge p'\ge \delta$ or $p'=0$ if $p=0$, and  $k'\ge k$. Then there exists a constant $C(k',\tau_1-\tau_0)$ such that
given $\psi$ as above, 
\[
\Xpprimekprime(\tau_0,\tau_1)
 \leq C(k',\tau_1-\tau_0) \, \Epprimekprime (\tau_0).
\]

\end{proposition}

\begin{proof}
This can be easily proven using estimate~\eqref{Afd}.
We leave the details to the reader.
\end{proof}

Note also the following easy corollary, which we will use as an extension criterion:
\begin{corollary}[Continuation criterion]
\label{continuash}
Fix $p=0$ or $\delta\le p\le 2-\delta$, and let $k\ge k_{\rm loc}$ and $\varepsilon_{\rm loc}$
be as in Proposition~\ref{localexistence}.  There exists a constant $C>0$ and an $\epsilon>0$ such that the following is true.

Let  $\tau_f>\tilde\tau_0$ and suppose that there exists a smooth solution $\psi$ of~\eqref{theequation} 
on 
$\cup_{\tilde{\tau}_0<\tau<\tau_f}\mathcal{R}(\tilde\tau_0,\tau)$
such that 
\[
 \Epk(\tau)\leq  \varepsilon_{\rm loc} 
\]
for $k\ge k_{\rm loc}$ and all $\tilde\tau_0\le \tau<\tau_f$. Then, defining $\tau_1:= \tau_f+\epsilon$,
 $\psi$ extends uniquely to a smooth function 
on $\mathcal{R}(\tilde\tau_0,\tau_1)$ satisfying the equation~\eqref{theequation} on this set
and, setting $\tau_0:=\max\{ \tau_f-\epsilon, \tilde{\tau}_0\}$, satisfying the estimate
\begin{equation}
\label{extensionhere}
\Xpk(\tau_0,\tau_1) \leq C\varepsilon_{\rm loc}.
\end{equation}

Moreover, for all $k'\ge k$ and all $2-\delta \ge p'\ge p$, 
there exist constants $C(k')$ such that
\begin{equation}
\label{extensionhigherhigher}
\Xpprimekprime (\tau_0,\tau_1) \leq 
 C(k') \, \Epprimekprime (\tau_0).
\end{equation}
\end{corollary}

\begin{remark}
\label{remarkonconstantshere}
Recall that by our conventions the constants $C$, $k_{\rm loc}$, $\varepsilon_{\rm loc}$ depend in addition on $k$. 
We may assume that $k_{\rm loc}$ is sufficiently large so that $k\ge k_{\rm loc}$ satisfies the largeness constraint
 of all previous propositions in this Section, and that $\varepsilon_{\rm loc}\le \varepsilon_{\rm prelim}$.
\end{remark}

\begin{remark}
In view of Remark~\ref{semiremsemi}, and the fact that one needs only~\eqref{Afd},
 all the statements of this section also hold in the semilinear case,
under the relaxed assumptions described there.
\end{remark}

\section{The main theorem: global existence and stability}
\label{maintheoremsection}
We may now state the main result of this paper.

\begin{theorem}[Global existence and stability]
\label{themaintheorem}
Consider $(\mathcal{M},g_0)$ 
satisfying the assumptions of Sections~\ref{backgroundgeom} and~\ref{waveequassump}
(for cases (i), (ii) or (iii))
and equation~\eqref{theequation} satisfying the assumptions
of Section~\ref{classofequations} and Assumption~\ref{nullcondassumphere}
of Section~\ref{formofnullcondsec}.

There exists a positive integer $k_{\rm global}\ge k_{\rm local}$ sufficiently large,
such that, given $k\ge k_{\rm global}$, there exists a positive $0<\varepsilon_{\rm global}<\varepsilon_{\rm loc}$ sufficiently small,
and a positive  constant $C>0$ sufficiently large, so that the following holds.

Fix $\tau_0=1$
and consider as in Proposition~\ref{localexistence}
initial data $(\uppsi,\uppsi')$ on $\Sigma(\tau_0)$ for~\eqref{theequation} satisfying
\[
 \Epk[\uppsi,\uppsi']
 \leq\varepsilon_0\leq \varepsilon_{\rm global} ,\qquad
\Eonek[\uppsi,\uppsi']+  \Etwominusdelkminusone[\uppsi,\uppsi']
 \leq\varepsilon_0\leq \varepsilon_{\rm global},\qquad
\Eonek[\uppsi,\uppsi'] +  \Etwominusdelkminustwo[\uppsi,\uppsi']
 \leq\varepsilon_0\leq \varepsilon_{\rm global},
\]
according to case (i), (ii) and (iii) respectively, where in the former case $\delta\le p\le 2-\delta$.
Then the $\psi$ given by Proposition~\ref{localexistence} 
extends to a unique globally defined solution of~\eqref{theequation} in $\mathcal{R}(\tau_0,\infty)$,
satisfying the estimates
\begin{equation}
\label{basicestimate}
{}^\rho\Xpk(\tau_0,\tau) \leq  C\varepsilon_0, 
\qquad{}^\rho\Xonek(\tau_0,\tau)+  {}^\rho\Xtwominusdelkminusone(\tau_0,\tau) \leq  C\varepsilon_0,
\qquad
{}^\rho\Xonek(\tau_0,\tau)+
{}^\rho\Xtwominusdelkminustwo(\tau_0,\tau) \leq  C\varepsilon_0,
\end{equation}
according to cases (i), (ii), (iii), respectively, 
for all $\tau\ge \tau_0$, in particular 
\[
\Epk(\tau) \leq  C\varepsilon_0, 
\qquad \Eonek(\tau)+ \Etwominusdelkminusone(\tau) \leq  C\varepsilon_0,
\qquad
\Eonek(\tau)+
\Etwominusdelkminustwo(\tau) \leq  C\varepsilon_0,
\]
respectively.

The solution will satisfy moreover estimates~\eqref{finalforcasei} in case (i), estimates~\eqref{finalboundedness}
and~\eqref{resumefromhere}--\eqref{weimprovehere} in
case (ii) and estimates~\eqref{finalboundednessforiii}--\eqref{weimproveheretwoforiii} in case (iii).
\end{theorem}

We shall give the proof of this theorem in Section~\ref{estimatehier}.

\begin{bigremark} 
\label{bigremarkonsemi}
Our theorem  also applies
under the relaxed assumptions of Remark~\ref{semiremsemi}, 
where~\eqref{theequation} is however required to be semilinear.  Here we can again distinguish case (i) where $\chi=1$ and
case (ii) where $\chi$ is allowed to degenerate. In case (i), the theorem holds as stated except for~\eqref{finalforcasei}, which
no longer applies. Similarly, in case (ii), the theorem holds as stated except for~\eqref{finalboundedness} which no longer applies.  
The necessary modifications will be collected in a series of remarks in Sections~\ref{easycase} and~\ref{itscaseiinow}. 
(See already Section~\ref{thereisnoalternative} for the completion of case (ii).)
\end{bigremark}

We reiterate (cf.~Propositions~\ref{itholdsthenull},  Theorem~\ref{AppendixTheo} and Theorem~\ref{trueforKerr}) that the above theorem in particular holds for equations~\eqref{theequation} on the Kerr spacetime $(\mathcal{M},g_{a,M})$ for $|a|\ll M$, with quasilinear terms
as described in Section~\ref{classofequations} and the general semilinear terms considered in~\cite{MR3082240}, and, by the above remark,
restricted to  the purely semilinear case,
for the full range of parameters $|a|<M$.

\section{The estimate hierarchies and global existence and decay}
\label{estimatehier}

We
proceed with the proof of Theorem~\ref{themaintheorem} in cases (i), (ii) and (iii). These will
be treated in Sections~\ref{easycase},~\ref{itscaseiinow} 
and~\ref{andnotitsiii}  below, respectively.

Note that we will now suppress some of the cumbersome notation with the following conventions:

In this section, the fraktur energies $\mathfrak{E}$, $\mathfrak{F}$ will everywhere
denote $\mathfrak{E}(g)$, $\mathfrak{F}(g)$, so we will drop explicit reference to $g$ here.
When the region $\mathcal{R}(\tau_0,\tau_1)$ is assumed, we will omit the $\tau_0$ or $(\tau_0,\tau_1)$ arguments 
in various energies, etc.,
and denote $\mathcal{X}(\tau_0,\tau_1)$ by $\mathcal{X}$.

The smallness parameter $\varepsilon_{\rm global}$ will be determined in the context of the proof. 
We will in fact introduce other parameters on the way; these will be in the relation
\begin{equation}
\label{discofeps}
\varepsilon_{\rm global} \le  \hat{\varepsilon}_{\rm slab} \le \varepsilon_{\rm slab} \le \varepsilon_{\rm boot} \le \varepsilon_{\rm loc},
\end{equation}
where parameters $\hat{\varepsilon}_{\rm slab}$,  $\varepsilon_{\rm slab}$ will only appear for cases (ii) and (iii).
The parameter $\varepsilon_{\rm boot}$ will be used in the context of a continuity or bootstrap argument. Let us remark already that 
if the basic bootstrap estimate~\eqref{basicbootstrap} holds for an $\varepsilon\le \varepsilon_{\rm boot}$, then by  Remark~\ref{remarkonconstantshere},
 $\varepsilon$ will satisfy the smallness requirements of all propositions of Section~\ref{prelimsec},   for the regions under consideration.

\subsection{Case (i)}
\label{easycase}

Case (i) is the most elementary case, where, moreover only the minimal nontrivial
$r^p$ weights for $2-\delta \ge p\ge \delta$
are required on data. (For instance, one may fix $p=\delta$.)
We present first the fundamental estimate in Section~\ref{caseonehierarchy} that holds
under the bootstrap assumption~\eqref{basicbootstrap} and then
carry out the global existence proof in Section~\ref{globalexistenceproof}.

The proof will also hold for the semilinear case 
 with the relaxed assumptions of  Remark~\ref{semiremsemi}, if we are in the analogue of case (i), i.e.~if $\chi=1$.
The modifications necessary to treat this case are described in a series of remarks (see already 
Remarks~\ref{remforcaseisemi} and~\ref{caseisemifinal}).

\subsubsection{The fundamental estimate}
\label{caseonehierarchy}

The fundamental estimate is given simply by the following:
\begin{proposition}
\label{ihierprop}
Let $\delta\le p\le 2-\delta$ and let $k$ be sufficiently large, and let us assume the case (i) assumptions.
There exists an $\varepsilon_{\rm boot}>0$ sufficiently small so that the following is true.

Let $\psi$ solve~\eqref{theequation} on $\mathcal{R}(\tau_0,\tau_1)$, satisfying 
moreover~\eqref{basicbootstrap} with our chosen $p$ and with $0<\varepsilon\le \varepsilon_{\rm boot}$.  Then we have:
\begin{equation}
\label{pweighted}
\Xpk \lesssim \Epk(\tau_0) , \qquad \sup_{\tau_0\le \tau\le \tau_1} \Efancypk(\tau) \leq \Efancypk(\tau_0) (1+C\varepsilon^{1/2}).
\end{equation}
\end{proposition}
\begin{proof}
Let~$\varepsilon_{\rm boot}\le \varepsilon_{\rm loc}$.  The assumption of Proposition~\ref{finalpropest} is satisfied in view
of our discussion following~\eqref{discofeps}.
From estimate~\eqref{addtermoneherewithflux} of 
Proposition~\ref{finalpropest}, since in case (i) we have $\rho=1$, $A=0$, $\tilde{A}=0$, it follows that 
\begin{equation}
\label{pweightedproof}
\Xpk \lesssim \Epk(\tau_0) 
+ \Xpk
\sqrt{\Xplesslessk}  ,\qquad \sup_{\tau_0\le \tau\le \tau_1} \Efancypk(\tau)  \leq \Efancypk(\tau_0) 
+ C\Xpk
\sqrt{\Xplesslessk}.
\end{equation}
By our bootstrap assumption~\eqref{basicbootstrap}, we have
\begin{equation}
\label{andtheboot}
\Xplesslessk \leq \varepsilon.
\end{equation}
Thus, possibly requiring $\varepsilon_{\rm boot}$ to be even smaller, we may absorb the nonlinear term in the first inequality
of~\eqref{pweightedproof} to obtain the first inequality of~\eqref{pweighted}.

To obtain the second inequality of~\eqref{pweighted}, we remark that Corollary~\ref{nomoreroman} 
applies to yield in particular
\begin{equation}
\label{heretheyreequi}
\Epk(\tau_0) \sim \Efancypk(\tau_0).
\end{equation}
The second inequality of~\eqref{pweighted} now immediately follows from the 
 the second inequality of~\eqref{pweightedproof}, plugging in the
 first inequality of~\eqref{pweighted}
just established,~\eqref{heretheyreequi}
 and~\eqref{andtheboot}. 
\end{proof}

\begin{remark}
\label{remforcaseisemi}
Under the relaxed assumptions of Remark~\ref{semiremsemi}, 
where~\eqref{theequation} is however required to be semilinear and we are in the analogue of case (i), i.e.~$\chi=1$,
then from~\eqref{addtermoneherechizeroforsemi} we obtain the first inequality of~\eqref{pweightedproof}. Thus, the above proposition
holds as stated for the first inequality of~\eqref{pweighted}.
\end{remark}

\subsubsection{Proof of Theorem~\ref{themaintheorem} in case (i)}
\label{globalexistenceproof}

We now carry out the proof of Theorem~\ref{themaintheorem} proper in case (i).

Let $\delta\le p\le 2-\delta$ be as in the statement and
recall the assumption
\[
 \Epk(\tau_0) \leq \varepsilon_0\leq\varepsilon_{\rm global}
 \]
on initial data, for a sufficiently small $\varepsilon_{\rm global}$ to be determined.

Consider
the set $\mathfrak{B}$ consisting of all $\tau_f>\tau_0$ such that a solution $\psi$ of~\eqref{theequation}
obtaining the data exists on $\mathcal{R}(\tau_0,\tau_f)$,
and such that moreover 
the energy bootstrap assumption~\eqref{basicbootstrap} 
holds on $\mathcal{R}(\tau_0,\tau_1:=\tau_f)$ with our chosen $p$ and where $0< \varepsilon\le \varepsilon_{\rm boot}$ is chosen to 
satisfy
\begin{equation}
\label{satisfyingapp}
1\gg \varepsilon\gg \varepsilon_{\rm global}.
\end{equation}
(The above relation in particular constrains $\varepsilon_{\rm global}$ to be small.)
By the local well posedness statement Proposition~\ref{localexistence},
it follows
that since $k\ge k_{\rm loc}$ and $\varepsilon_0\leq \varepsilon_{\rm global}\leq  \varepsilon_{\rm loc}$,
we have that $\tau_0+\tau_{\rm exist}\in \mathfrak{B}$ and thus $\mathfrak{B}\ne\emptyset$.
Also note that a fortiori, if $\tau_f\in \mathfrak{B}$, then $(\tau_0,\tau_f]\subset \mathfrak{B}$
and thus $\mathfrak{B}$  is manifestly a connected  subset of $(\tau_0,\infty)$.

For any $\tau_1\in \mathfrak{B}$, Proposition~\ref{ihierprop} applies in $\mathcal{R}(\tau_0,\tau_1)$.
The first estimate of~\eqref{pweighted} then yields
\begin{equation}
\label{whatwehavehere}
\Xpk\lesssim \Epk(\tau_0) \lesssim \varepsilon_0.
\end{equation}

If follows that
for sufficiently small $\varepsilon_{\rm glob}$, and all $\varepsilon_0\leq \varepsilon_{\rm glob}$,
the continuation criterion Corollary~\ref{continuash} 
applies to obtain that $\psi$ extends as a solution to~\eqref{theequation}
on some $\mathcal{R}(\tau_0,\tau_f+\epsilon)$, where $\epsilon$ is independent of $\tau_f$, and that, considering now $\tau_1:= \tau_f+\epsilon$,
inequality~\eqref{whatwehavehere} holds also in 
$\mathcal{R}(\tau_0,\tau_1:=\tau_f+\epsilon)$
(with a slightly different implicit constant than that of~\eqref{whatwehavehere}).
(To infer this from~\eqref{extensionhere}, which was an estimate on $\Xpk(\tau_f,\tau_f+\epsilon)$,  note
that $\Xpk(\tau_0,\tau_1)\lesssim \Xpk(\tau_0,\tau_f)+\Xpk(\tau_f,\tau_f+\epsilon)$.)

Finally we note that for $\varepsilon$ satisfying~\eqref{satisfyingapp},
inequality~\eqref{whatwehavehere} in $\mathcal{R}(\tau_0,\tau_1:=\tau_f+\epsilon)$
shows in particular that inequalities~\eqref{basicbootstrap} hold in this set.
It follows by the definition of $\mathfrak{B}$ that  $\tau_f+\epsilon\in \mathfrak{B}$ and thus, given also its connectedness, the set $\mathfrak{B}$ is open.
Since $\epsilon$ does not depend on $\tau_f$, it follows that $\mathfrak{B}$ is also closed, as any limit point $\tau_{\rm limit}$
of $\mathfrak{B}$ in $(\tau_0,\infty)$ satisfies $\tau_{\rm limit} \le \tau_f+\epsilon$ for some $\tau_f \in \mathfrak{B}$.
We have shown that $\mathfrak{B}$ is a non-empty open and closed subset of $(\tau_0,\infty)$ and thus,
$\mathfrak{B}=(\tau_0,\infty)$. Hence  the solution $\psi$ 
exists globally in $\mathcal{R}(\tau_0,\infty)$ and 
satisfies~\eqref{whatwehavehere} in $\mathcal{R}(\tau_0,\tau_1)$ where $\tau_1$ is now
any $\tau_1>\tau_0$. This gives~\eqref{basicestimate}.

Revisiting~\eqref{pweighted} 
we see finally that we have the following more precise global  estimate:
\begin{equation}
\label{finalforcasei}
\Efancypk(\tau) \leq \Efancypk(\tau_0)(1 +C \varepsilon_0^{1/2}).
\end{equation}
This completes the proof.

\begin{remark}
\label{caseisemifinal}
Note that in view of Remark~\ref{remforcaseisemi}, 
the above proof goes through,
except for the last paragraph involving~\eqref{finalforcasei},
 also under the relaxed assumptions of Remark~\ref{semiremsemi}, 
where~\eqref{theequation} is however required to be semilinear and we are in the analogue of case (i), i.e.~$\chi=1$.
Thus, in this case, one indeed obtains the statement as given in Remark~\ref{bigremarkonsemi}. 
\end{remark}

\subsection{Case (ii)}
\label{itscaseiinow}

Case (ii) is the second simplified case which we shall consider.
Unlike case (i), we shall require proving actual $\tau$-decay above a certain threshold in order to close, and this
in turn will require raising $p$ to $p=2-\delta$ in our initial energy assumptions.

The logic of the proof is a simplified version of the scheme we shall use for the general case
and which has already been summarised in the introduction:

\begin{enumerate}
\item
Rather than bootstrap directly $t$-weighted estimates, we formulate the fundamental estimates, depending only on
the bootstrap assumption~\eqref{basicbootstrap},
as a hierarchy of estimates on a spacetime slab of $\tau$-length at most $L$. 
This is the content of Section~\ref{hierarchyfortwo}.

\item
We shall then show by a bootstrap argument, \emph{restricted to such a slab}, how global existence
and estimates on the slab can be proven, with suitable assumptions on the initial data of the slab, which now
involve $L$.  This is the content of Section~\ref{globonellfortwo}.

\item Still restricted to a given slab of length $L$, we will
show that one can in fact a posteriori improve the above estimates
under suitable additional assumptions on the initial data,
and, using a pigeonhole argument,  show moreover an improved estimate
for any  $\Sigma(\tau')$ slice near the top of
the slab.  This is the content of Section~\ref{pigeons}.

\item
Finally, global existence and $\tau$ decay now follow by iterating the above estimates on 
a consecutive sequence of spacetime slabs of dyadic time-length $L_i=2^i$. This is the content of Section~\ref{theiterationfortwo}. 
\end{enumerate}

We shall in addition give an alternative iteration proof in Section~\ref{thereisnoalternative}, which does not use the exact boundedness
statement of the fraktur energies.
The advantage of this alternative proof is that it allows us to treat the semilinear case with the relaxed assumptions of  Remark~\ref{semiremsemi}. 
The modifications necessary to treat this case are described in a series of remarks (see already 
Remarks~\ref{remforcaseiisemi},~\ref{relaxedremark},~\ref{holdsinSemilinear} and~\ref{procremark}).

\subsubsection{The hierarchy of inequalities} 
\label{hierarchyfortwo}

\begin{proposition}
\label{iihierprop}
Let $k$ be sufficiently large and let us assume the case (ii) assumptions. 
There exist constants $C>0$, $c>0$ and an $\varepsilon_{\rm boot}>0$ sufficiently small 
such that the following is true.

Consider a region $\mathcal{R}(\tau_0,\tau_1)$
and a $\psi$ solving~\eqref{theequation} on  $\mathcal{R}(\tau_0,\tau_1)$, satisfying 
moreover~\eqref{basicbootstrap} with $p=0$ and with $0<\varepsilon\le \varepsilon_{\rm boot}$.
Let us assume moreover that 
\[
\tau_1\le \tau_0+L
\]
for some arbitrary $L>0$.
We have the following hierarchy of inequalities on $\mathcal{R}(\tau_0,\tau_1)$:
\begin{eqnarray}
 \label{pweightedii}
\Ffancytwominusdelk(v,\tau_1),\quad
\Efancytwominusdelk(\tau_1) , \quad  c\,{}^\chi\Xtwominusdelk
&\leq& \Efancytwominusdelk(\tau_0)  +C\, \left(  {}^\chi\Xtwominusdelk\sqrt{\Xzerolesslessk} + \sqrt{{}^\chi\Xtwominusdelk}\sqrt{ {}^\chi\Xzerok}
\sqrt{\Xtwominusdeltalesslessk} \right)+
C\, {}^\chi\Xzerok\sqrt{\Xzerolesslessk}\sqrt{L},
  \\
\label{pminusoneweightedii}
\Ffancyonek(v,\tau_1),\quad
\Efancyonek(\tau_1),\quad c\, {}^\chi \Xonek&\leq& \Efancyonek(\tau_0) 
+ C\, {}^\chi \Xonek
\sqrt{\Xonelesslessk} 
 + C\, {}^\chi\Xzerok\sqrt{\Xzerolesslessk}\sqrt{L},   \\
\nonumber
\Ffancyzerok(v,\tau_1),\quad \Efancyzerok(\tau_1), \quad c\,{}^\chi \Xzerok &\leq& 
\Efancyzerok(\tau_0)
+ C\left( \,{}^\chi \Xzerok + (\,{}^\chi \Xzerok)^{\frac{1-\delta}{1+\delta}} (\, {}^\chi \Xonek)^{\frac{2\delta}{1+\delta}} \right)
\sqrt{\Xzerolesslessk+ (\,\Xzerolesslessk)^{\frac{1-\delta}{1+\delta}} (\, \Xonelesslessk)^{\frac{2\delta}{1+\delta}}}\\
\label{nonpweightedii}
&&\qquad 
+ C\, {}^\chi\Xzerok\sqrt{\Xzerolesslessk}\sqrt{L}.
\end{eqnarray}
\end{proposition}

\begin{proof}
Again we recall that if $\varepsilon_{\rm boot}\le \varepsilon_{\rm loc}$, then the assumption of Proposition~\ref{finalpropest} holds. 
The first two inequalities follow from the estimate~\eqref{addtermoneherewithflux} of
 Proposition~\ref{finalpropest}
 applied to $p=2-\delta$, $p=1$
 in view of the fact that 
 $\rho=\chi$,  where we have used
 the relations~\eqref{therelationweknow}
 and~\eqref{implicationhere} to replace the far-away supported nonlinear terms with those displayed above,
 together with the estimate
\[
 \int_{\tau_0}^{\tau_1} \Ezerok (\tau') \sqrt{\,\,\Ezerominusoneminusdeltalesslessk'(\tau')  }d\tau'
 \lesssim \sup_{\tau_0\le \tau\le\tau_1} \Ezerok(\tau)  \sqrt{\int_{\tau_0}^{\tau_1} \Ezerominusoneminusdeltalesslessk'(\tau')d\tau' }
 \, \cdot \, \sqrt{L} \lesssim {}^\chi\Xzerok\sqrt{\Xzerolesslessk}\sqrt{L}.
\]
Note  that for the $p=1$ estimate we retain only the less precise expression ${}^\chi \Xonek
\sqrt{\Xonelesslessk}$ which will be sufficient for our purposes. 

The third inequality follows similarly using estimate~\eqref{addtermoneherewithfluxzero},
where we 
use also the interpolation inequality~\eqref{interpolationstatementbulk}
to obtain
\[
 \Xzeroplusk_{\frac{8R}9} (\tau_0,\tau_1) \sqrt{ \Xzeropluslesslessk_{\frac{8R}9}(\tau_0,\tau_1)}
 \lesssim 
 \left( \,{}^\chi \Xzerok + (\,{}^\chi \Xzerok)^{\frac{1-\delta}{1+\delta}} (\, {}^\chi \Xonek)^{\frac{2\delta}{1+\delta}} \right)
\sqrt{\Xzerolesslessk+ (\,\Xzerolesslessk)^{\frac{1-\delta}{1+\delta}} (\, \Xonelesslessk)^{\frac{2\delta}{1+\delta}}}.
\]
\end{proof}

\begin{remark}
\label{remforcaseiisemi}
We note that inequalities~\eqref{pweightedii}--\eqref{nonpweightedii} imply of course 
\begin{eqnarray}
 \label{pweightediialt}
{}^\chi\Xtwominusdelk
&\lesssim& \Etwominusdelk(\tau_0)  + \left(  {}^\chi\Xtwominusdelk\sqrt{\Xzerolesslessk} + \sqrt{{}^\chi\Xtwominusdelk}\sqrt{ {}^\chi\Xzerok}
\sqrt{\Xtwominusdeltalesslessk} \right)+
 {}^\chi\Xzerok\sqrt{\Xzerolesslessk}\sqrt{L},
  \\
\label{pminusoneweightediialt}
{}^\chi \Xonek&\lesssim& \Eonek(\tau_0) 
+  {}^\chi \Xonek
\sqrt{\Xonelesslessk} 
 +  {}^\chi\Xzerok\sqrt{\Xzerolesslessk}\sqrt{L},   \\
\label{nonpweightediialt}
{}^\chi \Xzerok &\lesssim& 
\Ezerok(\tau_0)
+ \left( \,{}^\chi \Xzerok + (\,{}^\chi \Xzerok)^{\frac{1-\delta}{1+\delta}} (\, {}^\chi \Xonek)^{\frac{2\delta}{1+\delta}} \right)
\sqrt{\Xzerolesslessk+ (\,\Xzerolesslessk)^{\frac{1-\delta}{1+\delta}} (\, \Xonelesslessk)^{\frac{2\delta}{1+\delta}}}
+  {}^\chi\Xzerok\sqrt{\Xzerolesslessk}\sqrt{L}.
\end{eqnarray}
Following the above proof but now using
equations~\eqref{addtermoneherechinewhere} and~\eqref{addtermoneherechizeroforsemi}, 
we may in fact directly deduce the inequalities~\eqref{pweightediialt}--\eqref{nonpweightediialt} 
under the relaxed assumptions of Remark~\ref{semiremsemi}, 
where~\eqref{theequation} is however required to be semilinear and we are in the analogue of case (ii). Thus, the analogue
of Proposition~\ref{iihierprop} holds in that case where~\eqref{pweightedii}--\eqref{nonpweightedii}
 are replaced by~\eqref{pweightediialt}--\eqref{nonpweightediialt}.
\end{remark}

\subsubsection{Global existence in $L$-slabs}
\label{globonellfortwo}

The presence of positive powers of $L$ in~\eqref{pweightedii}--\eqref{nonpweightedii} 
means that our smallness assumptions must involve negative powers of $L$ in order for the
estimates to close.

\begin{proposition}
\label{globproptwo}
Let $k-1\ge k_{\rm loc}$ be sufficiently large and let us assume the case (ii) assumptions.  Then  there exists an
 $0<\varepsilon_{\rm slab}\leq\varepsilon_{\rm loc}$ and a constant $C>0$ implicit in the sign $\lesssim$ below such that the following  is true.

Given arbitrary $L\ge 1$,
 $\tau_0\ge 0$,  $0<\varepsilon_0\leq \varepsilon_{\rm slab}$ and initial data $(\uppsi,\uppsi')$ 
 on $\Sigma(\tau_0)$ as in Proposition~\ref{localexistence}, satisfying moreover
\begin{equation}
\label{assumptionondatahere}
\qquad  \Eonek(\tau_0) \leq \varepsilon_0,
\qquad \Ezerokminusone(\tau_0) \leq \varepsilon_0L^{-1},
\end{equation}
then the unique solution of Proposition~\ref{localexistence} obtaining the data can be extended to 
a $\psi$
defined on the entire  spacetime slab $\mathcal{R}(\tau_0,\tau_0+L)$
satisfying the equation~\eqref{theequation}
and
 the estimates
 \begin{equation}
 \label{newformatestimates}
\qquad {}^\chi  \Xonek  \lesssim \varepsilon_0,
\qquad {}^\chi \Xzerokminusone  \lesssim \varepsilon_0L^{-1}.
 \end{equation}
 \end{proposition}

\begin{proof}
Consider the set $\mathfrak{B}\subset (\tau_0,\tau_0+L]$ 
consisting of all $\tau_0+L\ge \tau_f\ge \tau_0$ such that a solution $\psi$ of~\eqref{theequation}
obtaining the data exists on $\mathcal{R}(\tau_0,\tau_f)$ 
and such that
the boostrap assumption~\eqref{basicbootstrap} with $p=1$
and also the additional bootstrap assumption
\begin{equation}
\label{newversionadditionalboot}
\Xzerolesslessk \leq  \varepsilon L^{-1}
\end{equation}
 hold in $\mathcal{R}(\tau_0,\tau_1:=\tau_f)$, 
where
$0<\varepsilon\le \varepsilon_{\rm boot}$ is a small constant satisfying
\begin{equation}
\label{providedprovidedbefore}
1\gg \varepsilon\gg \varepsilon_{\rm slab}.
\end{equation}
(The above relation in particular already constrains $\varepsilon_{\rm slab}$ to be sufficiently small.)

By the local well posedness statement Proposition~\ref{localexistence},
it follows
that since $k-1\ge k_{\rm loc}$ and $\varepsilon_0\leq \varepsilon_{\rm slab}\leq  \varepsilon_{\rm loc}$,
we have that $\tau_0+\tau_{\rm exist}\in \mathfrak{B}$ and thus $\mathfrak{B}\ne\emptyset$, provided that
$\varepsilon$ satisfies~\eqref{providedprovidedbefore}.
Also note that a fortiori, if $\tau_f\in \mathfrak{B}$, then $(\tau_0,\tau_f]\in \mathfrak{B}$
and thus $\mathfrak{B}$  is manifestly a connected  subset of $(\tau_0,\infty)$.

For $\tau_1:=\tau_f\in \mathcal{B}$, Proposition~\ref{iihierprop} applies in $\mathcal{R}(\tau_0,\tau_1)$.
From~\eqref{pminusoneweightedii}   and~\eqref{assumptionondatahere} we obtain
\begin{equation}
\label{middleinhier}
{}^\chi  \Xonek \lesssim \varepsilon_0  + \varepsilon^{1/2} \, {}^\chi  \Xonek,
\end{equation}
where we have used the bootstrap assumptions~\eqref{basicbootstrap} and~\eqref{newversionadditionalboot}.

From~\eqref{nonpweightedii} for $k-1$  and~\eqref{assumptionondatahere}, we obtain
\begin{equation}
\label{lowestinhier}
{}^\chi \Xzerokminusone \lesssim \varepsilon_0 L^{-1} +
 \varepsilon^{1/2} (\,{}^\chi \Xzerokminusone)^{\frac{1-\delta}{1+\delta}} (\, {}^\chi \Xonekminusone)^{\frac{2\delta}{1+\delta}}
L^{- \frac12\frac{1-\delta}{1+\delta}}
+ \varepsilon^{1/2} \, {}^\chi  \Xzerokminusone.
\end{equation}

It follows that for $\varepsilon$ satisfying~\eqref{providedprovidedbefore}, we obtain
\begin{equation}
\label{decoupled}
{}^\chi  \Xtwominusdelk  \lesssim \varepsilon_0, \qquad {}^\chi  \Xzerokminusone  \lesssim \varepsilon_0L^{-\frac12\frac{1-\delta}{1+\delta}}
\end{equation}
and plugging the second inequality of~\eqref{decoupled} into~\eqref{lowestinhier} we obtain
\begin{equation}
\label{preresultforlowest}
{}^\chi  \Xzerokminusone  \lesssim \varepsilon_0 L^{-1} +\varepsilon^{1/2} \varepsilon_0
L^{-\frac12 (\frac{1-\delta}{1+\delta})^2} L^{-\frac12\frac{1-\delta}{1+\delta}}.
\end{equation}
Now  defining 
\[
\gamma_0:= \frac12 \left(\frac{1-\delta}{1+\delta}\right)^2 +\frac12\frac{1-\delta}{1+\delta} ,
\]
we have
\begin{equation}
\label{newcentred}
{}^\chi  \Xzerokminusone  \lesssim \varepsilon_0L^{-\gamma_0},
\end{equation}
whence, plugging~\eqref{newcentred} into~\eqref{lowestinhier}  
we improve to 
\begin{equation}
\label{preresultforlowestag}
{}^\chi  \Xzerokminusone  \lesssim \varepsilon_0 L^{-1} +\varepsilon^{1/2} \varepsilon_0
L^{-\gamma_0  (\frac{1-\delta}{1+\delta})^2} L^{-\frac12\frac{1-\delta}{1+\delta}}.
\end{equation}
Defining $\gamma_i$ iteratively by $\gamma_i= \gamma_{i-1} \left(\frac{1-\delta}{1+\delta}\right)^2 +\frac12\frac{1-\delta}{1+\delta}$,
then, in view of the restriction~\eqref{fixeddelta}, there exists a first $i$ such that $\gamma_i\ge 1$, whence we obtain
\begin{equation}
\label{resultforlowest}
{}^\chi  \Xzerokminusone  \lesssim \varepsilon_0 L^{-1}.
\end{equation}

We may now already apply our continuation criterion Corollary~\ref{continuash}, applied with $p=0$ and $k-1$,
to assert the existence of an $\epsilon$, independent of $\tau_1$, such that now defining $\tau_1:= \min \{ \tau_f+\epsilon, \tau_0+ L\}$,
the solution $\psi$ extends to a smooth solution of~\eqref{theequation}
on $\mathcal{R}(\tau_0,\tau_1)$, and moreover, from~\eqref{extensionhere}, that
\begin{equation}
\label{additionalbootherehereiiwithC}
{}^\chi \Xzerokminusone 
   \lesssim \varepsilon_0 L^{-1}
\end{equation}
holds on $\mathcal{R}(\tau_0,\tau_1)$,
and from~\eqref{extensionhigherhigher}, that
\begin{equation}
\label{alsoextendingthis}
{}^\chi \Xonek 
   \lesssim \varepsilon_0 
\end{equation}
holds on $\mathcal{R}(\tau_0,\tau_1)$.
It follows that~\eqref{newversionadditionalboot} and~\eqref{basicbootstrap} (with $p=1$) hold on $\mathcal{R}(\tau_0,\tau_1)$ and thus
 $\tau_1= \min \{ \tau_f+\epsilon, \tau_0+ L\} \in \mathfrak{B}$.

Since $\epsilon$ is independent of $\tau_f$, and in view also of the connectivity of $\mathfrak{B}$, 
it follows that $\mathfrak{B}$ is a nonempty open and closed subset of $(\tau_0, \tau_0+L]$
and thus $\mathfrak{B}=(\tau_0,\tau_0+L]$ and hence the solution exists in the entire spacetime slab
$\mathcal{R}(\tau_0,\tau_0+L)$.

The estimates~\eqref{decoupled} and~\eqref{resultforlowest} thus hold in the entire spacetime slab. 
This gives~\eqref{newformatestimates}.
\end{proof}

\begin{remark}
\label{relaxedremark}
Examining the proof, it is clear that we have only used~\eqref{pweightediialt}--\eqref{nonpweightediialt} and not the
full~\eqref{pweightedii}--\eqref{nonpweightedii}
of Proposition~\ref{iihierprop}. Thus, in view of the comments of Remark~\ref{remforcaseiisemi},
the proof holds also for the semilinear case under the relaxed assumptions described in Remark~\ref{semiremsemi}.
In general, examining the proof, we may in fact relax the assumption of the first inequality of~\eqref{assumptionondatahere} to the assumption
\begin{equation}
\label{relaxedwithbetaassumption}
\Eonek(\tau_0)\leq \varepsilon_0 L^\beta
\end{equation}
for a sufficiently small $\beta$,
in which case the first inequality of~\eqref{newformatestimates}  is replaced by
\begin{equation}
\label{relaxedwithbetaconclusion}
{}^\chi\Xonek(\tau_0)\leq \varepsilon_0 L^\beta.
\end{equation}
This will again be useful for the semilinear case.
\end{remark}

\subsubsection{The pigeonhole argument}
\label{pigeons}

The above assumptions on initial data are sufficient for global existence in the slab, but are not sufficient
to iterate. For this we shall need  strengthened assumptions.

\begin{proposition}
\label{bettertimeplease}
Under the assumptions of Proposition~\ref{globproptwo},
there exists a constant $C>0$, implicit in the inequalities $\lesssim$ below, a  parameter $\alpha_0\gg 1$
and, for all $\alpha\ge \alpha_0$, a parameter~$\hat{\varepsilon}_{\rm slab}(\alpha)$
such that for all $0<\hat\varepsilon_0 \le\hat{\varepsilon}_{\rm slab}(\alpha)$  the following  holds.

Let us assume that in addition to~\eqref{assumptionondatahere} we have 
\begin{equation}
\label{additassum}
\Efancyonek(\tau_0) \leq \hat\varepsilon_0, \qquad
 \Efancytwominusdelkminusone(\tau_0) \leq \hat\varepsilon_0,  \qquad
 \Efancyzerokminusone(\tau_0) \leq \hat\varepsilon_0 \alpha L^{-1},
 \qquad
   \Efancyonekminustwo(\tau_0) \leq  \hat\varepsilon_0\alpha L^{-1+\delta},
\qquad \Efancyzerokminusthree(\tau_0) \leq \hat\varepsilon_0\alpha^2 L^{-2+\delta}.
\end{equation}
Then the solution $\psi$ of~\eqref{theequation}
on $\mathcal{R}(\tau_0,\tau_0+L)$ given by Proposition~\ref{globproptwo} 
satisfies the additional estimates
\begin{equation}
\label{addedherenotfancy}
\sup_{\tau_0\le \tau\le \tau_0+L} \Eonek(\tau) \leq   \alpha\hat\varepsilon_0,
\end{equation}
\begin{equation}
\label{basicnewformiiagainkaivourgio}
\sup_{\tau_0\le \tau \le \tau_0+L} \Efancyonek(\tau) \leq \hat\varepsilon_0 (1+ \hat\varepsilon_0^{1/4}L^{(-1+\delta)/2}) ,
\end{equation}
\begin{equation}
\label{basicnewformiiagainallokaivourgio}
\sup_{\tau_0\le \tau \le \tau_0+L} 
\Efancytwominusdelkminusone(\tau) \leq \hat\varepsilon_0 (1+ \hat\varepsilon_0^{1/4}L^{-1/2})   , 
\end{equation}
\begin{equation}
\label{kiallobasicnewformiiagain}
{}^\chi \Xonek\lesssim\hat \varepsilon_0, \qquad {}^\chi \Xtwominusdelkminusone\lesssim\hat\varepsilon_0,\qquad
{}^\chi \Xzerokminusone \lesssim \hat\varepsilon_0\alpha L^{-1}, \qquad
{}^\chi \Xonekminustwo \lesssim\hat\varepsilon_0\alpha L^{-1+\delta}, \qquad 
 {}^\chi \Xzerokminusthree
\lesssim\hat\varepsilon_0\alpha^2 L^{-2+\delta}.
\end{equation}

 Moreover, 
for all times $\tau'$ with $L\ge \tau'-\tau_0 \ge L/2$, we have that
\begin{equation}
\label{betterjust}
\Efancyzerokminusone(\tau') \leq \frac12  \hat\varepsilon_0 \alpha L^{-1},
\end{equation}
\begin{equation}
\label{betterone}
\Efancyonekminusone(\tau') \leq \frac12  \hat\varepsilon_0 \alpha L^{-1+\delta},
\end{equation}
\begin{equation}
\label{bettertwo}
\Efancyzerokminustwo(\tau') \leq\frac14  \hat\varepsilon_0  \alpha^2 L^{- 2+\delta}.
\end{equation}
\end{proposition}

\begin{proof}
Note that by the statement of the proposition, we are in particular
also assuming a priori~\eqref{assumptionondatahere} for \emph{some} $\varepsilon_0\leq \varepsilon_{\rm slab}$,
since this is included in the assumptions of Proposition~\ref{globproptwo}.
In view now of Corollary~\ref{nomoreroman}, however, for $\alpha_0$ sufficiently large,
if say $\hat{\varepsilon}_{\rm slab}(\alpha)\ll \alpha^{-3} \varepsilon_{\rm slab}$,
it follows from the additional assumptions~\eqref{additassum} that the estimates~\eqref{assumptionondatahere} 
in fact hold with the specific constant $\varepsilon_0:=\hat\varepsilon_0\alpha^3$ for all $\alpha\ge \alpha_0$.

To obtain~\eqref{kiallobasicnewformiiagain} we argue as follows.
We note from the proof of Proposition~\ref{globproptwo} that we have the following system of inequalities
\begin{equation}
\label{newmiddleinhieragain}
{}^\chi  \Xonek \lesssim \hat\varepsilon_0 + \varepsilon_0^{1/2} \, {}^\chi  \Xonek,
\end{equation}
\begin{equation}
\label{highestinhieragain}
{}^\chi  \Xtwominusdelkminusone \lesssim \hat\varepsilon_0 + \varepsilon_0^{1/2} \, {}^\chi  \Xtwominusdelkminusone,
\end{equation}
\begin{equation}
\label{lowestinhieragainallaoxi}
{}^\chi \Xzerokminusone \lesssim \hat\varepsilon_0  \alpha L^{-1} +
(\,{}^\chi \Xzerokminusone)^{\frac{1-\delta}{1+\delta}} (\, {}^\chi \Xonekminusone)^{\frac{2\delta}{1+\delta}}
\sqrt{\Xzerolesslessk+ (\,\Xzerolesslessk)^{\frac{1-\delta}{1+\delta}} (\, \Xonelesslessk)^{\frac{2\delta}{1+\delta}}}
+  \varepsilon_0^{1/2} \,  {}^\chi\Xzerokminusone,
\end{equation}
\begin{equation}
\label{middleinhieragain}
{}^\chi  \Xonekminustwo \lesssim \hat\varepsilon_0 \alpha L^{-1+\delta} + \varepsilon_0^{1/2} \, {}^\chi  \Xonekminustwo,
\end{equation}
\begin{equation}
\label{lowestinhieragain}
{}^\chi \Xzerokminusthree \lesssim \hat\varepsilon_0  \alpha^2  L^{-2+\delta} +
(\,{}^\chi \Xzerokminusthree)^{\frac{1-\delta}{1+\delta}} (\, {}^\chi \Xonekminusthree)^{\frac{2\delta}{1+\delta}}
\sqrt{\Xzerolesslessk+ (\,\Xzerolesslessk)^{\frac{1-\delta}{1+\delta}} (\, \Xonelesslessk)^{\frac{2\delta}{1+\delta}}}
+  \varepsilon_0^{1/2} \,  {}^\chi\Xzerokminusthree,
\end{equation}
where we have used also the initial data assumptions~\eqref{additassum}.

For $\hat\varepsilon_{\rm slab}(\alpha)$ sufficiently small, we have $\varepsilon_0\ll 1$ and thus
we  obtain immediately that 
\begin{equation}
\begin{aligned}
\label{twoandthenathird}
{}^\chi \Xonek\lesssim \hat\varepsilon_0, \qquad
{}^\chi  \Xtwominusdelkminusone  \lesssim \hat \varepsilon_0, \qquad
{}^\chi  \Xonekminustwo  \lesssim \hat \varepsilon_0 \alpha L^{-1+\delta}, \\
{}^\chi \Xzerokminusone \lesssim \hat\varepsilon_0  \alpha L^{-1} +
(\,{}^\chi \Xzerokminusone)^{\frac{1-\delta}{1+\delta}} (\, {}^\chi \Xonekminusone)^{\frac{2\delta}{1+\delta}}
\sqrt{\Xzerolesslessk+ (\,\Xzerolesslessk)^{\frac{1-\delta}{1+\delta}} (\, \Xonelesslessk)^{\frac{2\delta}{1+\delta}}},
\\
{}^\chi \Xzerokminusthree \lesssim \hat\varepsilon_0  \alpha^2  L^{-2+\delta} +
(\,{}^\chi \Xzerokminusthree)^{\frac{1-\delta}{1+\delta}} (\, {}^\chi \Xonekminusthree)^{\frac{2\delta}{1+\delta}}
\sqrt{\Xzerolesslessk+ (\,\Xzerolesslessk)^{\frac{1-\delta}{1+\delta}} (\, \Xonelesslessk)^{\frac{2\delta}{1+\delta}}}.
\end{aligned}
\end{equation}
The first two inequalities of~\eqref{twoandthenathird} yield the first two inequalities of~\eqref{kiallobasicnewformiiagain},
while the third inequality of~\eqref{twoandthenathird}  yields the fourth inequality of~\eqref{kiallobasicnewformiiagain}.
On the other hand,  from the latter inequality of~\eqref{twoandthenathird} we obtain 
\begin{equation}
\label{fromlatterweobt}
{}^\chi \Xzerokminusthree  \lesssim   \hat\varepsilon_0 \alpha^2 L^{-2+\delta} +
\hat\varepsilon_0\varepsilon_0^{1/2}L^{-\frac32\frac{(1-\delta)(1+2\delta)}{1+\delta}}.
\end{equation}

Note that $-\frac32\frac{(1-\delta)(1+2\delta)}{1+\delta} > -2+\delta$.
We may however improve~\eqref{fromlatterweobt} iteratively as follows:
If 
\[
{}^\chi \Xzerokminusthree  \lesssim  \hat\varepsilon_0 \alpha^2  L^{-\gamma},
\]
for $\gamma<2-\delta$,
then 
\[
\Xzerolesslessk \lesssim \hat \varepsilon_0 \alpha^2L^{-\gamma}\lesssim \varepsilon_0 L^{-\gamma},
\]
hence, plugging this  again into the final inequality of~\eqref{twoandthenathird} to estimate the nonlinear term,
we obtain
\begin{equation}
\label{inductivehere}
{}^\chi \Xzerokminusthree  \lesssim   \hat\varepsilon_0 \alpha^2 L^{-2+\delta}  + \hat\varepsilon_0 \varepsilon_0^{1/2} 
L^{-\frac32\frac{(1-\delta)(\gamma+2\delta)}{1+\delta}}.
\end{equation}
Setting $\gamma_0=\frac32\frac{(1-\delta)(1+2\delta)}{1+\delta}$, and defining inductively
$\gamma_i=\frac32\frac{(1-\delta)(\gamma_{i-1}+2\delta)}{1+\delta}$,
we have that there exists a first $i\ge 1$ such that $\frac32\frac{(1-\delta)(\gamma_{i}+2\delta)}{1+\delta} >2-\delta$.
It follows that~\eqref{inductivehere} holds for $\gamma=\gamma_i$, hence
\[
{}^\chi \Xzerokminusthree  \lesssim   \hat\varepsilon_0 \alpha^2 L^{-2+\delta}.
\]
This yields the fifth inequality of~\eqref{kiallobasicnewformiiagain}.
Note that this implies 
\begin{equation}
\label{solow}
\Xzerolesslessk \lesssim \hat \varepsilon_0 \alpha^2L^{-2+\delta}\lesssim \varepsilon_0  L^{-2+\delta} .
\end{equation}
Note also that third inequality of~\eqref{twoandthenathird}  implies
\begin{equation}
\label{solowone}
\Xonelesslessk \lesssim \hat \varepsilon_0 \alpha L^{-1+\delta} \lesssim \varepsilon_0 L^{-1+\delta}.
\end{equation}
Finally, plugging these bounds into the fourth inequality of~\eqref{twoandthenathird} yields the third inequality of~\eqref{kiallobasicnewformiiagain},
completing the proof of~\eqref{kiallobasicnewformiiagain}.

To obtain~\eqref{basicnewformiiagainkaivourgio}, we recall~\eqref{pminusoneweightedii}  from Proposition~\ref{iihierprop}.
This gives, for sufficiently small $\hat\varepsilon_{\rm slab}$,
\begin{eqnarray*}
\sup_{\tau_0\le \tau \le \tau_0+L} \Efancyonek(\tau) &\leq& \Efancyonek(\tau_0) 
+ C\, {}^\chi \Xonek
\sqrt{\Xonelesslessk} 
 + C\, {}^\chi\Xzerok\sqrt{\Xzerolesslessk}\sqrt{L}, \\
 &\leq& \hat\varepsilon_0 + C\hat\varepsilon_0 \varepsilon_0^{1/2} L^{(-1+\delta)/2} +C\hat\varepsilon_0 \varepsilon_0^{1/2} L^{(-2+\delta)/2}L^{1/2}\\
 &\leq& \hat\varepsilon_0(1+\hat\varepsilon_0^{1/4} L^{(-1+\delta)/2})
\end{eqnarray*}
where we have used the first inequality of~\eqref{additassum}, the first inequality of~\eqref{twoandthenathird},  
the estimate~\eqref{solowone} and
the estimate~\eqref{solow} and we have  replaced
$\varepsilon_0^{1/2}$ by $\hat\varepsilon_0^{1/4}$ 
in the penultimate line, sacrificing a quarter power to absorb the resulting $\alpha$ term and the constant $C$.

To obtain~\eqref{basicnewformiiagainallokaivourgio}, we now recall~\eqref{pweightedii} from Proposition~\ref{iihierprop}. 
This gives, for sufficiently small $\hat\varepsilon_{\rm slab}$,  
\begin{eqnarray*}
\sup_{\tau_0\le \tau \le \tau_0+L} 
\Efancytwominusdelkminusone(\tau) 
&\leq&  \Efancytwominusdelkminusone(\tau_0)  +C\, \left(  {}^\chi\Xtwominusdelkminusone\sqrt{\Xzerolesslessk} + 
\sqrt{{}^\chi\Xtwominusdelkminusone}\sqrt{ {}^\chi\Xzerokminusone}
\sqrt{\Xtwominusdeltalesslessk} \right) +
C\, {}^\chi\Xzerokminusone\sqrt{\Xzerolesslessk}\sqrt{L} \\
&\leq&  \hat\varepsilon_0 +C   \hat\varepsilon_0 \varepsilon_0^{1/2} L^{(-2+\delta)/2}
+C\hat\varepsilon_0\alpha^{1/2} \varepsilon_0^{1/2}  L^{-1/2}
+C\hat\varepsilon_0 \alpha \varepsilon_0^{1/2}  L^{(-2+\delta)/2}  \\
&\leq&  \hat\varepsilon_0(1 + \hat\varepsilon_0^{1/4} L^{-1/2}) ,
\end{eqnarray*}
where we have used the second inequality of~\eqref{additassum},
the second and third inequalities of~\eqref{kiallobasicnewformiiagain},
 the estimate~\eqref{solow} and the first inequality of~\eqref{fromlatterweobt}, and we have  replaced
$\varepsilon_0^{1/2}$ by $\hat\varepsilon_0^{1/4}$ 
in the penultimate line, sacrificing a quarter power to absorb the resulting $\alpha$ term and the constant $C$.

To show~\eqref{betterjust}--\eqref{bettertwo}, we first apply the pigeonhole principle as in~\cite{DafRodnew}
to 
the inequality
\[
\int_{\tau_0}^{\tau_0+L} \left( \,
\Ezerokminusone' (\tau')+ 
\Eoneminusdelkminustwo' (\tau')  
+ \alpha^{-1} L^{1-\delta}  \Ezerokminusthree'(\tau')\right) d\tau' \lesssim \hat \varepsilon_0 ,
\]
which, upon addition, follows from  the first, second and fourth 
inequalities of the estimate~\eqref{kiallobasicnewformiiagain} already shown. 
Recalling from~\eqref{higherorderfluxbulkrelation} that for both $p=2-\delta$ and $p=1$, we have
\[
\Epminusonekminustwo' \gtrsim \Epminusonekminustwo, \qquad \Epminusonekminusthree' \gtrsim \Epminusonekminusthree,
\]
we obtain 
that there exists $\tau'' \in [\tau_0,\tau_0+L/2]$, whose precise value depends on the solution, such that
\begin{eqnarray}
\label{withimplicitconstonenewii}
\Ezerokminusone(\tau'') &\lesssim& \hat{\varepsilon}_0 \cdot L^{-1},\\
\label{withimplicitconstone} 
\Eoneminusdelkminustwo(\tau'') &\lesssim& \hat{\varepsilon}_0 \cdot L^{-1},\\
\label{withimplicitconstwo}
\Ezerokminusthree(\tau'')		&\lesssim& 
\hat{\varepsilon}_0 \alpha L^{-1+\delta} \cdot L^{-1}.
\end{eqnarray}
Now in view of the interpolation estimate~\eqref{interpolationstatementnewback} of Proposition~\ref{interpolationprop} we have
\[
\Eonekminustwo(\tau'')\lesssim  \left( \,\, \Eoneminusdelkminustwo(\tau'') \right)^{1-\delta} 
 \left( \,\, \Etwominusdelkminustwo(\tau'')\right)^{\delta}  
 \leq  \left( \,\, \Eoneminusdelkminustwo(\tau'') \right)^{1-\delta} 
 \left( \sup_{\tau\in [\tau_0,\tau_0+L]} \Etwominusdelkminustwo(\tau) \right)^{\delta}  
\]
 and thus
 \begin{equation}
 \label{afterinterpolationhere}
\Eonekminustwo(\tau'')  \lesssim \hat{\varepsilon}_0 L^{-1+\delta}
\end{equation}
where we have used~\eqref{withimplicitconstone} and the estimate
for $\Etwominusdelkminustwo$ contained in~\eqref{kiallobasicnewformiiagain}.

Now
we apply~\eqref{pminusoneweightedii} 
and~\eqref{nonpweightedii}  again, with $\tau''$  in place of $\tau_0$, 
using~\eqref{withimplicitconstonenewii},~\eqref{afterinterpolationhere} and~\eqref{withimplicitconstwo} to bound the initial data,
to obtain that 
for all $\tau_0+L\ge \tau'\ge \tau_0+L/2$, we have
\begin{eqnarray}
\label{withimplicitconstonelaternewforii}
\Efancyzerokminusone(\tau')	&\sim& \Ezerokminusone(\tau') \lesssim \hat{\varepsilon}_0  L^{-1},\\
\label{withimplicitconstonelater} 
\Efancyonekminustwo(\tau')	&\sim& \Eonekminustwo(\tau') \lesssim \hat{\varepsilon}_0  L^{-1+\delta},\\
\label{withimplicitconsttwolater}
\Efancyzerokminusthree(\tau') 	&\sim& \Ezerokminusthree(\tau')\lesssim 
\hat{\varepsilon}_0 \alpha L^{-2+\delta}.
\end{eqnarray}

Thus, in view of the requirement $\alpha \ge \alpha_0 \gg 1$, for sufficiently large $\alpha_0$ we may 
absorb the constants implicit in $\lesssim$ by explicit constants of our choice by adding extra positive $\alpha$
powers to the right hand side of~\eqref{withimplicitconstonelater} and~\eqref{withimplicitconsttwolater}.
In this way,  we obtain the specific 
estimates~\eqref{betterone} and~\eqref{bettertwo}  for all $\tau'\ge\tau''$ and thus in  particular for
all $L\ge \tau'-\tau_0\ge L/2$. In the same way, we also obtain the specific constant of the  estimate
of~\eqref{addedherenotfancy} which will be convenient in our scheme. 
\end{proof}

Let us state an alternative ``relaxed'' version of the above proposition.

\begin{proposition}
\label{bettertimepleaseweaker}
Under the relaxed assumptions for Proposition~\ref{globproptwo} 
described in Remark~\ref{relaxedremark},
there exists a constant $C>0$, implicit in the inequalities $\lesssim$ below, a parameter $\beta>0$ sufficiently small, a  parameter $\alpha_0\gg 1$
and, for all $\alpha\ge \alpha_0$, a parameter~$\hat{\varepsilon}_{\rm slab}(\alpha)$
such that the following  holds.

Given $0<\hat\varepsilon_0 \leq \hat{\varepsilon}_{\rm slab}$,
let us assume in addition  that
\begin{equation}
\label{additassumalt}
\Eonek(\tau_0) \leq \hat\varepsilon_0 L^{\beta}, \quad
 \Etwominusdelkminusone(\tau_0) \leq \hat\varepsilon_0 L^{\beta},  \quad
 \Ezerokminusone(\tau_0) \leq \hat\varepsilon_0 \alpha^{1/2} L^{-1+\beta},
 \quad
   \Eonekminustwo(\tau_0) \leq  \hat\varepsilon_0\alpha L^{-1+\delta+\beta},
\quad \Ezerokminusthree(\tau_0) \leq \hat\varepsilon_0\alpha^2 L^{-2+\delta+\beta}.
\end{equation}
Then the solution $\psi$ of~\eqref{theequation}
on $\mathcal{R}(\tau_0,\tau_0+L)$ given by Proposition~\ref{globproptwo} 
satisfies the additional estimates
\begin{equation}
\label{addedherenotfancyalt}
\sup_{\tau\le \tau\le \tau_0+L} \Eonek(\tau) \leq  \alpha^{\beta}  \hat\varepsilon_0L^\beta ,  \qquad
\sup_{\tau_0\le \tau \le \tau_0+L} 
\Etwominusdelkminusone(\tau) \leq \alpha^{\beta} \hat\varepsilon_0 L^\beta   ,
\end{equation}
\begin{equation}
\label{kiallobasicnewformiiagainalt}
{}^\chi \Xonek \lesssim\hat \varepsilon_0 L^\beta, \quad {}^\chi \Xtwominusdelkminusone\lesssim\hat\varepsilon_0 L^\beta,\quad
{}^\chi \Xzerokminusone \lesssim \hat\varepsilon_0\alpha L^{-1+\beta}, \quad
{}^\chi \Xonekminustwo \lesssim\hat\varepsilon_0\alpha L^{-1+\delta+\beta}, \quad 
 {}^\chi \Xzerokminusthree
\lesssim\hat\varepsilon_0\alpha^2 L^{-2+\delta+\beta}.
\end{equation}

 Moreover, 
for all times $\tau'$ with $L\ge \tau'-\tau_0 \ge L/2$, we have that
\begin{equation}
\label{betterjustalt}
\Ezerokminusone(\tau') \leq  \hat\varepsilon_0 \alpha^{2\beta} L^{-1+\beta },\qquad 
\Eonekminusone(\tau') \leq   \hat\varepsilon_0 \alpha^{2\beta} L^{-1+\delta+\beta},\qquad
\Ezerokminustwo(\tau') \leq  \hat\varepsilon_0  \alpha^{1+2\beta} L^{- 2+\delta+\beta}.
\end{equation}
\end{proposition}

\begin{proof}
The proof is similar to that of Proposition~\ref{bettertimeplease} and is left to the reader.
\end{proof}

\begin{remark}
\label{holdsinSemilinear}
The above proposition  can now immediately be seen to also hold in
 the semilinear case, with the relaxed assumptions described in Remark~\ref{semiremsemi}.
\end{remark}

\subsubsection{The iteration: proof of Theorem~\ref{themaintheorem} in case (ii)}
\label{theiterationfortwo}
We may now prove Theorem~\ref{themaintheorem} in case (ii).

We shall proceed iteratively.

We define $\tau_0=1$, $L_0=1$, $L_i = 2^i$, $\tau_{i+1}=\tau_i+L_i$,
and fix $\alpha\ge \alpha_0$  so that the statement of Proposition~\ref{bettertimeplease} holds.
(Note that we shall no longer track the dependence of constants and parameters on $\alpha$ since it is now fixed; one could simply take $\alpha=\alpha_0$.)

For $0<\varepsilon_0 \le \varepsilon_{\rm global}$ and $\varepsilon_{\rm global}$ sufficiently small,
since 
\[
\Eonek(\tau_0)+ \Etwominusdelkminusone(\tau_0) \leq\varepsilon_0,
\]
then in view of Corollary~\ref{nomoreroman}, we have 
\[
\Efancyonek(\tau_0)\lesssim\varepsilon_0, \qquad
\Efancytwominusdelkminusone(\tau_0) \lesssim \varepsilon_0.
\]
Thus, for sufficiently small $\varepsilon_{\rm global}\ll \hat\varepsilon_{\rm slab}$, 
it
follows that there exists a $\hat\varepsilon_0(\varepsilon_0) \sim \varepsilon_0$, 
satisfying moreover  $2\alpha \hat\varepsilon_0 \leq \varepsilon_{\rm slab}$,
such that 
\[
\hat\varepsilon_0\leq \frac12\hat\varepsilon_{\rm slab}
\]
and
\[
\Efancyonek(\tau_0)\leq\hat\varepsilon_0,\qquad 
\Efancytwominusdelkminusone(\tau_0)\leq\hat\varepsilon_0, \qquad
\Efancyzerokminusone(\tau_0)\leq \hat\varepsilon_0 \alpha L^{-1},
 \qquad \Efancyonekminustwo(\tau_0)\leq\hat\varepsilon_0\alpha L^{-1+\delta},
\qquad \Efancyzerokminusthree(\tau_0)\leq \hat\varepsilon_0\alpha^2 L^{-2+\delta}.
\]
Finally, with our restriction on the definition of $\varepsilon_0$,  we may also write
\[
\Eonek(\tau_0)\leq \alpha\hat\varepsilon_0.
\]

In general, given
$\tau_i\ge 1$ defined above, $\hat{\varepsilon}_i\leq \hat{\varepsilon}_{\rm slab}$, and a solution $\psi$
of~\eqref{theequation} on $\mathcal{R}(\tau_0,\tau_i)$
 such that
\begin{equation}
\label{notsofancy}
\Eonek(\tau_i) \leq \alpha \hat\varepsilon_i
\end{equation}
and
\begin{equation}
\label{thosefancyones}
\Efancyonek(\tau_i) \leq \hat\varepsilon_i, \qquad
\Efancytwominusdelkminusone(\tau_i)\leq\hat\varepsilon_i, 
\qquad \Efancyzerokminusone(\tau_i)\leq\hat\varepsilon_i \alpha L^{-1},
\qquad \Efancyonekminusone(\tau_i)\leq\hat\varepsilon_i \alpha 
L_i^{-1+\delta},
\qquad \Efancyzerokminustwo(\tau_i)\leq \hat\varepsilon_i\alpha^2 L_i^{-2+\delta}   ,
\end{equation}
note that by our restrictions on $\hat\varepsilon_{\rm slab}$, the assumptions of Proposition~\ref{globproptwo}
hold with $\alpha^2 \hat\varepsilon_i$ in  place of $\varepsilon_0$ (here we are using~\eqref{notsofancy} to invoke Corollary~\ref{nomoreroman}
to rewrite~\eqref{thosefancyones} in terms of the calligraphic energies; cf.~the
first lines of the proof of Proposition~\ref{bettertimeplease})  and the assumptions of 
Proposition~\ref{bettertimeplease} then apply with $\hat\varepsilon_i$ in  place of $\hat\varepsilon_0$, 
where both propositions are understood now with
 $\tau_i$, $\tau_{i+1}$ in place of $\tau_0$, $\tau_1$.
 It follows that the solution $\psi$ extends to a solution defined also in
$\mathcal{R}_i:=\mathcal{R}(\tau_i,\tau_i+L_i)$ satisfying the estimates
\eqref{addedherenotfancy}--\eqref{kiallobasicnewformiiagain} while for $\tau'=\tau_{i+1}=\tau_i+L_i$ 
we  have in addition~\eqref{betterone}--\eqref{bettertwo}.
We have thus
\begin{align}
\nonumber
\Eonek(\tau_{i+1})&\leq \alpha \hat\varepsilon_i&\leq&\alpha\hat\varepsilon_{i+1} ,\\
\nonumber
\Efancyonek(\tau_{i+1})&\leq \hat\varepsilon_i(1+\hat{\varepsilon}_i^{1/4}L_i^{-(1+\delta)/2})&\leq&\hat\varepsilon_{i+1}  \\
\nonumber
\Efancytwominusdelkminusone(\tau_{i+1})&\leq \hat\varepsilon_i(1+\hat{\varepsilon}_i^{1/4}L_i^{-1/2})&\leq&\hat\varepsilon_{i+1}   , \\
\nonumber
\Efancyonekminusone(\tau_{i+1})&\leq\frac12 \hat\varepsilon_i \alpha L_i^{-1}&\leq &\hat\varepsilon_{i+1}\alpha L_{i+1}^{-1}  , \\
\nonumber
\Efancyonekminustwo(\tau_{i+1})&\leq\frac12 \hat\varepsilon_i \alpha L_i^{-1+\delta}&\leq &\hat\varepsilon_{i+1}\alpha L_{i+1}^{-1+\delta}  , \\
\nonumber
\Efancyzerokminusthree(\tau_{i+1})&\leq \frac14 \hat\varepsilon_i\alpha^2 L_i^{-2+\delta}&\leq & \hat\varepsilon_{i+1} \alpha^2 L_{i+1}^{-2+\delta} ,
\end{align}
as long as 
\begin{equation}
\label{inductdef}
\hat\varepsilon_{i+1}:=\hat\varepsilon_i(1+\hat\varepsilon_i^{1/4} L_i^{-(1+\delta)/2}).
\end{equation}

Now note that for $\varepsilon_{\rm global}$ sufficiently small, the above inductive definition~\eqref{inductdef} 
of $\hat\varepsilon_i$
ensures that $\hat\varepsilon_i\le 2\hat\varepsilon_0$ for all $i$.

It follows that a solution exists in $\mathcal{R}(\tau_0,\infty)=\cup \mathcal{R}(\tau_i,\tau_i+L_i)$
and in each interval the estimates \eqref{kiallobasicnewformiiagain} hold.

We obtain finally that for all $\tau\ge 1$ we have
\begin{equation}
\label{finalboundedness}
\Efancyonek(\tau)\leq \hat\varepsilon_0+C\hat\varepsilon_0^{3/2}, \qquad
\Efancytwominusdelkminusone(\tau)\leq \hat\varepsilon_0+C\hat\varepsilon_0^{3/2}, 
\end{equation}
\begin{equation}
\label{mustimprovethis}
\int_{\tau_0}^\tau {}^\chi\Eoneminusdelkminusone'+ \Eoneminusdelkminustwo'  \lesssim \hat\varepsilon_0\log(\tau+1),
\end{equation}
\begin{equation}
\label{resumefromhere}
\Eonekminustwo(\tau) \lesssim \hat\varepsilon_0\tau^{-1+\delta}, \qquad \Fonekminustwo(v, \tau)\lesssim \hat\varepsilon_0\tau^{-1+\delta} ,
\end{equation}
\begin{equation}
\label{resumefromheretwo}
\int_\tau^\infty {}^\chi\Ezerokminustwo'+ \Ezerokminusthree'  \lesssim \hat\varepsilon_0\tau^{-1+\delta} ,
\end{equation}
\begin{equation}
\label{resumefromherethree}
\Ezerokminusthree(\tau) \lesssim \hat \varepsilon_0\tau^{-2+\delta}, 
\qquad \Fzerokminusthree(v,\tau)\lesssim \hat\varepsilon_0\tau^{-2+\delta} ,
\end{equation}
\begin{equation}
\label{resumefromherefour}
\int_\tau^\infty{}^\chi\Ezerominusoneminusdeltakminusthree'+ 
\Ezerominusoneminusdeltakminusfour'  \lesssim \hat\varepsilon_0\tau^{-2+\delta} .
\end{equation}
Here, by $\mathcal{F}(v,\tau)$ we denote the restriction of the flux on $\underline{C}_v$ to $J^+(\Sigma(\tau))$.

Let us note that we may improve, a posteriori, the inequality~\eqref{mustimprovethis} to 
\begin{equation}
\label{weimprovehere}
\int_{\tau_0}^\infty{}^\chi\Eoneminusdelkminusone'+ \Eoneminusdelkminustwo'  \lesssim \hat\varepsilon_0.
\end{equation}
To show~\eqref{weimprovehere}, for all $\tau\ge \tau_0$,
we reapply the estimates of Proposition~\ref{finalpropest} globally on $\mathcal{R}(\tau_0,\tau)$. To control the
nonlinear error bulk integrals, redecompose the domain in dyadic time slabs and apply the estimates
proven and sum. In view of the estimates,
all these error bulk integrals can be controlled by $C\hat\varepsilon_0^{3/2}$.

\subsubsection{Alternative proof using Proposition~\ref{bettertimepleaseweaker} and the semilinear case}
\label{thereisnoalternative}

We note that we may prove the theorem alternatively using Proposition~\ref{bettertimepleaseweaker}.
We first fix $\alpha\ge \alpha_0$, and now define $L_i: = \alpha^{i}$.
Here our iterative assumption is 
\[
\Eonek(\tau_0) \leq \hat\varepsilon_0 L_i^{\beta}, \quad
 \Etwominusdelkminusone(\tau_0) \leq \hat\varepsilon_0 L_i^{\beta},  \quad
 \Ezerokminusone(\tau_0) \leq \hat\varepsilon_0 \alpha L_i^{-1+\beta},
 \quad
   \Eonekminustwo(\tau_0) \leq  \hat\varepsilon_0\alpha L_i^{-1+\delta+\beta},
\quad \Ezerokminusthree(\tau_0) \leq \hat\varepsilon_0\alpha^2 L_i^{-2+\delta+\beta}.
\]
In view of~\eqref{addedherenotfancyalt} and~\eqref{betterjustalt} and the definition of $L_i$, we see that this is closed under iteration. 
We obtain thus existence in $\mathcal{R}(\tau_0,\infty)$, and at first instance the bounds
\[
\Eonek(\tau)\lesssim  \hat\varepsilon_0 \tau^\beta, \qquad \Etwominusdelkminusone(\tau) \lesssim \hat\varepsilon_0 \tau^\beta ,
\]
\[
 \Ezerokminusone(\tau) \lesssim \hat\varepsilon_0 \tau^{-1+\beta},
 \qquad
   \Eonekminustwo(\tau) \lesssim  \hat\varepsilon_0 \tau^{-1+\delta+\beta},
\qquad \Ezerokminusthree(\tau_0) \lesssim \hat\varepsilon_0\tau^{-2+\delta+\beta}.
\]

These bounds are of course weaker than those of~\eqref{finalboundedness}--\eqref{resumefromherefour}.
The $\tau^\beta$ factors terms may be, however, removed a posteriori, as in the last paragraph of Section~\ref{theiterationfortwo},
by  revisiting the estimates 
globally on $\mathcal{R}(\tau_0,\tau)$, controlling the nonlinear error bulk integrals by redecomposing into dyadic time slabs,
applying the estimates already proven, and summing.

Thus this proof in the end yields finally the  same estimates as before, but where~\eqref{finalboundedness} is replaced by
\begin{equation}
\label{orbitalfinal}
\Ezerok \lesssim \hat\varepsilon_0,  \qquad   \Eonekminustwo(\tau) \lesssim \hat\varepsilon_0.
\end{equation}

\begin{remark}
\label{procremark}
In view of Remark~\ref{holdsinSemilinear}, this proof in particular 
applies to the semilinear case with the relaxed assumptions of Remark~\ref{semiremsemi}, 
giving the result for case (ii) as stated in Remark~\ref{bigremarkonsemi}.
\end{remark}

\begin{remark}
The disadvantage of this alternative proof over the proof given in Section~\ref{theiterationfortwo}  is that to obtain  the fundamental
top-order orbital stability statement~\eqref{orbitalfinal},
one must revisit the estimates globally. We thus believe that the proof in Section~\ref{theiterationfortwo}  better reflects the dyadically localised
philosophy of our method. 
\end{remark}

\subsection{Case (iii)}
\label{andnotitsiii}
We now turn to the most general case we consider, case (iii), where 
we do not assume that~\eqref{inhomogeneous} follows from a physical space identity~\eqref{energyidentity}, but only
assume~\eqref{inhomogeneous} as a black box, together with the weaker physical space identity described
in Section~\ref{caseiiiidentity}.

This case is slightly more involved because we must combine the estimate originating from
 Section~\ref{caseiiiidentity} 
with the estimate originating in Section~\ref{basicblackboxestimatesec}.
This alters a bit the numerology of the number of derivatives we must take, but
the basic  scheme is the same as case (ii).

\subsubsection{The hierarchy of inequalities}
\label{hierarchiiisec}

\begin{proposition}
\label{withthehierarchy}
Let $k$ be sufficiently large and let us assume the case (iii) assumptions. 
There exist constants $C>0$, $c>0$ and an $\varepsilon_{\rm boot}>0$ sufficiently small 
such that the following is true.

Consider a region $\mathcal{R}(\tau_0,\tau_1)$
and a $\psi$ solving~\eqref{theequation} on  $\mathcal{R}(\tau_0,\tau_1)$, satisfying 
moreover~\eqref{basicbootstrap} for $p=0$ and $0< \varepsilon\le \varepsilon_{\rm boot}$.
Let us assume moreover that 
\[
\tau_1\le \tau_0+L
\]
for some arbitrary $L> 0$.
We have the following hierarchy of inequalities:
\begin{eqnarray}
 \label{pweightedtopiii}
\Ffancytwominusdelk(v,\tau_1), \quad\Efancytwominusdelk(\tau_1)  , \quad  c\,{}^\rho\Xtwominusdelk
&\leq& \Efancytwominusdelk(\tau_0)  +   AC\,  {}^\chi\Xzerokminusone  +   C\, \left(  {}^\rho\Xtwominusdelk\sqrt{\Xzerolesslessk} + \sqrt{{}^\rho\Xtwominusdelk}\sqrt{ {}^\rho\Xzerok}
\sqrt{\Xtwominusdeltalesslessk} \right) +
C\, {}^\rho\Xzerok\sqrt{\Xzerolesslessk}\sqrt{L} \, , 
  \\
  \label{pweightednottopiii}
 {}^\chi\Xtwominusdelkminusone
&\lesssim& \Etwominusdelkminusone(\tau_0)  +    {}^\rho\Xtwominusdelkminusone\sqrt{\Xzerolesslessk} + \sqrt{{}^\rho\Xtwominusdelkminusone}\sqrt{ {}^\rho\Xzerokminusone}
\sqrt{\Xtwominusdeltalesslessk}+
 {}^\rho\Xzerok\sqrt{\Xzerolesslessk}\sqrt{L} \, , 
  \\
\label{pminusoneweightedtopiii}
\Ffancyonek(v,\tau_1) , \quad \Efancyonek(\tau_1) ,\quad c\, {}^\rho \Xonek&\leq& \Efancyonek(\tau_0) 
  +  AC\,  {}^\chi\Xzerokminusone        +C\, {}^\rho \Xonek
\sqrt{\Xonelesslessk} 
 + C\, {}^\rho\Xzerok\sqrt{\Xzerolesslessk}\sqrt{L}   \,  ,  \\
 \label{pminusoneweightednottopiii}
{}^\chi \Xonekminusone&\lesssim& \Eonekminusone(\tau_0) 
  +   C\, {}^\rho \Xonekminusone
\sqrt{\Xzerolesslessk} 
 + C\, {}^\rho\Xzerok\sqrt{\Xzerolesslessk}\sqrt{L}   \,    ,  \\
\nonumber
\Ffancyzerok(v,\tau_1),\quad \Efancyzerok(\tau_1) , \quad c\,{}^\rho \Xzerok &\leq& 
\Efancyzerok(\tau_0)
 +AC\,  {}^\chi\Xzerokminusone 
+ C\left( \,{}^\rho \Xzerok + (\,{}^\rho \Xzerok)^{\frac{1-\delta}{1+\delta}} (\, {}^\rho \Xonek)^{\frac{2\delta}{1+\delta}} \right)
\sqrt{\Xzerolesslessk+ (\,\Xzerolesslessk)^{\frac{1-\delta}{1+\delta}} (\, \Xonelesslessk)^{\frac{2\delta}{1+\delta}}}\\
\label{nonpweightedtopiii}
&&\qquad+ C\, {}^\rho\Xzerok\sqrt{\Xzerolesslessk}\sqrt{L}   \,   , 
 \\
\nonumber
{}^\chi \Xzerokminusone &\lesssim& 
\Ezerokminusone(\tau_0)
+  \left( \,{}^\rho \Xzerokminusone + (\,{}^\rho \Xzerokminusone)^{\frac{1-\delta}{1+\delta}} (\, {}^\rho \Xonekminusone)^{\frac{2\delta}{1+\delta}} \right)
\sqrt{\Xzerolesslessk+ (\,\Xzerolesslessk)^{\frac{1-\delta}{1+\delta}} (\, \Xonelesslessk)^{\frac{2\delta}{1+\delta}}}\\
\label{nonpweightednottopiii}
&&\qquad
+  {}^\rho\Xzerok\sqrt{\Xzerolesslessk}\sqrt{L}.
\end{eqnarray}
\end{proposition}

\begin{proof}
Again we recall that if $\varepsilon_{\rm boot}\le \varepsilon_{\rm loc}$, then the assumption of Proposition~\ref{finalpropest} holds. The proposition then  follows  from Proposition~\ref{finalpropest}
as in the proof of Proposition~\ref{iihierprop}, where we have used also the crucial estimate
\[
A\int_{\tau_0}^{\tau_1}\Eerrorkminusone'(\tau')d\tau'  \lesssim 
AC\,  {}^\chi\Xzerokminusone ,
\]
which follows from 
the fundamental relation~\eqref{glue}.
\end{proof}

\subsubsection{Global existence in $L$-slabs}

To show global existence in 
$L$-slabs, we require a minor modification
of the assumptions of Proposition~\ref{globproptwo}.

\begin{proposition}
\label{globpropthree}
Let $k-2\ge k_{\rm loc}$ be sufficiently large and let us assume the case (iii) assumptions.  Then  there exists a positive constant $\varepsilon_{\rm slab}\leq\varepsilon_{\rm loc}$ and a constant $C>0$ implicit in the sign $\lesssim$ below such that the following  is true.

Given arbitrary $L\ge 1$,
 $\tau_0\ge 0$,  $0<\varepsilon_0\leq \varepsilon_{\rm slab}$ and initial data $(\uppsi,\uppsi')$ 
 on $\Sigma(\tau_0)$ as in Proposition~\ref{localexistence}, satisfying moreover
\begin{equation}
\label{assumptionondataherepartiii}
\Eonek(\tau_0) \leq \varepsilon_0, 
\qquad  \Ezerokminustwo(\tau_0) \leq \varepsilon_0L^{-1}, 
\end{equation}
then the unique solution of Proposition~\ref{localexistence} obtaining the data can be extended to 
a $\psi$
defined on the entire  spacetime slab $\mathcal{R}(\tau_0,\tau_0+L)$
satisfying the equation~\eqref{theequation}
and
 the estimates
 \begin{equation}
 \label{newformatestimatesiii}
 {}^\rho  \Xonek + {}^\chi  \Xonekminusone  \lesssim \varepsilon_0,
\qquad {}^\rho \Xzerokminustwo  +{}^\chi \Xzerokminusthree \lesssim \varepsilon_0L^{-1}.
 \end{equation}
\end{proposition}

\begin{proof}
Consider the set $\mathfrak{B}\subset (\tau_0,\tau_0+L]$ 
consisting of all $\tau_0+L\ge \tau_f\ge \tau_0$ such that a solution $\psi$ of~\eqref{theequation}
obtaining the data exists on $\mathcal{R}(\tau_0,\tau_f)$ 
and such that
the boostrap assumption~\eqref{basicbootstrap} (with $p=1$)
and also the additional bootstrap assumption
\begin{equation}
\label{newversionadditionalbootthree}
\Xzerolesslessk \leq  \varepsilon L^{-1}
\end{equation}
 hold in $\mathcal{R}(\tau_0,\tau_1:=\tau_f)$, 
where
$0<\varepsilon\le \varepsilon_{\rm boot}$ 
 is a small constant satisfying
\begin{equation}
\label{providedprovided}
1\gg \varepsilon\gg \varepsilon_{\rm slab}.
\end{equation}
(The above relation in particular already constrains $\varepsilon_{\rm slab}$ to be sufficiently small.)

By the local well posedness statement Proposition~\ref{localexistence},
it follows
that since $k-2\ge k_{\rm loc}$ and $\varepsilon_0\leq \varepsilon_{\rm slab}\leq  \varepsilon_{\rm loc}$,
we have that $\tau_0+\tau_{\rm exist}\subset \mathfrak{B}$ and thus $\mathfrak{B}\ne\emptyset$, provided that
$\varepsilon$ satisfies~\eqref{providedprovided}.
Also note that a fortiori, if $\tau_f\in \mathfrak{B}$, then $(\tau_0,\tau_f]\in \mathfrak{B}$
and thus $\mathfrak{B}$  is manifestly a connected  subset of $(\tau_0,\infty)$.

Proposition~\ref{withthehierarchy} holds for $\mathcal{R}(\tau_0,\tau_1)$ with $\tau_1:=\tau_f$ for any $\tau_f\in \mathfrak{B}$.
Adding the equations~\eqref{pminusoneweightedtopiii}--\eqref{nonpweightednottopiii} pairwise with a suitable constant
so as to absorb the term multiplying~$A$, with $k-2$ replacing $k$ for the latter pair,
 we obtain the
system:
\begin{eqnarray}
\label{pminusoneweightedaddediii}
 {}^\rho \Xonek + {}^\chi \Xonekminusone &\lesssim& \Efancyonek(\tau_0) 
  +   {}^\rho \Xonek
\sqrt{\Xonelesslessk} 
 +  {}^\rho\Xzerok\sqrt{\Xzerolesslessk}\sqrt{L},    \\
\label{nonpweightedaddediii}
{}^\rho \Xzerokminustwo + {}^\chi \Xzerokminusthree &\lesssim& 
\Efancyzerokminustwo(\tau_0)
+\left( \,{}^\rho \Xzerokminustwo + (\,{}^\rho \Xzerokminustwo)^{\frac{1-\delta}{1+\delta}} 
(\, {}^\rho \Xonekminustwo)^{\frac{2\delta}{1+\delta}} \right)
\sqrt{\Xzerolesslessk+ (\,\Xzerolesslessk)^{\frac{1-\delta}{1+\delta}} (\, \Xonelesslessk)^{\frac{2\delta}{1+\delta}}}
+ {}^\rho\Xzerokminustwo\sqrt{\Xzerolesslessk}\sqrt{L}.
\end{eqnarray}

We now obtain that
\begin{equation}
\label{whatweveobtained}
{}^\rho \Xonek + {}^\chi \Xonekminusone \lesssim \varepsilon_0, \qquad
{}^\rho \Xzerokminustwo + {}^\chi \Xzerokminusthree \lesssim \varepsilon_0L^{-1}.
\end{equation}

We may now already apply our continuation criterion Corollary~\ref{continuash}, applied with $p=0$ and $k-2$,
to assert the existence of an $\epsilon$, independent of $\tau_1$, such that now defining $\tau_1:= \min \{ \tau_f+\epsilon, \tau_0+ L\}$,
the solution $\psi$ extends to a smooth solution of~\eqref{theequation}
on $\mathcal{R}(\tau_0,\tau_1)$, 
and moreover, from~\eqref{extensionhere}, that
\begin{equation}
\label{additionalbootherehereiiiwithC}
{}^\chi \Xzerokminusthree+
{}^\rho \Xzerokminustwo
   \lesssim \varepsilon_0 L^{-1}
\end{equation}
holds on $\mathcal{R}(\tau_0,\tau_1)$
and, from~\eqref{extensionhigherhigher}, that
\begin{equation}
\label{alsoextendingthisiii}
{}^\chi \Xonek 
   \lesssim \varepsilon_0 
\end{equation}
holds on $\mathcal{R}(\tau_0,\tau_1)$.
It follows that~\eqref{newversionadditionalbootthree} and~\eqref{basicbootstrap} (with $p=1$) hold on $\mathcal{R}(\tau_0,\tau_1)$ and thus
 $\tau_1= \min \{ \tau_f+\epsilon, \tau_0+ L\} \in \mathfrak{B}$.

Since $\epsilon$ is independent of $\tau_f$, and in view also of the connectivity of $\mathfrak{B}$, 
it follows that $\mathfrak{B}$ is a nonempty open and closed subset of $(\tau_0, \tau_0+L]$
and thus $\mathfrak{B}=(\tau_0,\tau_0+L]$ and hence the solution exists in the entire spacetime slab
$\mathcal{R}(\tau_0,\tau_0+L)$.

The estimates~\eqref{whatweveobtained} thus hold in the entire spacetime slab. 
This gives~\eqref{newformatestimatesiii}.
\end{proof}

\subsubsection{The pigeonhole argument}
\label{pigeonsiii}

The above assumptions on initial data are sufficient for global existence in the slab, but are not sufficient
to iterate. For this we shall need  strengthened assumptions.

\begin{proposition}
\label{bettertimepleaseiii}
Under the assumptions of Proposition~\ref{globpropthree}, 
there exists a constant $C>0$, implicit in the inequalities $\lesssim$ below,
a  parameter $\alpha_0\gg 1$
and, for all $\alpha\ge \alpha_0$, a parameter~$\hat{\varepsilon}_{\rm slab}(\alpha)$
such that for all $0<\hat\varepsilon_0 \le\hat{\varepsilon}_{\rm slab}(\alpha)$ the following  holds.

Let us assume in addition to~\eqref{assumptionondataherepartiii} that we have
\begin{equation}
\label{additiii}
\Efancyonek(\tau_0) \leq \hat\varepsilon_0 ,\qquad
\Efancytwominusdelkminustwo(\tau_0)\leq  \hat\varepsilon_0, \qquad
\Efancyzerokminustwo(\tau_0)\leq \hat\varepsilon_0\alpha L^{-1}, \qquad
\Efancyonekminusfour(\tau_0) \leq  \hat\varepsilon_0\alpha L^{-1+\delta}, 
\qquad
\Efancyzerokminussix(\tau_0) \leq \hat\varepsilon_0\alpha^2 L^{-2+\delta }.
\end{equation}

Then the solution $\psi$ of~\eqref{theequation}
on $\mathcal{R}(\tau_0,\tau_0+L)$ given by Proposition~\ref{globproptwo} 
satisfies the additional estimates
\begin{equation}
\label{addedherenotfancyiii}
\sup_{\tau_0\le \tau\le \tau_0+L} \,  \Eonek(\tau)   \leq   \alpha \hat\varepsilon_0,
\end{equation}
\begin{equation}
\label{basicnewformiiiagain}
\sup_{\tau_0\le \tau\le \tau_0+L} \, \Efancyonek(\tau) \leq \hat\varepsilon_0(1+ \alpha L^{-1/4})  ,    
\end{equation}
\begin{equation}
\label{basicestimateoneiiiagainnottop}
\sup_{\tau_0\le \tau\le \tau_0+L}\,  \Efancytwominusdelkminustwo(\tau) \leq\hat\varepsilon_0(1+ \alpha L^{-1/4})  ,
\end{equation}
\begin{equation}
\label{kiallobasicnewformiiiagain}
{}^\rho \Xonek+ {}^\chi \Xonekminusone \lesssim \hat\varepsilon_0,
\qquad {}^\rho \Xtwominusdelkminustwo+{}^\chi \Xtwominusdelkminusthree \lesssim\hat\varepsilon_0, 
\qquad
{}^\rho \Xzerokminustwo+{}^\chi \Xzerokminusthree \lesssim\hat\varepsilon_0\alpha L^{-1}, 
\qquad {}^\rho \Xonekminusfour+ {}^\chi \Xonekminusfive \lesssim \hat\varepsilon_0 \alpha L^{-1+\delta},
\end{equation}
\begin{equation}
\label{kiallobasicnewformiiiagainlastone}
 {}^\rho \Xzerokminussix + {}^\chi \Xzerokminusseven
\lesssim\hat\varepsilon_0\alpha^2 L^{-2+\delta}.
\end{equation}
  
 Moreover, 
for all times $\tau'$ with $L\ge \tau'-\tau_0 \ge L/2$, we have that
\begin{equation}
\label{betterjustiii}
\Efancyzerokminustwo(\tau') \leq \frac12  \hat\varepsilon_0 \alpha L^{-1},
\end{equation}
\begin{equation}
\label{betteroneiii}
\Efancyonekminusfour(\tau') \leq \frac12  \hat\varepsilon_0 \alpha L^{-1+\delta},
\end{equation}
\begin{equation}
\label{bettertwoiii}
\Efancyzerokminussix(\tau') \leq\frac14  \hat\varepsilon_0  \alpha^2 L^{- 2+\delta}.
\end{equation}

\end{proposition}

\begin{proof}
The proof follows closely that of Proposition~\ref{bettertimeplease}.

Note again that by the statement of the proposition, we are in particular
also assuming a priori~\eqref{assumptionondataherepartiii} for \emph{some} $\varepsilon_0\leq \varepsilon_{\rm slab}$,
since this is included in the assumptions of Proposition~\ref{globpropthree}.
In view now of  Corollary~\ref{nomoreroman}, for $\alpha_0$ sufficiently large,
if say $\hat{\varepsilon}_{\rm slab}(\alpha)\ll \alpha^{-3} \varepsilon_{\rm slab}$,
it follows from the additional assumptions~\eqref{additiii} that 
the estimates~\eqref{assumptionondataherepartiii} 
in fact hold with the specific constant $\varepsilon_0:=\hat\varepsilon_0\alpha^3$ for all $\alpha\ge \alpha_0$.

We  now revisit~\eqref{pweightedtopiii}--\eqref{nonpweightednottopiii} and add again pairwise.
We now obtain that
\begin{equation}
\label{thefirstthree}
{}^\rho\Xonek + {}^\chi\Xonekminusone \lesssim  \hat\varepsilon_0 , \qquad
{}^\rho\Xtwominusdelkminustwo + {}^\chi\Xtwominusdelkminusthree\lesssim  \hat\varepsilon_0, \qquad
{}^\rho \Xonekminusfour + {}^\chi \Xonekminusfive \lesssim \hat\varepsilon_0\alpha L^{-1+\delta} , 
\end{equation}
\begin{equation}
\label{needstobeimprovedbefore}
{}^\rho\Xzerokminustwo + {}^\chi\Xzerokminusthree\lesssim  \hat\varepsilon_0\alpha L^{-1}+ 
 \left(\,{}^\rho \Xzerokminustwo)^{\frac{1-\delta}{1+\delta}} 
(\, {}^\rho \Xonekminustwo)^{\frac{2\delta}{1+\delta}} \right)
\sqrt{\Xzerolesslessk+ (\,\Xzerolesslessk)^{\frac{1-\delta}{1+\delta}} (\, \Xonelesslessk)^{\frac{2\delta}{1+\delta}}},
\end{equation}
\begin{equation}
\label{needstobeimproved}
{}^\rho \Xzerokminussix + {}^\chi \Xzerokminusseven \lesssim \hat\varepsilon_0\alpha^2 L^{-2+\delta} +\hat\varepsilon_0 \varepsilon_0^{1/2} 
L^{-\frac32\frac{(1-\delta)(\gamma+2\delta)}{1+\delta}}.
\end{equation}

The inequalities~\eqref{thefirstthree} give the first second and fourth inequality of~\eqref{kiallobasicnewformiiiagain}.

Note that $-2+\delta < -\frac32+\delta$.
We may however improve~\eqref{needstobeimproved} iteratively as follows:
If 
\[
{}^\rho \Xzerokminussix + {}^\chi \Xzerokminusseven  \lesssim  \hat\varepsilon_0 \alpha^2  L^{-\gamma},
\]
for $\gamma \le 2-\delta$,
then 
\[
\Xzerolesslessk \lesssim \hat \varepsilon_0 \alpha^2L^{-\gamma} \lesssim \varepsilon_0 L^{-\gamma},
\]
hence, plugging this  again into~\eqref{nonpweightedaddediii},
we obtain

\begin{equation}
\label{inductivehereiii}
{}^\rho \Xzerokminussix + {}^\chi \Xzerokminusseven \lesssim   \hat\varepsilon_0 \alpha^2 L^{-2+\delta} 
 + \hat\varepsilon_0 \varepsilon_0^{1/2} L^{-\frac32\frac{(1-\delta)(\gamma+2\delta)}{1+\delta}}.
\end{equation}
Setting $\gamma_0=\frac32\frac{(1-\delta)(1+2\delta)}{1+\delta}$, and defining inductively
$\gamma_i=\frac32\frac{(1-\delta)(\gamma_{i-1}+2\delta)}{1+\delta}$,
we have that there exists a first $i\ge 1$ such that $\frac32\frac{(1-\delta)(\gamma_{i}+2\delta)}{1+\delta} >2-\delta$.
It follows that~\eqref{inductivehereiii} holds for $\gamma=\gamma_i$, hence
\[
{}^\rho \Xzerokminussix + {}^\chi \Xzerokminusseven  \lesssim   \hat\varepsilon_0 \alpha^2 L^{-2+\delta}.
\]
This yields~\eqref{kiallobasicnewformiiiagainlastone}.
Note that this implies 
\begin{equation}
\label{solowiii}
\Xzerolesslessk \lesssim \hat \varepsilon_0 \alpha^2L^{(-2+2\delta)z}\lesssim \varepsilon_0  L^{-2+\delta} .
\end{equation}
On the other hand, the fourth inequality of~\eqref{kiallobasicnewformiiiagain} implies 
\begin{equation}
\label{solowiiinew}
\Xonelesslessk \lesssim \hat\varepsilon_0 \alpha L^{-1+\delta} \lesssim \varepsilon_0 L^{-1+\delta}.
\end{equation}
We may now infer the third inequality of~\eqref{kiallobasicnewformiiiagain} by plugging in the derived bounds into~\eqref{needstobeimprovedbefore}.
This concludes the proof of~\eqref{kiallobasicnewformiiiagain}.

To show~\eqref{betterjustiii}--\eqref{bettertwoiii}, respectively, we first apply the pigeonhole principle as in~\cite{DafRodnew}
to 
the inequality
\[
\int_{\tau_0}^{\tau_0+L} \left(\, \Ezerokminustwo' (\tau')
+ \Eoneminusdelkminusfour' (\tau')
+ \alpha^{-1} L^{1-\delta} \Ezerokminussix' (\tau')\right) d\tau' \lesssim \hat\varepsilon_0,
\]
which, upon addition, follows from the  inequalities of the estimate~\eqref{kiallobasicnewformiiiagain} already shown. 
Recalling from~\eqref{higherorderfluxbulkrelation} that  we have

\[
\Ezerokminustwo' \gtrsim \Ezerokminustwo, \qquad \Eoneminusdelkminusfour' \gtrsim \Eoneminusdelkminusfour, \qquad
\Ezerokminussix' \gtrsim \Ezerokminussix, 
\]
we obtain 
that there exists $\tau''  \in [\tau_0,\tau_0+L/2]$, whose precise value depends on the solution, such that
\begin{eqnarray}
\label{withimplicitconstoneiiinew}
\Ezerokminustwo(\tau'') &\lesssim& \hat{\varepsilon}_0 \cdot L^{-1},\\
\label{withimplicitconstoneiii}
\Eoneminusdelkminusfour(\tau'') &\lesssim& \hat{\varepsilon}_0  \cdot L^{-1},\\
\label{withimplicitconstwoiii}
\Ezerokminussix(\tau'')	&\lesssim&
\hat{\varepsilon}_0 \alpha L^{-1+\delta} \cdot L^{-1}.
\end{eqnarray}
Now in view of the interpolation estimate~\eqref{interpolationstatementnewback} of Proposition~\ref{interpolationprop}, 
we have
\begin{equation}
\label{reftosecondfactor}
\Eonekminusfour(\tau'')\lesssim  \left( \,\, \Eoneminusdelkminusfour(\tau'') \right)^{1-\delta} 
 \left( \,\, \Etwominusdelkminusfour(\tau'')\right)^{\delta}  
\leq \left( \,\, \Eoneminusdelkminusfour(\tau'') \right)^{1-\delta} 
 \left( \,\, \sup_{\tau_0\le \tau\le \tau_0+L} \Etwominusdelkminusfour(\tau) \right)^{\delta}  
\end{equation}
 and thus
 \begin{equation}
 \label{afterinterpolationhereiii}
\Eonekminusfour(\tau'')  \lesssim \hat{\varepsilon}_0   L^{-1+\delta} ,
\end{equation}
where we have used~\eqref{withimplicitconstoneiii} and
the estimate for the second factor on the right hand side of~\eqref{reftosecondfactor} contained in the second inequality of~\eqref{kiallobasicnewformiiiagain}.

Now
we apply~\eqref{pminusoneweightedaddediii} 
and~\eqref{nonpweightedaddediii}  again, with $\tau''$ in place of $\tau_0$,  
using~\eqref{withimplicitconstoneiiinew}--\eqref{withimplicitconstwoiii} to bound the initial data,
to obtain that 
for all $\tau_0+L\ge \tau'\ge \tau_0+L/2$, we have
\begin{align}
\label{withimplicitconstonelateriiinewhere}
\Efancyonekminustwo(\tau')\sim  \Eonekminustwo(\tau) &\lesssim  \hat{\varepsilon}_0L^{-1},\\
\label{withimplicitconstonelateriii}
\Efancyonekminusfour(\tau')\sim  \Eonekminusfour(\tau) &\lesssim  \hat{\varepsilon}_0  L^{-1+\delta}, \\
\label{withimplicitconsttwolateriii}
\Efancyzerokminussix(\tau') \sim \Ezerokminussix(\tau) &\lesssim 
\hat{\varepsilon}_0 \alpha L^{-2+\delta}.
\end{align}

Thus, in view of the requirement $\alpha \ge \alpha_0 \gg 1$, for sufficiently large $\alpha_0$ we may 
absorb the constants implicit in $\lesssim$ by explicit constants of our choice by adding extra positive $\alpha$
powers to the right hand side of~\eqref{withimplicitconstonelateriii} and~\eqref{withimplicitconsttwolateriii}.
In this way,  we obtain the specific 
estimates~\eqref{betterjustiii},~\eqref{betteroneiii} and~\eqref{bettertwoiii}. In the same way, we also obtain the specific constant of the  estimate
of~\eqref{addedherenotfancyiii}
 which will be convenient in our scheme.

We finally turn to~\eqref{basicnewformiiiagain} and~\eqref{basicestimateoneiiiagainnottop}.
Let us first note that revisiting~\eqref{nonpweightednottopiii}, with a little bit of averaging, we have the bounds
\begin{eqnarray*}
 {}^\chi\Xzerokminusone(\tau_0+\frac12, \tau_0+L) &\lesssim&
\int_{\tau_0}^{\tau_0+\frac12} \Ezerokminusone(\tau)d\tau 
+  \left( \,{}^\rho \Xzerokminusone + (\,{}^\rho \Xzerokminusone^{\frac{1-\delta}{1+\delta}} (\, {}^\rho \Xonekminusone)^{\frac{2\delta}{1+\delta}} \right)
\sqrt{\Xzerolesslessk+ (\,\Xzerolesslessk)^{\frac{1-\delta}{1+\delta}} (\, \Xonelesslessk)^{\frac{2\delta}{1+\delta}}}
+  {}^\rho\Xzerok\sqrt{\Xzerolesslessk}\sqrt{L}\\
&\lesssim &\sqrt{\int_{\tau_0}^{\tau_0+\frac12}\Ezerok(\tau)d\tau }\sqrt{\int_{\tau_0}^{\tau_0+\frac12}\Ezerokminustwo(\tau)d\tau}
+ \hat\varepsilon_0 \varepsilon^{\frac12} L^{-\frac14}\\
&\lesssim &\hat\varepsilon_0 \alpha^{\frac12} L^{-\frac12}+  \hat\varepsilon_0 \varepsilon^{\frac12} L^{-\frac14}\\
&\lesssim &\hat\varepsilon_0 \alpha^{\frac12} L^{-\frac14},
\end{eqnarray*}
where we have used also the interpolation inequality~\eqref{interpolationwithk} (and that $L\ge \frac12$).

We now notice that the above bound clearly holds also for
${}^\chi\Xzerokminusone':= {}^\chi\Xzerokminusone(\tau_0+\frac12, \tau_0+L)+\int_{\tau_0}^{\tau_0+\frac12} \Ezerokminusone(\tau)d\tau $,
For~\eqref{basicnewformiiiagain},  we revisit~\eqref{pminusoneweightedtopiii}, noting that we may replace
${}^\chi\Xzerokminusone$ by ${}^\chi\Xzerokminusone'$ on the right hand side, 
and thus  we may bound 
\begin{eqnarray*}
\sup_{\tau\in [\tau_0,\tau_0+L]} 
\Efancyonek(\tau) 				&\leq& \Efancyonek(\tau_0) 
  +   AC\, {}^\chi\Xzerokminusone'+C\, {}^\rho \Xonek
\sqrt{\Xonelesslessk} 
 + C\, {}^\rho\Xzerok\sqrt{\Xzerolesslessk}\sqrt{L} \\
 &\leq&\hat \varepsilon_0
  +  AC (\hat\varepsilon_0 \alpha^{\frac12} L^{-\frac14}) +C\hat\varepsilon_0\varepsilon_0^{1/2}L^{(-1+\delta)/2}
 + C\hat\varepsilon_0\varepsilon_0^{12} L^{-\frac32+\delta}\\
    &\leq&\hat \varepsilon_0
  + \alpha\hat\varepsilon_0L^{-\frac14} ,
\end{eqnarray*}
and we have used a higher power of $\alpha$
to absorb constants.

Finally, for~\eqref{basicestimateoneiiiagainnottop}, we revisit~\eqref{nonpweightednottopiii} for $k-2$ replacing $k-1$,
to obtain
\begin{eqnarray*}
 {}^\chi\Xzerokminustwo&\lesssim&
 \Ezerokminustwo(\tau_0)
+  \left( \,{}^\rho \Xzerokminustwo + (\,{}^\rho \Xzerokminustwo)^{\frac{1-\delta}{1+\delta}} (\, {}^\rho \Xonekminustwo)^{\frac{2\delta}{1+\delta}} \right)
\sqrt{\Xzerolesslessk+ (\,\Xzerolesslessk)^{\frac{1-\delta}{1+\delta}} (\, \Xonelesslessk)^{\frac{2\delta}{1+\delta}}}
+  {}^\rho\Xzerokminusone\sqrt{\Xzerolesslessk}\sqrt{L}\\
&\lesssim &\hat\varepsilon_0 \alpha L^{-1}
\end{eqnarray*}
whence we have
\begin{eqnarray*}
\sup_{\tau\in [\tau_0,\tau_0+L]} 
\Efancytwominusdelkminusone(\tau)	&\leq& 	
 \Efancytwominusdelkminusone(\tau_0) 
  +   A\, {}^\chi\Xzerokminustwo+   C\, \left(  {}^\rho\Xtwominusdelkminusone\sqrt{\Xzerolesslessk} + \sqrt{{}^\rho\Xtwominusdelkminusone}\sqrt{ {}^\rho\Xzerokminusone}
\sqrt{\Xtwominusdeltalesslessk} \right) +
C\, {}^\rho\Xzerokminusone\sqrt{\Xzerolesslessk}\sqrt{L}\\
 &\le&
\hat\varepsilon_0 + 
 A\hat\varepsilon_0 \alpha^{\frac12} L^{-1}  +C\hat\varepsilon_0\varepsilon_0^{1/2} L^{(-2+\delta)/2} 
 + C\hat\varepsilon_0\varepsilon_0^{1/2} L^{(-2+\delta)/2}L^{1/2}\\
&\leq&\hat\varepsilon_0+\alpha\hat\varepsilon_0L^{-\frac14},
\end{eqnarray*}
where we have again used a higher power of $\alpha$ to absorb all other constants.
\end{proof}

\subsubsection{The iteration: proof of Theorem~\ref{themaintheorem} in case (iii)}
\label{theiterationforthree}
We may now prove Theorem~\ref{themaintheorem} in case (iii).

As in case (ii), we shall proceed iteratively. The proof is a simple modification of 
that of Section~\ref{theiterationfortwo}.

We define $\tau_0=1$, $L_0=1$, $L_i = 2^i$, $\tau_{i+1}=\tau_i+L_i$,
and fix $\alpha\ge \alpha_0$  so that the statement of Proposition~\ref{bettertimepleaseiii} holds,
for instance $\alpha:=\alpha_0$.  Define the parameter
\begin{equation}
\label{dparameterdef}
d:=\Pi_{i=1}^\infty (1+ 2^{-\frac{i}4}\alpha)<\infty.
\end{equation}

(Note that we shall no longer note the dependence of constants and parameters on $\alpha$ since it is now considered fixed.
Thus implicit constants depending on the choice of $\alpha$ will from now on be incorporated in the notations $\sim$ and  $\lesssim$.)

For $0<\varepsilon_0 \le \varepsilon_{\rm global}$ and $\varepsilon_{\rm global}$ sufficiently small,
since 
\begin{equation}
\label{afouedw}
\Eonek(\tau_0) \leq\varepsilon_0, \qquad
\Etwominusdelkminustwo(\tau_0)\leq\varepsilon_0,
\end{equation}
then in view of Corollary~\ref{nomoreroman}, we have 
\[
\Efancyonek (\tau_0) \lesssim \varepsilon_0, \qquad 
\Efancytwominusdelkminustwo(\tau_0) \lesssim \varepsilon_0.
\]
Thus, for sufficiently small $\varepsilon_{\rm global}\ll \hat\varepsilon_{\rm slab}$, 
and all $0<\varepsilon_0\le \varepsilon_{\rm global}$, then if~\eqref{afouedw} is satisfied,
it
follows that there exists a $\hat\varepsilon_0(\varepsilon_0) \sim \varepsilon_0$, 
satisfying 
\begin{equation}
\label{smallenoughhere}
\alpha  \hat\varepsilon_0 \leq \frac1d \varepsilon_{\rm slab}, \qquad \hat\varepsilon_0\leq \frac1d \hat\varepsilon_{\rm slab}
\end{equation}
and
\[
\Efancyonek(\tau_0) \leq \hat\varepsilon_0 ,\qquad
\Efancytwominusdelkminustwo(\tau_0)\leq  \hat\varepsilon_0, \qquad
\Efancyzerokminustwo(\tau_0)\leq \hat\varepsilon_0\alpha L_0^{-1}, \qquad
\Efancyonekminusfour(\tau_0) \leq  \hat\varepsilon_0\alpha L_0^{-1+\delta}, 
\qquad
\Efancyzerokminussix(\tau_0) \leq \hat\varepsilon_0\alpha^2 L_0^{-2+\delta }.
\]
Finally, with our restriction on the definition of $\varepsilon_0$,  we may also write
\[
\Eonek(\tau_0)\leq \alpha \hat\varepsilon_0.
\]

In general, given
$\tau_i\ge 1$ defined above, $\hat{\varepsilon}_i\leq \hat{\varepsilon}_{\rm slab}$, and a solution $\psi$
of~\eqref{theequation} on $\mathcal{R}(\tau_0,\tau_i)$
 such that
\begin{equation}
\label{notsofancyiii}
\Eonek(\tau_i)\leq \alpha \hat\varepsilon_i,
\end{equation}
and
\begin{equation}
\label{thosefancyonesiii}
\Efancyonek(\tau_i) \leq \hat\varepsilon_i ,\qquad
\Efancytwominusdelkminustwo(\tau_i)\leq  \hat\varepsilon_i, \qquad
\Efancyzerokminustwo(\tau_i)\leq \hat\varepsilon_i\alpha L_i^{-1}, \qquad
\Efancyonekminusfour(\tau_i) \leq  \hat\varepsilon_i\alpha L_i^{-1+\delta}, 
\qquad
\Efancyzerokminussix(\tau_i) \leq \hat\varepsilon_i\alpha^2 L_i^{-2+\delta },
\end{equation}
note that by our restrictions on $\hat\varepsilon_{\rm slab}$, the assumptions of Proposition~\ref{globproptwo}
hold with $\alpha^2 \hat\varepsilon_i$ in  place of $\varepsilon_0$ (here we are using~\eqref{notsofancyiii} to invoke Corollary~\ref{nomoreroman}
to rewrite~\eqref{thosefancyones} in terms of the calligraphic energies; cf.~the
first lines of the proof of Proposition~\ref{bettertimepleaseiii})  and the assumptions of 
Proposition~\ref{bettertimepleaseiii} then apply with $\hat\varepsilon_i$ in  place of $\hat\varepsilon_0$, 
where both propositions are understood now with
 $\tau_i$, $\tau_{i+1}$ in place of $\tau_0$, $\tau_1$.
 It follows that the solution $\psi$ extends to a solution defined also in
$\mathcal{R}_i:=\mathcal{R}(\tau_i,\tau_i+L_i)$ satisfying the estimates
\eqref{addedherenotfancyiii}--\eqref{kiallobasicnewformiiiagain} while for $\tau'=\tau_{i+1}=\tau_i+L_i$ 
we  have in addition~\eqref{betteroneiii}--\eqref{bettertwoiii}.
We have thus
\begin{align}
\nonumber
\Eonek(\tau_{i+1})&\leq\hat\varepsilon_i \alpha	&\leq&\hat\varepsilon_{i+1}\alpha ,\\
\nonumber
\Efancyonek(\tau_{i+1})&\leq \hat\varepsilon_i(1+\alpha  L_i^{-\frac14})&\leq&\hat\varepsilon_{i+1}   , \\
\nonumber
\Efancytwominusdelkminustwo(\tau_{i+1})&\leq \hat\varepsilon_i(1+\alpha  L_i^{-\frac14})&\leq&\hat\varepsilon_{i+1}   , \\
\nonumber
\Efancyonekminustwo(\tau_{i+1})&\leq \frac12\hat\varepsilon_i\alpha L_i^{-1}  &\leq&\hat\varepsilon_{i+1} \alpha L_{i+1}^{-1}  , \\
\nonumber
\Efancyonekminusfour(\tau_{i+1})&\leq\frac12 \hat\varepsilon_i \alpha L_i^{-1+\delta}&\leq &\hat\varepsilon_{i+1}\alpha L_{i+1}^{-1+\delta}  , \\
\nonumber
\Efancyzerokminussix(\tau_{i+1})&\leq \frac14 \hat\varepsilon_i\alpha^2 L_i^{-2+\delta}&\leq & \hat\varepsilon_{i+1} \alpha^2 L_{i+1}^{-2+\delta} ,
\end{align}
as long as 
\begin{equation}
\label{inductdefforthree}
\hat\varepsilon_{i+1}:=\hat\varepsilon_i(1+\alpha  L_i^{-\frac14}).
\end{equation}
By our requirement~\eqref{smallenoughhere} and the 
definition~\eqref{dparameterdef} of the parameter $d$, then defining
$\hat\varepsilon_{i+1}$ inductively by~\eqref{inductdefforthree}, it follows
that $\hat\varepsilon_{i+1}\leq \hat\varepsilon_{\rm slab}$ and $\alpha\hat\varepsilon_{i+1}\leq \varepsilon_{\rm slab}$.

It follows that a solution exists in $\mathcal{R}(\tau_0,\infty)=\cup \mathcal{R}(\tau_i,\tau_i+L_i)$
and in each interval the estimates \eqref{kiallobasicnewformiiiagain}--\eqref{kiallobasicnewformiiiagainlastone} hold,
with $L=L_i$.

We obtain finally that for all $\tau\ge 1$ we have (among other estimates)
\begin{equation}
\label{noneedtoimprove}
\Eonek(\tau)\lesssim \varepsilon_0 , 
\end{equation}
\begin{equation}
\label{canalsobeimproved}
\int_{\tau_0}^\tau {}^\rho\Ezerok' + {}^\chi\Ezerokminusone'+ \Ezerokminustwo' 
\lesssim \varepsilon_0 
\log (\tau+1),
\end{equation}
\begin{equation}
\label{finalboundednessforiii}
\Etwominusdelkminustwo(\tau)\lesssim \varepsilon_0, 
\end{equation}
\begin{equation}
\label{resumefromhereforiii}
\Eonekminusfour(\tau) \lesssim \varepsilon_0\tau^{-1+\delta}, \qquad \Fzerokminusfour(v,\tau)\lesssim \varepsilon_0\tau^{-1+\delta},
\end{equation}
\begin{equation}
\label{resumefromheretwoforiii}
\int_\tau^\infty {}^\rho \Ezerokminusfour' +{}^\chi \Ezerokminusfive'+ \Ezerokminussix'  \lesssim \varepsilon_0\tau^{-1+\delta},
\end{equation}
\begin{equation}
\label{resumefromheresevenforiii}
\Ezerokminussix(\tau) \lesssim \varepsilon_0\tau^{-2+\delta}, \qquad \Fzerokminussix(v,\tau)\lesssim \varepsilon_0\tau^{-2+\delta},
\end{equation}
\begin{equation}
\label{resumefromhereeightforiii}
\int_\tau^\infty{}^\rho\Ezerominusoneminusdeltakminussix'+ {}^\chi\Ezerominusoneminusdeltakminusseven'
+\Ezerominusoneminusdeltakminuseight'  \lesssim \varepsilon_0\tau^{-2+\delta}.
\end{equation}

Finally, let us note that we may  improve~\eqref{canalsobeimproved} 
 to
\begin{equation}
\label{weimproveheretwoforiii}
\int_{\tau_0}^\infty {}^\rho \Ezerok'+{}^\chi\Ezerokminusone'+ \Ezerokminustwo' \lesssim \varepsilon_0.
\end{equation}
This follows, for all $\tau$,  by applying 
again both estimates of Proposition~\ref{finalpropest} appropriate to case (iii)
globally in $\mathcal{R}(\tau_0,\tau)$.
One may now re-estimate all nonlinear spacetime integrals arising in the estimates  on dyadic intervals and sum.

\appendix

\section{The  physical space identity on very slowly rotating~Kerr}
\label{howto}

In this appendix, we show  that very slowly rotating Kerr metrics indeed satisfy the assumptions
of Section~\ref{caseiiiidentity}. 
\begin{theorem}
\label{AppendixTheo}
Consider the manifold $(\mathcal{M}, g_{a,M})$ of
Section~\ref{kerrsubsubsection}, where  $g_{a,M}$ denotes the Kerr metric and
$|a|\ll M$.
Then there exist currents $J^{V,w,q,\varpi}$, $K^{V,w,q}$ associated to 
the wave operator $\Box_{g_{a,M}}$, satisfying the assumptions of Section~\ref{caseiiiidentity}.
\end{theorem}

Theorem~\ref{AppendixTheo} 
actually holds for general suitably small stationary perturbations of Schwarzschild satisfying the assumptions of Section~\ref{backgroundgeom}
and appropriate decay at infinity. So as not to complicate matters more, here we will simply do explicitly the computation for slowly rotating Kerr $|a|\ll  M$.

The nontrivial computations necessary to produce  $J^{V,w,q,\varpi}$ and $K^{V,w,q}$
all involve the exterior region. As a result, we may work entirely in Boyer--Lindquist coordinates.
Thus, {\bf  in this appendix, $r$ will denote the Boyer--Lindquist coordinate, and \underline{not} the function~\eqref{rdefhere} of Section~\ref{thatwhichunderlies}}. 
We will show explicitly the positivity
properties in the region $r>r_+$, which can be covered globally by such Boyer--Lindquist coordinates. 
The extension of the coercivity properties to the manifold of Section~\ref{kerrsubsubsection} follows softly, without computation.
(See already the end of Section~\ref{addingredshift}).

Our currents will be explicit, but rather than give formulae for the final $(V,w,q,\varpi)$, we will build the current from its natural constituent pieces.
In particular, so as to directly generate positive  $0$'th order terms in the boundary currents, we will make use of the twisted energy momentum tensor
as defined in~\cite{HOLZEGEL20142436}.

This appendix is organised as follows.
We first recall  the Kerr metric in Boyer--Lindquist coordinates in Section~\ref{thekerrmanifold}. 
We shall then review in Section~\ref{twistedsec} the twisted energy momentum tensor and its basic properties. 
In the next two sections we shall build our current as a combination of a ``Morawetz current'', constructed
in Section~\ref{Morawetziden} (vanishing, however, identically in an open set contained trapped null geodesics), and the twisted
stationary current and redshift currents, constructed in Section~\ref{addingredshift}. Section~\ref{auxiliarycomps} contains some auxiliary calculations.

\subsection{The Kerr metric in Boyer--Lindquist coordinates}
\label{thekerrmanifold}

We define
\begin{align}
\Delta = r^2 -2 M r +a^2 \ \ \ , \ \ \ \Sigma=r^2 + a^2 \cos^2 \theta \, .
\end{align}
We recall the $(t,r, \theta,\phi)$ Boyer--Lindquist coordinates and the associated $(t,r^\star, \theta,\phi)$
coordinates (which we will also refer to as Boyer--Lindquist coordinates) with the familiar relation
\[
\frac{dr^\star}{dr} = \frac{r^2+a^2}{\Delta}  \, .
\]
We shall consider these in the domain $\mathbb R \times (r_+,\infty) \times \mathbb S^2$, where $(\theta,\phi)$ are identified with the usual
spherical coordinates of $\mathbb S^2$.
The Kerr metric then has the form
\[
g=g_{a,M} = g_{tt} dt\otimes dt +g_{t\phi} (dt \otimes d\phi + d\phi\otimes dt) + g_{r^\star r^\star} dr^\star\otimes dr^\star + g_{\theta \theta} d\theta\otimes d\theta + g_{\phi \phi} d\phi\otimes d\phi\, , 
\]
with
\begin{align}
g_{tt} &= -\left(1-\frac{2Mr}{\Sigma}\right) \ \ \ , \ \ \ 
 g_{r^\star r^\star} = \frac{\Sigma \Delta}{(r^2+a^2)^2} \ \ \ , \ \ \ g_{\theta \theta} = \Sigma \, , \nonumber \\
 g_{t\phi} &= -\frac{2Mr}{\Sigma} a \sin^2 \theta \ \ \ , \ \ \ \ \ 
 g_{\phi \phi} = \left(r^2+a^2 + \frac{2Mr a^2}{\Sigma} \sin^2 \theta\right) \sin^2\theta \, .
 \end{align}
 We compute the inverse metric components as  
 \begin{align}
g^{tt} &= -\frac{1}{\Delta} \left(r^2+a^2 + \frac{2Mr a^2}{\Sigma} \sin^2 \theta \right) \ \ \ , \ \ \ 
 g^{r^\star r^\star} = \frac{(r^2+a^2)^2} {\Sigma \Delta} \ \ \ , \ \ \ g^{\theta \theta} = \frac{1}{\Sigma} \, , \nonumber \\
 g^{t\phi} &= -\frac{2Mr}{\Sigma \Delta} a  \ \ \ , \ \ \ \ \ 
 g^{\phi \phi} = \frac{\Delta - a^2 \sin^2 \theta}{\Sigma \Delta \sin^2 \theta}  \, .
 \end{align}
 For the determinant in these coordinates we note
 \begin{align} \label{detformula}
 \sqrt{|g|} 
 = \frac{\Delta}{r^2+a^2} \Sigma \sin \theta=g_{r^\star r^\star} \left(r^2+a^2\right) \sin \theta \, .
\end{align}

\subsection{The twisted energy momentum tensor}
\label{twistedsec}

We consider the covariant wave equation $\Box_g \psi = g^{\alpha \beta } \nabla_\alpha \nabla_\beta \psi = 0$ for $g=g_{a,M}$ the Kerr metric. 
Note that this equation can be rewritten as follows:
With the function $\beta=\left(\sqrt{r^2+a^2}\right)^{-1}$, we
 define as in~\cite{HOLZEGEL20142436} the twisted operators
\begin{align}
\tilde{\nabla}_\mu (\cdot) = \beta \nabla_\mu \left(\beta^{-1} \cdot \right) \ \ \ , \ \ \ \tilde{\nabla}^\dagger_\mu (\cdot) = -\beta^{-1} \nabla_\mu \left(\beta \cdot \right)
\end{align}
and rewrite the wave equation as 
\begin{align} \label{waved}
0=\Box_g \psi &= - \tilde{\nabla}_\mu^\dagger \tilde{\nabla}^\mu \psi - \mathcal{V} \psi = 0 \nonumber \\
&=\sqrt{r^2+a^2} \nabla_\mu \left((r^2+a^2)^{-1} \nabla^\mu \left(\psi \sqrt{r^2+a^2}\right)\right) -\mathcal{V} \psi \, 
\end{align}
where
\begin{align}
\mathcal{V} &= -\frac{\Box_g \beta}{\beta} = -\frac{1}{g_{r^\star r^\star}(r^2+a^2)\beta } \partial_{r^\star} \left(  \left(r^2+a^2\right) \partial_{r^\star} \beta \right) \nonumber \\
&=  \frac{1}{\Sigma} \cdot \frac{2Mr^3 +a^2 r (r-4M) +a^4}{(r^2+a^2)^2} =: \frac{1}{\Sigma} \mathcal{V}_0 \, . 
\end{align}

We now define the twisted energy momentum tensor
\begin{align} \label{twistedT}
\tilde{T}_{\mu \nu}[\psi] &= \tilde{\nabla}_\mu \psi \tilde{\nabla}_\nu \psi - \frac{1}{2} g_{\mu \nu} \left( \tilde{\nabla}^\alpha \psi \tilde{\nabla}_\alpha \psi + \mathcal{V} \psi^2 \right) \\
&=\beta^2 \left[ {\nabla}_\mu (\beta^{-1}\psi) {\nabla}_\nu (\beta^{-1}\psi) - \frac{1}{2} g_{\mu \nu} \left( {\nabla}^\alpha (\beta^{-1} \psi) {\nabla}_\alpha (\beta^{-1} \psi) + 
\mathcal{V} \beta^{-2} \psi^2 \right)\right] . \nonumber 
\end{align}

We note that for all $|a|<M$, $\mathcal{V}$ is strictly positive in the exterior, in fact
\begin{equation}
\label{thispositivityhere}
\mathcal{V}\geq  \frac{2Mr(r-M)(r+M)}{\Sigma ( r^2+a^2)^2} \gtrsim r^{-3},
\end{equation}
since $r \ge r_+ > M$,
and $\tilde{T}_{\mu\nu}[\psi]$ satisfies the dominant energy condition, 
i.e.~$\tilde{T}_{\mu\nu}[\psi]\xi^\mu \xi^\nu$ is nonnegative  if $\xi$ is timelike,
and in fact controls coercively first derivatives of $\psi$ as well as $\psi$ itself (the latter with the weight $r^{-3}$).
(Indeed, this direct control of the zero'th order term is our motivation for considering~\eqref{twistedT} in place of the usual $T_{\mu\nu}$.).

From Proposition 3 of~\cite{HOLZEGEL20142436} we infer:
\begin{proposition} \label{prop:holwar}
For $\psi$ a $C^2$ solution of $\Box_g \psi = 0$ and $X$ a smooth spacetime vector field we have the identity
\begin{align}
\nabla^\mu \tilde{J}^X_\mu [\psi] = \tilde{K}^X [\psi] \, , 
\end{align}
where 
\begin{align}
\tilde{J}^X_\mu [\psi ] &= \tilde{T}_{\mu \nu} [\psi] X^\nu \, , \\
\tilde{K}^X [\psi] &= {}^{(X)}\pi^{\mu \nu} \tilde{T}_{\mu \nu} [\psi] + X^\nu \tilde{S}_\nu [\psi] \,  
\end{align}
and
\begin{align}
\tilde{S}_\mu [\psi] = \frac{\tilde{\nabla}^\dagger_\mu (\beta \mathcal{V})}{2\beta} \psi^2 + \frac{\tilde{\nabla}^\dagger_\mu \beta}{2\beta}  \tilde{\nabla}^\nu \psi \tilde{\nabla}_\nu \psi \, .
\end{align}
\end{proposition}

We also have the Lagrangian identity for an arbitrary spacetime function $w$:
\begin{proposition} \label{prop:lagrangian}
For $\psi$ a $C^2$ solution of $\Box_g \psi = 0$ and $w$ a smooth spacetime function we have the identity
\begin{align}
\nabla^\mu \tilde{J}_\mu^{aux,w}  = \tilde{K}^{aux, w}[\psi]
\end{align}
with
\begin{align}
 \tilde{K}^{aux,w}[\psi]  &:= w \beta^2 \nabla^\alpha (\beta^{-1} \psi) \nabla_\alpha (\beta^{-1} \psi)  +(\beta^{-1}\psi)^2 \left(- \frac{1}{2}  \nabla^\mu \left(\beta^2 \nabla_\mu w\right) + \mathcal{V} w \beta^2 \right)\, ,  \nonumber \\
 \tilde{J}^{aux,w}_\mu &:= w \beta^2 \left( \beta^{-1} \psi \nabla_\mu (\beta^{-1} \psi) \right) - \frac{1}{2} \psi^2 \nabla_\mu w. \nonumber
\end{align}
\end{proposition}
\begin{proof}
This follows from 
\begin{align}
\nabla^\mu \left( w  \left( \psi \beta^{-1}  \beta^2 \nabla_\mu (\beta^{-1} \psi) \right) \right) =& w \beta^2 \nabla^\mu (\beta^{-1} \psi) \nabla^\mu (\beta^{-1} \psi) + w \psi \left(- \tilde{\nabla}_\mu^\dagger \tilde{\nabla}^\mu \psi \right) \nonumber \\
&+ (\nabla^\mu w) \beta^2 \frac{1}{2} \nabla_\mu (\beta^{-1}\psi)^2
\end{align}
after rearranging and inserting (\ref{waved}).
\end{proof} 

\begin{remark} \label{rem:ch}
Note that when $w$ is a function of $r$ only (as will always be the case in the applications below) we have
\[
\nabla^\mu \left(\beta^2 \nabla_\mu w \right) = \frac{1}{\sqrt{g}} \partial_{r^\star} \left(\sqrt{g} g^{r^\star r^\star} \beta^2 \partial_{r^\star} w \right) = \frac{r^2+a^2}{\Delta \Sigma} \partial_{r^\star}^2 w \, ,
\]
a formula which is useful in the computations below.
\end{remark}

\subsection{A Morawetz current vanishing identically in a  neighbourhood of trapping}
\label{Morawetziden}

In this section, we shall produce a current giving the desired bulk positivity properties modulo suitable zero'th order terms, but with additional degeneration
at the horizon, and without the boundary positivity properties. (The horizon degeneration and the boundary positivity properties 
will  be dealt with in Section~\ref{addingredshift}.)

The point about this current is that it will completely vanish in the set $11M/4\le r\le 7M/2$, which contains all trapped null geodesics for $|a|\ll M$.
Thus, \emph{this current is insensitive to the precise nature of the dynamics near trapping}. The requirement of vanishing on such a set will necessarily generate
lower order terms, however, with an unfavourable sign. These too will be supported away from trapping.

Our current will be defined combining those of Proposition~\ref{prop:holwar} and~\ref{prop:lagrangian}, and the
 vector field component of the current will be in the direction of $\partial_{r^\star}$. We begin with a computation (note that $'$ below will denote $\frac{d}{dr^*}$):

\begin{lemma} \label{lem:divid}
For $\beta^{-1} = \sqrt{r^2+a^2}$, $h=\left(1-\frac{3M}{r}\right)\frac{\Delta}{(r^2+a^2)^2}$ and $X=f (r^*) \partial_{r^\star}$ (with $f$ arbitrary) we have for any
 $\updelta \in \mathbb{R}$ the  divergence identity
\begin{align}
\nabla^\mu \left(\tilde{J}^X_\mu [\psi ] + \tilde{J}^{aux,\frac{1}{2}f^\prime}_{\mu} [\psi] - \updelta \cdot \tilde{J}^{aux, f \cdot h}_{\mu} [\psi] \right) \nonumber \\= \tilde{K}^X[\psi] + \tilde{K}^{aux,\frac{1}{2}f^\prime}[\psi] -\updelta \cdot \tilde{K}^{aux, f \cdot h}_{\mu} [\psi] \, ,
\end{align}
where
\begin{align} \label{crosste}
\tilde{J}^X_\mu [\psi ] &+ \tilde{J}_{\mu}^{aux,\frac{1}{2}f^\prime} [\psi] -\updelta \tilde{J}^{aux, f \cdot h}_{\mu} [\psi]  \nonumber \\
&=   \tilde{T}_{\mu \nu} [\psi] X^\nu + f^\prime \beta^2 \left( \beta^{-1} \psi \nabla_\mu (\beta^{-1} \psi) \right) - \frac{1}{2} \psi^2 \nabla_\mu f^\prime \, 
\nonumber \\
&\qquad -\updelta \left[f \cdot h \beta^2 \left( \beta^{-1} \psi \nabla_\mu (\beta^{-1} \psi) \right) - \frac{1}{2} \psi^2 \nabla_\mu f \cdot h\right] 
\end{align} 
and
\begin{align} \label{lookatbulk}
\tilde{K}^{f\partial_{r^\star}}[\psi] &+ \tilde{K}^{aux,\frac{1}{2}f^\prime}[\psi] - \updelta \tilde{K}^{aux, f \cdot h}_{\mu} [\psi] \nonumber \\
&= \beta^2 \left( g^{r^\star r^\star} \partial_{r^\star} f - \updelta f h \right) \left(\partial_{r^\star}(\beta^{-1} \psi) \right)^2 \nonumber \\
&\qquad  +\frac{1}{2} f \beta^2  \sum_{\mu,\nu \neq r^\star} \left( \mathcal{A}^{\mu \nu} - 2 \updelta h g^{\mu \nu} \right) \partial_\mu (\beta^{-1}\psi) \partial_{\nu} (\beta^{-1} \psi) \nonumber \\
&\qquad +\frac{r^2+a^2}{\Delta \Sigma} \Bigg( - \frac{1}{4}   f^{\prime \prime \prime}   -\frac{1}{2} f\partial_{r^\star} \left(\frac{\Delta}{(r^2+a^2)^2} \mathcal{V}_0\right) \nonumber \\
& \ \ \ \ \ \ \ \ \ \ \ \ \ \ \ \ \ + \updelta \left[ \frac{1}{2} \left(f \cdot h\right)^{\prime \prime} -\mathcal{V}_0 f \cdot h \frac{\Delta}{(r^2+a^2)^2}\right]  \Bigg) (\beta^{-1} \psi)^2 
\end{align}
with
\begin{align}
\mathcal{A}^{\mu \nu} =   -\partial_{r^\star} (g^{\mu \nu}) +   g^{\mu \nu} g_{r^\star r^\star} \partial_{r^\star} (g^{r^\star r^\star}) \, .
\end{align}
\end{lemma}

\begin{proof}
We compute (see Section~\ref{auxiliarycomps})
\begin{align} \label{target1}
\tilde{K}^{f\partial_{r^\star}}[\psi] 
&= \beta^2 g^{r^\star r^\star} \partial_{r^\star} f \left(\partial_{r^\star}(\beta^{-1} \psi) \right)^2 \nonumber \\
&\qquad-\frac{1}{2} f \beta^2 \sum_{\mu,\nu \neq r^\star}  \left[ \partial_{r^\star} (g^{\mu \nu}) -   g^{\mu \nu} g_{r^\star r^\star} \partial_{r^\star} (g^{r^\star r^\star})\right]\partial_\mu (\beta^{-1}\psi) \partial_{\nu} (\beta^{-1} \psi) \nonumber \\
&\qquad-\frac{1}{2} \beta^2 \left( \partial_{r^\star} f  \right) g^{\alpha \beta} \partial_\alpha (q^{-1} \psi) \partial_{\beta}(\beta^{-1} \psi)-\frac{1}{2} f\partial_{r^\star} 
(\beta^2 \mathcal{V}) (\beta^{-1} \psi)^2 \nonumber \\
&\qquad -\frac{1}{2} \left( \partial_{r^\star} f + f\frac{1}{\sqrt{g}} \partial_{r^\star} \sqrt{g} \right) \mathcal{V} \psi^2 .
\end{align}
Adding the Lagrangian identity of Proposition~\ref{prop:lagrangian} with $w= \frac{1}{2} f^\prime$ (and using Remark~\ref{rem:ch}) we deduce the result for $\updelta=0$. For arbitrary $\updelta$ we simply add the Lagrangian identity of Proposition \ref{prop:lagrangian} with $w= \left(1-\frac{3M}{r}\right)\frac{\Delta}{(r^2+a^2)} f= f \cdot h$ and group terms.
\end{proof}

We now exploit the divergence identity of Lemma \ref{lem:divid}. For this we first define the following (disjoint) decomposition of the range of the $r$-variable:
\begin{align}
[r_+, \infty) &= R_1 \cup R_2 \cup R_3 \cup R_4 \cup R_5 \nonumber \\
&= \left[r_+, \frac{5}{2} M\right) \cup \left[\frac{5}{2}M, \frac{11}{4}M\right) \cup \left[\frac{11}{4}M, \frac{7}{2}M\right) \cup \left[\frac{7}{2}M, 4M\right) \cup \left[4M,\infty\right) \nonumber .
\end{align}
Note that $\mathcal{M} \cap \{r \in R_3\}$ includes the  region containing all trapped null geodesics  if $\frac{a}{M}$ is suitably small. 

\begin{proposition} \label{prop:choices}
There exists an (explicit) function $f$ such that for all $\frac{|a|}{M}$ sufficiently small we can choose $\updelta>0$ sufficiently small (depending only on $M$) such that the estimate
\begin{align} \label{sigu}
&\tilde{K}^{f \partial_{r^\star}} [\psi] + \tilde{K}^{aux,\frac{1}{2}f^\prime}[\psi] -\updelta \cdot \tilde{K}^{aux, \frac{f \Delta}{(r^2+a^2)^2}}_{\mu} [\psi] \geq \nonumber \\
& \qquad c \mathbbm{1}_{R_1 \cup R_5} \left(\frac{\left(\partial_{t} (\beta^{-1} \psi)\right)^2+ \left(\partial_{r^\star} (\beta^{-1} \psi)\right)^2+ (\beta^{-1} \psi)^2}{r^4}   + \frac{1}{r^5}| \mathring{\nabla} (\beta^{-1} \psi)|^2    \right) \nonumber \\
& \qquad - C \mathbbm{1}_{R_2 \cup R_4} (\beta^{-1} \psi)^2 \,
\end{align}
holds for constants $c$ and $C$ depending only on $M$. Here we have defined the shorthand $| \mathring{\nabla} (\beta^{-1} \psi)|^2 :=  \left(\partial_\theta (\beta^{-1} \psi)\right)^2 +   \frac{1}{\sin^2\theta}   \left(\partial_\phi (\beta^{-1} \psi)\right)^2$.
\end{proposition}

\begin{proof}
In the proof, we let $\tilde{c}$ and $\tilde{C}$ be constants depending only on $M$ (but which might change from line to line).
Starting from (\ref{lookatbulk}) we define $\mathcal{B}^{\mu \nu} = \mathcal{A}^{\mu \nu} - 2 \updelta \left(1-\frac{3M}{r}\right) \frac{ \Delta}{(r^2+a^2)^2}g^{\mu \nu}$ and compute
\begin{align}
\mathcal{B}^{tt} &= 2a^2\cdot \frac{r^3-3Mr^2+a^2 r + a^2 M}{ \Sigma (r^2+a^2)^2}\sin^2 \theta +2\updelta \left(1-\frac{3M}{r}\right) \frac{ r^2+a^2 + \frac{2Mr a^2}{\Sigma} \sin^2 \theta}{(r^2+a^2)^2} , \nonumber \\
\mathcal{B}^{t\phi} &=  \frac{2aM(a^2-3r^2)}{ \Sigma (r^2+a^2)^2} + 2\updelta \cdot a \left(1-\frac{3M}{r}\right) \frac{2Mr}{\Sigma (r^2+a^2)^2} , \nonumber \\
\mathcal{B}^{\phi \phi} &= 2\frac{r^2 (r-3M) + a^2 r \cos[2\theta]+a^2 M}{\Sigma (r^2+a^2)^2 \sin^2 \theta} -2 \updelta \left(1-\frac{3M}{r}\right) \frac{\Delta - a^2 \sin^2\theta}{\Sigma (r^2+a^2)^2 \sin^2\theta} ,  \nonumber \\
\mathcal{B}^{\theta \theta} &= 2\frac{r^2 (r-3M) + a^2 r +a^2 M}{ \Sigma (r^2+a^2)^2 \sin^2 \theta} - 2 \updelta \left(1-\frac{3M}{r}\right) \frac{ \Delta}{\Sigma (r^2+a^2)^2} \, .
\end{align}
From this one sees that for $\frac{|a|}{M}$ sufficiently small we can choose $\updelta$ sufficiently small (depending only on $M$), such that in $\mathcal{M} \cap \{ r \notin R_3\}$ (where $|1-\frac{3M}{r}| \geq \frac{1}{12}$) the estimates
\begin{align} \label{qf1}
\frac{\mathcal{B}^{tt} }{(1-\frac{3M}{r})} \geq \frac{\tilde{c}}{r^2} \ \ \ , \ \ \ \frac{\mathcal{B}^{\phi \phi} }{(1-\frac{3M}{r})} \geq \frac{\tilde{c}}{r^3 \sin^2 \theta} \ \ \ , \ \ \ \frac{\mathcal{B}^{\theta \theta} }{(1-\frac{3M}{r})} \geq \frac{\tilde{c}}{r^3} 
\end{align}
hold for a constant $\tilde{c}$ depending only on $M$. Moreover, we have 
\begin{align} \label{qf2}
\Big| \frac{\mathcal{B}^{t\phi} }{(1-\frac{3M}{r})}\Big| \leq \frac{|a|}{M} \frac{\tilde{C}}{r^4}
\end{align}
We next choose $f : \left[r_+,\infty\right) \rightarrow \mathbb{R}$ to be bounded, monotonically increasing as follows
\begin{equation*}
  f = \left\{
    \begin{array}{rl}
      -\frac{M}{r}& \text{if } r \in R_1,\\
      \textrm{$C^3$ interpolate}& \text{if } r \in R_2,\\
      0 & \text{if } r \in R_3,\\
      \textrm{$C^3$ interpolate}& \text{if } r \in R_4,\\ 
      1-\frac{M}{r} & \text{if } r \in R_5.
    \end{array} \right.
\end{equation*}
Specifically, we do the $C^3$ interpolation such that we have $f$ monotonically increasing and
\begin{align} \label{fderel}
f^\prime \geq \tilde{c} (-f) \ \ \ \ \textrm{ in $R_2$}  \ \ \ \textrm{and} \ \ \ f^\prime \geq \tilde{c} f \ \ \ \ \textrm{ in $R_4$} 
\end{align}
for a fixed constant $\tilde{c}>0$ depending only on $M$. In particular (potentially making $\updelta$ slightly smaller), we can achieve that 
\[
g^{r^\star r^\star} \partial_{r^\star} f - \updelta f \left(1-\frac{3M}{r}\right) \frac{\Delta}{(r^2+a^2)^2} \geq f \left(1-\frac{3M}{r}\right) \frac{\tilde{c}}{r^2}
\]
holds in all of $\mathcal{M} \cap \{ r \notin R_3\}$ for a $\tilde{c}$ depending only on $M$. Since $f\left(1-\frac{3M}{r}\right)$ is globally non-negative and in fact bounded uniformly below by a $\tilde{c}$ in $R_1 \cup R_5$, the estimate~\eqref{sigu} now follows except for the zeroth order term on the right hand side. (For this the only thing to notice is that in
\[
\frac{1}{2}f \beta^2 \left(1-\frac{3M}{r}\right) \left[ \sum_{\mu,\nu \neq r^\star} \frac{\mathcal{B}^{\mu \nu}}{1-\frac{3M}{r}}  \partial_\mu (\beta^{-1}\psi) \partial_{\nu} (\beta^{-1} \psi) \right]
\]
the quadratic form in the square bracket is positive definite by the estimates~\eqref{qf1} and~\eqref{qf2}.)

For the zeroth order term, clearly we only need to establish a lower bound in $R_1 \cup R_5$.
We start with $R_5$, where we have for $f= 1- \frac{M}{r}$ that 
\[
-\frac{1}{2} f^{\prime \prime \prime} - f \left(\frac{\Delta}{(r^2+a^2)^2}V_0\right)^\prime = \frac{3M}{r^4} + \mathcal{O}(r^{-5})
\]
and
\[
\frac{r^2+a^2}{r^2-2Mr+a^2} \left( -\frac{1}{2} f^{\prime \prime \prime} - f \left(\frac{\Delta}{(r^2+a^2)^2}\mathcal{V}_0\right)^\prime \right) = \frac{M}{r^4} \left(3-\frac{2M}{r}-14\frac{M^2}{r^2}\right)  \ \ \ \textrm{if $a=0$.}
\]
Since the left hand side of the second identity is continuous in $a$ and the bracket on the right hand side uniformly bounded below by $\frac{1}{3}$ for $r \geq 3M$, the estimate claimed for the zeroth order term follows for sufficiently small $\frac{|a|}{M}$ after potentially making $\updelta$ smaller (note $ \updelta \frac{r^2+a^2}{r^2-2Mr+a^2}\left[ \frac{1}{2} \left(f \cdot h\right)^{\prime \prime} -\mathcal{V}_0 f \cdot h \frac{\Delta}{(r^2+a^2)^2}\right] \leq \frac{\tilde{C}\updelta}{r^4})$.

Similarly, in $R_1$ we have for $f=-\frac{M}{r}$ the identity
\[
\frac{r^2+a^2}{r^2-2Mr+a^2} \left( -\frac{1}{2} f^{\prime \prime \prime} - f (w\mathcal{V}_0)^\prime \right) = -\frac{M}{r^6} \left(14M^2-14Mr+3r^2\right)  \ \ \ \textrm{if $a=0$.}
\]
The expression on the right hand side is easily shown to be uniformly bounded below by $\frac{M^3}{4r^6}$ for $r \in [9M/5,11M/4]$. Since the expression on the left is in particular continuous in $a$ on $[r_+,11M/4]$ the estimate claimed for the zeroth order term also follows in $R_1$.
\end{proof}

\subsection{Adding the redshift and stationary currents and completing the proof}
\label{addingredshift}

Let $T$ denote the stationary Killing field of Kerr (which in Boyer--Lindquist coordinates is $T=\partial_t$) and $N$  the timelike (including on $\mathcal{H}^+$) redshift vector 
field constructed in Theorem~7.1 of~\cite{Mihalisnotes}. We have that $\tilde{K}^T=0$ while
\begin{align}
\tilde{K}^N[\psi] \geq \tilde{c}  \Big(&\left(\partial_{t} (\beta^{-1} \psi)\right)^2+ \left(\partial_{r^\star} (\beta^{-1} \psi)\right)^2  + | \mathring{\nabla} (\beta^{-1} \psi)|^2 + (\beta^{-1} \psi)^2 \nonumber \\
&+ \left[\frac{r^2+a^2}{\Delta} \left(\partial_t -\partial_{r^\star} + \frac{a}{r^2+a^2} \partial_\phi\right) (\beta^{-1}\psi)\right]^2 \Big) \ \ \textrm{in $r \leq \frac{9}{4}M$} \nonumber
\end{align}
provided $a$ is sufficiently small.  We also have $N=T$ in the complement of $R_1$.

We shall now  complete the proof of Theorem~\ref{AppendixTheo} by adding $\Upsilon \cdot \tilde{J}^T_{\mu} [\psi] + \eta \tilde{J}^N_{\mu} [\psi]$ to the current
of Proposition~\ref{prop:choices},   for constants $\Upsilon>0$ (large) and $\eta>0$ (small).

We consider the divergence identity of Lemma \ref{lem:divid} with $\updelta$ and $f$ now chosen as in Proposition \ref{prop:choices}. 
We  expand the identity as follows for constants $\Upsilon>0$ (large) and $\eta>0$ (small) to be chosen below as follows:
\begin{align} \label{finaldiv}
\nabla^\mu \left(\tilde{J}^X_\mu [\psi ] + \tilde{J}^{aux,\frac{1}{2}f^\prime}_{\mu} [\psi] - \updelta \cdot \tilde{J}^{aux, f \cdot h}_{\mu} [\psi] + \Upsilon \cdot \tilde{J}^T_{\mu} [\psi] + \eta \cdot \tilde{J}^N_{\mu} [\psi] \right) \nonumber \\= \tilde{K}^X[\psi] + \tilde{K}^{aux,\frac{1}{2}f^\prime}[\psi] -\updelta \cdot \tilde{K}^{aux, f \cdot h}_{\mu} [\psi] +
\eta\cdot  \tilde{K}^N[\psi]  \, .
\end{align}
 It is now clear that we can choose $\eta$ sufficiently small (depending only on $M$) such that  for $r \geq \frac{9}{4} M$ we can absorb the term $\eta \tilde{K}^N[\psi]$, which is supported only in $R_1$, by the positivity of $\tilde{K}^X[\psi] + \tilde{K}^{aux,\frac{1}{2}f^\prime}[\psi] -\updelta \cdot \tilde{K}^{aux, f \cdot h}_{\mu} [\psi]$ established in Proposition \ref{prop:choices} to deduce
\begin{align} \label{sigu2}
&\tilde{K}^{f \partial_{r^\star}} [\psi] + \tilde{K}^{aux,\frac{1}{2}f^\prime}[\psi] -\updelta  \tilde{K}^{aux, \frac{f \Delta}{(r^2+a^2)^2}}_{\mu} [\psi]+ \eta \tilde{K}^N[\psi]\geq\\
&+ c \mathbbm{1}_{R_1 \cup R_5} \left(\frac{\left(\partial_{t} (\beta^{-1} \psi)\right)^2+ \left(\partial_{r^\star} (\beta^{-1} \psi)\right)^2+ (\beta^{-1} \psi)^2}{r^4}   + \frac{1}{r^5}| \mathring{\nabla} (\beta^{-1} \psi)|^2    \right) \nonumber \\
&+ c \mathbbm{1}_{R_1}  \left[\frac{r^2+a^2}{\Delta} \left(\partial_t -\partial_{r^\star} + \frac{a}{r^2+a^2} \partial_\phi\right) (\beta^{-1}\psi)\right]^2 - C \mathbbm{1}_{R_2 \cup R_4} |(\beta^{-1} \psi)|^2 \, ,  \nonumber 
\end{align}
where $c$ may be smaller than in~\eqref{sigu} but still depends only on $M$. We finally claim that we can choose $\Upsilon$ sufficiently large in~\eqref{finaldiv} such that 
\begin{align}
&\left( \tilde{J}^X_\mu [\psi ] + \tilde{J}^{aux,\frac{1}{2}f^\prime}_{\mu} [\psi] - \updelta  \tilde{J}^{aux, f \cdot h}_{\mu} [\psi] + \Upsilon  \tilde{J}^T_{\mu} [\psi] + \eta \tilde{J}^N_{\mu} [\psi] \right) n_{\Sigma_{\tau}}^\mu \nonumber \\
&= \boxed{ \tilde{T}_{\mu \nu} [\psi] \left( X^\nu + \Upsilon T^\nu + \eta N^\nu\right) n_{\Sigma_{\tau}}^\mu} + f^\prime \beta^2 \left( \beta^{-1} \psi \nabla_\mu (\beta^{-1} \psi) \right) n_{\Sigma_{\tau}}^\mu \nonumber \\
&\ \ \ - \frac{1}{2} \psi^2 \nabla_\mu f^\prime n_{\Sigma_{\tau}}^\mu
 -\updelta \left[f \cdot h \beta^2 \left( \beta^{-1} \psi \nabla_\mu (\beta^{-1} \psi) \right) - \frac{1}{2} \psi^2 \nabla_\mu f \cdot h\right]n_{\Sigma_{\tau}}^\mu \nonumber \\
 \label{boundarycoercivity}
&\geq \beta^2 \left( |L(\beta^{-1}\psi)|^2 + \iota_{r \leq R} |\underline{L}(\beta^{-1}\psi)|^2 + \frac{1}{r^2} | \mathring{\nabla} (\beta^{-1} \psi)|^2 +
 \frac{1}{r^3} (\beta^{-1} \psi)^2 \right).
\end{align}
For the boxed term, the estimate is an immediate consequence of the 
twisted energy momentum tensor (\ref{twistedT}) satisfying the dominant energy condition, as discussed in Section~\ref{twistedsec}, 
the positivity~\eqref{thispositivityhere}, 
and the fact that the vector field $\Upsilon T + \eta N + X$ is timelike for sufficiently large $\Upsilon$, provided that we then restrict to $|a|$ sufficiently small.
Moreover, for the boxed term the zeroth order term in the estimate scales with $\Upsilon$, i.e.~the larger we choose $\Upsilon$ the larger the zeroth order term becomes. This can be used to absorb the remaining terms using the Cauchy--Schwarz inequality, for $|a|$ sufficiently small depending on this final choice of $\Upsilon$.

The relations~\eqref{sigu2} and~\eqref{boundarycoercivity} give the coercivity properties~\eqref{bulkunweightedcoercivity} and the first of~\eqref{boundunweightedcoercivity},
restricted to the exterior region,
where we define $\rho= \mathbbm{1}_{R_1 \cup R_5}$ and $\xi= \mathbbm{1}_{R_2 \cup R_4}$,
  except that the $r$ decay for the bulk current is not optimised to the $r^{-1-\delta}$ and $r^{-3-\delta}$ weights for the first and zeroth order terms, respectively, and the $r$ decay for the boundary  current is not optimised to the $r^{-2}$ weight for the zeroth order term. (Note 
 that the Boyer--Lindquist $r$ and the $r$ of~\eqref{rdefhere} are comparable for large $r$ values so one can compare directly
 the $r$-decay as if the coordinates were the same.) To remedy this, it suffices to add  $\epsilon J^{\hat\chi V, \hat\chi w, \hat\chi q, \hat\chi \varpi}$, 
 $\epsilon K^{\hat\chi V, \hat\chi w, \hat\chi q}$, respectively to the currents, where $(V,w,q,\varpi)$ here denotes the current of Appendix~\ref{subsec:VcurrentMink},
  $\hat\chi$ is a cutoff supported far away, and $\epsilon$ is sufficiently small. The resulting currents now indeed satisfy~\eqref{bulkunweightedcoercivity} and the first inequality of~\eqref{boundunweightedcoercivity}, restricted to the exterior, with weights as stated. We note finally that for small $|a|$, we may take the $\chi$ in
  the  estimate~\eqref{inhomogeneous} proven
  in~\cite{DafRodsmalla} or~\cite{partiii} to be identically one  outside of a small neighbourhood of $r=3M$. Thus, our $\xi$ indeed satisfies~\eqref{xivanishing}.

The currents trivially may be extended to the slightly larger bigger domain of Section~\ref{kerrsubsubsection}. Note finally that
the additional positivity statements of~\eqref{boundunweightedcoercivity} are easily shown to hold.
Rewriting the current in terms of a single quadruple $(V, w,q,\varpi)$, one easily sees that the boundedness statements~\eqref{seedboundedness} hold as well.
This completes the proof of Theorem~\ref{AppendixTheo}.

\subsection{Computation of the $X$-deformation tensor}
\label{auxiliarycomps}
We collect here some computations which were used in the proof of Theorem~\ref{AppendixTheo}.

We have 
\begin{align}
-2{}^{(X)}\pi^{\mu \nu} = X^\alpha \partial_\alpha (g^{\mu \nu}) - \partial_\alpha X^\mu g^{\alpha \nu}  - \partial_\alpha X^\nu g^{\alpha \mu} \, ,
\end{align}
hence for $X=f(r) \partial_{r^\star}$
\begin{align}
-2{}^{(X)}\pi^{\mu \nu} = f \partial_{r^\star} (g^{\mu \nu}) - \partial_\alpha X^\mu g^{\alpha \nu}  - \partial_\alpha X^\nu g^{\alpha \mu} \, ,
\end{align}
which we can write (using $g^{r^\star r^\star} g_{r^\star r^\star}=1$ for Kerr in Boyer-Lindquist) as 
\begin{align}
-2{}^{(X)}\pi^{\mu \nu} &= f\partial_{r^\star} (g^{\mu \nu}) + f  g^{\mu \nu} g_{r^\star r^\star} \partial_{r^\star} (g^{r^\star r^\star})- f  g^{\mu \nu} g_{r^\star r^\star} \partial_{r^\star} (g^{r^\star r^\star}) \ \ \ \textrm{unless $\mu=\nu=r^\star$} \nonumber \\
-2{}^{(X)}\pi^{r^\star r^\star} &= f  g^{r^\star r^\star} g_{r^\star r^\star} \partial_{r^\star} (g^{r^\star r^\star}) - 2g^{r^\star r^\star} \partial_{r^\star}  f , 
\end{align}
so
\begin{equation}
-2{}^{(X)}\pi^{\mu \nu} = g^{\mu \nu} \cdot f g_{r^\star r^\star} \partial_{r^\star} (g^{r^\star r^\star}) + \left\{
\begin{array}{lr}
 -2g^{r^\star r^\star} \partial_{r^\star} f  &\textrm{if $\mu=\nu=r^\star$}  \\
 f\partial_{r^\star} (g^{\mu \nu}) - f  g^{\mu \nu} g_{r^\star r^\star} \partial_{r^\star} (g^{r^\star r^\star}) & \textrm{otherwise.} \nonumber 
\end{array}
\right.
\end{equation}
We note also 
\begin{align}
tr  {}^{(X)}\pi = g_{\mu \nu} {}^{(X)}\pi^{\mu \nu} &= -\frac{1}{2} g_{\mu \nu}  f \partial_{r^\star} (g^{\mu \nu}) + \partial_{r^\star}  f =  f\frac{1}{\sqrt{g}} \partial_{r^\star} \sqrt{g} + \partial_{r^\star} f.
\end{align}
Therefore
\begin{align}
{}^{(X)}\pi^{\mu \nu} \tilde{T}_{\mu \nu} &= {}^{(X)}\pi^{\mu \nu}\tilde{\nabla}_{\mu} \psi \tilde{\nabla}_{\nu} \psi - \frac{1}{2} (tr  {}^{(X)}\pi) g^{\mu \nu} \tilde{\nabla}_\mu \psi \tilde{\nabla}_{\nu}\psi  - \frac{1}{2} (tr  {}^{(X)}\pi) \mathcal{V} \psi^2 \nonumber \\
&= \beta^2 g^{r^\star r^\star} \partial_{r^\star} f \left(\partial_{r^\star}(\beta^{-1} \psi) \right)^2 \nonumber \\
& \qquad-\frac{1}{2} \beta^2 f \sum_{\mu,\nu \neq r^\star}  \left[ \partial_{r^\star} (g^{\mu \nu}) -   g^{\mu \nu} g_{r^\star r^\star} \partial_{r^\star} (g^{r^\star r^\star})\right]\partial_\mu (\beta^{-1} \psi) \partial_{\nu} (\beta^{-1} \psi) \nonumber \\
&\qquad -\frac{1}{2} \left( \partial_{r^\star} f + f\frac{1}{\sqrt{g}} \partial_{r^\star} \sqrt{g} + f g_{r^\star r^\star} \partial_{r^\star} (g_{r^\star r^\star})^{-1}  \right) g^{\mu \nu} \tilde{\nabla}_\mu \psi \tilde{\nabla}_{\nu}\psi \nonumber \\
&\qquad -\frac{1}{2} \left( \partial_{r^\star} f + f\frac{1}{\sqrt{g}} \partial_{r^\star} \sqrt{g} \right) \mathcal{V} \psi^2.
\end{align}
Noting that the determinant $\sqrt{g} = \frac{\Delta}{r^2+a^2} \Sigma \sin \theta=g_{r^\star r^\star} \left(r^2+a^2\right) \sin \theta$ for Kerr in Boyer--Lindquist $(t,r^\star,\theta,\phi)$,
\begin{align}
{}^{(X)}\pi^{\mu \nu} \tilde{T}_{\mu \nu}
&= \beta^2 g^{r^\star r^\star} \partial_{r^\star} f \left(\partial_{r^\star}(\beta^{-1} \psi) \right)^2 \nonumber \\
&\qquad -\frac{1}{2} \beta^2 f \sum_{\mu,\nu \neq r^\star}  \left[ \partial_{r^\star} (g^{\mu \nu}) -   g^{\mu \nu} g_{r^\star r^\star} \partial_{r^\star} (g^{r^\star r^\star})\right]\partial_\mu (\beta^{-1} \psi) \partial_{\nu} (\beta^{-1} \psi) \nonumber \\
&\qquad-\frac{1}{2} \left( \partial_{r^\star} f + f \frac{2r\Delta}{(r^2+a^2)^2}  \right) g^{\mu \nu} \partial_\mu \psi \partial_{\nu}\psi -\frac{1}{2} \left( \partial_{r^\star} f + f\frac{1}{\sqrt{g}} \partial_{r^\star} \sqrt{g} \right) \mathcal{V} \psi^2 . \nonumber 
\end{align}
We now recall from Proposition \ref{prop:holwar} the definition $\tilde{K}^X [\psi] = {}^{(X)}\pi^{\mu \nu} \tilde{T}_{\mu \nu} [\psi] + X^\nu \tilde{S}_\nu [\psi]$ and compute
\begin{align}
\tilde{K}^X[\psi]
&= \beta^2 g^{r^\star r^\star} \partial_{r^\star} f \left(\partial_{r^\star}(\beta^{-1} \psi) \right)^2 \nonumber \\
&\qquad -\frac{1}{2} f \beta^2 \sum_{\mu,\nu \neq r^\star}  \left[ \partial_{r^\star} (g^{\mu \nu}) -   g^{\mu \nu} g_{r^\star r^\star} \partial_{r^\star} (g^{r^\star r^\star})\right]\partial_\mu (\beta^{-1}\psi) \partial_{\nu} (\beta^{-1} \psi) \nonumber \\
&\qquad-\frac{1}{2} \beta^2 \left( \partial_{r^\star} f + f \frac{2r\Delta}{(r^2+a^2)^2} + f \frac{\partial_{r^\star} (\beta^2)}{\beta^2} \right) g^{\mu \nu} \partial_\mu (\beta^{-1} \psi) \partial_{\nu}(\beta^{-1} \psi) \nonumber \\
&\qquad-\frac{1}{2} f\partial_{r^\star} (\beta^2 \mathcal{V}) (\beta^{-1} \psi)^2 -\frac{1}{2} \left( \partial_{r^\star} f + f\frac{1}{\sqrt{g}} \partial_{r^\star} \sqrt{g} \right) 
\mathcal{V} \psi^2,
\end{align}
which indeed simplifies to~\eqref{target1}.

\section{Energy currents in Minkowski space}
\label{section:appedixcurrents}

In this section the energy currents used to obtain the assumptions of Section~\ref{caseidiscussion} and Section~\ref{rpsection} are described in Minkowski space.  In Section~\ref{subsec:VcurrentMink} a quadruple $(V,w,q,\varpi)$ satisfying the assumptions of Section \ref{caseidiscussion} is introduced, and in 
Section~\ref{subsec:pcurrentMink} a quadruple $(\tilde{V}_{\rm far}, \tilde{w}_{\rm far}, \tilde{q}_{\rm far}, \tilde{\varpi}_{\rm far})$ satisfying the assumptions of 
Section~\ref{rpsection} is introduced.

Recall that, for a given spacetime $(\mathcal{M},g)$ and suitably regular function $\psi \colon \mathcal{M} \to \mathbb{R}$, for a given vector field $V$, a function $w$, a one form $q$ and a two form $\varpi$, the energy current $J^{V,w,q,\varpi}$ takes the form
\[
	J^{V,w,q,\varpi}_\mu [g, \psi] := T_{\mu\nu} [g,\psi] V^\nu + w  \psi \partial_\mu \psi + \psi^2 q_{\mu} + * d \big( \psi^2 \varpi \big)_{\mu}.
\]
Here, for $0\leq k \leq 4$, $* \colon \Lambda^k \mathcal{M} \to \Lambda^{4-k} \mathcal{M}$ denotes the Hodge star operator which satisfies, for all $\alpha, \beta \in \Lambda^k \mathcal{M}$,
\[
	\alpha \wedge * \beta = g(\alpha, \beta) d\mathrm{Vol}_{\mathcal{M}}.
\]
The divergence of $J^{V,w,q,\varpi}$ takes the form
\[
	\nabla^{\mu} J^{V,w,q,\varpi}_\mu [g, \psi]
	=
	K^{V,w,q}[g,\psi ] + H^{V,w} [\psi] \Box_{g} \psi,
\]
where
\begin{align*}
	K^{V,w,q}[g, \psi]
	&
	:=
	\pi^{V}_{\mu\nu} [g]  T^{\mu\nu}[g, \psi]
	+
	\psi \nabla^{\mu} w \nabla_{\mu} \psi
	+
	w \nabla^{\mu} \psi \nabla_{\mu} \psi
	+
	\psi^2 \nabla^{\mu} q_{\mu}
	+
	2 \psi g^{\mu \nu} q_{\mu} \partial_{\nu} \psi,
	\\
	H^{V,w}[\psi]
	&
	:= V^\mu\partial_\mu \psi + w\psi.
\end{align*}

The divergence of a one form $\xi\in \Lambda^1 \mathcal{M}$ can be expressed in terms of $d$ and $*$ by
\[
	\mathrm{Div} \xi = * \, d * \xi.
\]
In particular it follows that, for any two form $\varpi \in \Lambda^2 \mathcal{M}$ and any function $\psi \in C^{\infty}(\mathcal{M})$,
\[
	\mathrm{Div} * d \big( \psi^2 \varpi \big) = - * d d \big( \psi^2 \varpi \big) = 0.
\]
Thus the choice of $\varpi$ in the current $J^{V,w,q,\varpi}_\mu[g, \psi]$ never contributes to the associated $K^{V,w,q}[g, \psi]$ or $H^{V,w}[g, \psi]$.  Moreover $q$ does not contribute to $H^{V,w}[g, \psi]$.

Throughout this section $(\mathcal{M},g_0)$ denotes Minkowski space (see Section \ref{Minkowexam}),
 and
\begin{equation}
\label{newrdef}
	r = \sqrt{(x^1)^2+(x^2)^2+(x^3)^2}
\end{equation}
denotes the standard radial coordinate  and \underline{not}
  the $r$ of~\eqref{rdefhere}.
 (Note that~\eqref{newrdef} was denoted $\tilde{r}$ in Section~\ref{Minkowexam}. Since~\eqref{rdefhere} and~\eqref{newrdef}
 are comparable for large $r$, the associated $r$-weighted coercivity properties will be equivalent.)
Recall the $(t,r,\vartheta, \varphi)$ and the $(u,v,\vartheta, \varphi)$ coordinate systems.  The Minkowski metric takes the form
\[
	g_0
	=
	-dt^2 + dr^2 + r^2( d \vartheta^2 + \sin^2 \vartheta d \varphi^2)
	=
	- dudv + r^2( d \vartheta^2 + \sin^2 \vartheta d \varphi^2).
\]
The following basic properties of Minkowski space are used.  The spacetime volume form of Minkowski space can be written
\[
	d\mathrm{Vol}_{\mathcal{M}}
	=
	r^2 \sin \vartheta \ d t \wedge dr \wedge d \vartheta \wedge d \varphi
	=
	\frac{1}{2} r^2 \sin \vartheta \ d u \wedge dv \wedge d \vartheta \wedge d \varphi,
\]
and so the Hodge star operator $*$ in particular satisfies
\begin{equation} \label{eq:starMink1}
	* \left( r^2 \sin \vartheta \ d r \wedge d \vartheta \wedge d \varphi \right)
	=
	-
	dt,
	\qquad
	* \left( r^2 \sin \vartheta \ d t \wedge d \vartheta \wedge d \varphi \right)
	=
	-
	dr,
\end{equation}
\begin{equation} \label{eq:starMink2}
	* \left( r^2 \sin \vartheta \ d u \wedge d \vartheta \wedge d \varphi \right)
	=
	du,
	\qquad
	* \left( r^2 \sin \vartheta \ d v \wedge d \vartheta \wedge d \varphi \right)
	=
	- dv.
\end{equation}
The normals to the hypersurfaces $\Sigma(\tau)$ and $\underline{C}_{v}$ take the form
\begin{equation} \label{eq:Minknormals}
	n_{\Sigma_{\tau}} = \partial_v + \iota_{r\le R} \partial_u,
	\qquad
	n_{\underline{C}_{v}} = \partial_u.
\end{equation}

\subsection{The $J^{V,w,q,\varpi}$ current}
\label{subsec:VcurrentMink}

In the case that $(\mathcal{M},g_0)$ is Minkowski space, the tuple $(V,w,q,\varpi)$ of Section \ref{caseidiscussion} can be defined as follows.
Consider
\[
	V_1 = T = \partial_u + \partial_v,
	\qquad
	w_1 = 0,
	\qquad
	q_1 = 0,
	\qquad
	\varpi_1=0,
\]
\[
	V_2 = \delta_1 T = \delta_1(\partial_u + \partial_v),
	\qquad
	w_2 =  0,
	\qquad
	q_2 = 0,
	\qquad
	\varpi_2 = -\frac{\delta_1}{2} r^{-1} r^2 \sin \vartheta d\vartheta \wedge d \varphi,
\]
\[
	V_3 = \frac{\delta_2}{2} \left( 1 - \frac{\delta_3}{(1+r)^{\delta}} \right) ( \partial_v - \partial_u),
	\qquad
	w_3 =  \frac{\delta_2}{r} \left( 1 - \frac{\delta_3}{(1+r)^{\delta}} \right),
	\qquad
	(q_3)_{\mu} = - \frac{\partial_{x^{\mu}}w_3}{2},
	\qquad
	\varpi_3=0,
\]
for appropriate
\begin{equation} \label{eq:delta23}
	0 < \delta_3 \ll \delta_2 \ll \delta_1 \ll 1,
\end{equation}
and define
\begin{equation} \label{eq:VwqvarpiMink}
	V = \sum_{i=1}^3 V_i,
	\qquad
	w = \sum_{i=1}^3 w_i,
	\qquad
	q = \sum_{i=1}^3 q_i,
	\qquad
	\varpi = \sum_{i=1}^3 \varpi_i,
\end{equation}
and then, for a given $\psi$, define currents $J^{V_i,w_i,q_i,\varpi_i}_{\mu}[\psi]$ and $J^{V,w,q,\varpi}_{\mu} [\psi]$ by \eqref{generalJdef}. 
These satisfy the boundedness properties~\eqref{seedboundedness}. 
 (Note that the current $J^{V_2,w_2,q_2,\alpha_2}_{\mu}[g_0,\psi]$ can be viewed as arising from contracting the twisted energy momentum tensor, defined in \eqref{twistedT}, with the Killing vector field $\delta_1 T$.)

\begin{proposition}[The $J^{V,w,q,\varpi}$ current in Minkowski space] \label{prop:VcurrentMink}
	With $(V,w,q,\varpi)$ defined as above, if $\delta_1$, $\delta_2$ and $\delta_3$ are chosen according to \eqref{eq:delta23}, the current $J^{V,w,q,\varpi}_\mu [\psi]$ satisfies the coercivity relations \eqref{insymbolsiboundary} and the corresponding $K^{V,w,q} [\psi]$ satisfies the coercivity relation \eqref{insymbolsi}.  More precisely,
	\begin{align}
		J^{V,w,q,\varpi}_{\mu}[\psi] n_{\Sigma_{\tau}}^{\mu}
		&
		\gtrsim
		(\partial_v \psi)^2
		+
		|\slashed\nabla \psi|^2
		+
		r^{-2} \vert \partial_v (r\psi) \vert^2
		+
		\iota_{r\le R} \left( ( \partial_u \psi)^2 + r^{-2} \vert \partial_u (r\psi) \vert^2 \right)
		+
		(1+r)^{-2} \psi^2,
		\label{eq:VcurrentMink1}
		\\
		J^{V,w,q,\varpi}_{\mu}[\psi] n_{\underline{C}_{v}}^{\mu}
		&
		\gtrsim
		(\partial_u \psi)^2
		+
		r^{-2} \vert \partial_u (r\psi) \vert^2
		+
		|\slashed\nabla \psi|^2
		+
		(1+r)^{-2} \psi^2,
		\label{eq:VcurrentMink2}
		\\
		K^{V,w,q}[\psi]
		&
		\gtrsim
		r^{-1}\vert \nablaslash \psi \vert^2
		+
		(1+r)^{-1-\delta}
		\left(
		(\partial_u \psi)^2 
		+
		(\partial_v \psi)^2 
		\right)
		+
		(1+r)^{-3-\delta} \psi^2.
		\label{eq:VcurrentMink3}
	\end{align}
	Moreover the corresponding $H^{V,w}[\psi]$ satisfies
\begin{equation} \label{eq:HVwMink}
	\vert H^{V,w}[\psi] \vert
	\lesssim
	\vert \partial_u \psi \vert
	+
	\vert \partial_v \psi \vert
	+
	\frac{1}{r} \vert \psi \vert.
\end{equation}
\end{proposition}

\begin{proof}
Recalling \eqref{eq:starMink1}, \eqref{eq:starMink2}, note first that
\[
	* d \big( \psi^2 \varpi_2 \big)
	=
	- \delta_1 \left(\frac{\psi^2}{r^2} + \frac{2}{r} \psi \partial_r \psi \right) dt - \delta_1 \frac{2}{r^2} \psi \partial_t \psi \ x^i dx^i
	=
	\delta_1 \left( - \frac{\psi^2}{2r^2} + \frac{2}{r} \psi \partial_u \psi \right) du
	-
	\delta_1 \left( \frac{\psi^2}{2r^2} + \frac{2}{r} \psi \partial_v \psi \right) dv.
\]
Recall the expressions \eqref{eq:Minknormals} for the normals to the hypersurfaces $\Sigma_{\tau}$ and $\underline{C}_v$.  The fluxes corresponding to $J^{V_i,w_i,q_i,\varpi_i}[\psi]$ satisfy
\[
	J^{V_1,w_1,q_1,\varpi_1}_{\mu}[\psi] n_{\Sigma_{\tau}}^{\mu}
	=
	(\partial_v \psi)^2 + \frac{1}{4} |\slashed\nabla \psi|^2+ \iota_{r\le R}( \partial_u \psi)^2,
	\qquad
	J^{V_1,w_1,q_1,\varpi_1}_{\mu}[\psi] n_{\underline{C}_{v}}^{\mu}
	=
	(\partial_u \psi)^2 + \frac{1}{4} |\slashed\nabla \psi|^2,
\]
\begin{align*}
	J^{V_2,w_2,q_2,\varpi_2}_{\mu}[\psi] n_{\Sigma_{\tau}}^{\mu}
	&
	=
	\frac{\delta_1}{r^2} (\partial_v (r\psi))^2 + \frac{\delta_1}{4r^2} |\slashed\nabla (r\psi)|^2+ \frac{\delta_1}{r^2} \iota_{r\le R}( \partial_u (r\psi))^2
	\\
	&
	=
	\delta_1 J^{V_1,w_1,q_1,\varpi_1}_{\mu}[\psi] n_{\Sigma_{\tau}}^{\mu}
	+
	\delta_1 \left(
	\frac{1}{4r^2} \psi^2
	+
	\frac{1}{r} \psi \partial_v \psi
	+
	\frac{1}{4r^2} \psi^2 \iota_{r\le R}
	-
	\frac{1}{r} \psi \partial_u \psi \iota_{r\le R}
	\right),
	\\
	J^{V_2,w_2,q_2,\varpi_2}_{\mu}[\psi] n_{\underline{C}_{v}}^{\mu}
	&
	=
	\frac{\delta_1}{4r^2} (\partial_u (r\psi))^2 + \frac{\delta_1}{8r^2} |\slashed\nabla (r\psi)|^2
	\\
	&
	=
	\delta_1 J^{V_1,w_1,q_1,\varpi_1}_{\mu}[\psi] n_{\underline{C}_{v}}^{\mu}
	+
	\delta_1 \left(
	\frac{1}{4r^2} \psi^2
	-
	\frac{1}{r} \psi \partial_u \psi
	\right),
	\\
	\vert J^{V_3,w_3,q_3,\varpi_3}_{\mu}[\psi] n_{\Sigma_{\tau}}^{\mu} \vert
	&
	\leq
	\delta_2 C ( (\partial_v \psi)^2 + |\slashed\nabla \psi|^2 + (1+r)^{-2} \vert \psi\vert^2 + \iota_{r\le R}( \partial_u \psi)^2 ),
	\\
	\vert J^{V_3,w_3,q_3,\varpi_3}_{\mu}[\psi] n_{\underline{C}_{v}}^{\mu} \vert
	&
	\leq
	\delta_2 C ( (\partial_u \psi)^2 + |\slashed\nabla \psi|^2 + (1+r)^{-2} \vert \psi\vert^2 ),
\end{align*}
and the bulk terms satisfy
\[
	K^{V_1,w_1,q_1}[\psi] = 0,
	\qquad
	H^{V_1,w_1}[\psi] 
	=
	(\partial_u \psi + \partial_v \psi),
\]
\[
	K^{V_2,w_2,q_2}[\psi] = 0,
	\qquad
	H^{V_2,w_2}[\psi] 
	=
	\delta_1(\partial_u \psi + \partial_v \psi),
\]
\begin{align*}
	K^{V_3,w_3,q_3}[\psi] 
	=
	\
	&
	\delta_2 \left( \frac{1}{r} - \delta_3
	\Big(
	\frac{1}{r(1+r)^{\delta}}
	+
	\frac{\delta}{2(1+r)^{1+\delta}} 
	\Big)
	\right)\vert \nablaslash \psi \vert^2
	\\
	&
	+
	\frac{\delta_2 \delta_3 \delta}{2(1+r)^{1+\delta}}
	\left(
	(\partial_t \psi)^2 
	+
	(\partial_r \psi)^2 
	\right)
	+
	\frac{\delta_2 \delta_3\delta(1+\delta)}{2r(1+r)^{2+\delta}} \psi^2
	,
	\\
	H^{V_3,w_3}[\psi] 
	=
	\
	&
	\frac{\delta_2}{2} \left( 1 - \frac{\delta_3}{(1+r)^{\delta}} \right) ( \partial_v \psi - \partial_u \psi)
	+
	\frac{\delta_2}{r} \left( 1 - \frac{\delta_3}{(1+r)^{\delta}} \right) \psi.
\end{align*}
It thus follows that the coercivity relations \eqref{eq:VcurrentMink1}--\eqref{eq:VcurrentMink3}, along with the property \eqref{eq:HVwMink}, hold if $\delta_1$, $\delta_2$ and $\delta_3$ are chosen according to \eqref{eq:delta23}.
\end{proof}

\subsection{The $J_{\rm far}$ current}
\label{subsec:pcurrentMink}

The quadruples $(\accentset{\scalebox{.6}{\mbox{\tiny $(p)$}}}{{V}}^{\circ}_{\rm far}, \accentset{\scalebox{.6}{\mbox{\tiny $(p)$}}}{{w}}^{\circ}_{\rm far}, \accentset{\scalebox{.6}{\mbox{\tiny $(p)$}}}{{q}}^{\circ}_{\rm far}, \accentset{\scalebox{.6}{\mbox{\tiny $(p)$}}}{{\varpi}}^{\circ}_{\rm far})$ and $(\tilde{V}_{\rm far}, \tilde{w}_{\rm far}, \tilde{q}_{\rm far}, \tilde{\varpi}_{\rm far})$ of Section \ref{rpsection} are defined as follows.  Consider some
\begin{equation} \label{eq:delta6}
	0 <  \delta_6 \ll \delta_5 \ll \delta_4 \ll \delta_3 \ll \delta_2 \ll \delta_1 \ll 1.
\end{equation}
Noting that $-\frac{1}{2} \partial_{x^{\mu}} u dx^{\mu} = - \frac{1}{2} du = (\partial_v)^{\flat}$, define
\[
	\accentset{\scalebox{.6}{\mbox{\tiny $(p)$}}}{V}^{\circ}_{\rm far} = r^p \partial_v,
	\qquad
	\accentset{\scalebox{.6}{\mbox{\tiny $(p)$}}}{w}^{\circ}_{\rm far}
	= 
	\frac{r^{p-1}}{2},
	\qquad
	(\accentset{\scalebox{.6}{\mbox{\tiny $(p)$}}}{q}^{\circ}_{\rm far})_{\mu}
	=
	- \frac{\partial_{x^{\mu}} \accentset{\scalebox{.6}{\mbox{\tiny $(p)$}}}{w}^{\circ}_{\rm far}}{2}
	-
	\Big(
	\frac{p}{4}
	-
	\frac{\delta_4}{2} 
	\Big)
	r^{p-2} \partial_{x^{\mu}} u,
\]
\[
	\accentset{\scalebox{.6}{\mbox{\tiny $(p)$}}}{\varpi}^{\circ}_{\rm far}
	= 
	\Big(
	-
	\frac{r^{p-1}}{4}
	+
	2 \delta_5 r^{\frac{p}{2}-1} 
	\Big)
	r^2 \sin \vartheta d \vartheta \wedge d \varphi.
\]
Define then
\[
	\tilde{V}_{\rm far}
	=
	\frac{1}{\delta_6}V,
	\qquad
	\tilde{w}_{\rm far}
	=
	\frac{1}{\delta_6}w,
	\qquad
	\tilde{q}_{\rm far}
	=
	\frac{1}{\delta_6} q,
	\qquad
	\tilde{\varpi}_{\rm far}
	=
	\frac{1}{\delta_6} \varpi,
\]
where $V$, $w$, $q$, $\varpi$ are defined by \eqref{eq:VwqvarpiMink}, so that
\[
	\accentset{\scalebox{.6}{\mbox{\tiny $(p)$}}}V_{\rm far}
	=
	\accentset{\scalebox{.6}{\mbox{\tiny $(p)$}}}{V}^{\circ}_{\rm far} + \frac{1}{\delta_6}V,
	\qquad
	\accentset{\scalebox{.6}{\mbox{\tiny $(p)$}}}w_{\rm far}
	=
	\accentset{\scalebox{.6}{\mbox{\tiny $(p)$}}}{w}^{\circ}_{\rm far} + \frac{1}{\delta_6}w,
	\qquad
	\accentset{\scalebox{.6}{\mbox{\tiny $(p)$}}}q_{\rm far}
	=
	\accentset{\scalebox{.6}{\mbox{\tiny $(p)$}}}{q}^{\circ}_{\rm far} + \frac{1}{\delta_6}q,
	\qquad
	\accentset{\scalebox{.6}{\mbox{\tiny $(p)$}}}\varpi_{\rm far}
	=
	\accentset{\scalebox{.6}{\mbox{\tiny $(p)$}}}{\varpi}^{\circ}_{\rm far} + \frac{1}{\delta_6}\varpi.
\]

\begin{proposition}[The $J_{\rm far}$ current in Minkowski space] \label{prop:pcurrentMink}
	With $(\accentset{\scalebox{.6}{\mbox{\tiny $(p)$}}}V_{\rm far}, \accentset{\scalebox{.6}{\mbox{\tiny $(p)$}}}w_{\rm far}, \accentset{\scalebox{.6}{\mbox{\tiny $(p)$}}}q_{\rm far},\accentset{\scalebox{.6}{\mbox{\tiny $(p)$}}}\varpi_{\rm far})$ defined as above, if $\delta_1,\ldots,\delta_6$ are chosen according to \eqref{eq:delta6}, the associated currents $\Jp_{\rm far}$, $\Kp_{\rm far}$, defined by \eqref{generalJdef},~\eqref{generalKdef} satisfy the weighted bulk coercivity properties \eqref{insymbolsiiwithpweightKfar} and the weighted boundary coercivity property \eqref{insymbolsiiwithpweightJfar}.  More precisely, for $\delta \leq p \leq 2-\delta$ and $r\geq R$,
	\begin{align}
		\Jp_{\rm far}{}_{\mu}[\psi] n_{\Sigma_{\tau}}^{\mu}
		&
		\gtrsim
		r^{p-2} \vert \partial_v (r \psi) \vert^2
		+
		r^{\frac{p}{2}} (\partial_v \psi)^2
		+
		\vert \nablaslash \psi \vert^2
		+
		r^{\frac{p}{2} - 2} \psi^2,
		\label{eq:pcurrentMink1}
		\\
		\Jp_{\rm far}{}_{\mu}[\psi] n_{\underline{C}_{v}}^{\mu}
		&
		\gtrsim
		(\partial_u \psi)^2
		+
		r^{p} \vert \nablaslash \psi \vert^2
		+
		r^{p-2} \psi^2,
		\label{eq:pcurrentMink2}
		\\
		\Kp_{\rm far}[\psi]
		&
		\gtrsim
		r^{p-3} \vert \partial_v (r \psi) \vert^2
		+
		r^{p-1} \vert \partial_v \psi \vert^2
		+
		r^{p-1} \vert \nablaslash \psi \vert^2
		+
		r^{-1-\delta} \vert \partial_u \psi \vert^2
		+
		r^{p-3} \psi^2.
		\label{eq:pcurrentMink3}
	\end{align}
	Moreover the corresponding $\Hap_{\rm far} [\psi]$ satisfies
\begin{equation} \label{eq:HfarMink}
	\vert \Hap_{\rm far} [\psi] \vert
	\lesssim
	r^{p-1} \vert \partial_v( r \psi) \vert
	+
	\vert \partial_u \psi \vert
	+
	\vert \partial_v \psi \vert
	+
	\frac{1}{r} \vert \psi \vert.
\end{equation}
\end{proposition}

\begin{proof}
Recall again \eqref{eq:starMink2} and note that
\begin{align*}
	* \, d (\psi^2 \accentset{\scalebox{.6}{\mbox{\tiny $(p)$}}}{\varpi}^{\circ}_{\rm far})
	=
	\
	&
	\left(
	\Big( \frac{p+1}{8} r^{p-2} - \delta_5 \frac{p+2}{2} r^{\frac{p}{2} - 2} \Big) \psi^2 
	- 
	\Big( \frac{r^{p-1}}{2} - 4 \delta_5 r^{\frac{p}{2} - 1} \Big) \psi \partial_u \psi
	\right) du
	\\
	&
	+
	\left(
	\Big( \frac{p+1}{8} r^{p-2} - \delta_5 \frac{p+2}{2} r^{\frac{p}{2} - 2} \Big) \psi^2
	+
	\Big( \frac{r^{p-1}}{2} - 4 \delta_5 r^{\frac{p}{2} - 1} \Big) \psi \partial_v \psi
	\right) dv.
\end{align*}

Note moreover that
\[
	(\pi^{\accentset{\scalebox{.6}{\mbox{\tiny $(p)$}}}{V}^{\circ}_{\rm far}})^{\sharp \sharp}
	=
	p r^{p-1} \partial_v \otimes \partial_v
	-
	\frac{pr^{p-1}}{2} \big( \partial_u \otimes \partial_v + \partial_v \otimes \partial_u \big)
	+
	\frac{r^{p-3}}{2} \big( \partial_{\vartheta} \otimes \partial_{\vartheta} + \sin^{-2} \vartheta \partial_{\varphi} \otimes \partial_{\varphi} \big)
	,
\]
and
\[
	\nabla^{\mu} (\accentset{\scalebox{.6}{\mbox{\tiny $(p)$}}}{q}^{\circ}_{\rm far})_{\mu}
	=
	\frac{p r^{p-3}}{4}
	-
	\frac{\delta_4 p r^{p-3}}{2}.
\]
The fluxes corresponding to $J^{\accentset{\scalebox{.6}{\mbox{\tiny $(p)$}}}{V}^{\circ}_{\rm far},\accentset{\scalebox{.6}{\mbox{\tiny $(p)$}}}{w}^{\circ}_{\rm far},\accentset{\scalebox{.6}{\mbox{\tiny $(p)$}}}{q}^{\circ}_{\rm far},\accentset{\scalebox{.6}{\mbox{\tiny $(p)$}}}{\varpi}^{\circ}_{\rm far}}[\psi]$ satisfy, for $r \geq R$,
\begin{align}
	\label{eq:outfluxV4}
	J^{\accentset{\scalebox{.6}{\mbox{\tiny $(p)$}}}{V}^{\circ}_{\rm far},\accentset{\scalebox{.6}{\mbox{\tiny $(p)$}}}{w}^{\circ}_{\rm far},\accentset{\scalebox{.6}{\mbox{\tiny $(p)$}}}{q}^{\circ}_{\rm far},\accentset{\scalebox{.6}{\mbox{\tiny $(p)$}}}{\varpi}^{\circ}_{\rm far}}_{\mu}[\psi] n_{\Sigma_{\tau}}^{\mu}
	&
	=
	r^{p-2} \vert \partial_v (r \psi) \vert^2
	-
	\delta_5
	\left( \frac{p+2}{2} r^{\frac{p}{2} - 2} \psi^2 + 4 r^{\frac{p}{2} - 1} \psi \partial_v \psi \right)
	,
	\\
	&
	=
	r^{p-2} \vert \partial_v (r \psi) \vert^2
	+
	\delta_5
	\left( \frac{(2-p)}{2} r^{\frac{p}{2} - 2} \psi^2 - 4 r^{\frac{p}{2} - 2} \psi \partial_v (r\psi) \right)
	,
	\nonumber
	\\
	\label{eq:influxV4}
	J^{\accentset{\scalebox{.6}{\mbox{\tiny $(p)$}}}{V}^{\circ}_{\rm far},\accentset{\scalebox{.6}{\mbox{\tiny $(p)$}}}{w}^{\circ}_{\rm far},\accentset{\scalebox{.6}{\mbox{\tiny $(p)$}}}{q}^{\circ}_{\rm far},\accentset{\scalebox{.6}{\mbox{\tiny $(p)$}}}{\varpi}^{\circ}_{\rm far}}_{\mu}[\psi] n_{\underline{C}_{v}}^{\mu}
	&
	=
	\frac{r^{p}}{4} \vert \nablaslash \psi \vert^2
	+
	\frac{\delta_4}{2} r^{p-2} \psi^2
	-
	\delta_5
	\left( \frac{p+2}{2} r^{\frac{p}{2} - 2} \psi^2 - 4 r^{\frac{p}{2} - 1} \psi \partial_u \psi \right),
\end{align}
and the bulk terms satisfy
\begin{align}
	K^{\accentset{\scalebox{.6}{\mbox{\tiny $(p)$}}}{V}^{\circ}_{\rm far},\accentset{\scalebox{.6}{\mbox{\tiny $(p)$}}}{w}^{\circ}_{\rm far},\accentset{\scalebox{.6}{\mbox{\tiny $(p)$}}}{q}^{\circ}_{\rm far}}[\psi]
	=
	\
	&
	pr^{p-3} \vert \partial_v (r \psi) \vert^2
	+
	\frac{(2-p)r^{p-1}}{4} \vert \nablaslash \psi \vert^2
	+
	\frac{\delta_4 (2-p)}{2} r^{p-3} \psi^2
	- 
	2 \delta_4 r^{p-3} \psi \partial_v (r\psi)
	\label{eq:bulkV4}
	\\
	H^{\accentset{\scalebox{.6}{\mbox{\tiny $(p)$}}}{V}^{\circ}_{\rm far},\accentset{\scalebox{.6}{\mbox{\tiny $(p)$}}}{w}^{\circ}_{\rm far}}[\psi] 
	=
	\
	&
	r^{p-1} \partial_v (r \psi).
\end{align}
In particular, if $\delta_4$ and $\delta_5$ are chosen according to \eqref{eq:delta6} (depending on $p$), then
\begin{align}
	\label{eq:V4coercivity1}
	J^{\accentset{\scalebox{.6}{\mbox{\tiny $(p)$}}}{V}^{\circ}_{\rm far},\accentset{\scalebox{.6}{\mbox{\tiny $(p)$}}}{w}^{\circ}_{\rm far},\accentset{\scalebox{.6}{\mbox{\tiny $(p)$}}}{q}^{\circ}_{\rm far},\accentset{\scalebox{.6}{\mbox{\tiny $(p)$}}}{\varpi}^{\circ}_{\rm far}}_{\mu}[\psi] n_{\Sigma_{\tau}}^{\mu}
	&
	\gtrsim
	r^{p-2} \vert \partial_v (r \psi) \vert^2
	+
	r^{\frac{p}{2}} (\partial_v \psi)^2
	+
	r^{\frac{p}{2} - 2} \psi^2,
	\\
	\label{eq:V4coercivity2}
	J^{\accentset{\scalebox{.6}{\mbox{\tiny $(p)$}}}{V}^{\circ}_{\rm far},\accentset{\scalebox{.6}{\mbox{\tiny $(p)$}}}{w}^{\circ}_{\rm far},\accentset{\scalebox{.6}{\mbox{\tiny $(p)$}}}{q}^{\circ}_{\rm far},\accentset{\scalebox{.6}{\mbox{\tiny $(p)$}}}{\varpi}^{\circ}_{\rm far}}_{\mu}[\psi] n_{\underline{C}_{v}}^{\mu}
	+
	(\partial_u \psi)^2
	&
	\gtrsim
	r^{p} \vert \nablaslash \psi \vert^2
	+
	r^{p-2} \psi^2,
\end{align}
and
\begin{align}
	\label{eq:V4coercivity3}
	K^{\accentset{\scalebox{.6}{\mbox{\tiny $(p)$}}}{V}^{\circ}_{\rm far},\accentset{\scalebox{.6}{\mbox{\tiny $(p)$}}}{w}^{\circ}_{\rm far},\accentset{\scalebox{.6}{\mbox{\tiny $(p)$}}}{q}^{\circ}_{\rm far}}[\psi]
	&
	\gtrsim
	p r^{p-3} \vert \partial_v (r \psi) \vert^2
	+
	(2-p)r^{p-1} \vert \nablaslash \psi \vert^2
	+
	p(2-p)r^{p-1} \vert \partial_v \psi \vert^2 + (2-p) r^{p-3} \psi^2.
\end{align}

It follows from Proposition \ref{prop:VcurrentMink} that the currents $\Jp_{\rm far}$, $\Kp_{\rm far}$ satisfy the weighted bulk coercivity properties \eqref{eq:pcurrentMink1}--\eqref{eq:pcurrentMink3} provided $\delta_1,\ldots,\delta_6$ are chosen according to \eqref{eq:delta6}, and that $\Hap_{\rm far}$ satisfies \eqref{eq:HfarMink}.
\end{proof}

The inequalities \eqref{eq:HVwMink} and \eqref{eq:HfarMink} in particular mean that the $r^{-1} \vert \nablaslash \psi \vert$ term on the left hand side of the assumed inequality~\eqref{nullcondassump} is superfluous in the case that $(\mathcal{M},g_0)$ is Minkowski space.  This term is estimated in the proof of Proposition~\ref{itholdsthenull} nonetheless in order to illustrate how it is estimated in the case of Kerr.

\section{Verifying the null condition assumption}
\label{nonlineartermsatinf}

In this section the proof of Proposition~\ref{itholdsthenull} is given.  Proposition~\ref{itholdsthenull} can be more precisely stated as follows.

\begin{proposition} \label{prop:appendixbmain}
	Assumption~\ref{nullcondassumphere} holds for the classical null condition of Klainerman~\cite{KlNull} on Minkowski space and more generally the class of equations on Kerr considered in Luk~\cite{MR3082240}.

More precisely, 
  if $(\mathcal{M},g_0)$ is either Minkowski space or a member of the Kerr family and if
	\[
		F = N^{\mu\nu}(\psi, x) \partial_\mu \psi \partial_\nu \psi
	\]
	satisfies the assumption~\eqref{assumponNzero}, then there exists $k_{\rm null}>0$, and for all $k\ge k_{\rm null}$, there exists 
  $\varepsilon_{\rm null}>0$, such that we  have the following.
  
Let $\psi$ be a smooth function in $\mathcal{R}(\tau_0,\tau_1,v)$ 
satisfying~\eqref{basicbootstrapinnullcond}  for
$0<\varepsilon\le \varepsilon_{\rm null}$.  
Then for all  $\delta \leq p \leq 2 - \delta$,
	\begin{multline*}
		\sum_{|{\bf k}|\le k} \int_{\mathcal{R}(\tau_0,\tau',v)\cap \{ r \geq R \} }
		\Big(
		\vert \underline{L} \Dk \psi \vert
		+
		\vert L \Dk \psi \vert
		+
		\frac{1}{r} \vert \nablaslash \Dk \psi \vert
		+
		\frac{1}{r} \vert \Dk \psi \vert
		\Big)
		\left| \mathfrak{D}^{\bf k} F \right|
		\\
		+
		\int_{\mathcal{R}(\tau_0,\tau',v)\cap \{ r \geq R \}}(\mathfrak{D}^{\bf k} F)^2
		\lesssim
		\sqrt{\Xzeropluslesslessk_{\frac{8R}9}(\tau_0,\tau_1)}
		\Xzeroplusk_{R}(\tau_0,\tau_1)
		,
	\end{multline*}
	and moreover,
	\begin{multline*}
		\sum_{|{\bf k}|\le k}\int_{\mathcal{R}(\tau_0,\tau',v) \cap \{ r \geq R \} }
		\Big(
		r^{p-1}\vert L (r\Dk \psi) \vert
		+
		\vert \underline{L} \Dk \psi \vert
		+
		\vert L \Dk \psi \vert
		+
		\frac{1}{r} \vert \nablaslash \Dk \psi \vert
		+
		\frac{1}{r} \vert \Dk \psi \vert
		\Big)
		\left|\mathfrak{D}^{\bf k}F\right|
		\\
		+
		\int_{\mathcal{R}(\tau_0,\tau',v)\cap \{ r \geq R \} } (\mathfrak{D}^{\bf k} F)^2
		\lesssim
		\sqrt{\Xzerolesslessk_{\frac{8R}9}(\tau_0,\tau_1)}\Xpk_{R}(\tau_0,\tau_1)
		+
		\sqrt{\Xplesslessk_{R}(\tau_0,\tau_1)}\sqrt{\Xpk_{R}(\tau_0,\tau_1)}\sqrt{\Xzerok_{R}(\tau_0,\tau_1)}.
	\end{multline*}
\end{proposition}

The estimate of Proposition \ref{prop:appendixbmain} is only nontrivial for large $r$ and, thus, the main content appears already around Minkowski space.  Accordingly, the proof will be given in detail for the case that $(\mathcal{M},g_0)$ is Minkowski space, for the two model nonlinearities
\[
	N^{\mu\nu}(\psi, x) \partial_\mu \psi \partial_\nu \psi
	= 
	\partial_u \psi \partial_v \psi,
	\qquad
	N^{\mu\nu}(\psi, x) \partial_\mu \psi \partial_\nu \psi
	= 
	\slashed{g}^{AB} \nablaslash_A \psi \nablaslash_B \psi,
\]
where $\slashed{g} = r^2\gamma$ is the induced round metric and $\nablaslash$ its associated connection.  See Section \ref{subsec:Minkowskispace}.  The more general nonlinearities and the Kerr case are discussed in Section \ref{section:Kerrcase}.

\subsection{Two model nonlinearities on Minkowski space}
\label{subsec:Minkowskispace}

Throughout this section $(\mathcal{M},g_0)$ denotes Minkowski space (see Section \ref{Minkowexam}).
Note that in the region $r\ge \frac{8R}9$ which will be considered here,
the $r$ of~\eqref{rdefhere} coincides with the standard radial coordinate~\eqref{newrdef} used
 in Appendix~\ref{section:appedixcurrents} and thus coincides
with what was denoted by as $\tilde{r}$ in Section~\ref{Minkowexam}.
 
 Recall the $(u,v,\theta, \phi)$ coordinate system.  In the region $r \geq R$ the null vector fields take the form
\[
	L = \partial_v, \qquad \underline{L} = \partial_u.
\]
The spheres $\Sigma(\tau) \cap \underline{C}_v$ are round, $\slashed{g} = r^2\gamma$ denotes the induced round metric and $\nablaslash$ its associated connection.
Define
\[
	\Fstar := \sup_{\tau_0 \leq \tau \leq \tau_1} \accentset{\scalebox{.6}{\mbox{\tiny $(0)$}}}{\underaccent{\scalebox{.8}{\mbox{\tiny $k$}}}{\mathcal{F}}}(\tau),
	\qquad
	\Epstar := \sup_{\tau_0 \leq \tau \leq \tau_1} \accentset{\scalebox{.6}{\mbox{\tiny $(p)$}}}{\underaccent{\scalebox{.8}{\mbox{\tiny $k$}}}{\mathcal{E}}}(\tau),
	\qquad
	\intEpone' := \int_{\tau_0}^{\tau_1}\accentset{\scalebox{.6}{\mbox{\tiny $(p-1)$}}}{\underaccent{\scalebox{.8}{\mbox{\tiny $k$}}}{\mathcal{E}}}'(\tau) d \tau,
\]
and similarly for $\ll k$ replacing $k$, and with $R$ subscripts added, etc.  This notation will be used throughout this section.

The main result of this section is the following.

\begin{proposition}[Model nonlinearities on Minkowski space] \label{prop:appendixBmainMink}

Fix $\delta \leq p \leq 2-\delta$ and  $k \geq 7$ and let 
 $\psi$ be a smooth function in $\mathcal{R}(\tau_0,\tau_1,v)\cap \{r\ge 8R/9\}$
 satisfying~\eqref{basicbootstrapinnullcond}.
 
Then the nonlinearities
\[ 
		F = \partial_u \psi \partial_v \psi
		\qquad
		\text{and}
		\qquad
		F
		= 
		\slashed{g}^{AB} \nablaslash_A \psi \nablaslash_B \psi
	\]
both satisfy
	\begin{multline*}
		\sum _{|{\bf k}|\le k} \int_{\mathcal{R}(\tau_0,\tau',v)\cap \{ r \geq R \} }
		\Big(
		\vert \partial_u \Dk \psi \vert
		+
		\vert \partial_v \Dk \psi \vert
		+
		\frac{1}{r} \vert \nablaslash \Dk \psi \vert
		+
		\frac{1}{r} \vert \Dk \psi \vert
		\Big)
		\left|\mathfrak{D}^{\bf k} F\right|
		+
		\int_{\mathcal{R}(\tau_0,\tau',v) \cap \{ r \geq R \}}(\mathfrak{D}^{\bf k} F)^2
		\\
		\lesssim
		\sqrt{\Fstarll_R +\Ezerostarll_{\frac{8R}{9}}
		+ \intEdeltall_R' \ }
		\Big(
		\Fstar_R
		+
		\Ezerostar_R
		+
		\int^*\accentset{\scalebox{.6}{\mbox{\tiny $(\delta-1)$}}}{\underaccent{\scalebox{.8}{\mbox{\tiny $k$}}}{\mathcal{E}}}_R'
		\
		\Big)
		+
		\sqrt{
		\int^*\accentset{\scalebox{.6}{\mbox{\tiny $(\delta-1)$}}}{\underaccent{\scalebox{.8}{\mbox{\tiny $\ll \mkern-6mu k$}}}{\mathcal{E}}}_R'}
		\
		\bigg(
		\Fstar_R + \intEdelta_R \ 
		\bigg)
	\end{multline*}
	and
	\begin{multline*}
		\sum_{|{\bf k}|\le k} \int_{\mathcal{R}(\tau_0,\tau',v)\cap \{ r \geq R \} }
		\Big(
		r^{p-1}\vert \partial_v (r\Dk \psi) \vert
		+
		\vert \partial_u \Dk \psi \vert
		+
		\vert \partial_v \Dk \psi \vert
		+
		\frac{1}{r} \vert \nablaslash \Dk \psi \vert
		+
		\frac{1}{r} \vert \Dk \psi \vert
		\Big)
		\left|\mathfrak{D}^{\bf k} F\right|
		\\
		+
		\int_{\mathcal{R}(\tau_0,\tau',v)\cap \{ r \geq R \}} (\mathfrak{D}^{\bf k} F)^2
		\lesssim
		\sqrt{\int^*\accentset{\scalebox{.6}{\mbox{\tiny $(p-1)$}}}{\underaccent{\scalebox{.8}{\mbox{\tiny $\ll \mkern-6mu k$}}}{\mathcal{E}}}_R' \ }
		\sqrt{\Fstar_R}
		\sqrt{\intEpone_R' \ }
		+
		\sqrt{\Fstarll_R + \Ezerostarll_{\frac{8R}{9}} + \intEdeltall_R '\ }
		\Bigg(
		\Fstar_R
		+
		\Epstar_R
		+
		\intEpone_R'
		\Bigg)
		.
	\end{multline*}
\end{proposition}

The following properties of Minkowski space will be used in the proof of Proposition \ref{prop:appendixBmainMink}.  First, for any $\bf k$ and any function $\psi$, for $r \geq R$,
\begin{equation} \label{eq:DkdvdvDk}
	\vert \Dk L \psi \vert
	\lesssim
	\sum_{\vert \widetilde{\bf k} \vert \leq \vert \bf k \vert}
	\vert L \mathfrak{D}^{\widetilde{\bf k}} \psi \vert,
	\qquad
	\vert \Dk \underline{L} \psi \vert
	\lesssim
	\sum_{\vert \widetilde{\bf k} \vert \leq \vert \bf k \vert}
	\vert \underline{L} \mathfrak{D}^{\widetilde{\bf k}} \psi \vert
	.
\end{equation}
Indeed, \eqref{eq:DkdvdvDk} follows that $L = \partial_v$, $\underline{L} = \partial_u$, $T = N = \partial_u + \partial_v$, and $[\partial_u,\Omega_i] = [\partial_v,\Omega_i] = 0$ for $i=1,2,3$.  Thus $\partial_u$ and $\partial_v$ commute with all components of $\mathfrak{D}$.  Second, the fact that, for any $\bf k$ and any function $\psi$, for $r \geq R$,
\begin{equation} \label{eq:nablaMinkproperty}
	\vert \Dk ( \gslash^{AB} \nablaslash_A \psi \nablaslash_B \psi ) \vert
	\lesssim
	\sum_{\vert \mathbf{k}_1 \vert + \vert \mathbf{k}_2 \vert \leq \vert \mathbf{k} \vert}
	\vert \nablaslash \mathfrak{D}^{\mathbf{k}_1} \psi \vert
	\vert \nablaslash \mathfrak{D}^{\mathbf{k}_2} \psi \vert,
\end{equation}
will also be used.

In the above and what follows, note that  when the volume form is omitted, the
usual spacetime volume form is to be understood, and this 
contains in particular an omitted $r^2 \sin \vartheta$ factor.

The proof of Proposition \ref{prop:appendixBmainMink} relies on weighted Sobolev inequalities, which are discussed in Section \ref{subsec:Sobolev}.  In Section \ref{subsec:modelnonlin1} the proof of Proposition \ref{prop:appendixBmainMink} for the case of the nonlinearity $F = \partial_u \psi \partial_v \psi$ is given, followed by the proof for the nonlinearity $F = \slashed{g}^{AB} \nablaslash_A \psi \nablaslash_B \psi$ in Section \ref{subsec:modelnonlin2}.

\subsubsection{Weighted Sobolev inequalities}
\label{subsec:Sobolev}

The proof of Proposition~\ref{prop:appendixbmain} relies on certain weighted Sobolev inequalities.

Let $d\omega$ denote the unit volume form on $S^2$
\[
	d \omega = \sin \vartheta \, d\vartheta \, d\varphi.
\]
Recall (see Section \ref{Minkowexam}) that $r(u,v) = \frac{1}{2}(v-u)$.  Define $v(R) = v(R,u) = 2R + u$, so that $r(u,v(R,u)) = R$, and similarly define $u(R,v) = v-2R$, so that $r(u(R,v),v) = R$.

Recall the region $\mathcal{R}(\tau_0,\tau_1,v)$.  In order to avoid the blow up of constants in Sobolev inequalities, the region near the corner $\{u=\tau_0\} \cap \{r=R\}$ will typically be considered separately from its complement.  Accordingly, define
\begin{equation} \label{eq:Rcheck}
	\mathcal{R}_R(\tau_0,\tau_1,v) := \mathcal{R}(\tau_0,\tau_1,v) \cap \{ r \geq R\},
	\qquad
	\check{\mathcal{R}}_R(\tau_0,\tau_1,v) := \mathcal{R}_R(\tau_0,\tau_1,v) \smallsetminus \{ v'\leq v(R,\tau_0 + R)\},
\end{equation}
and
\begin{equation} \label{eq:Rcorner}
	\mathcal{R}_{\mathrm{Corner}}(\tau_0,\tau_1,v) := \mathcal{R}(\tau_0,\tau_1,v) \cap \{ r \geq R\} \cap \{ v'\leq v(R,\tau_0 + R)\}.
\end{equation}
Note that
\begin{equation} \label{eq:Rcornerbound}
	r \leq \frac{3R}{2} \qquad \text{in} \quad \mathcal{R}_{\mathrm{Corner}}(\tau_0,\tau_1,v).
\end{equation}
Define also
\begin{equation} \label{eq:T1T1prime}
	T_1 = T_1(v) = \min\{ \tau_1,u(R,v)\}, \qquad \text{and} \qquad T_1' = T_1(v'),
\end{equation}
and
\[
	\Sigma_{R}(\tau,v) = \Sigma(\tau,v) \cap \{ r \geq R\}.
\]

\begin{proposition}[Sobolev inequality on incoming cones $\underline{C}_v$] \label{prop:Sobolevin}
	For any $\tau_0 \leq u \leq \tau_1$, $v(R,\tau_0+R) \leq v' \leq v$, $(\vartheta,\varphi) \in S^2$ (i.\@e.\@ for any $(u,v',\vartheta,\varphi) \in \check{\mathcal{R}}_R(\tau_0,\tau_1,v)$), and any function $f \colon \mathcal{R}(\tau_0,\tau_1,v) \to \mathbb{R}$,
	\begin{multline*}
		r \vert f (u,v',\vartheta,\varphi) \vert
		\lesssim
		\big( 1 + (\tau_1-\tau_0)^{-\frac{1}{2}} \big)
		\sum_{k \leq 1}
		\Big(
		\int_{\tau_0}^{T_1'} \int_{S^2}
		\big( \vert \underline{L}^{k} f \vert^2
		+
		\sum_{i=1}^3  \vert \Omega_i \underline{L}^{k} f \vert^2
		\\
		+
		\sum_{i,j=1}^3  \vert \Omega_i \Omega_j \underline{L}^{k} f \vert^2
		\big) (u',v',\vartheta,\varphi)
		r^2 d\omega du'
		\Big)^{\frac{1}{2}}.
	\end{multline*}
\end{proposition}

\begin{proof}
	For simplicity, consider first the case that $v'\geq v(R,\tau_1)$, so that $T_1' = \tau_1$.  Suppose first that $u \geq \frac{\tau_1+\tau_2}{2}$.  Let $\phi$ be a smooth cut off function such that $\phi(\tau_1) = 0$, $\phi(\tau) = 1$ for $\tau \geq \frac{\tau_1+\tau_2}{2}$, and
	\[
		\vert \phi'(\tau) \vert \leq \frac{4}{\tau_2-\tau_1},
		\quad
		\text{for }
		\tau_1 \leq \tau \leq \frac{\tau_1+\tau_2}{2}.
	\]
	Now
	\[
		\partial_u \left( \phi (u) \vert f(u,v',\vartheta,\varphi) \vert^2 \right)
		=
		\phi' (u) \vert f(u,v',\vartheta,\varphi) \vert^2
		+
		2 \phi (u) f(u,v',\vartheta,\varphi) \partial_uf(u,v',\vartheta,\varphi).
	\]
	Integrating from $\tau_0$ to $u$ gives
	\[
		r^2 \vert f(u,v',\vartheta,\varphi) \vert^2
		\lesssim
		\big( 1 + (\tau_1-\tau_0)^{-1} \big) \int_{\tau_0}^{\tau_1} \vert f(u',v',\vartheta,\varphi) \vert^2 r^2 du'
		+
		\int_{\tau_0}^{\tau_1} \vert \partial_uf (u',v',\vartheta,\varphi) \vert^2 r^2 du'.
	\]
	The result then follows from the standard Sobolev inequality on $S^2$,
	\[
		\sup_{(\vartheta,\varphi)\in S^2} \vert  f(u',v',\vartheta,\varphi) \vert
		\lesssim
		r^{-1}
		\Big(
		\int_{S^2}
		\big( \vert f \vert^2
		+
		\sum_{i=1}^3  \vert \Omega_i f \vert^2
		+
		\sum_{i,j=1}^3  \vert \Omega_i \Omega_j f \vert^2
		\big) (u',v',\vartheta,\varphi)
		r^2 \sin \vartheta d \vartheta d \varphi
		\Big)^{\frac{1}{2}}.
	\]
	Similarly for $u \leq \frac{\tau_1+\tau_2}{2}$, and similarly for $v(R,\tau_0+R) \leq v'\leq v(R,\tau_1)$.
\end{proof}

Proposition \ref{prop:Sobolevin} in particular implies that
\begin{equation} \label{eq:pointwise1}
	\sup_{\check{\mathcal{R}}_R(\tau_0,\tau_1,v)}
	\sum_{\vert \mathbf{k} \vert \leq k-3}
	\Big(
	r \vert \partial_u \Dk \psi \vert
	+
	r \vert \nablaslash( \Dk \psi) \vert
	\Big)
	\lesssim
	\big( 1 + (\tau_1-\tau_0)^{-\frac{1}{2}} \big) \sqrt{ \Fstar },
\end{equation}
and, for any $s\geq 0$, and any $\tau_0 \leq u \leq \tau_1$, $v(R,\tau_0+R) \leq v' \leq v$, $(\vartheta,\varphi) \in S^2$,
\begin{equation} \label{eq:dvpointwisein}
	r^s \vert \partial_v \Dk \psi (u,v',\vartheta,\varphi) \vert
	\lesssim
	\big( 1 + (\tau_1-\tau_0)^{-\frac{1}{2}} \big)
	\sum_{\vert \widetilde{\mathbf{k}} \vert \leq \vert \mathbf{k} \vert+3}
	\Big(
	\int_{\tau_0}^{T_1'} \int_{S^2}
	r^{s-2}
	\vert \partial_v \Dk \psi (u',v',\vartheta,\varphi) \vert
	r^2 d\omega du'
	\Big)^{\frac{1}{2}}.
\end{equation}
It is in the forms \eqref{eq:pointwise1} and \eqref{eq:dvpointwisein} that Proposition \ref{prop:Sobolevin} will be used below.

\begin{proposition}[Sobolev inequality on $\Sigma(\tau)$] \label{prop:Sobolevout}
	For any $\tau_0 \leq \tau \leq \tau_1$, $v(R,u) \leq v' \leq v$, $(\vartheta,\varphi) \in S^2$, and any function $f \colon \mathcal{R}(\tau_0,\tau_1,v) \to \mathbb{R}$,
	\begin{multline*}
		\vert f (\tau,v',\vartheta,\varphi) \vert
		\lesssim
		\sum_{k =0}^1
		\Big(
		\int_{\Sigma_{R}(\tau,v)}
		r^{-2}
		\big( \vert L^{k} f \vert^2
		+
		\sum_{i=1}^3  \vert \Omega_i L^{k} f \vert^2
		+
		\sum_{i,j=1}^3 \vert \Omega_i \Omega_j L^{k} f \vert^2
		\big)
		r^2 d\omega dv'
		\Big)^{\frac{1}{2}}
		\\
		+
		\Big(
		\int_{\Sigma_{\frac{8R}{9}}(\tau,v)}
		r^{-2}
		\vert f \vert^2
		r^2 d\omega dv'
		\Big)^{\frac{1}{2}}.
	\end{multline*}
\end{proposition}

\begin{proof}
	The proof follows from the fundamental theorem of calculus and the standard Sobolev inequality on $S^2$, as in the proof of Proposition \ref{prop:Sobolevin}.
\end{proof}

Proposition \ref{prop:Sobolevout} in particular implies that
\begin{equation} \label{eq:pointwiseout}
	\sup_{\mathcal{R}_{\mathrm{Corner}}(\tau_0,\tau_1,v)}
	\sum_{\vert \mathbf{k} \vert \leq k-4}
	\vert \Dk \psi \vert
	\lesssim
	\sqrt{ \Ezerostar_{\frac{8R}{9}} },
\end{equation}
where the region $\mathcal{R}_{\mathrm{Corner}}$ is defined in \eqref{eq:Rcorner}.  It is in the form \eqref{eq:pointwiseout} that Proposition \ref{prop:Sobolevout} will typically be used.

\subsubsection{The proof of Proposition \ref{prop:appendixBmainMink} for the nonlinearity $F = \partial_u \psi \partial_v \psi$}
\label{subsec:modelnonlin1}

Proposition \ref{prop:appendixBmainMink} for the nonlinearity
\[
	F = \partial_u \psi \partial_v \psi,
\] 
follows from the following proposition, whose proof is the subject of the present section.

\begin{proposition}[Nonlinear error estimates for $F = \partial_u \psi \partial_v \psi$ on Minkowski space] \label{prop:Minkowskinonlin1}
	Under the assumptions of Proposition \ref{prop:appendixBmainMink}, for each $\vert \mathbf{k} \vert \leq k$,
	\begin{multline*}
		\sum_{\vert \mathbf{k}_1 \vert + \vert \mathbf{k}_2 \vert \leq \vert \mathbf{k} \vert}
		\Bigg(
		\int_{\mathcal{R}(\tau_0,\tau_1,v)\cap \{ r \geq R \}} 
		\Big(
		\vert \partial_u \Dk \psi \vert
		+
		\vert \partial_v \Dk \psi \vert
		+
		\frac{1}{r} \vert \nablaslash \Dk \psi \vert
		+
		\frac{1}{r} \vert \Dk \psi \vert
		\Big)
		\cdot 
		\vert \Dkone \partial_u \psi \vert
		\cdot
		\vert \Dktwo \partial_v \psi \vert
		\\
		+
		\int_{\mathcal{R}(\tau_0,\tau_1,v)\cap \{ r \geq R \}}
		\vert \Dkone \partial_u \psi \vert^2
		\cdot
		\vert \Dktwo \partial_v \psi \vert^2
		\Bigg)
		\\
		\lesssim
		\sqrt{\Fstarll_R + \intEdeltall_R' \ }
		\Big(
		\Fstar_R
		+
		\Ezerostar_R
		+
		\int^*\accentset{\scalebox{.6}{\mbox{\tiny $(\delta-1)$}}}{\underaccent{\scalebox{.8}{\mbox{\tiny $k$}}}{\mathcal{E}}}_R'
		\
		\Big)
		+
		\sqrt{\Ezerostarll_{\frac{8R}{9}}
		+
		\int^*\accentset{\scalebox{.6}{\mbox{\tiny $(\delta-1)$}}}{\underaccent{\scalebox{.8}{\mbox{\tiny $\ll \mkern-6mu k$}}}{\mathcal{E}}}_R'}
		\
		\bigg(
		\Fstar_R + \intEdelta_R \ 
		\bigg)
		,
	\end{multline*}
	and
	\begin{multline*}
		\sum_{\vert \mathbf{k}_1 \vert + \vert \mathbf{k}_2 \vert \leq \vert \mathbf{k} \vert}
		\int_{\mathcal{R}(\tau_0,\tau_1,v)\cap \{ r \geq R \}} 
		r^{p-1} \vert \partial_v( r \Dk \psi) \vert
		\cdot 
		\vert \Dkone \partial_u \psi \vert
		\cdot
		\vert \Dktwo \partial_v \psi \vert
		\\
		\lesssim
		\sqrt{\int^*\accentset{\scalebox{.6}{\mbox{\tiny $(p-1)$}}}{\underaccent{\scalebox{.8}{\mbox{\tiny $\ll \mkern-6mu k$}}}{\mathcal{E}}}_R' \ }
		\sqrt{\Fstar_R}
		\sqrt{\intEpone_R' \ }
		+
		\sqrt{\Fstarll_R + \Ezerostarll_{\frac{8R}{9}} + \intEdeltall_R '\ }
		\Bigg(
		\Fstar_R
		+
		\Epstar_R
		+
		\intEpone_R'
		\Bigg)
		.
	\end{multline*}
\end{proposition}

\begin{proof}
The separate terms are estimated individually.
We will use~\eqref{basicbootstrapinnullcond} to replace quartic bounds with cubic ones without further comment.
  In view of the fact \eqref{eq:DkdvdvDk}, the proof follows from the estimates \eqref{eq:cornerterms}--\eqref{eq:rperrornonlin1} below.

The notation of Section \ref{subsec:Sobolev} will be used throughout.  Recall, in particular, the regions $\check{\mathcal{R}}_R$ and $\mathcal{R}_{\mathrm{Corner}}$ defined in \eqref{eq:Rcheck} and \eqref{eq:Rcorner} respectively, the boundedness property \eqref{eq:Rcornerbound} of $r$ in $\mathcal{R}_{\mathrm{Corner}}$, and the $T_1$ and $T_1'$ notation defined in \eqref{eq:T1T1prime}.

The constant in the Sobolev inequality \eqref{eq:pointwise1} blows up as $\tau_1 - \tau_0 \to 0$.  Accordingly, the cases $\tau_1 - \tau_0 \geq 1$ and $\tau_1 - \tau_0 < 1$ are typically considered separately.  Similarly, the region $\mathcal{R}_{\mathrm{Corner}}$, near the corner $\{u=\tau_0\} \cap \{r=R\}$, is treated separately.

\noindent \textbf{Estimate in the corner region $\{ v'\leq v(R,\tau_0 + R)\}$:}
	For any $\vert \tilde{\mathbf{k}} \vert \leq k+1$,
	\begin{equation} \label{eq:cornerterms}
		\sum_{\vert \mathbf{k}_1 \vert + \vert \mathbf{k}_2 \vert \leq \vert \tilde{\mathbf{k}} \vert -1}
		\int_{\mathcal{R}_{\mathrm{Corner}}(\tau_0,\tau_1,v)} 
		\vert \mathfrak{D}^{\tilde{\mathbf{k}}} \psi \vert
		\cdot 
		\vert \partial_u \Dkone \psi \vert
		\cdot
		\vert \partial_v \Dktwo \psi \vert
		\lesssim
		\sqrt{\Ezerostarll_{\frac{8R}{9}}} \intEdelta_R'
		.
	\end{equation}
	Indeed, consider some $\vert \mathbf{k}_1 \vert + \vert \mathbf{k}_2 \vert \leq \vert \tilde{\mathbf{k}} \vert -1$.  Since $k \geq 7$, it follows that either $\vert \mathbf{k}_1 \vert \leq k-4$, or $\vert \mathbf{k}_2 \vert \leq k-4$.  In the former case,
	\begin{multline*}
		\int_{\mathcal{R}_{\mathrm{Corner}}(\tau_0,\tau_1,v)} 
		\vert \mathfrak{D}^{\tilde{\mathbf{k}}} \psi \vert
		\cdot 
		\vert \partial_u \Dkone \psi \vert
		\cdot
		\vert \partial_v \Dktwo \psi \vert
		\\
		\lesssim
		\sup_{\mathcal{R}_{\mathrm{Corner}(\tau_0,\tau_1,v)}} \vert \partial_u \Dkone \psi \vert
		\Big(
		\int_{\mathcal{R}_{\mathrm{Corner}}(\tau_0,\tau_1,v)}
		\vert \mathfrak{D}^{\tilde{\mathbf{k}}} \psi \vert
		\Big)^{\frac{1}{2}}
		\Big(
		\int_{\mathcal{R}_{\mathrm{Corner}}(\tau_0,\tau_1,v)}
		\vert \partial_v \Dktwo \psi \vert
		\Big)^{\frac{1}{2}}
		\lesssim
		\sqrt{\Ezerostarll_{\frac{8R}{9}}} \intEdelta_R',
	\end{multline*}
	by the Sobolev inequality \eqref{eq:pointwiseout}.  Similarly when $\vert \mathbf{k}_2 \vert \leq k-4$.
	
	It follows, in view of the boundedness property \eqref{eq:Rcornerbound}, that the estimates of the Proposition are trivial in the corner region.  The remainder of the proof thus concerns the region $\check{\mathcal{R}}_R(\tau_0,\tau_1,v)$.

\noindent \textbf{Estimate for $\partial_u \psi$ term:}
	First note that, for any $\vert \mathbf{k} \vert \leq k$,
	\begin{equation} \label{eq:partialuterm}
		\sum_{\vert \mathbf{k}_1 \vert + \vert \mathbf{k}_2 \vert \leq \vert \mathbf{k} \vert}
		\int_{\check{\mathcal{R}}_R(\tau_0,\tau_1,v)} 
		\vert \partial_u \Dk \psi \vert
		\cdot 
		\vert \partial_u \Dkone \psi \vert
		\cdot
		\vert \partial_v \Dktwo \psi \vert
		\lesssim
		\sqrt{\Fstarll}
		\Big(
		\Ezerostar
		+
		\int^*\accentset{\scalebox{.6}{\mbox{\tiny $(\delta-1)$}}}{\underaccent{\scalebox{.8}{\mbox{\tiny $k$}}}{\mathcal{E}}}'
		\Big)
		+
		\sqrt{\Ezerostarll
		+
		\int^*\accentset{\scalebox{.6}{\mbox{\tiny $(\delta-1)$}}}{\underaccent{\scalebox{.8}{\mbox{\tiny $\ll \mkern-6mu k$}}}{\mathcal{E}}}'}
		\
		\bigg(
		\Fstar + \intEdelta \ 
		\bigg)
		.
	\end{equation}
	Indeed, suppose $\vert \mathbf{k}_1 \vert + \vert \mathbf{k}_2 \vert \leq \vert \mathbf{k} \vert$.  Since $k \geq 7$, it follows that either $\vert \mathbf{k}_1 \vert \leq k-3$ or $\vert \mathbf{k}_2 \vert \leq k-3$.  If $\vert \mathbf{k}_1 \vert \leq k-3$, then, for $\tau_1 - \tau_0 \geq 1$,
	\begin{align*}
		&
		\int_{\check{\mathcal{R}}_R(\tau_0,\tau_1,v)} 
		\vert \partial_u \Dk \psi \vert
		\cdot 
		\vert \partial_u \Dkone \psi \vert
		\cdot
		\vert \partial_v \Dktwo \psi \vert
		\\
		&
		\qquad
		\lesssim
		\int_{v(R,\tau_0+R)}^v
		\sup_{u,\theta} r \vert \partial_u \Dkone \psi \vert
		\int_{\tau_0}^{T_1'} \int_{S^2}
		\big(
		r^{-1-\delta} \vert \partial_u \Dk \psi \vert^2
		+
		r^{-1+\delta} \vert \partial_v \Dktwo \psi \vert^2
		\big)
		r^2 d\omega du
		dv'
		\lesssim
		\sqrt{\Fstarll}
		\int^*\accentset{\scalebox{.6}{\mbox{\tiny $(\delta-1)$}}}{\underaccent{\scalebox{.8}{\mbox{\tiny $k$}}}{\mathcal{E}}}',
	\end{align*}
	by the Sobolev inequality on the incoming cones \eqref{eq:pointwise1}.  For $\tau_1 - \tau_0 < 1$,
	\begin{align*}
		&
		\int_{\check{\mathcal{R}}_R(\tau_0,\tau_1,v)} 
		\vert \partial_u \Dk \psi \vert
		\cdot 
		\vert \partial_u \Dkone \psi \vert
		\cdot
		\vert \partial_v \Dktwo \psi \vert
		\\
		&
		\lesssim
		\int_{v(R,\tau_0)}^v \sup_{\substack{\tau_0 \leq u \leq T_1' \\ \theta \in S^2}} \vert r \partial_u \Dkone \psi \vert
		\left( \int_{\tau_0}^{T_1'} \int_{S^2} r^{-1-\delta} \vert \partial_u \Dk \psi \vert r^2 d\omega du \right)^{\frac{1}{2}}
		\cdot 
		\left( \int_{\tau_0}^{T_1'} \int_{S^2} r^{-1+\delta} \vert \partial_v \Dktwo \psi \vert r^2 d\omega du \right)^{\frac{1}{2}}
		dv'
		\\
		&
		\lesssim
		\frac{\sqrt{\Fstarll}}{(\tau_1-\tau_0)^{\frac{1}{2}} }
		\left( \int_{v(R,\tau_0)}^v  \int_{\tau_0}^{T_1'} \int_{S^2} r^{-1-\delta} \vert \partial_u \Dk \psi \vert r^2 d\omega du dv' \right)^{\frac{1}{2}}
		\cdot 
		\left( \int_{\tau_0}^{T_1} \int_{v(R,u)}^v \int_{S^2} r^{-1+\delta} \vert \partial_v \Dktwo \psi \vert r^2 d\omega dv' du \right)^{\frac{1}{2}}
		\\
		&
		\lesssim
		\frac{\sqrt{\Fstarll}}{(\tau_1-\tau_0)^{\frac{1}{2}} }
		\sqrt{\intEdelta \ } \sqrt{\Ezerostar}
		\left( \int_{\tau_0}^{T_1} du \right)^{\frac{1}{2}}
		\lesssim
		\sqrt{\Ezerostar} \sqrt{\intEdelta \ } \sqrt{\Fstarll},
	\end{align*}
	by the Sobolev inequality on the incoming cones \eqref{eq:pointwise1}.
	
	If $\vert \mathbf{k}_2 \vert \leq k-2$, then
	\begin{align*}
		&
		\int_{\check{\mathcal{R}}_R(\tau_0,\tau_1,v)} 
		\vert \partial_u \Dk \psi \vert
		\cdot 
		\vert \partial_u \Dkone \psi \vert
		\cdot
		\vert \partial_v \Dktwo \psi \vert
		\\
		&
		\qquad \qquad
		\lesssim
		\int_{v(R,\tau_0+R)}^v
		\sup_{u,\theta} \vert r^{\frac{1+\delta}{2}} \partial_v \Dktwo \psi \vert
		\int_{\tau_0}^{T_1'} \int_{S^2}
		r^{-\frac{1+\delta}{2}}\vert \partial_u \Dk \psi \vert
		\vert \partial_u \Dkone \psi \vert
		r^2 d\omega du
		dv'
		\\
		&
		\qquad \qquad
		\lesssim
		\big(1+ (\tau_1-\tau_0)^{-\frac{1}{2}} \big) \sum_{\vert \tilde{\mathbf{k}} \vert \leq \vert \mathbf{k}_2\vert+3}
		\int_{v(R)}^v
		\Big(
		\int_{\tau_0}^{T_1'} \int_{S^2} r^{-1+\delta} \vert \partial_v \mathfrak{D}^{\tilde{\mathbf{k}}} \psi \vert^2
		r^2 d\omega du
		\Big)^{\frac{1}{2}}
		\\
		&
		\qquad \qquad \qquad \qquad \qquad
		\times
		\Big(
		\int_{\tau_0}^{T_1'} \int_{S^2}
		r^{-1-\delta}\vert \partial_u \Dk \psi \vert^2
		r^2 d\omega du
		\Big)^{\frac{1}{2}}
		\Big(
		\int_{\tau_0}^{T_1'} \int_{S^2}
		\vert \partial_u \Dkone \psi \vert^2
		r^2 d\omega du
		\Big)^{\frac{1}{2}}
		dv'
		\\
		&
		\qquad \qquad
		\lesssim
		\big(1+ (\tau_1-\tau_0)^{-\frac{1}{2}} \big) \sqrt{\Fstar} \sqrt{\intEdelta \ } \sum_{\vert \tilde{\mathbf{k}} \vert \leq \vert \mathbf{k}_2\vert+3}
		\Big(
		\int_{\tau_0}^{\tau_1}
		\int_{v(R,u)}^v
		\int_{S^2} r^{-1+\delta} \vert \partial_v \mathfrak{D}^{\tilde{\mathbf{k}}} \psi \vert^2
		r^2 d\omega du dv'
		\Big)^{\frac{1}{2}},
	\end{align*}
	by the Sobolev inequality on the incoming cones \eqref{eq:dvpointwisein}.  Hence, if $\tau_1-\tau_0 \geq 1$,
	\[
		\int_{\check{\mathcal{R}}_R(\tau_0,\tau_1,v)} 
		\vert \partial_u \Dk \psi \vert
		\cdot 
		\vert \partial_u \Dkone \psi \vert
		\cdot
		\vert \partial_v \Dktwo \psi \vert
		\lesssim
		\sqrt{\Fstar} \sqrt{\intEdelta \ }
		\sqrt{\int^*\accentset{\scalebox{.6}{\mbox{\tiny $(\delta-1)$}}}{\underaccent{\scalebox{.8}{\mbox{\tiny $\ll \mkern-6mu k$}}}{\mathcal{E}}}'},
	\]
	and, if $\tau_1-\tau_0 <1$,
	\[
		\int_{\check{\mathcal{R}}_R(\tau_0,\tau_1,v)} 
		\vert \partial_u \Dk \psi \vert
		\cdot 
		\vert \partial_u \Dkone \psi \vert
		\cdot
		\vert \partial_v \Dktwo \psi \vert
		\lesssim
		\sqrt{\Fstar} \sqrt{\intEdelta \ }
		\sqrt{\Ezerostarll}.
	\]

\noindent \textbf{Estimate for $\partial_v \psi$ term:}
	Next, for any $\vert \mathbf{k} \vert \leq k$,
	\begin{equation} \label{eq:partialvterm}
		\sum_{\vert \mathbf{k}_1 \vert + \vert \mathbf{k}_2 \vert \leq \vert \mathbf{k} \vert}
		\int_{\check{\mathcal{R}}_R(\tau_0,\tau_1,v)}  
		\vert \partial_v \Dk \psi \vert
		\cdot 
		\vert \partial_u \Dkone \psi \vert
		\cdot
		\vert \partial_v \Dktwo \psi \vert
		\lesssim
		\sqrt{\Fstarll + \intEdeltall' \ }
		\Big(
		\Fstar
		+
		\Ezerostar
		+
		\int^*\accentset{\scalebox{.6}{\mbox{\tiny $(\delta-1)$}}}{\underaccent{\scalebox{.8}{\mbox{\tiny $k$}}}{\mathcal{E}}}'
		\
		\Big)
		.
	\end{equation}
Indeed, if $\vert \mathbf{k}_1 \vert \leq k-2$, then
	\begin{align*}
		&
		\int_{\check{\mathcal{R}}_R(\tau_0,\tau_1,v)} 
		\vert \partial_v \Dk \psi \vert
		\cdot 
		\vert \partial_u \Dkone \psi \vert
		\cdot
		\vert \partial_v \Dktwo \psi \vert
		\\
		&
		\qquad
		\lesssim
		\sup_{u,v,\theta} \vert r \partial_u \Dkone \psi \vert
		\Big(
		\int_{\check{\mathcal{R}}_R(\tau_0,\tau_1,v)} 
		r^{-1}
		\vert \partial_v \Dk \psi \vert^2
		\Big)^{\frac{1}{2}}
		\Big(
		\int_{\check{\mathcal{R}}_R(\tau_0,\tau_1,v)} 
		r^{-1}
		\vert \partial_v \Dktwo \psi \vert^2
		\Big)^{\frac{1}{2}}
		\\
		&
		\qquad
		\lesssim
		\big(1+ (\tau_1-\tau_0)^{-\frac{1}{2}} \big)
		\sqrt{\Fstarll} \sqrt{\int^*\accentset{\scalebox{.6}{\mbox{\tiny $(\delta-1)$}}}{\underaccent{\scalebox{.8}{\mbox{\tiny $k$}}}{\mathcal{E}}}'}
		\Big(
		\int_{\check{\mathcal{R}}_R(\tau_0,\tau_1,v)} 
		r^{-1}
		\vert \partial_v \Dk \psi \vert^2
		\Big)^{\frac{1}{2}},
	\end{align*}
	by the Sobolev inequality on the incoming cones \eqref{eq:pointwise1}.  Hence, if $\tau_1-\tau_0 \geq 1$,
	\[
		\int_{\check{\mathcal{R}}_R(\tau_0,\tau_1,v)} 
		\vert \partial_v \Dk \psi \vert
		\cdot 
		\vert \partial_u \Dkone \psi \vert
		\cdot
		\vert \partial_v \Dktwo \psi \vert
		\lesssim
		\sqrt{\Fstarll} \int^*\accentset{\scalebox{.6}{\mbox{\tiny $(\delta-1)$}}}{\underaccent{\scalebox{.8}{\mbox{\tiny $k$}}}{\mathcal{E}}}',
	\]
	and, if $\tau_1-\tau_0 <1$,
	\[
		\int_{\check{\mathcal{R}}_R(\tau_0,\tau_1,v)} 
		\vert \partial_v \Dk \psi \vert
		\cdot 
		\vert \partial_u \Dkone \psi \vert
		\cdot
		\vert \partial_v \Dktwo \psi \vert
		\lesssim
		\sqrt{\Fstarll} \sqrt{\Ezerostar} \sqrt{\int^*\accentset{\scalebox{.6}{\mbox{\tiny $(\delta-1)$}}}{\underaccent{\scalebox{.8}{\mbox{\tiny $k$}}}{\mathcal{E}}}'}.
	\]

	If $\vert \mathbf{k}_2 \vert \leq \vert \mathbf{k} \vert-2$ then one similarly estimates,
	\begin{align*}
		&
		\int_{\check{\mathcal{R}}_R(\tau_0,\tau_1,v)} 
		\vert \partial_v \Dk \psi \vert
		\cdot 
		\vert \partial_u \Dkone \psi \vert
		\cdot
		\vert \partial_v \Dktwo \psi \vert
		\\
		&
		\qquad
		\lesssim
		\int_{v(R,\tau_0+R)}^v \sup_{\substack{\tau_0 \leq u \leq T_1' \\ \theta \in S^2}} \vert r^{\frac{1}{2}} \partial_v \Dktwo \psi \vert
		\left( \int_{\tau_0}^{T_1'} \int_{S^2} r^{-1} \vert \partial_v \Dk \psi \vert r^2 d\omega du \right)^{\frac{1}{2}}
		\cdot 
		\left( \int_{\tau_0}^{T_1'} \int_{S^2} \vert \partial_u \Dkone \psi \vert r^2 d\omega du \right)^{\frac{1}{2}}
		dv'
		\\
		&
		\qquad
		\lesssim
		\big(1+ (\tau_1-\tau_0)^{-\frac{1}{2}} \big)
		\sqrt{\Fstar}
		\\
		&
		\qquad \qquad \quad
		\times
		\sum_{\vert \tilde{\mathbf{k}} \vert \leq \vert \mathbf{k}_2\vert+3}
		\int_{v(R,\tau_0+R)}^v
		\left( \int_{\tau_0}^{T_1'} \int_{S^2} r^{-1-\delta} \vert \partial_v \mathfrak{D}^{\tilde{\mathbf{k}}} \psi \vert r^2 d\omega du \right)^{\frac{1}{2}}
		\left( \int_{\tau_0}^{T_1'} \int_{S^2} r^{-1+\delta} \vert \partial_v \Dk \psi \vert r^2 d\omega du \right)^{\frac{1}{2}}
		dv'
		\\
		&
		\qquad
		\lesssim
		\big(1+ (\tau_1-\tau_0)^{-\frac{1}{2}} \big)
		\sqrt{\Fstar} \sqrt{\intEdeltall'}
		\left( 
		\int_{v(R,\tau_0+R)}^v
		\int_{\tau_0}^{T_1'} \int_{S^2} r^{-1+\delta} \vert \partial_v \Dk \psi \vert r^2 d\omega du dv'
		\right)^{\frac{1}{2}},
	\end{align*}
	by the Sobolev inequality on the incoming cones \eqref{eq:dvpointwisein}.  Hence, if $\tau_1-\tau_0 \geq 1$,
	\[
		\int_{\check{\mathcal{R}}_R(\tau_0,\tau_1,v)} 
		\vert \partial_v \Dk \psi \vert
		\cdot 
		\vert \partial_u \Dkone \psi \vert
		\cdot
		\vert \partial_v \Dktwo \psi \vert
		\lesssim
		\sqrt{\Fstar} \sqrt{\intEdeltall'}
		\sqrt{\int^*\accentset{\scalebox{.6}{\mbox{\tiny $(\delta-1)$}}}{\underaccent{\scalebox{.8}{\mbox{\tiny $k$}}}{\mathcal{E}}}'}
		,
	\]
	and, if $\tau_1-\tau_0 <1$,
	\[
		\int_{\check{\mathcal{R}}_R(\tau_0,\tau_1,v)} 
		\vert \partial_v \Dk \psi \vert
		\cdot 
		\vert \partial_u \Dkone \psi \vert
		\cdot
		\vert \partial_v \Dktwo \psi \vert
		\lesssim
		\sqrt{\Fstar}
		\sqrt{\intEdeltall'}
		\sqrt{\Ezerostar}.
	\]

\noindent \textbf{Estimate for $r^{-1}\nablaslash \psi$ and $r^{-1} \psi$ terms:}
	For any $\vert \mathbf{k} \vert \leq k$,
	\begin{equation} \label{eq:nablaslashpsiterm}
		\sum_{\vert \mathbf{k}_1 \vert + \vert \mathbf{k}_2 \vert \leq \vert \mathbf{k} \vert}
		\int_{\check{\mathcal{R}}_R(\tau_0,\tau_1,v)} 
		r^{-1}\vert \nablaslash \Dk \psi \vert
		\cdot 
		\vert \partial_u \Dkone \psi \vert
		\cdot
		\vert \partial_v \Dktwo \psi \vert
		\lesssim
		\sqrt{\Fstarll + \intEdeltall' \ }
		\Big(
		\Fstar
		+
		\Ezerostar
		+
		\int^*\accentset{\scalebox{.6}{\mbox{\tiny $(\delta-1)$}}}{\underaccent{\scalebox{.8}{\mbox{\tiny $k$}}}{\mathcal{E}}}'
		\
		\Big)
		,
	\end{equation}
	and
	\begin{equation} \label{eq:psiterm}
		\sum_{\vert \mathbf{k}_1 \vert + \vert \mathbf{k}_2 \vert \leq \vert \mathbf{k} \vert}
		\int_{\check{\mathcal{R}}_R(\tau_0,\tau_1,v)} 
		r^{-1}\vert \Dk \psi \vert
		\cdot 
		\vert \partial_u \Dkone \psi \vert
		\cdot
		\vert \partial_u \Dktwo \psi \vert
		\lesssim
		\sqrt{\Fstarll + \intEdeltall' \ }
		\Big(
		\Fstar
		+
		\Ezerostar
		+
		\int^*\accentset{\scalebox{.6}{\mbox{\tiny $(\delta-1)$}}}{\underaccent{\scalebox{.8}{\mbox{\tiny $k$}}}{\mathcal{E}}}'
		\
		\Big)
		.
	\end{equation}
	Indeed, \eqref{eq:nablaslashpsiterm} and \eqref{eq:psiterm} follow as in the proof of \eqref{eq:partialvterm}, using now the fact that
	\[
		\int_{\check{\mathcal{R}}_R(\tau_0,\tau_1,v)} 
		r^{-1-\delta} \vert r^{-1} \nablaslash \psi \vert^2
		+
		r^{-1-\delta} \vert r^{-1} \psi \vert^2
		\leq
		\intEdelta'
		,
	\]
	and
	\[
		\int_{\check{\mathcal{R}}_R(\tau_0,\tau_1,v)} 
		r^{-1+\delta} \vert r^{-1} \nablaslash \psi \vert^2
		+
		r^{-1+\delta} \vert r^{-1} \psi \vert^2
		\leq
		\int^*\accentset{\scalebox{.6}{\mbox{\tiny $(\delta-1)$}}}{\underaccent{\scalebox{.8}{\mbox{\tiny $k$}}}{\mathcal{E}}}'
		,
	\qquad
		\int_{v(R,\tau)}^v \int_{S^2} \left( \vert r^{-1} \nablaslash \psi \vert^2 + \vert r^{-1} \psi \vert^2 \right) r^2 d\omega dv' \leq \Ezerostar.
	\]

\noindent \textbf{Estimate for quartic term:}
	For any $\vert \mathbf{k} \vert \leq k$,
	\begin{equation} \label{eq:quarticnonlin1}
		\sum_{\vert \mathbf{k}_1 \vert + \vert \mathbf{k}_2 \vert \leq \vert \mathbf{k} \vert}
		\int_{\check{\mathcal{R}}_R(\tau_0,\tau_1,v)} 
		\vert \partial_u \Dkone \psi \vert^2
		\cdot
		\vert \partial_v \Dktwo \psi \vert^2
		\lesssim
		\Big( \ \Ezerostar + \intEdelta' \ \Big)\Fstarll
		+
		\Big( \ \Ezerostarll + \intEdeltall' \ \Big)\Fstar
		.
	\end{equation}
	Indeed, supposing first that $\vert \mathbf{k}_1 \vert \leq k-3$,
	\begin{multline*}
		\int_{\check{\mathcal{R}}_R(\tau_0,\tau_1,v)}
		\vert \partial_u \Dkone \psi \vert^2
		\cdot
		\vert \partial_v \Dktwo \psi \vert^2
		\\
		\lesssim
		\sup_{u,\theta} \left( r \vert \partial_u \Dkone \psi \vert \right)^2
		\int_{\check{\mathcal{R}}_R(\tau_0,\tau_1,v)}
		r^{-2} \vert \partial_v \Dktwo \psi \vert^2
		\lesssim
		\Fstarll 
		\big(1+ (\tau_1-\tau_0)^{-\frac{1}{2}} \big)
		\int_{\check{\mathcal{R}}_R(\tau_0,\tau_1,v)}
		r^{-2} \vert \partial_v \Dktwo \psi \vert^2,
	\end{multline*}
	by the Sobolev inequality on the incoming cones \eqref{eq:pointwise1}.  Hence, if $\tau_1-\tau_0 \geq 1$,
	\[
		\int_{\check{\mathcal{R}}_R(\tau_0,\tau_1,v)}
		\vert \partial_u \Dkone \psi \vert^2
		\cdot
		\vert \partial_v \Dktwo \psi \vert^2
		\lesssim
		\Fstarll
		\intEdelta',
	\]
	and, if $\tau_1-\tau_0 <1$,
	\[
		\int_{\check{\mathcal{R}}_R(\tau_0,\tau_1,v)}
		\vert \partial_u \Dkone \psi \vert^2
		\cdot
		\vert \partial_v \Dktwo \psi \vert^2
		\lesssim
		\Fstarll
		\Ezerostar.
	\]
	
	Similarly, if $\vert \mathbf{k}_2 \vert \leq k-3$,
	\begin{multline*}
		\int_{\check{\mathcal{R}}_R(\tau_0,\tau_1,v)}
		\vert  \partial_u \Dkone \psi \vert^2
		\cdot
		\vert \partial_v \Dktwo \psi \vert^2
		\lesssim
		\int_{v(R,\tau_0+R)}^v
		\sup_{u,\theta}  \vert \partial_v \Dktwo \psi \vert^2
		\int_{\tau_0}^{\tau_1} \int_{S^2}
		\vert \partial_v \Dktwo \psi \vert^2
		r^2 d\omega du dv'
		\\
		\lesssim
		\Fstar
		\big(1+ (\tau_1-\tau_0)^{-\frac{1}{2}} \big)
		\sum_{\vert \widetilde{\textbf{k}}\vert \leq \vert \textbf{k}_2\vert +3}
		\int_{\check{\mathcal{R}}_R(\tau_0,\tau_1,v)}
		r^{-2} \vert \partial_v \mathfrak{D}^{\widetilde{\textbf{k}}} \psi \vert^2,
	\end{multline*}
	and so one similarly has
	\[
		\int_{\check{\mathcal{R}}_R(\tau_0,\tau_1,v)}
		\vert \partial_u \Dkone \psi \vert^2
		\cdot
		\vert \partial_v \Dktwo \psi \vert^2
		\lesssim
		\Big( \ \Ezerostarll + \intEdeltall' \ \Big)\Fstar.
	\]

\noindent \textbf{Estimate for $r^{p-1} \partial_v (r\psi)$ term:}
	Finally, for any $\vert \mathbf{k} \vert \leq k$,
	\begin{multline} \label{eq:rperrornonlin1}
		\sum_{\vert \mathbf{k}_1 \vert + \vert \mathbf{k}_2 \vert \leq \vert \mathbf{k} \vert}
		\int_{\check{\mathcal{R}}_R(\tau_0,\tau_1,v)}  
		r^{p-1} \vert \partial_v (r \Dk \psi) \vert
		\cdot 
		\vert \partial_u \Dkone  \psi \vert
		\cdot
		\vert \partial_v \Dktwo \psi \vert
		\\
		\lesssim
		\bigg(
		\Fstar
		+
		\Epstar
		+
		\intEpone'
		\bigg)
		\sqrt{\Fstarll + \Ezerostarll}
		+
		\sqrt{\Fstar}
		\sqrt{\intEpone' \ }
		\sqrt{\int^*\accentset{\scalebox{.6}{\mbox{\tiny $(p-1)$}}}{\underaccent{\scalebox{.8}{\mbox{\tiny $\ll \mkern-6mu k$}}}{\mathcal{E}}}' \ }
		.
	\end{multline}
	Indeed, consider first the case $\tau_1 - \tau_0 \geq 1$.  Suppose first that $\vert \mathbf{k}_1 \vert \leq k-3$.  
	Then, if $\tau_1-\tau_0 \geq 1$,
	\begin{align*}
		&
		\int_{\check{\mathcal{R}}_R(\tau_0,\tau_1,v)} 
		r^{p-1} \vert \partial_v (r \Dk \psi) \vert
		\cdot 
		\vert \partial_u \Dkone \psi \vert
		\cdot
		\vert \partial_v \Dktwo \psi \vert
		\\
		&
		\quad
		\lesssim
		\sup_{\check{\mathcal{R}}_R(\tau_0,\tau_1,v)} (r \vert \partial_u \Dkone \psi \vert)
		\Big(
		\int_{\mathcal{R}_R(\tau_0,\tau_1,v)}
		r^{p-3}
		\vert \partial_v (r \Dk \psi) \vert^2
		\Big)^{\frac{1}{2}}
		\Big(
		\int_{\mathcal{R}_R(\tau_0,\tau_1,v)}
		r^{p-1} 
		\vert \partial_v \Dktwo \psi \vert^2
		\Big)^{\frac{1}{2}}
		\lesssim
		\sqrt{\Fstarll}
		\intEpone ',
	\end{align*}
	by the Sobolev inequality on the incoming cones \eqref{eq:pointwise1}.  Similarly, if $\tau_1-\tau_0 < 1$,
	\begin{align*}
		&
		\int_{\check{\mathcal{R}}_R(\tau_0,\tau_1,v)} 
		r^{p-1} \vert \partial_v (r \Dk \psi) \vert
		\cdot 
		\vert \partial_u \Dkone \psi \vert
		\cdot
		\vert \partial_v \Dktwo \psi \vert
		\\
		&
		\quad
		\lesssim
		(\tau_1-\tau_0)^{\frac{1}{2}} \sqrt{\Fstarll}
		\Big(
		\int_{\mathcal{R}_R(\tau_0,\tau_1,v)}
		r^{p-2}
		\vert \partial_v (r \Dk \psi) \vert^2
		\Big)^{\frac{1}{2}}
		\Big(
		\int_{\mathcal{R}_R(\tau_0,\tau_1,v)}
		r^{p-2} 
		\vert \partial_v \Dktwo \psi \vert^2
		\Big)^{\frac{1}{2}}
		\lesssim
		\sqrt{\Fstarll}
		\sqrt{\Epstar}
		\sqrt{\Ezerostar}.
	\end{align*}
	
	Suppose now $\vert \mathbf{k}_2 \vert \leq k-3$.  Then, if $\tau_1-\tau_0 \geq 1$,
	\begin{align*}
		&
		\int_{\check{\mathcal{R}}_R(\tau_0,\tau_1,v)} 
		r^{p-1} \vert \partial_v (r \Dk \psi) \vert
		\cdot 
		\vert \partial_u \Dkone \psi \vert
		\cdot
		\vert \partial_v \Dktwo \psi \vert
		\\
		&
		\quad
		\lesssim
		\sqrt{\Fstar}
		\int_{v(R,\tau_0)}^v 
		\sup_{u,\theta} \big( r^{\frac{p+1}{2}} \vert \partial_v \Dktwo \psi \vert \big)
		\left( \int_{\tau_0}^{\tau_1} \int_{S^2}
		r^{p-3} \vert \partial_v (r \Dk \psi) \vert^2
		r^2 d\omega du \right)^{\frac{1}{2}} dv
		\\
		&
		\quad
		\lesssim
		\sqrt{\Fstar}
		\left(
		\int_{\mathcal{R}_R(\tau_0,\tau_1,v)} 
		r^{p-3} \vert \partial_v (r \Dk \psi) \vert^2
		r^2 d\omega du dv \right)^{\frac{1}{2}}
		\sum_{\vert \widetilde{\textbf{k}} \vert \leq \vert \mathbf{k}_2 \vert + 3}
		\left(
		\int_{\mathcal{R}_R(\tau_0,\tau_1,v)} 
		r^{p-1} \vert \partial_v \mathfrak{D}^{\widetilde{\mathbf{k}}} \psi \vert^2
		r^2 d\omega du dv \right)^{\frac{1}{2}}
		\\
		&
		\quad
		\lesssim
		\sqrt{\Fstar}
		\sqrt{\intEpone' \ }
		\sqrt{\int^*\accentset{\scalebox{.6}{\mbox{\tiny $(p-1)$}}}{\underaccent{\scalebox{.8}{\mbox{\tiny $\ll \mkern-6mu k$}}}{\mathcal{E}}}' \ }
		,
	\end{align*}
	using again the Sobolev inequality on the incoming cones \eqref{eq:dvpointwisein}.  Similarly, if $\tau_1 - \tau_0 < 1$,
	\begin{align*}
		&
		\int_{\check{\mathcal{R}}_R(\tau_0,\tau_1,v)} 
		r^{p-1} \vert \partial_v (r \Dk \psi) \vert
		\cdot 
		\vert \partial_u \Dkone \psi \vert
		\cdot
		\vert \partial_v \Dktwo \psi \vert
		\\
		&
		\quad
		\lesssim
		(\tau_1 - \tau_0)^{- \frac{1}{2}}\sqrt{\Fstar}
		\left(
		\int_{\mathcal{R}_R(\tau_0,\tau_1,v)} 
		r^{p-2} \vert \partial_v (r \Dk \psi) \vert^2
		\right)^{\frac{1}{2}}
		\sum_{\vert \widetilde{\textbf{k}} \vert \leq \vert \mathbf{k}_2 \vert + 3}
		\left(
		\int_{\mathcal{R}_R(\tau_0,\tau_1,v)} 
		r^{p-2} \vert \partial_v \mathfrak{D}^{\widetilde{\mathbf{k}}} \psi \vert^2
		\right)^{\frac{1}{2}}
		\\
		&
		\quad
		\lesssim
		\sqrt{\Fstar}
		\sqrt{\Epstar}
		\sqrt{\Ezerostarll}.
	\end{align*}

\end{proof}

\subsubsection{The proof of Proposition \ref{prop:appendixBmainMink} for the nonlinearity $F = \slashed{g}^{AB} \nablaslash_A \psi \nablaslash_B \psi$}
\label{subsec:modelnonlin2}

Consider now the nonlinearity
\[
	F = \slashed{g}^{AB} \nablaslash_A \psi \nablaslash_B \psi.
\]
In view of the property \eqref{eq:nablaMinkproperty}, the proof of Proposition \ref{prop:appendixBmainMink} follows from the following proposition.

\begin{proposition}[Nonlinear error estimates for $F = \slashed{g}^{AB} \nablaslash_A \psi \nablaslash_B \psi$ on Minkowski space] \label{prop:Minkowskinonlin2}
	Under the assumptions of Proposition \ref{prop:appendixBmainMink}, for each $\vert \mathbf{k} \vert \leq k$,
	\begin{multline*}
		\sum_{\vert \mathbf{k}_1 \vert + \vert \mathbf{k}_2 \vert \leq \vert \mathbf{k} \vert}
		\Bigg(
		\int_{\mathcal{R}(\tau_0,\tau_1,v)\cap \{ r \geq R \}} 
		\Big(
		\vert \partial_u \Dk \psi \vert
		+
		\vert \partial_v \Dk \psi \vert
		+
		\frac{1}{r} \vert \nablaslash \Dk \psi \vert
		+
		\frac{1}{r} \vert \Dk \psi \vert
		\Big)
		\cdot 
		\vert \nablaslash \Dkone \psi \vert
		\cdot
		\vert \nablaslash \Dktwo \psi \vert
		\\
		+
		\int_{\mathcal{R}(\tau_0,\tau_1,v)\cap \{ r \geq R \}}
		\vert \nablaslash \Dkone \psi \vert^2
		\cdot
		\vert \nablaslash \Dktwo \psi \vert^2
		\Bigg)
		\lesssim
		\sqrt{\Fstarll_R + \Ezerostarll_{\frac{8R}{9}}}
		\Big(
		\Ezerostar_R
		+
		\int^*\accentset{\scalebox{.6}{\mbox{\tiny $(\delta-1)$}}}{\underaccent{\scalebox{.8}{\mbox{\tiny $k$}}}{\mathcal{E}}}_R' \
		\Big)
	\end{multline*}
	and
	\[
		\sum_{\vert \mathbf{k}_1 \vert + \vert \mathbf{k}_2 \vert \leq \vert \mathbf{k} \vert}
		\int_{\mathcal{R}(\tau_0,\tau_1,v)\cap \{ r \geq R \}} 
		r^{p-1} \vert \partial_v( r \Dk \psi) \vert
		\cdot 
		\vert \nablaslash \Dkone \psi \vert
		\cdot
		\vert \nablaslash \Dktwo \psi \vert
		\lesssim
		\sqrt{\Fstarll_R + \Ezerostarll_{\frac{8R}{9}}}
		\Big(
		\Epstar_R
		+
		\intEpone_R' \
		\Big).
	\]
\end{proposition}

\begin{proof}

The separate terms are estimated individually, and the proof follows from the estimates \eqref{eq:nonlin2cornerterms}--\eqref{eq:nonlin2rpterm} below.

The notation of Section \ref{subsec:Sobolev} will again be used throughout.  As in the proof of Proposition \ref{prop:Minkowskinonlin1}, the cases $\tau_1 - \tau_0 \geq 1$ and $\tau_1 - \tau_0 < 1$ are typically considered separately.  Similarly, the region near the corner $\{u=\tau_0\} \cap \{r=R\}$ is treated separately.

\noindent \textbf{Estimate in the corner region $\{ v'\leq v(R,\tau_0 + R)\}$:}
	For any $\vert \tilde{\mathbf{k}} \vert \leq k+1$,
	\begin{equation} \label{eq:nonlin2cornerterms}
		\sum_{\vert \mathbf{k}_1 \vert + \vert \mathbf{k}_2 \vert \leq \vert \tilde{\mathbf{k}} \vert -1}
		\int_{\mathcal{R}_{\mathrm{Corner}}(\tau_0,\tau_1,v)} 
		\vert \mathfrak{D}^{\tilde{\mathbf{k}}} \psi \vert
		\cdot 
		\vert \nablaslash \Dkone \psi \vert
		\cdot
		\vert \nablaslash \Dktwo \psi \vert
		\lesssim
		\sqrt{\Ezerostarll_{\frac{8R}{9}}} \intEdelta_R'
		.
	\end{equation}
	Indeed, \eqref{eq:nonlin2cornerterms} follows exactly as in the proof of Proposition \ref{prop:Minkowskinonlin1}, using the Sobolev inequality \eqref{eq:pointwiseout}.

	The remainder of the proof concerns the region $\check{\mathcal{R}}_R(\tau_0,\tau_1,v)$.

\noindent \textbf{Estimate for $\partial_u \psi$ term:}
	First note that, for any $\vert \mathbf{k} \vert \leq k$,
	\begin{equation} \label{eq:nonlin2duterm}
		\sum_{\vert \mathbf{k}_1 \vert + \vert \mathbf{k}_2 \vert \leq \vert \mathbf{k} \vert}
		\int_{\check{\mathcal{R}}(\tau_0,\tau_1,v)} 
		\vert \partial_u \Dk \psi \vert
		\cdot 
		\vert \nablaslash \Dkone \psi \vert
		\cdot
		\vert \nablaslash \Dktwo \psi \vert
		\lesssim
		\sqrt{\Fstarll}
		\int^*\accentset{\scalebox{.6}{\mbox{\tiny $(\delta-1)$}}}{\underaccent{\scalebox{.8}{\mbox{\tiny $k$}}}{\mathcal{E}}}'.
	\end{equation}
	Indeed, for $\tau_1 - \tau_0 \geq 1$, assuming without loss of generality that $\vert \mathbf{k}_1 \vert \leq k-3$,
	\begin{align*}
		&
		\int_{\check{\mathcal{R}}(\tau_0,\tau_1,v)}
		\vert \partial_u \Dk \psi \vert
		\cdot 
		\vert \nablaslash \Dkone \psi \vert
		\cdot
		\vert \nablaslash \Dktwo \psi \vert
		\\
		&
		\qquad \qquad
		\lesssim
		\sup_{u,v,\theta} r \vert \nablaslash \Dkone \psi \vert
		\Big(
		\int_{\check{\mathcal{R}}(\tau_0,\tau_1,v)}
		r^{-1-\delta} \vert \partial_u \Dk \psi \vert^2
		\Big)^{\frac{1}{2}}
		\Big(
		\int_{\check{\mathcal{R}}(\tau_0,\tau_1,v)}
		r^{-1+\delta} \vert \nablaslash \Dktwo \psi \vert^2
		\Big)^{\frac{1}{2}}
		\lesssim
		\sqrt{\Fstarll}
		\int^*\accentset{\scalebox{.6}{\mbox{\tiny $(\delta-1)$}}}{\underaccent{\scalebox{.8}{\mbox{\tiny $k$}}}{\mathcal{E}}}',
	\end{align*}
	by the Sobolev inequality on the incoming cones \eqref{eq:pointwise1}, since $p \geq \delta$.  For $\tau_1 - \tau_0 < 1$,
	\[
		\int_{\check{\mathcal{R}}(\tau_0,\tau_1,v)}
		\vert \partial_u \Dk \psi \vert
		\cdot 
		\vert \nablaslash \Dkone \psi \vert
		\cdot
		\vert \nablaslash \Dktwo \psi \vert
		\lesssim
		\sqrt{\Fstarll}
		\sqrt{\Ezerostar}
		\sqrt{\Fstar}.
	\]

\noindent \textbf{Estimates for $\partial_v \psi$, $r^{-1}\nablaslash \psi$, and $r^{-1} \psi$ terms:}
	Next, for any $\vert \mathbf{k} \vert \leq k$,
	\begin{equation} \label{eq:nonlin2dvterm}
		\sum_{\vert \mathbf{k}_1 \vert + \vert \mathbf{k}_2 \vert \leq \vert \mathbf{k} \vert}
		\int_{\check{\mathcal{R}}(\tau_0,\tau_1,v)} 
		\big(
		\vert \partial_v \Dk \psi \vert
		+
		r^{-1} \vert \nablaslash \Dk \psi \vert
		+
		r^{-1} \vert \Dk \psi \vert
		\big)
		\cdot 
		\vert \nablaslash \Dkone \psi \vert
		\cdot
		\vert \nablaslash \Dktwo \psi \vert
		\lesssim
		\sqrt{\Fstarll}
		\Big(
		\Ezerostar
		+
		\int^*\accentset{\scalebox{.6}{\mbox{\tiny $(\delta-1)$}}}{\underaccent{\scalebox{.8}{\mbox{\tiny $k$}}}{\mathcal{E}}}' \
		\Big).
	\end{equation}
	Indeed, for $\tau_1-\tau_0 \geq 1$, assuming again without loss of generality that $\vert \mathbf{k}_1 \vert \leq k-3$,
	\begin{align*}
		&
		\int_{\check{\mathcal{R}}(\tau_0,\tau_1,v)} 
		\vert \partial_v \Dk \psi \vert
		\cdot 
		\vert \nablaslash \Dkone \psi \vert
		\cdot
		\vert \nablaslash \Dktwo \psi \vert
		\\
		&
		\qquad
		\lesssim
		\sup_{u,v,\theta} r \vert \nablaslash \Dkone \psi \vert
		\Big(
		\int_{\check{\mathcal{R}}(\tau_0,\tau_1,v)}
		r^{-1}
		\vert \partial_v \Dk \psi \vert^2
		\Big)^{\frac{1}{2}}
		\Big(
		\int_{\check{\mathcal{R}}(\tau_0,\tau_1,v)}
		r^{-1} \vert \nablaslash \Dktwo \psi \vert^2
		\Big)^{\frac{1}{2}}
		\lesssim
		\sqrt{\Fstarll}
		\int^*\accentset{\scalebox{.6}{\mbox{\tiny $(\delta-1)$}}}{\underaccent{\scalebox{.8}{\mbox{\tiny $k$}}}{\mathcal{E}}}',
	\end{align*}
	and similarly, for $\tau_1 - \tau_0 <1$,
	\begin{align*}
		\int_{\check{\mathcal{R}}(\tau_0,\tau_1,v)} 
		\vert \partial_v \Dk \psi \vert
		\cdot 
		\vert \nablaslash \Dkone \psi \vert
		\cdot
		\vert \nablaslash \Dktwo \psi \vert
		\lesssim
		\sqrt{\Fstarll} \ 
		\Ezerostar,
	\end{align*}
	using again the Sobolev inequality on the incoming cones \eqref{eq:pointwise1}.  Similarly for the $r^{-1} \psi$ and $r^{-1} \nablaslash \psi$ terms, using the fact that
	\[
		\int_{\check{\mathcal{R}}(\tau_0,\tau_1,v)}
		\big(
		r^{-1}
		\vert \nablaslash \Dk \psi \vert^2
		+
		r^{-3} \vert \Dk \psi \vert^2
		\big)
		\leq
		\int^*\accentset{\scalebox{.6}{\mbox{\tiny $(\delta-1)$}}}{\underaccent{\scalebox{.8}{\mbox{\tiny $k$}}}{\mathcal{E}}}'
		,
		\qquad
		\int_{v(R,u)}^v \int_{S^2} \left( \vert r^{-1} \nablaslash \psi \vert^2 + \vert r^{-1} \psi \vert^2 \right) r^2 d\omega dv'
		\leq
		\Ezerostar.
	\]

\noindent \textbf{Estimate for quartic term:}
	For any $\vert \mathbf{k} \vert \leq k$,
	\begin{equation} \label{eq:nonlin2quarticterm}
		\sum_{\vert \mathbf{k}_1 \vert + \vert \mathbf{k}_2 \vert \leq \vert \mathbf{k} \vert}
		\int_{\check{\mathcal{R}}(\tau_0,\tau_1,v)}
		\vert \nablaslash \Dkone \psi \vert^2
		\cdot
		\vert \nablaslash \Dktwo \psi \vert^2
		\lesssim
		\Fstarll
		\Big(
		\Ezerostar
		+
		\intEdelta \
		\Big)
		,
	\end{equation}
	which follows as an easy consequence of the Sobolev inequality \eqref{eq:pointwise1}.

\noindent \textbf{Estimate for $r^{p-1} \partial_v (r\psi)$ term:}
	Finally, for any $\vert \mathbf{k} \vert \leq k$,
	\begin{equation} \label{eq:nonlin2rpterm}
		\sum_{\vert \mathbf{k}_1 \vert + \vert \mathbf{k}_2 \vert \leq \vert \mathbf{k} \vert}
		\int_{\check{\mathcal{R}}(\tau_0,\tau_1,v)}
		r^{p-1} \vert \partial_v (r \Dk \psi) \vert
		\cdot 
		\vert \nablaslash \Dkone \psi \vert
		\cdot
		\vert \nablaslash \Dktwo \psi \vert
		\lesssim
		\sqrt{\Fstarll}
		\Big(
		\Epstar
		+
		\intEpone' \
		\Big).
	\end{equation}
	Assume again, without loss of generality, that $\vert \mathbf{k}_1 \vert \leq k-3$.  Then, using \eqref{eq:pointwise1}, for $\tau_1 - \tau_0 \geq 1$,
	\begin{align*}
		&
		\int_{\check{\mathcal{R}}(\tau_0,\tau_1,v)}
		r^{p-1} \vert \partial_v (r \Dk \psi) \vert
		\cdot 
		\vert \nablaslash \Dkone \psi \vert
		\cdot
		\vert \nablaslash \Dktwo \psi \vert
		\\
		&
		\qquad \qquad
		\lesssim
		\sup_{\check{\mathcal{R}}(\tau_0,\tau_1,v)} r \vert \nablaslash \Dkone \psi \vert
		\Big(
		\int_{\check{\mathcal{R}}(\tau_0,\tau_1,v)}
		r^{p-3}
		\vert \partial_v (r \Dk \psi) \vert^2
		\Big)^{\frac{1}{2}}
		\Big(
		\int_{\check{\mathcal{R}}(\tau_0,\tau_1,v)}
		r^{p-1}
		\vert \nablaslash \Dktwo \psi \vert^2
		\Big)^{\frac{1}{2}}
		\lesssim
		\sqrt{ \Fstarll}
		\intEpone'.
	\end{align*}
	Similarly, for $\tau_1 - \tau_0 <1$,
	\begin{align*}
		&
		\int_{\check{\mathcal{R}}(\tau_0,\tau_1,v)}
		r^{p-1} \vert \partial_v (r \Dk \psi) \vert
		\cdot 
		\vert \nablaslash \Dkone \psi \vert
		\cdot
		\vert \nablaslash \Dktwo \psi \vert
		\\
		&
		\qquad \qquad
		\lesssim
		(\tau_1 - \tau_0)^{-\frac{1}{2}}
		\sqrt{ \Fstarll}
		\Big(
		\int_{\check{\mathcal{R}}(\tau_0,\tau_1,v)}
		r^{p-2}
		\vert \partial_v (r \Dk \psi) \vert^2
		\Big)^{\frac{1}{2}}
		\Big(
		\int_{\check{\mathcal{R}}(\tau_0,\tau_1,v)}
		r^{p-2}
		\vert \nablaslash \Dktwo \psi \vert^2
		\Big)^{\frac{1}{2}}
		\lesssim
		\sqrt{ \Fstarll}
		\sqrt{\Epstar} \sqrt{\Ezerostar}.
	\end{align*}

\end{proof}

\subsection{More general nonlinearities and Kerr}
\label{section:Kerrcase}

For more general nonlinearities on Minkowski space of the form
\[
	F
	=
	N^{\mu \nu} (\psi,x) \partial_{\mu} \psi \partial_{\nu} \psi
	, 
\]
with $N$ satisfying a null condition of the form
\begin{equation} \label{assumponNzero}
	\sup_{\vert \xi \vert \leq 1, r\ge R} \sum_{\vert \textbf{k} \vert + s \leq k} \sum_{A,B=1,2}
	\Dk \partial_{\xi}^{s} 
	\big( 
	r N^{uu}
	+
	N^{uv}
	+
	N^{vv}
	+
	r N^{Au}
	+
	r N^{Av}
	+
	r^2 N^{AB}
	\big) (\xi,x)
	\lesssim
	D_k
	,
\end{equation}
where the $A$, $B$ indices refer to coordinates  $\vartheta,  \varphi$ (recall the $(u,v,\vartheta,\varphi)$ coordinate system of Section \ref{Minkowexam}), and $D_k$ are
arbitrary constants, the proof of 
Proposition~\ref{prop:appendixbmain} follows exactly as above, where the bound~\eqref{basicbootstrapinnullcond} is now also used
to obtain pointwise bounds on $\psi$ through an easy weighted Sobolev estimate which allow as to invoke~\eqref{assumponNzero}.

\begin{remark}
Note that, besides the nonlinearities of the classical null condition~\cite{KlNull}, 
 our class~\eqref{assumponNzero} includes for instance also nonlinearities which for large $r$ take 
the form $F= (\sin r) \partial_u\psi \partial_v\psi$. It does not, however, include  even more
general examples like $F=(\sin x) \partial_u\psi \partial_v\psi$ considered recently in~\cite{andersonzbarsky}, due to the presence of
the weighted vector fields $\Omega_i$ in our set of
commutation vector fields $\Dk$.
\end{remark}

More generally, if $g_0$ is a metric with asymptotics suitably close to those of Minkowski space, with extra terms suitably small, then the above proof again applies.

We will consider explicitly the case that $g_0$ is the Kerr metric.
Define double null $(u,v,\theta^1, \theta^2)$ coordinates on Kerr, when $0 < \vert a \vert < M$, in terms of the Boyer--Lindquist coordinates $(t,r,\vartheta,\phi)$ of
Appendix~\ref{thekerrmanifold}, by 
\[
	u = t-r_*, \qquad v = t+r_*,
\]
with $\theta^1$ defined implicitly by the relation
\begin{equation} \label{eq:theta1}
	F(\theta^1,r, \vartheta) = 0,
	\quad
	\text{where}
	\quad
	F(\theta^1,r, \vartheta) = \int_{\theta^1}^{\vartheta} \frac{1}{a \sqrt{\sin^2 \theta^1 - \sin^2 \theta'}} d\theta'
	+
	\int_r^{\infty} \frac{1}{\sqrt{((r')^2+a^2)^2 - a^2 \sin^2 \theta^1 \Delta'}} dr',
\end{equation}
and $\theta^2$ defined by
\[
	\theta^2 = \phi + h(r_*,\theta^1),
	\quad
	\text{where $h$ satisfies}
	\quad
	\partial_{r_*} h(r_*,\theta^1) = \frac{2Mar}{\Sigma \mathrm{R}^2},
	\quad
	\lim_{r_*\to \infty} h(r_*,\theta^1) = 0.
\]
See for instance~\cite{Pretorius, DafLuk1}.
Here
\begin{equation} \label{eq:rstar}
	r_*(r,\vartheta)
	=
	\int\frac{r^2+a^2}{\Delta} dr
	+
	\int_r^{\infty} \frac{(r')^2+a^2 - \sqrt{((r')^2+a^2)^2 - a^2 \sin^2 \theta^1 \Delta'}}{\Delta'} dr'
	+
	\int^{\vartheta}_{\theta^1} a \sqrt{\sin^2 \theta^1 - \sin^2 \theta'} d\theta',
\end{equation}
(note that the expression \eqref{eq:rstar} is independent of $\theta^1$ in view of \eqref{eq:theta1}), where $\int\frac{r^2+a^2}{\Delta} dr$ is a function satisfying
\[
	\partial_r \int\frac{r^2+a^2}{\Delta} dr = \frac{r^2+a^2}{\Delta},
\]
and
\[
	\Delta = r^2+a^2-2Mr,
	\qquad
	\Delta' = (r')^2+a^2-2Mr',
	\qquad
	\Sigma = r^2 + a^2\cos^2 \vartheta,
	\qquad
	\mathrm{R}^2 = r^2 + a^2 + \frac{2Ma^2 r \sin^2 \vartheta}{\Sigma}.
\]
Note in particular that $r_*$ is distinct from $r^{\star}$ defined in Section \ref{thekerrmanifold}.

In these double null coordinates the Kerr metric takes the form
\[
	g_{a,M}
	=
	- \Omega^2 du dv + \slashed{g}_{AB}(d\theta^A - b^Adv)(d\theta^B - b^B dv),
\]
where
\[
	\Omega^2 = \frac{\Delta}{\mathrm{R}^2},
	\qquad
	b^1 = 0, \qquad b^2 = \frac{4Mar}{\Sigma \mathrm{R}^2},
\]
and
\[
	\slashed{g}_{11}
	=
	\frac{a^2 (\partial_{\theta^1}F)^2 (\sin^2 \theta^1 - \sin^2 \vartheta)((r^2+a^2)^2 - a^2 \sin^2 \theta^1 \Delta)}{\mathrm{R}^2}
	+
	(\partial_{\theta^1} h)^2 \mathrm{R}^2 \sin^2 \vartheta,
\]
\[
	\slashed{g}_{22} = \mathrm{R}^2 \sin^2 \vartheta,
	\qquad
	\slashed{g}_{12} = \slashed{g}_{21} = - \mathrm{R}^2 \sin^2 \vartheta \partial_{\theta^1} h.
\]
In particular, one has
\[
	L = \partial_v + b^A \partial_{\theta^A},
	\qquad
	\underline{L} = \partial_u.
\]
Moreover $T = \partial_u + \partial_v$ and $\Omega_i = {\Omega_i}(\theta^1,\theta^2)^A \partial_{\theta^A}$ for $i=1,2,3$, for functions ${\Omega_i}(\theta^1,\theta^2)^A$, so that, in particular,
\[
	[\partial_u,T] = [\partial_u,\Omega_i] = [\partial_v,T] = [\partial_v,\Omega_i] = 0.
\]
The function $b^2 = b^2(u,v,\theta^1,\theta^2)$ satisfies, for any $k \geq 0$,
\[
	\sum_{k_1+k_2+\vert \textbf{k}_3 \vert \leq k}
	\vert \partial_u^{k_1} (r\partial_v)^{k_2} \Omega^{\textbf{k}_3} b^2 \vert
	\leq
	\frac{C_k}{r^3}
	,
\]
for large $r$, and hence the commutation relations
\begin{align} \label{eq:DkdvdvDkKerr1}
	\vert \Dk L \psi \vert
	&
	\lesssim
	\sum_{\vert \widetilde{\bf k} \vert \leq \vert \bf k \vert}
	\vert L\mathfrak{D}^{\widetilde{\bf k}} \psi \vert
	+
	\frac{1}{r^3} \sum_{i=1}^3 \sum_{\vert \widetilde{\bf k} \vert \leq \vert \bf k \vert - 1}
	\vert \Omega_i \mathfrak{D}^{\widetilde{\bf k}} \psi \vert
	,
	\\
	\vert \Dk \underline{L} \psi \vert
	&
	\lesssim
	\sum_{\vert \widetilde{\bf k} \vert \leq \vert \bf k \vert}
	\vert \underline{L} \mathfrak{D}^{\widetilde{\bf k}} \psi \vert
	+
	\frac{1}{r^3} \sum_{i=1}^3 \sum_{\vert \widetilde{\bf k} \vert \leq \vert \bf k \vert - 1}
	\vert \Omega_i \mathfrak{D}^{\widetilde{\bf k}} \psi \vert
	,
	\label{eq:DkdvdvDkKerr2}
\end{align}
hold for large $r$.  Moreover \eqref{eq:nablaMinkproperty} remains true.  The proof of Proposition \ref{prop:appendixbmain} in this case then follows just as in
 the  case where $g_0$ was the Minkowski metric, using now \eqref{eq:DkdvdvDkKerr1} and \eqref{eq:DkdvdvDkKerr2} in place of \eqref{eq:DkdvdvDk}.

We note finally that the condition~\eqref{assumponNzero} applied to Kerr indeed includes in particular the nonlinearities considered in~\cite{MR3082240}.

\section{The inhomogeneous estimate~\eqref{inhomogeneous}  on Kerr
in the full subextremal case~$|a|<M$}
\label{theinhomo}

The main theorem of~\cite{partiii} only states~\eqref{inhomogeneous}  for the homogeneous case, i.e.~the case $F=0$. 
In this section, we explicitly address the issue of 
the inhomogeneous estimate~\eqref{inhomogeneous}. We have the following:

\begin{theorem}
\label{trueforKerr}
The inhomogeneous estimate~\eqref{inhomogeneous}  holds on Kerr
in the full subextremal case $|a|<M$.
\end{theorem}

Though the inhomogeneous estimate~\eqref{inhomogeneous} for general $F$, precisely as stated, 
indeed can be shown to hold in the full subextremal case $|a|<M$,
it is a little bit delicate
to produce the physical space expression corresponding to the middle term
on the  right hand side of~\eqref{inhomogeneous}, except in the case $|a|\ll M$.
(This difficulty is entirely due to the presence of the term $\Ezero(\tau)$  on the left hand side.)
The weaker version of the estimate in Remark~\ref{basicblackboxestimatesec}, on the other hand, 
which is all that is actually used here,
can essentially be immediately read off  from the proof of the main
result of~\cite{partiii}. 
Since we have no real use for the stronger statement, we prefer here to
give the details of how to directly obtain 
the weaker statement (though in the case $|a|\ll M$, we note that
our argument indeed reviews the stronger statement).

\begin{proof}[Proof (of the weaker version of Remark~\ref{weakerassumption})]
 We will obtain the estimate in three steps.
The reader should refer to~\cite{partiii} to follow along.

\paragraph{Step 1.}
We first obtain~\eqref{inhomogeneous} (in its original form) for general $F$,
but without the 
future boundary terms on the left hand side, and where
the regions are all restricted to $r\ge r_+$.

For this,  note that one can assume without loss of generality that $F$ is compactly supported
in spacetime, and thus 
it is clearly sufficient to prove Proposition~9.1.1 of~\cite{partiii} now for solutions of~\eqref{linearhomogeq}, with the extra inhomogeneous
terms of~\eqref{inhomogeneous} now on the right hand side. 

We note that the cutoffs used in the proof of Proposition 9.1.1~already gave rise 
to inhomogeneous
terms 
\begin{equation}
\label{freqlocterm}
\mathcal{T}:= \int_{-\infty}^\infty \sum_{m\ell}\left(\int_{-\infty}^\infty H\cdot (f, h, y, \chi)\cdot (u, u')\right) d\omega \, dr^*
\end{equation}
(see e.g.~(165) of~\cite{partiii}). 
An inhomogeneity $F$ 
on the right hand side of~\eqref{linearhomogeq} contributes
an additional term of the form~\eqref{freqlocterm} in the proof of this proposition,
where $H$ is replaced by the Carter separated coefficients of $F$ according to formula (43)
of~\cite{partiii}. Let us denote this term as~$\mathcal{T}_F$.
One easily sees, however,  that, in view of the inclusion of our inhomogeneity $F$, 
the problem can be reduced to the case where one has trivial initial data,
in which case one may work directly with $\psi$
without  applying a cutoff. Thus we may assume  now  that we \emph{only} have $\mathcal{T}_F$ on the right hand side, and $H$ below will denote the Carter separated coefficients of $F$. We must now 
 essentially repeat the steps of the proof of the proposition in order
to produce the desired right hand side.

For this, we first immediately partition  $\mathcal{T}_F:= \mathcal{T}_{F}^1+
\mathcal{T}_{F}^2+
\mathcal{T}_{F}^3$, where the summands correspond to the integral of~\eqref{freqlocterm} over the region $R_-\le r\le R_+$ (i.e.~$R_-^*\le r^*\le R_+^*$, etc.),
$\{ r\le R_-\} \cup \{R_+\le r\le R\}$ and $r\ge R$, respectively, 
 where $R_\pm$, $R_{\infty}$ are as in~\cite{partiii}
and we require
$R\ge R_\infty$.

Considering $ \mathcal{T}^1_{F}$, we further partition
$ \mathcal{T}^1_{F}= \mathcal{T}^1_{F,  \omega} +
\mathcal{T}^1_{F, \slashed\omega}$
where $\mathcal{T}^1_{F,  \omega} $ is the contribution of the
currents in $H\cdot (f, h, y, \chi)\cdot (u, u')$ that multiply $\omega$,
and $\mathcal{T}^1_{F, \slashed\omega}$ denotes the rest. 
We note (see formula (102) of~\cite{partiii}) that the coefficients of the currents in the sum defining 
$\mathcal{T}^1_{F, \omega}$ 
are then frequency independent for  all ``trapped'' frequencies 
$(\omega, m, \ell)\in \mathcal{G}_\natural$ (the expression is simply $E\omega {\rm Im} (H\bar u)$).
Thus, by adding a suitable compensating term $\mathcal{T}^1_{\rm comp}$ 
supported
only in the untrapped frequencies, then 
 $\mathcal{T}^1_{F, \omega}+\mathcal{T}^1_{\rm comp}$
is a sum and integral over
the precise expression $E\omega {\rm Im} (H\bar u)$,
and thus this  integral and 
sum may be rewritten via the Plancherel formulae on page 820 of~\cite{partiii}
as a spacetime integral precisely
like the middle term of~\eqref{inhomogeneous}, restricted to $r\le R$, with $V_0=\partial_t$.

For the remaining terms $\mathcal{T}^1_{F, \slashed\omega}-\mathcal{T}^1_{\rm comp}$, 
note that these are supported entirely in the \emph{non-trapped} frequencies, i.e.~in the complement frequency set $\mathcal{G}_\natural^c$.
Recall now that in estimate (165), we in fact manifestly control a stronger left hand side than
what is written, namely we may substitute $\zeta$ with $(1-\frac{r_{\rm trap}}r)^2$
(which was the original form of the expression before (163) of~\cite{partiii} was invoked). 
In particular, since $r_{\rm trap}=0$  for all non-trapped frequencies, 
we control all derivatives without degeneration
for frequencies in $\mathcal{G}_{\natural}^c$.
One may 
apply Cauchy--Schwarz to the expression $H\cdot (f, h, y, \chi)\cdot (u, u')$
and to the integrand summands of $\mathcal{T}^1_{\rm comp}$, absorbing the
resulting terms proportional to $|u|^2$, together  with any frequency coefficients $\omega^2$ and $\Lambda$, and proportional to $|u'|^2$, into the bulk controlled
in view of our nondegenerate bulk, at the expense of an additional
term $|H|^2$ on the right hand side. Upon application of  the Plancherel formulae of
page 820, this  
produces a bulk term $\int F^2$ of the form already  present on the right hand side of~\eqref{inhomogeneous}.   The term $\mathcal{T}^2_{F}$ may be treated exactly in the
same way, as here the controlled bulk term is nondegenerate for all frequencies, since the integral
is supported away from trapping.

This leaves the term $\mathcal{T}^3_F$.
One notes that for $r\ge R$,  the coefficients of the currents appearing in~\eqref{freqlocterm}  
are frequency independent
\emph{up to multipliers which decay exponentially in $r$}. Again, by adding and subtracting a compensating term as above, one may thus estimate
the additional term here by the middle
term on the right hand side of~\eqref{replacedterm} (now restricted to $r\ge R$), 
again with a suitable vector field $V_0$, after applying 
Plancherel, up to an additional term exponentially decaying term,
which may be estimated by
Cauchy--Schwarz, absorbing the term containing derivatives of $\psi$ into the bulk on the left hand side just as above, and producing
an extra bulk term $\int F^2$ of the form already  present on the right hand side of~\eqref{inhomogeneous}. 

Finally, we note that to complete the proof
of Proposition~9.1.1 in the inhomogeneous case, it remains to apply the analogue of
Proposition~9.7.1 (this is where the quantitative version of mode stability~\cite{SRT} is appealed to), 
which  also however produces an additional inhomogeneous term
when applied to~\eqref{linearhomogeq}.
Examining the original~\cite{SRT}, in particular the proof of Lemma 3.3 of that paper, one sees that
this extra term may be bounded by the $\int F^2$  term on the right hand side of~\eqref{inhomogeneous}. This completes the proof of Step 1.

\paragraph{Step 2.}
We now apply a sufficiently small multiple of the red-shift multiplier and of a cutoff version of the $J_\mu^{T}$-current to obtain
the terms $\sup_{v: \tau\le \tau(v)}\Fzero(v,\tau_0,\tau)
+ \Ezero_{\mathcal{S}}(\tau)$ on the left hand side of~\eqref{inhomogeneous}, and the part of the spacetime integral supported in $r_0\le r\le r_+$, exploiting that one may always  absorb the resulting error terms in the bulk already controlled in Step 1. We add these identities to the estimate
obtained above.
As these are physical space identities, the inhomogeneity $F$ manifestly generates a term of the form of the middle term in~\eqref{inhomogeneous} which may be combined
with the previous such terms to define a new vector field $V_0$.
One obtains in this way also a bound on the restriction of the integral defining $\Ezero(\tau)$ to $\{r\le R_-\}\cup\{r\ge R_+\}$.   

Note that in the case of very slow rotation $|a|\ll M$, one may chose the
support of the gradient of the cutoff applied to  $J_\mu^T$ in the identity above to in fact be in the region $r\le R_-$,
whereas  $J_\mu^T$ is moreover coercive in the region $r\ge R_-$. Thus, the above argument in fact
yields control of the full $\Ezero(\tau)$. This would then complete the proof. 
In this case, note that we thus obtain the estimate~\eqref{inhomogeneous} in its original from.

\paragraph{Step 3.}
In the general subextremal case, however, $J^T_\mu$ is not everywhere coercive in the region 
$R_-\le r\le R_+$, and 
it remains to control the contribution of the finite region $R_-\le r\le R_+$ to the energy integral of
the missing term~$\Ezero(\tau)$ from the left hand side of~\eqref{inhomogeneous}.

 To obtain
 the missing part of the energy flux we must  thus repeat the proof of Proposition 13.1 of~\cite{partiii},
allowing now an inhomogeneous term $F$.  Recall that in this proof one considered a
cutoff version $\tilde\psi$, decomposed $\tilde\psi=\tilde\psi_1+\ldots +\tilde\psi_N$ orthogonally
using Carter's separation, for some large $N$, and 
applying a distinct current $J^{V_i}_\mu$ to each summand $\psi_i$. 
 
 Due to the presence of cutoffs in the proof, one can again read off
the additional terms arising from the new inhomogeneity.
Since  outside
the region $R_-\le r\le R_+$,,
all errors from    can be absorbed as in the previous steps,  while the  cutoffs vanish inside the region $R_-\le r\le R_+$, all inhomogeneous
terms in $R_-\le r\le R_+$ now arise from $F$, and we may use the notation $\tilde{F}_i$
of the proof of Proposition 13.1 to denote precisely these new terms. The resulting inhomogeneous
terms that must be controlled  are then $ \sum_i \int_{R_-\le r\le R_+}V_i\tilde\psi_i \cdot \tilde{F}_i$.
By orthogonality, we may indeed manifestly
bound this term by the expression in the first term of~\eqref{replacedterm}.

Since  all terms of the form of the middle term on the right hand side of~\eqref{inhomogeneous} 
supported in $r\le R$ (which we generated in the previous two steps) can
of course manifestly be bounded by the  first term of~\eqref{replacedterm},   
we obtain finally that~\eqref{inhomogeneous} indeed holds with~\eqref{replacedterm}
replacing the middle term on the right hand side.
\end{proof}

We note that obtaining~\eqref{inhomogeneous} in the $|a|<M$ case
with its middle term in its original form
requires a slightly more delicate
decomposition of $\tilde\psi$ than the one used in Proposition 13.1 of~\cite{partiii}. 
Again, since this is not used, we spare the reader the 
details of this argument.

\bibliographystyle{DHRalpha}
{\footnotesize
\bibliography{quasilinearrefs}}

\end{document}